\documentclass[10pt,reqno]{amsart}

\usepackage[margin=1in]{geometry}   
\usepackage[utf8]{inputenc}  
\usepackage[T1]{fontenc}
\usepackage{upgreek}

\usepackage{amsmath,amsthm,amsfonts,amssymb,amscd}
\usepackage{mathtools}
\usepackage[mathscr]{euscript}
\usepackage{bm}
\usepackage{bbm}
\usepackage{stmaryrd}

\usepackage[dvipsnames]{xcolor}
\usepackage{soul}

\usepackage[square,numbers]{natbib}

\usepackage{etoolbox,xstring}
\makeatletter
\patchcmd{\NAT@test}{\else\NAT@nm}{\else\NAT@nmfmt{\NAT@nm}}{}{}
\let\NAT@up\scshape
\makeatother

\makeatletter
\renewcommand{\NAT@nmfmt}{\expandafter\aliNAT@nmfmt\expandafter}
\newcommand\aliNAT@nmfmt[1]{{%
  \noexpandarg
  \def~{}%
  \edef\temp#1\edef\temp{\detokenize\expandafter{\temp}}%
  \begingroup\edef\x{\endgroup
    \noexpand\StrSubstitute{\temp}{\detokenize{etal}}}\x
    {\textnormal{et\nobreakspace al}}[\tempa]%
  \textsc{\tempa}}}
\makeatother

\usepackage{graphicx}
\graphicspath{{graphics/}}

\usepackage{wrapfig}
\usepackage{afterpage}
\usepackage{caption}

\usepackage{subcaption}
\usepackage{tikz-cd}
\usepackage{tikz}
\usetikzlibrary{arrows.meta,shapes.geometric,positioning,fit,calc}
\usetikzlibrary{decorations.pathreplacing,angles,quotes,patterns}
\usepackage{pgfplots}
\pgfplotsset{compat=1.18}
\usepgfplotslibrary{fillbetween}

\usepackage{array}
\usepackage{multirow}
\newcolumntype{P}[1]{>{\centering\arraybackslash}p{#1}}		
\newcolumntype{M}[1]{>{\centering\arraybackslash}m{#1}}
\newcolumntype{L}[1]{>{\raggedright\arraybackslash}m{#1}}
\usepackage{multirow}
\usepackage{booktabs}

\usepackage{hyperref}
\hypersetup{colorlinks=True, linkcolor=blue, urlcolor=blue, citecolor=red}
\usepackage{cleveref}

\usepackage{float}
\usepackage{enumitem}
\usepackage{algorithm}
\usepackage{algpseudocode}
    \makeatletter
    \newenvironment{breakablealgorithm}
      {%
       \par
       \vspace{1em}
       \refstepcounter{algorithm}%
       \hrule height.8pt depth0pt \kern2pt
       \renewcommand{\caption}[2][\relax]{%
         {\raggedright\textbf{Algorithm~\thealgorithm: ##2}\par}%
         \kern2pt\hrule\kern2pt
       }%
      }
      {%
       \kern2pt\hrule\relax
       \vspace{1em}
       \par
      }
    \makeatother
\usepackage{verbatim}
\usepackage{cases}
\usepackage{needspace}

\usepackage{xfrac}
\usepackage{siunitx}
\numberwithin{equation}{section}

\newtheorem{theorem}{Theorem}[section]

\newtheorem{proposition}[theorem]{Proposition}

\newtheorem{example}[theorem]{Example}
\newtheorem{remark}[theorem]{Remark}

\def\tmax{{t}_{\mathsf{max}}}

\DeclareMathOperator*{\argmin}{arg\,min}

\def\p{\partial}

\def\Omegaref{{\Omega}_{\mathrm{ref}}}

\newcommand{\LGR}{\textsf{LGR}}
\newcommand{\mts}{\textsf{MTS}}
\def\lgr{{\mathsf{lgr}}}
\def\eul{{\mathsf{eul}}}
\def\Nlgr{N_{\mathsf{lgr}}}
\def\Neul{N_{\mathsf{eul}}}
\def\Gammalgr{\mathcal{L}}
\def\Gammaeul{\mathcal{E}}
\def\zlgr{\ell}
\def\zeul{e}

\def\lftle{\lambda^{\mathsf{tan}}}
\def\sftle{\sigma^{\mathsf{tan}}}

\def\Vlgr{V}
\def\Veul{W}
\def\Adj{\mathcal{A}}
\def\Event{\mathcal{T}}
\def\Nevent{N_{\mathsf{event}}}

\def\Defect{\mathcal{D}_h}
\def\Surg{\mathscr{S}}
\def\Plgr{\mathscr{P}}
\def\Peul{\mathscr{Q}}
\def\plgr{{p}}
\def\peul{{q}}
\def\Npiece{\mathscr{N}}
\def\Mpiece{\mathscr{M}}
\def\Vpiece{\mathscr{V}}
\def\Wpiece{\mathscr{W}}
\def\vpiece{{v}}
\def\wpiece{{w}}

\def\dref{d_{\mathsf{ref}}}
\def\dcrs{d_{\mathsf{crs}}}
\def\Nmin{N_{\mathsf{min}}}
\def\Lmin{L_{\mathsf{min}}}
\def\Amin{A_{\mathsf{min}}}
\def\wfit{w_{\mathsf{fit}}}
\def\wsm{w_{\mathsf{sm}}}
\def\wtan{w_{\mathsf{tan}}}

\def\Sigmam{\mathbf{\Sigma}\;\!}

\title[Microscale topological surgery for tracking filament breakup]
{Interface tracking with microscale topological surgery for two-dimensional 
filament breakup}
\author{Raaghav Ramani}
\address{Center for Nonlinear Studies, Los Alamos National Laboratory, Los Alamos, NM 87545}
\email{\href{raaghav.ramani@outlook.com}{raaghav.ramani@outlook.com}}

\begin{document}

\begin{abstract}
We design and implement a \emph{Microscale Topological Surgery}
(\mts) algorithm to detect and enforce topological transitions in
two-dimensional tracked interfaces.
The method combines classical Lagrangian tracking with an intermittent
topological processor that:
(\emph{i}) constructs Eulerian snapshots from which an 
interface family with microscale-resolved topology is extracted, 
(\emph{ii}) infers adjacency topology between dual Lagrangian and
Eulerian interface families, and
(\emph{iii}) performs interface surgery to stitch the two
families together across microscale defect regions.
A novel long-time nonlinear alternating-shear flow is introduced, in
which repeated stretching and folding generate rich multiscale
interface dynamics with filamentation at microscales.
Using the \mts\ algorithm and \emph{a posteriori}
geometric and material diagnostics, we compute and
visualize microscale filament-breakup dynamics.
Error analysis and scaling studies demonstrate
second-order geometric convergence and optimal computational scaling of
the \mts\ algorithm, with topology-processing costs comparable to those
of the underlying Lagrangian evolution. 
Ensemble simulations generated by pseudo-random perturbations of the flow 
further reveal coherent droplet size distributions and statistically robust filament-breakup dynamics.
\end{abstract}

\maketitle

\setcounter{tocdepth}{1}
{\small
\tableofcontents}

\section{Introduction}
\label{sec:intro}

We consider the evolution of a collection of two-dimensional
\emph{interfaces} $\Gamma_\alpha(t)$ under the flow of a given smooth
velocity field $u(x,t)$. Our primary interest is in multiscale
interface dynamics characterized by large-scale deformations at length
scale $L>0$ together with fine-scale structure at length scale $h>0$,
satisfying
\begin{equation}\label{h/L}
0 < \frac{h}{L} \ll 1.
\end{equation}
Such scale separation commonly arises through \emph{filamentation}, whereby
stretching and folding of the interface generate progressively thinner
interfacial structures while simultaneously maintaining large-scale
geometric features. In the specific case that
\[
\frac{h}{L}=\mathcal{O}(10^{-3}),
\]
we refer to $h$ as a \emph{microscale-resolving} length scale.
Numerical approaches for computing the interfacial evolution can be broadly
categorized into interface \emph{capturing} and interface \emph{tracking}
methods.

Eulerian interface-capturing methods represent the evolving interface
implicitly through one or more auxiliary fields defined on a fixed
computational mesh. Prominent examples include level-set (LS) methods
\cite{OsSe1988,PeMeOsZhKa1999,OlKr2005,DeMoPi2008},
volume-of-fluid (VOF) methods \cite{HiNi1981,RiKo1998}, and
moment-of-fluid (MOF) methods
\cite{DySh2008,ZhLi2008,AhSh2009,JeSuSh2015,HePhXi2023,HeLiPhXi2024}.
In these approaches, topological transitions are accommodated through a
finite-scale \emph{geometric regularization} associated with the
implicit interface representation. An alternative class of methods
replaces the sharp interface by a diffuse transition layer governed by
an auxiliary order parameter, thereby introducing a \emph{physical
regularization} of topological singularities through interfacial mixing
and diffusion \cite{LoTr1998,MiJaDo2017}. 
Both approaches naturally accommodate topological transitions without
requiring specialized reconstruction procedures, and their robustness
and flexibility have made them standard tools in multiphase-flow simulations.
The tradeoff is that the interface is
transported indirectly through advection on the background computational mesh, causing
numerical errors to accumulate over time. As a result, resolving fine-scale filamentary structures typically
requires substantial volumetric mesh refinement, thereby increasing the
number of spatial degrees of freedom and imposing more restrictive CFL
stability constraints, both of which contribute to the overall computational cost.

Lagrangian interface-tracking methods, by contrast, represent the 
interface explicitly by parametrized codimension--1 hypersurfaces moving with 
the flow \cite{ChGlMcPlYa1986,UnTr1992,PoZa1999}, thereby providing 
sharp interface descriptions and direct access to geometric quantities 
such as normals, curvature, and surface differential operators. 
Moreover, adaptive refinement may be
concentrated directly on the interface, making Lagrangian approaches
particularly efficient for problems involving highly localized
interfacial structure and large scale separations. 
For example, in the rotating-vortex benchmark considered below, we show
that Lagrangian tracking is approximately 1000 times more accurate 
than a baseline Eulerian capturing method, while running 100 times faster (\Cref{fig:LGR-vs-LS_diagnostics}).
The principal limitation of Lagrangian tracking methods is the
topological rigidity of the underlying interface representation.
Because interface connectivity is built directly into the underlying
curve or mesh data structures, the classical Lagrangian evolution alone cannot
alter the topology of the interface. Consequently, topological
transitions such as pinch-off, breakup, and coalescence
require specialized detection and reconstruction procedures.

\subsection{The core \mts\ methodology}

The contrasting characteristics of Lagrangian interface-tracking and
Eulerian interface-capturing naturally motivate hybrid approaches that
seek to combine the geometric accuracy and computational efficiency of
the former with the topological flexibility of the latter. Motivated by
the sharp-interface continuation perspective advocated by
\citet{LoTr1998}\footnote{
\citet{LoTr1998} argued that \emph{``...it is natural to conclude that the ideal
numerical approach should be based on the sharp interface model and then
use `jump conditions' to transit the sharp interface through a topology
change.''}
}, we propose a hybrid Eulerian--Lagrangian
\emph{Microscale Topological Surgery} (\mts) algorithm.  
In our numerical framework, the sharp interface model is realized
by classical Lagrangian tracking, while transitions through 
topology changes are handled by intermittent \mts\ processing.
The basic computational pipeline is:
\begin{equation}
\begin{tikzcd}[column sep=large, row sep=large]
& \text{Lagrangian tracking} \arrow[dl] & \\
\text{Eulerian snapshot} \arrow[r]
& \text{Adjacency topology} \arrow[r]
& \text{Interface surgery} \arrow[ul]
\end{tikzcd}
\tag{\(\mathsf{MTS}\)}
\end{equation}
The lower row of the diagram represents the core \mts\ procedure. 
Throughout the paper, we will use the terminology \mts\ interchangeably for 
both this core procedure and the complete hybrid tracking algorithm obtained by 
coupling with the underlying Lagrangian evolution.

\begin{figure}[p]
\centering
\footnotesize

\begin{tikzpicture}[
  >={Stealth[length=2.5mm]},
  line/.style={->, ultra thick},
  nodefont/.style={font=\sffamily},
  block/.style={
    rectangle, rounded corners=2pt, draw=black, align=left,
    inner xsep=3mm, inner ysep=2mm, nodefont, text width=58mm
  },
  magblock/.style={
    block, fill=magenta!20, rounded corners=5pt
  },
  orangeblock/.style={
    block, fill=orange!25
  },
  decision/.style={
      rectangle, rounded corners=12pt,
      draw=black, thick, align=center,
      inner xsep=4mm, inner ysep=2mm,
      nodefont, fill=green!25,
      text width=34mm
    },
  lbl/.style={font=\sffamily\scriptsize, inner sep=1pt},
  node distance=8mm and 16mm
]

\node[magblock, text width=54mm] (init)
{\textbf{Initialize interface family:} $t = 0$
\begin{itemize}[leftmargin=*, itemsep=2pt, topsep=1pt, parsep=0pt]
\item given initial interfaces $\{\Gamma_\alpha(0)\}_{\alpha=1}^{N(0)}$
\item construct parametrizations $\gamma_\alpha(s,0)$
\end{itemize}};

\node[orangeblock, text width=48mm, below=of init] (whilehdr)
{\textbf{Lagrangian tracking (\Cref{sec:lgr})}
while $t < \tmax$
\begin{itemize}[leftmargin=*, itemsep=2pt, topsep=1pt, parsep=0pt]
  \item evolve $\partial_t \gamma_\alpha = u \circ \gamma_\alpha$
  \item adaptively refine interfaces
\end{itemize}};

\node[decision, text width=24mm, below=14mm of whilehdr] (dec)
{\textbf{MTS time?}\\ $t=T_m$};

\coordinate (colX) at ($(dec)+(80mm,36.5mm)$);

\node[orangeblock, text width=48mm] (capture) at (colX)
{\textbf{Eulerian snapshot (\Cref{sec:extraction})}
\begin{itemize}[leftmargin=*, itemsep=2pt, topsep=1pt, parsep=0pt]
  \item define pre-processed Lagrangian interface 
  family $\{ \Gammalgr_\alpha, \Nlgr, \zlgr_\alpha \}$
  \item construct SDF $\phi$
  \item extract Eulerian interface family $\{ \Gammaeul_\beta, \Neul, \zeul_\beta \}$
\end{itemize}};

\node[orangeblock, text width=48mm, below=10mm of capture] (adj)
{\textbf{Adjacency topology (\Cref{sec:adjacency})}
\begin{itemize}[leftmargin=*, itemsep=2pt, topsep=1pt, parsep=0pt]
  \item construct adjacency matrix $\Adj$
  between $\{\Gammalgr_\alpha\}$ and $\{\Gammaeul_\beta\}$
  \item classify event family $\{ \Event_1,\ldots,\Event_{\Nevent} \}$
\end{itemize}};

\node[decision, text width=24mm, below=10mm of adj] (topodec)
{\textbf{Interface surgery?} \\ $\exists \, \Event_q \in $ \emph{branching}};

\node[orangeblock, text width=48mm, below=12mm of topodec] (surgery)
{\textbf{Interface surgery (\Cref{sec:surgery})}
\begin{itemize}[leftmargin=*, itemsep=2pt, topsep=1pt, parsep=0pt]
  \item compute topological defect measure $\Defect$
  \item slice surgical pieces $\Plgr$, $\Peul$
  \item surgically stitch $\Plgr$ and $\Peul$ to reconstruct $\{\Gamma_\alpha, N(T_m), z_\alpha\}$
\end{itemize}};

\node[orangeblock, text width=24mm, align=left, below=14mm of dec]
  (cont) {\textbf{Continue}
  \begin{itemize}[leftmargin=*, itemsep=2pt, topsep=1pt, parsep=0pt]
  \item $t \leftarrow t+\Delta t$
  \end{itemize}};

\draw[line] (init) -- (whilehdr);
\draw[line] (whilehdr) -- (dec);

\draw[line] (dec.south) -- (cont.north)
  node[lbl, midway, right]{\footnotesize{No}};

\coordinate (decRight) at ($(dec.east)+(25mm,0)$);

\draw[line]
    (dec.east)
    -- node[lbl, pos=0.5, above]{\footnotesize{Yes}} (decRight)
    -- (decRight |- capture.west)
    -- (capture.west);

\draw[line] (capture) -- (adj);
\draw[line] (adj) -- (topodec);

\draw[line] (topodec.west) -- (cont.east)
  node[lbl, midway, above]{\footnotesize{No}};
  
\draw[line] (topodec) -- (surgery)
  node[lbl, midway, right]{\footnotesize{Yes}};

\coordinate (surgToCont) at (surgery.west -| cont.center);

\draw[line]
  (surgery.west)
  -- (surgToCont)
  -- (cont.south);

\coordinate (contLeft) at ($(cont.west)+(-25mm,0)$);
\coordinate (contUp)   at ($(contLeft |- whilehdr.west)$);

\draw[line] (cont.west) -- (contLeft) -- (contUp) -- (whilehdr.west);

\end{tikzpicture}

\caption{Flowchart of the Lagrangian tracking algorithm with Microscale Topological Surgery (\mts).}
\label{fig:flowchart}

\vspace{1em}

\newlength{\catwidth}
\setlength{\catwidth}{0.16\linewidth}
\newlength{\symwidth}
\setlength{\symwidth}{0.30\linewidth}
\newlength{\descwidth}
\setlength{\descwidth}{0.50\linewidth}

\captionof{table}{Summary of principal notation, terminology, and parameters.}
\label{tab:notation}

\renewcommand{\arraystretch}{1.2}
\begin{tabular}{
>{\raggedright\arraybackslash}p{\catwidth}
>{\raggedright\arraybackslash}p{\symwidth}
>{\raggedright\arraybackslash}p{\descwidth}}
\hline
Category & Symbol & Description \\
\hline

\multirow{2}{=}{Temporal\\ parameters}
& $\tmax$, $\Delta t$, $\{ t_k \}$
& final time, time-step size, and computational times \eqref{lgr-times} \\

& $\{T_m\}$
& \mts\ times \eqref{str-times} \\

\hline

\multirow{2}{=}{Spatial\\ parameters}
& $h_0$, $h$
& coarse and fine mesh spacings \eqref{h0} \\

& $L$, $h/L$
& large-scale deformation length and scale-separation ratio \eqref{h/L} \\

\hline

\multirow{3}{=}{Fidelity \&\\ smoothing\\ parameters}
& $\dref = \tfrac{h}{2}$, $\dcrs=\tfrac{h}{10}$
& default refinement and coarsening thresholds \\

& $\Nmin=3$, $\Lmin=\tfrac{h}{10}$, $\Amin=\tfrac{h^2}{10}$
& default sub-microscale filtering thresholds \\

& $\wfit=10^{-5}$, $\wsm=h^2$, $\wtan=1$
& default interface surgery smoothing parameters \\

\hline

\multirow{6}{=}{Lagrangian\\ tracking}
& $\Gamma_\alpha(t)$, $N(t)$
& evolving interfaces and interface count \\

& $\gamma_\alpha(s,t)$
& interface parametrization \\

& $\tau_\alpha(s,t)$, $\nu_\alpha(s,t)$, $\kappa_\alpha(s,t)$
& tangent, normal, and curvature \\

& $\xi_\alpha(s,t)$, $\lftle_\alpha(s,t)$, $\sftle_\alpha(s,t)$
& (unnormalized) tangent field, stretch factor, and FTLE \\

& $\mathcal W_j(t)$, $c_j(t)$, $w_j(t)$
& Lagrangian observation window, center, and half-width \\

& $\Gamma_{\alpha,h}(t)$, $S_{\alpha,j}$, $N_\alpha$
& polygonal approximation, segment, and node count \\

\hline

\multirow{6}{=}{\mts\\ variables}
& $\Gammalgr_\alpha$, $\Nlgr$, $\zlgr_\alpha$
& pre-processed Lagrangian interface family \\

& $\Gammaeul_\beta$, $\Neul$, $\zeul_\beta$
& extracted Eulerian interface family \\

& $\phi$
& signed-distance function \\

& $\Adj$, $\{\Event_q\}$, $\Nevent$
& adjacency matrix, local events, and number of events \\

& $(\mu,\eta)$, $\Defect$
& topological defect densities and regularized defect measure \\

& $\Surg$, $\{\Plgr,\Peul\}$
& surgical region and surgical pieces \\

\hline
\end{tabular}

\end{figure}

An overview of the \mts\ algorithm is provided in the flowchart in 
\Cref{fig:flowchart} and summarized in the following.
The proposed method evolves the interface family
$\{\Gamma_\alpha(t)\}$ over the time interval $[0,\tmax]$, discretized
uniformly with time-step $\Delta t$,
\begin{subequations}\label{time-discrete}
\begin{equation}\label{lgr-times}
t_k = k \Delta t,
\qquad
k = 0,\ldots,K,
\qquad
K \Delta t = \tmax.
\end{equation}
We prescribe a sequence of \mts\ times
\begin{equation}\label{str-times}
\{T_m\}_{m=1}^{M}
\subset
\{t_k\}_{k=0}^{K},
\qquad
0 < T_1 < T_2 < \cdots < T_M \leq \tmax, 
\end{equation}
\end{subequations}
and evolve the interfaces by a classical Lagrangian tracking method for $t \in (T_m,T_{m+1})$.  
Meanwhile, at each time $t = T_m$, we process the interface family 
by the \mts\ algorithm to detect and enforce topological changes. 
The \mts\ algorithm evolves the interface geometry consistently down
to the prescribed fine scale $h$, while permitting topological
transitions below this threshold. In particular, filamentary structures
with thickness comparable to or smaller than $h$ are regarded as
topologically unresolved and are subsequently reconstructed by the
\mts\ algorithm. Because this topological reconstruction is decoupled
from the underlying Lagrangian geometric representation, \mts\
continues to track sub-microscale filamentary geometry even after the
topology has been modified at scale $h$.

Several existing methods combine Eulerian and Lagrangian information in
hybrid interface representations. Marker-particle level-set and VOF
methods, for example, augment Eulerian interface-capturing schemes with
Lagrangian information to improve accuracy
\cite{EnFeFeMi2002,AuMaSc2003}. The \mts\ methodology instead begins from an
explicitly tracked \emph{pre-processed} Lagrangian interface family
$\{\Gammalgr_\alpha\}$ and uses \emph{Eulerian snapshots}
to supply topological information.
Within the Eulerian-snapshot framework, a common strategy is to detect
topological changes through intersections between the tracked
Lagrangian interface and an underlying Eulerian grid
\cite{GlGrLiTa2000,DuFiFlJiLiLiWu2006,BoLiGlLi2011,ShYoJu2011}.\footnote{
Topological transitions may also be handled without constructing an
Eulerian snapshot
\cite{NoTr1996,GlGrLiTa2000,TrBuEsJuAlTaHaNaJa2001,Regnault2023}.
Such approaches operate directly on the Lagrangian representation and
do not employ an independent Eulerian topology oracle.
}
Rather than relying on grid-intersection tests, \mts\ constructs an
auxiliary Eulerian field whose extracted zero level set serves as a
topological ``oracle'', providing the topological source of truth used
to identify and enforce interface transitions.

Specifically, our \mts\ algorithm constructs a signed-distance function (SDF)
$\phi$ from the pre-processed Lagrangian interface family using an
efficient bounding-box accelerated closest-point algorithm together
with a ray-casting procedure for sign determination. The
zero-level set of $\phi$ defines the topological oracle, from which a
topologically corrected Eulerian interface family
$\{\Gammaeul_\beta\}$ is extracted using a standard
marching-squares procedure \cite{LoCl1998}. 
In this sense, \mts\ may be viewed as a realization
of the basic sharp-interface continuation philosophy advocated by
\citet{LoTr1998}; an essential distinction is that topological
transitions are resolved here through a purely \emph{geometric} construction, 
rather than through the \emph{diffuse-interface} dynamics described in \cite{LoTr1998}.
For related approaches based on reconstructed Eulerian interfaces we refer the reader to
\cite{ShJu2002,CeRo2005,SiSh2007,ShJu2009,WoThGrTu2009,Muller2009,CeRoSiVi2010,GeGoDeWa2025}.

Using the interface families
$\{\Gammalgr_\alpha\}$ and $\{\Gammaeul_\beta\}$ as input, the second
stage of the \mts\ algorithm infers the \emph{adjacency topology} between the
two families from their local geometric configuration. 
The adjacency topology provides a parent--child correspondence
structure between the two interface families, with topological
transitions manifested through departures from one-to-one
correspondence. Such correspondence structures underlie existing topological reconstruction algorithms
\cite{GlGrLiTa2000,DuFiFlJiLiLiWu2006,BoLiGlLi2011,GeGoDeWa2025}; in
\mts, they are made explicit through a graph-theoretic construction.
Specifically, we\ represent the adjacency topology by a bipartite
graph and use its connected components to decompose the topology into
independent local topological events, which are subsequently classified 
according to their parent--child connectivity as 
\emph{trivial}, \emph{vanishing}, or \emph{branching}. 

The first two event types may be resolved directly, whereas \emph{branching}
events require \emph{interface surgery}, performed in the third stage
of the \mts\ pipeline. Existing approaches have either employed
direct-replacement strategies, in which the Lagrangian
interfaces are globally replaced by topologically corrected Eulerian interfaces
\cite{ShJu2002,SiSh2007,Muller2009,ShYoJu2011,PaLoChScPoZa2024,GeGoDeWa2025},
or surgical reconstructions based on grid-intersection tests
\cite{GlGrLiTa2000,DuFiFlJiLiLiWu2006,WoThGrTu2009,BoLiGlLi2011,HeKaStEtYaWo2024}.
In the \mts\ methodology, we employ a surgical framework
based on the construction of localized surgical regions
identified through a collection of topological \emph{defect densities}
and associated \emph{defect measures}. 
Topologically inconsistent portions of the dual interface families are
identified through degeneracies in their closest-point
correspondences, and the resulting defect densities $(\mu,\eta)$
provide a localization of the underlying topological
inconsistencies. A related approach is employed in the ``kink-detection'' algorithm
proposed in \cite{HeCoLu2022} for level-set reinitialization
specialized to filamentary flows. See also the recently proposed
signature method underlying the ``manifold-death'' algorithm of
\cite{ChMaPoZa2022}, which likewise targets filamentation. 

An $h$-scale regularized defect measure $\Defect$ is then constructed
from the defect densities, thereby providing an Eulerian
representation of the microscale defect regions. The associated
surgical regions $\Surg$ are subsequently used to
decompose the Lagrangian and Eulerian interface families into
compatible surgical pieces $\{\Plgr,\Peul\}$, which are stitched
together through a graph-theoretic cycle-tracing procedure to produce
the post-transition interface family.

The principal application considered in this work is multiscale interface
evolution driven by filamentation. To assess the performance of the \mts\
methodology in this setting, three numerical examples of increasing complexity are considered.
The first is the classical rotating-vortex flow \cite{RiKo1998}, which serves as a standard
benchmark for interface-capturing and interface-tracking methods.
Using the \mts\ algorithm, we simulate the formation and subsequent
disappearance of satellite interfaces generated during the filamentation
process. This behavior is consistent with the qualitative dynamics
observed in interface-capturing simulations. 
The second example is the $\mathcal{S}$-flow benchmark, developed in the
context of filament-aware MOF methods
\cite{AhSh2009,JeSuSh2015,HePhXi2023,HeLiPhXi2024}.
Compared with the rotating-vortex flow, this second benchmark exhibits a
moderately more intricate multiscale structure.
The corresponding \mts\ solutions qualitatively reproduce the resulting
filament-breakup dynamics.

The third numerical example is a new filamentation benchmark 
based on a continuous-in-time 
embedding of a nonlinear alternating-shear map.  
The repeated stretching-and-folding of the interface produces substantially richer 
multiscale structure than either of the two previous benchmarks, along with 
a sustained sequence of topological transitions.
To analyze these dynamics, we employ a tangent
modification \cite{Thiffeault2004} of the classical finite-time Lyapunov exponent (FTLE)
\cite{Haller2000,ShLeMa2005,Haller2015}, which provides
a quantitative diagnostic of filament formation 
along the evolving interface.
We then use the tangent FTLE as an \emph{a posteriori} diagnostic tool 
to visualize complex filament-breakup processes at microscale resolution.
Using the analytical structure of the proposed benchmarks, we perform a
quantitative analysis demonstrating $\mathcal{O}(h^2)$
convergence of geometric errors competitive with 
existing filament-aware MOF schemes \cite{JeSuSh2015,HePhXi2023}. 
Moreover, we show that the computational cost of the \mts\ pipeline 
remains commensurate with the baseline cost of Lagrangian tracking, so that the 
combined algorithm exhibits optimal $\mathcal{O}(h^{-2})$ computational scaling.
Finally, we exploit the computational efficiency of the \mts\ framework to construct  
ensembles of microscale-resolving simulations, revealing a competition 
between filament-breakup dynamics and statistical coarse-graining: most breakup 
events are absorbed into a diffuse-interface description, while the breakup in the strongest filamenting
region survives coarse-graining and is statistically persistent.

\subsubsection*{Outline of the paper}

The remainder of the paper is organized as follows. In
\Cref{sec:lgr}, we review the classical Lagrangian interface-tracking
framework and introduce the tangent-FTLE diagnostic used in the
subsequent numerical examples. In \Cref{sec:lgr-examples}, we introduce 
two of the benchmark problems and establish baseline results using the 
classical Lagrangian tracking algorithm. The three
core stages of the \mts\ methodology are then developed in
\Cref{sec:extraction,sec:adjacency,sec:surgery}, where we introduce the
Eulerian snapshot construction, the adjacency topology framework, and
the interface surgery algorithm. Numerical results are presented in \Cref{sec:examples}, followed by
concluding remarks in \Cref{sec:conclusion}. 
Detailed descriptions of the supporting algorithms 
for the \mts\ processor are provided in \Cref{appendix:aux-algs}.

\section{Classical Lagrangian interface tracking}
\label{sec:lgr}

This section reviews the classical Lagrangian interface-tracking 
formulation adopted in the present work.
The goal is to establish notation, define the governing
evolution equations, and provide a reference
implementation for the subsequent \mts\ reconstruction framework.
In addition, we introduce a tangential finite-time Lyapunov exponent
(FTLE) diagnostic for characterizing the formation and development of
filamentary structures along evolving interfaces.

\subsection{Lagrangian formulation of the interface evolution problem}
\label{subsec:lagrangian_formulation}

A schematic of the interface tracking problem is shown in \Cref{fig:schematic}. 
We consider the evolution of a family of bounded, simply connected domains
\begin{subequations}\label{eq:lgr-prob}
\begin{equation}
\Omega_{\alpha}(t)
\subset
\mathbb{R}^2,
\qquad
t \in [0,\tmax],
\end{equation}
where the subscript $\alpha$ indexes the family.
The boundary of each domain is denoted by
\begin{equation}
\Gamma_{\alpha}(t)
\coloneqq
\partial \Omega_{\alpha}(t),
\end{equation}
\end{subequations}
and is referred to as an \emph{interface}.  Each interface is assumed to be a smooth,
simple, closed, oriented curve, and we assume that the interfaces remain 
mutually non-intersecting throughout the evolution so that
${\Gamma_{\alpha}(t)} \cap {\Gamma_{\beta}(t)} =\varnothing$ for $\alpha \neq \beta$.

\begin{figure}[ht]
\centering
\includegraphics[width=0.5\textwidth]{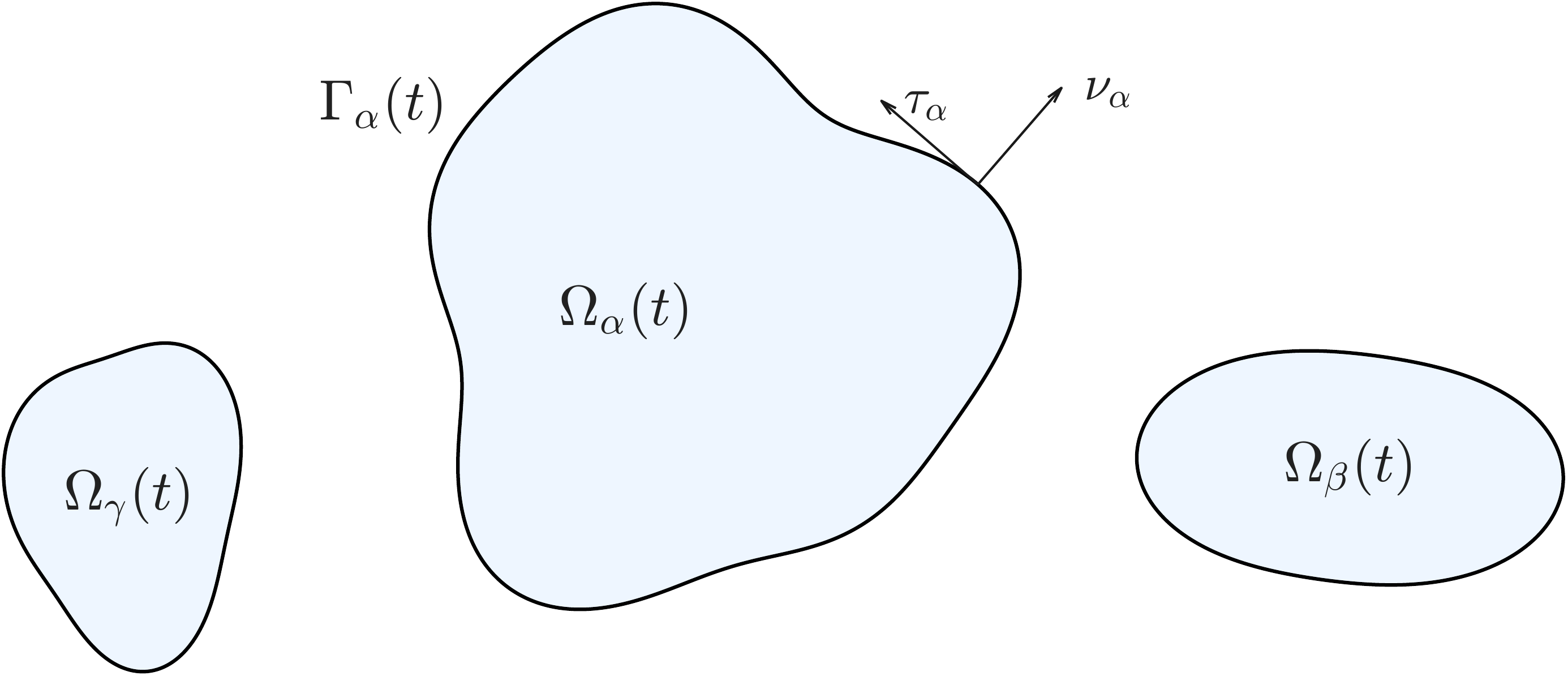}
\caption{Schematic of the interface evolution problem.}
\label{fig:schematic}
\end{figure}

As outlined in the flowchart in \Cref{fig:flowchart}, the collection $\{\Gamma_\alpha(t)\}$ is 
evolved primarily using a classical Lagrangian interface-tracking method, with 
topological changes enforced by the \mts\ algorithm at the prescribed 
set of \mts\ times \eqref{str-times}. The Lagrangian representation consists of an explicit time-dependent
parametrization of each interface
\begin{subequations}\label{eq:z}
\begin{equation}
{\gamma}_{\alpha}(s,t)
=
\left(
\gamma_{\alpha}^{1}(s,t),
\gamma_{\alpha}^{2}(s,t)
\right),
\qquad
(s,t) \in \mathbb{T} \times [0,\tmax],
\end{equation}
where $s$ denotes the parameter along the interface and 
$\mathbb{T} = \mathbb{R}/ 2\pi\mathbb{Z}$ denotes the 
one-dimensional periodic torus, identified
with the interval $[0,2\pi]$ equipped with periodic boundary conditions.
Each interface is therefore represented by
\begin{equation}
\Gamma_{\alpha}(t)
=
\left\{
{\gamma}_{\alpha}(s,t)
:
s \in \mathbb{T}
\right\}.
\end{equation}
\end{subequations}
The unit tangent and normal vectors to each 
interface $\Gamma_{\alpha}(t)$ are given by
\begin{subequations}\label{eq:lgr-geom}
\begin{equation}
\tau_\alpha(s,t)
=
\frac{\p_s \gamma_\alpha(s,t)}
{|\p_s \gamma_\alpha(s,t)|}
\qquad
\text{and}
\qquad
\nu_\alpha(s,t)
=
\tau_\alpha^\perp(s,t),
\end{equation}
where $(a,b)^\perp = (-b,a)$, and the curvature of the interface is given by
\begin{equation}
\kappa_\alpha (s,t) 
=
\frac{
| 
\p_s \tau_\alpha
\cdot
\nu_\alpha
|
}{
|\p_s \gamma_\alpha|
}
=
\frac{
\left|
\p_s \gamma_\alpha^{1}\,
\p_s^2 \gamma_\alpha^{2}
-
\p_s \gamma_\alpha^{2}\,
\p_s^2 \gamma_\alpha^{1}
\right|
}{
|\p_s \gamma_\alpha|^3
}.
\end{equation}
\end{subequations}

The interface index $\alpha$ ranges over
\begin{subequations}
\begin{equation}
\alpha = 1,\ldots,N,
\end{equation}
where $N$ is a non-negative integer.
In the classical Lagrangian interface-tracking method \cite{ChGlMcPlYa1986}, $N$ is a
constant fixed by its value at the initial time. 
In the proposed \mts\ method, interfaces are allowed to undergo topological transitions, so that $N(t)$ becomes a
time-dependent non-negative integer-valued function
\begin{equation}
N
:
[0,\tmax]
\to
\mathbb{N}. 
\end{equation}
In particular, $N(t)$ remains constant on each time interval
$t \in [T_m,T_{m+1})$, with changes permitted only at the \mts\ times $t=T_m$.
\end{subequations}

For times $t \in (T_m, T_{m+1})$, the interface 
parametrizations $\gamma_\alpha$ are evolved 
by solving the Lagrangian evolution equations
\begin{subequations}\label{eq:lgr}
\begin{equation}\label{eq:lgr-evolution}
\p_t \gamma_\alpha (s,t) 
= u(\gamma_\alpha(s,t),t), 
\qquad (s,t) \in \mathbb{T} \times (T_m, T_{m+1}), 
\quad \alpha = 1, \ldots, N(T_m), 
\end{equation}
with initial data 
\begin{equation}\label{eq:lgr-init}
\gamma_\alpha(s,T_m) \quad \text{generated by the \mts\ algorithm applied at time $t = T_m$}. 
\end{equation}
\end{subequations}
The \emph{classical Lagrangian tracking} algorithm is recovered by bypassing the 
\mts\ reconstruction step in \eqref{eq:lgr-init} and simply taking the initial data at each 
time $T_m$ to be the terminal interface configuration obtained 
from the preceding Lagrangian evolution.

We assume that the velocity field
\[
u : \mathbb{R}^2 \times [0,\tmax] \to \mathbb{R}^2
\]
is smooth. 
The proposed framework does not depend on how $u(x,t)$ is obtained; for instance, it may be 
prescribed analytically, computed by a coarse-scale 
fluid solver, or generated by a boundary integral method.

\subsection{Polygonal interface representation with adaptive refinement}
\label{subsec:adaptive_refinement}

The Lagrangian evolution equations \eqref{eq:lgr} are discretized 
in time using the standard TVD-RK3 scheme \cite{ShOh1988}
with fixed time step $\Delta t$.
In the following, we describe the spatial discretization scheme employed, including the adaptive 
refinement procedure used to maintain geometric resolution throughout the evolution. 
Adaptive refinement techniques for Lagrangian interface tracking are
well established in the literature, including three-dimensional implementations 
\cite{UnTr1992,GlGrLiTa2000,TrBuEsJuAlTaHaNaJa2001,DuFiFlJiLiLiWu2006,BoLiGlLi2011,GoEvWaDe2022,Regnault2023}.

We use Greek subscripts $\alpha,\beta,\ldots$ to index the interface
family and Roman subscripts $i,j,\ldots$ to index the spatial discretization along
each interface. At each time $t$, we denote by 
\begin{subequations}
\begin{equation}
s_{\alpha,i}(t) \in \mathbb{T},
\qquad
i = 1,\ldots,N_\alpha(t),
\qquad
\alpha = 1,\ldots,N(t), 
\end{equation}
a family of discretizations of the the parameter interval $\mathbb{T}$. 
Here $N_\alpha : [0,\tmax] \to \mathbb{N}$ denotes 
the number of \emph{nodes} in the corresponding 
discretization, which is allowed to vary in time 
under adaptive refinement and coarsening. We will often suppress the explicit 
dependence on time, writing $N_\alpha \equiv N_\alpha(t)$ with the understanding 
that $N_\alpha$ will typically change at the discrete times $t = t_k$.

The spatial approximation to the interface parametrization is denoted by
\begin{equation}
\gamma_{\alpha,i}(t)
\coloneqq
\gamma_\alpha(s_{\alpha,i}(t),t), 
\end{equation}
with periodic closure enforced by identifying
\begin{equation}
\gamma_{\alpha,N_\alpha+1}(t)
=
\gamma_{\alpha,1}(t).
\end{equation}
Each set of points $\{\gamma_{\alpha,i} \}_{i=1}^{N_\alpha}$ defines a polygonal approximation of the interface, 
\begin{equation}
\Gamma_{\alpha,h}(t)
=
\bigcup_{i=1}^{N_\alpha} 
S_{\alpha,i}(t),  
\qquad 
S_{\alpha,i}(t)
\coloneqq
\left[
\gamma_{\alpha,i}(t),
\gamma_{\alpha,i+1}(t)
\right],
\end{equation}
\end{subequations}
where $S_{\alpha,i}(t)$ denotes the oriented line segment 
connecting consecutive interface nodes. 

The subscript $h$ in $\Gamma_{\alpha,h}$ indicates that the 
resolution of the polygonal approximation is controlled by the fine-scale $h$. 
At this stage, $h$ may be viewed as prescribed; later, we will define $h$ 
as a dyadic refinement of a coarse {computational scale} $h_0$. 
As small scales develop in the evolving interface, a simple refinement 
algorithm is employed to maintain resolution. 
We introduce refinement and coarsening thresholds
\[
0 \leq \dcrs < \dref < h,
\]
which determine the admissible range of segment lengths in the
polygonal approximation. Segments with length exceeding $\dref$ are
refined through midpoint insertion and cubic interpolation, while segments with length below
$\dcrs$ are coarsened through node removal, producing a dyadic refinement of the parameter space $\mathbb{T}$. 
The resulting interface discretization is necessarily subgrid to the fine-scale $h$. 
The complete adaptive refinement procedure is described in \Cref{alg:air} in \Cref{appendix:aux-algs}.

\subsection{Characterizing filament formation by the tangent FTLE}
\label{subsec:tangent-ftle}

For visualization and qualitative analysis, it is useful to have a diagnostic that 
identifies the development of localized small scales and filamentary structures 
along the evolving interface. A common choice is the interface curvature; however, as the following 
examples illustrate, large curvature does not necessarily imply filament formation, while strongly 
filamentary structures may exhibit only modest curvature.

\begin{example}[A shrinking circle]
\label{ex:shrinking-circle}
Consider an exponentially shrinking circle
$\gamma(s,t)=R_0 e^{-\lambda t}(\cos s,\sin s)^\top$, $s\in\mathbb{T}$,
with $R_0>0$ and $\lambda>0$. The curvature is
$\kappa(s,t)=R_0^{-1}e^{\lambda t}$, which becomes arbitrarily large as
$t\to\infty$. Nevertheless, the interface remains perfectly circular for all
time and develops no thin elongated structures.
\end{example}

\begin{example}[Localized incompressible stretching]
\label{ex:localized-stretching}
Let $\gamma_0(s)=(a_0\cos s,b_0\sin s)^\top$, $s\in\mathbb{T}$, denote an ellipse
with semi-axes $a_0=4$ and $b_0=0.35$. We deform this curve by the
time-dependent family of maps
\[
\psi(x_1,x_2,t)
=
\left(
F(x_1,t), \,
\frac{x_2}{\partial_1 F(x_1,t)}
\right)^\top,
\qquad
t\in[0,1],
\]
where
\[
F(x_1,t)
=
x_1
+
\tfrac{\sqrt{\pi}}{2}
A t \operatorname{erf}(x_1),
\qquad
\partial_1 F(x_1,t)
=
1+A t e^{-x_1^2},
\]
with $A=4$. The deformation gradient is
\[
D\psi(x_1,x_2,t)
=
\begin{pmatrix}
\partial_1 F(x_1,t) & 0 \\[1mm]
-\tfrac{x_2 \partial_1^2 F(x_1,t)}{(\partial_1 F(x_1,t))^2}
&
\tfrac{1}{\partial_1 F(x_1,t)}
\end{pmatrix},
\]
so that $\det D\psi=1$.
The initial and final interfaces are shown in
\Cref{fig:ftle-example-int}. The localized increase in $\partial_1 F(x_1,t)$ produces
strong tangential stretching near $x_1=0$ together with transverse compression
through the factor $1/\partial_1 F(x_1,t)$, causing the interface to develop an elongated
filamentary region. Meanwhile, the curvature (shown in
\Cref{fig:ftle-example-kappa}) remains bounded by its initial maximum value.
\end{example}

\begin{figure}[ht]
\centering
\begin{subfigure}[t]{0.25\linewidth}
  \centering
  \includegraphics[width=0.9\linewidth]{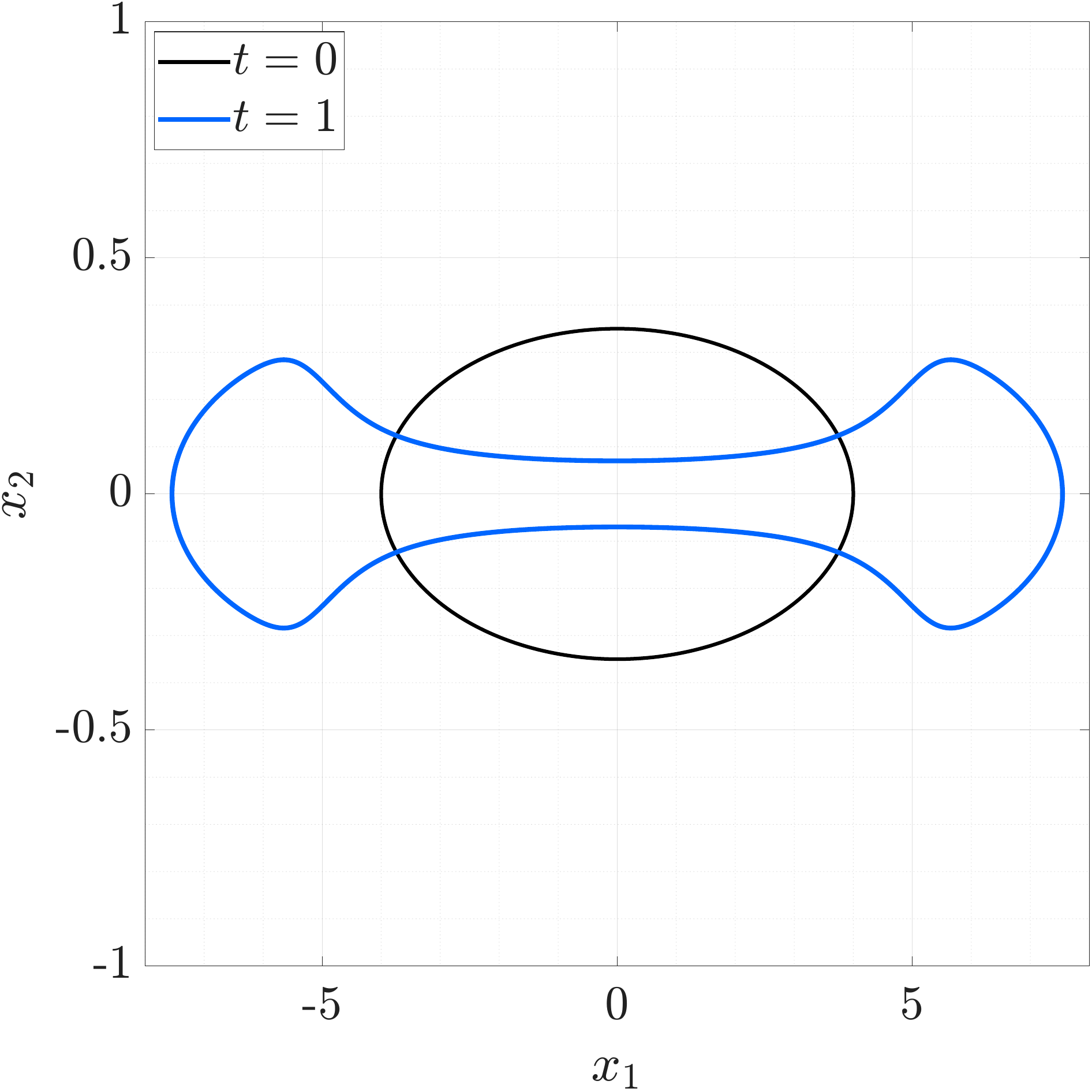}
  \caption{Interface evolution}
  \label{fig:ftle-example-int}
\end{subfigure}
\hspace{0.5em}
\begin{subfigure}[t]{0.25\linewidth}
  \centering
  \includegraphics[width=0.9\linewidth]{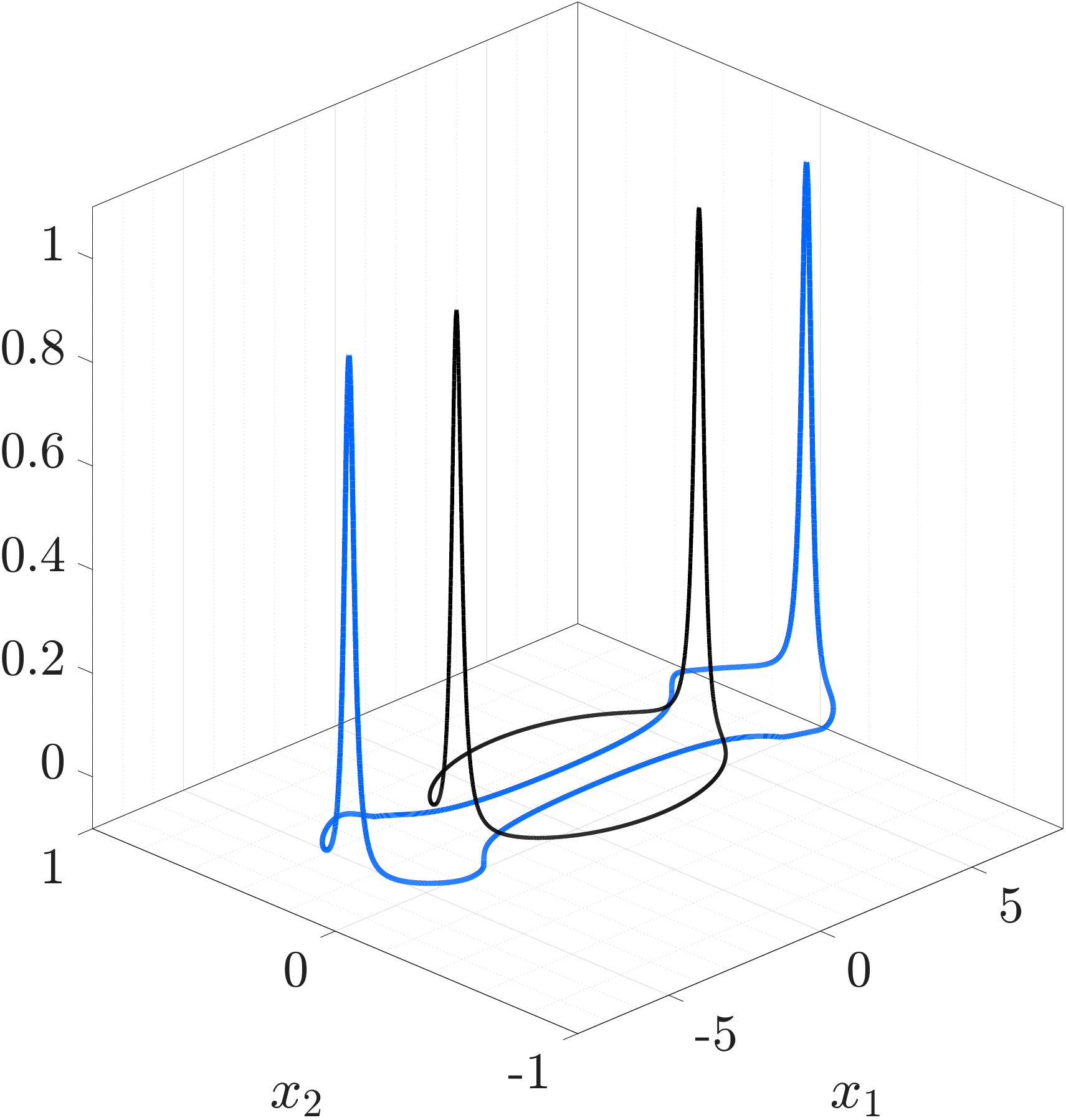}
  \caption{Normalized curvature}
  \label{fig:ftle-example-kappa}
\end{subfigure}
\hspace{0.5em}
\begin{subfigure}[t]{0.25\linewidth}
  \centering
  \includegraphics[width=0.9\linewidth]{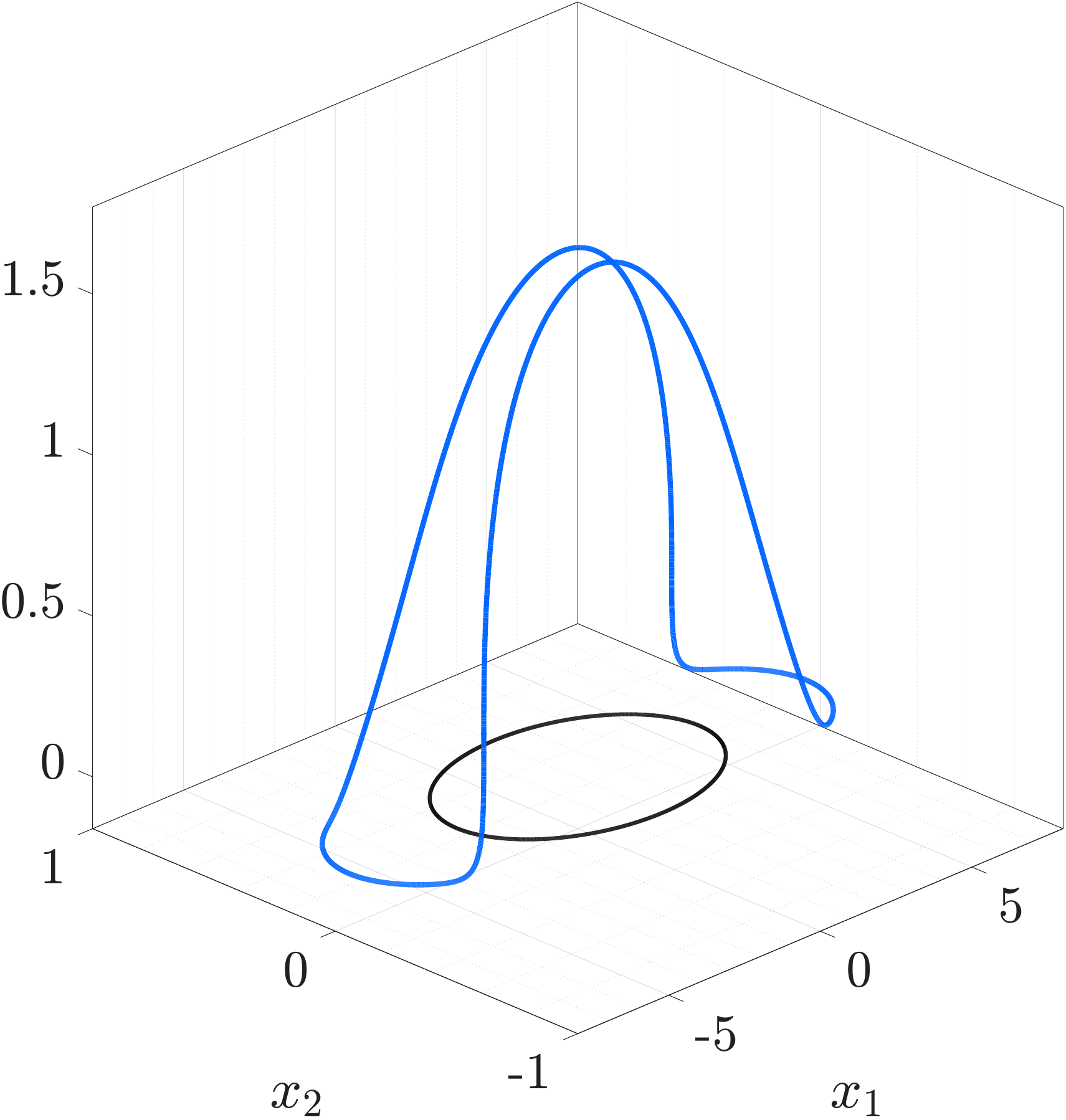}
  \caption{Tangent FTLE}
  \label{fig:ftle-example-sigma}
\end{subfigure}
\caption{
Localized incompressible stretching of the ellipse from
\Cref{ex:localized-stretching}. In all panels, quantities at $t=0$ are shown
in black and quantities at $t=1$ are shown in blue. The curvature and
tangent FTLE diagnostics are plotted directly on the evolving interface.
\textbf{Left:}
Formation of an elongated filament in the interface $\gamma(s,1)$.
\textbf{Center:}
Normalized curvature
$|\kappa(s,t)|/\max_{r\in\mathbb{T}}|\kappa(r,0)|$.
Despite the pronounced filamentation, the maximum curvature at the
final time remains bounded by its initial maximum value.
\textbf{Right:}
The tangent FTLE develops a pronounced localized maximum that
identifies the filamentary region.
}
\label{fig:ftle-example}
\end{figure}

To characterize the formation of filamentary structures, we introduce a
Lagrangian diagnostic based on the tangential deformation of neighboring
material points along the evolving interface family. Let
\begin{subequations}\label{flowmap}
\begin{equation}
\psi : \mathbb{R}^2 \times [0,\tmax] \to \mathbb{R}^2
\end{equation}
denote the flowmap associated with the velocity field $u(x,t)$ so that
\begin{equation}
\gamma_\alpha(s,t)
=
\psi( \gamma_\alpha(s,0),t).
\end{equation}
\end{subequations}
The classical finite-time Lyapunov exponent (FTLE) measures the
maximal finite-time stretching experienced by infinitesimal
perturbations under the flowmap
\cite{Haller2000,ShLeMa2005,Haller2015}.
Introducing the Cauchy--Green strain tensor
\[
C(x,t)
=
D\psi(x,t)^\top D\psi(x,t),
\]
where $D\psi$ denotes the deformation gradient of the flowmap, the classical FTLE is defined by
\[
\Lambda(x,t)
\coloneqq
\tfrac{1}{2t}
\log
\lambda_{\max}(C(x,t))
=
\tfrac{1}{2t}
\log
\left(
\max_{v\neq0}
\tfrac{v^\top C(x,t)v}{v^\top v}
\right)
=
\tfrac{1}{t}
\log
\left(
\max_{v\neq0}
\tfrac{|D\psi(x,t)v|}{|v|}
\right),
\]
where $\lambda_{\max}$ denotes the largest eigenvalue of $C(x,t)$.
Thus, the classical FTLE measures the maximal finite-time 
stretching rate over all perturbation directions.

For interface dynamics, however, the relevant perturbations are those
connecting neighboring material points on the interface. Since the
stretching and compression of such perturbations govern the local
concentration or separation of interface markers and thereby drive the
formation of filamentary structures, we follow the viewpoint
of \citet{Thiffeault2004} and restrict attention to perturbations in the
tangent direction of the material interface.
Introducing the unnormalized tangent field
\begin{subequations}\label{xi-def}
\begin{equation}
\xi_\alpha(s,t)
\coloneqq
\partial_s \gamma_\alpha(s,t),
\end{equation}
differentiation of \eqref{flowmap} with respect to $s$ then yields
\begin{equation}
\xi_\alpha(s,t)
=
D\psi(\gamma_\alpha(s,0),t)
\,\xi_\alpha(s,0).
\end{equation}
\end{subequations}
The tangential stretch factor
\begin{subequations}\label{ftle}
\begin{equation}\label{lftle}
\lftle_\alpha(s,t)
\coloneqq
\frac{
|\xi_\alpha(s,t)|
}{
|\xi_\alpha(s,0)|
}
=
\left|
D\psi(\gamma_\alpha(s,0),t)
\,\tau_\alpha(s,0)
\right|
\end{equation}
therefore measures the local tangential stretching or compression of
neighboring material points along the evolving interface.
The associated \emph{tangent finite-time Lyapunov exponent} (FTLE) is defined by
\begin{equation}\label{sftle}
\sftle_{\alpha}(s,t)
\coloneqq
\tfrac{1}{t}
\log
\lftle_{\alpha}(s,t).
\end{equation}
\end{subequations}
Positive values of $\sftle_\alpha$ correspond to tangential stretching 
and negative values to tangential compression. 
Localized extrema of $\sftle_\alpha$ identify regions of concentrated tangential 
deformation and therefore provide a useful diagnostic for the development of filamentary structures.

To compute the tangent FTLE in the \mts\ algorithm, we differentiate
the Lagrangian evolution equation \eqref{eq:lgr-evolution} with respect
to $s$, obtaining
\begin{subequations}\label{eq:xi}
\begin{equation}\label{eq:xi-evolution}
\partial_t \xi_\alpha(s,t)
=
Du(\gamma_\alpha(s,t),t)\,
\xi_\alpha(s,t),
\qquad
(s,t)\in\mathbb{T}\times(T_m,T_{m+1}),
\quad
\alpha=1,\ldots,N(T_m),
\end{equation}
with initial data
\begin{equation}\label{eq:xi-init}
\xi_\alpha(s,T_m)
\quad
\text{generated by the \mts\ algorithm applied at time $t=T_m$.}
\end{equation}
\end{subequations}
Here $Du$ denotes the velocity gradient.
As with the interface evolution \eqref{eq:lgr}, classical
Lagrangian tracking corresponds to bypassing the \mts\ reconstruction
step in \eqref{eq:xi-init} and directly continuing the tangent
field from the preceding evolution interval.
In the numerical method, \eqref{eq:xi-evolution}
is integrated alongside the interface evolution \eqref{eq:lgr-evolution},
and the tangent FTLE diagnostic \eqref{ftle}
is then evaluated \emph{a posteriori} from the computed tangent field.\footnote{
Whenever adaptive refinement or coarsening modifies the interface
discretization, both the tangent field $\xi_\alpha(s,t)$ and
the reference quantity $|\xi_\alpha(s,0)|$ are transferred to the new
interface nodes by interpolation.
}
The resulting quantity provides a simple Lagrangian measure of localized
tangential stretching and compression directly along the evolving
interface family.

In contrast to curvature, the tangent FTLE correctly distinguishes
between the two preceding examples. For the shrinking circle of
\Cref{ex:shrinking-circle},
$|\partial_s\gamma(s,t)|=R_0e^{-\lambda t}$ and
$|\partial_s\gamma(s,0)|=R_0$, so that
$\lftle(s,t)=e^{-\lambda t}$ and $\sftle(s,t)=-\lambda$.
Thus, the tangent FTLE is spatially uniform and constant in time despite the
unbounded growth of the curvature. For the localized-stretching example of
\Cref{ex:localized-stretching}, the tangent FTLE develops a pronounced
localized maximum near the center of the interface, as shown in
\Cref{fig:ftle-example-sigma}. This maximum identifies the concentrated
tangential stretching responsible for the filamentary region visible in
\Cref{fig:ftle-example-int}, whereas the curvature shown in
\Cref{fig:ftle-example-kappa} remains bounded by its initial maximum value.

The preceding examples show that curvature and tangent FTLE 
capture different aspects of interface evolution, with curvature serving as an 
instantaneous geometric descriptor and the tangent FTLE as an accumulated 
Lagrangian geometric descriptor. 
Nevertheless, near filament \emph{tips}, the two quantities become closely 
related \cite{Thiffeault2004,BeDrLaWi2022}. 
Using a local quadratic model,
\citet{Thiffeault2004} derived a universal relationship between
tangential stretching and curvature. 
In the present notation, this stretch--curvature relation takes the following form.

\begin{proposition}[Stretch--curvature relation at a filament tip]
Suppose that, near a filament tip, the interface admits the local
representation
\[
\gamma(s,t)
=
(\beta s,f(s)),
\qquad
f(s)
=
\tfrac12 a s^2
+
\mathcal O(s^3),
\qquad
a,\beta>0.
\]
Then, as $s\to0$,
\begin{subequations}\label{stretch-curve-law}
\begin{equation}\label{stretch-curve-law-a}
\lftle(s,t)
\sim
\kappa(s,t)^{-1/3},
\end{equation}
and consequently
\begin{equation}\label{stretch-curve-law-b}
\sftle(s,t)
=
-\tfrac{1}{3t}\log\kappa(s,t)
+
\mathcal O(t^{-1}).
\end{equation}
\end{subequations}
\end{proposition}

\begin{proof}
The local representation gives
$\partial_s\gamma=(\beta,f'(s))$ and
$\partial_s^2\gamma=(0,f''(s))$. Hence
\[
\kappa(s,t)
=
{\beta |f''(s)|} \,
{\bigl(\beta^2+f'(s)^2\bigr)^{-\frac32}}.
\]
Since $f''(s)=a+\mathcal O(s)$ as $s\to0$ and
$|\partial_s\gamma(s,t)|=(\beta^2+f'(s)^2)^{\frac12}$, we obtain
\[
\kappa(s,t)
\sim
|\partial_s\gamma(s,t)|^{-3},
\]
after absorbing the fixed prefactor $\beta |a|$ into the
asymptotic constant.
Using \eqref{lftle} and absorbing the factor $|\partial_s\gamma(s,0)|$, it follows that
\[
\kappa(s,t)
\sim
\lftle(s,t)^{-3},
\]
which upon rearranging yields \eqref{stretch-curve-law-a}.
Taking logarithms and using \eqref{sftle} then yields
\eqref{stretch-curve-law-b}.
\end{proof}

The asymptotic relation \eqref{stretch-curve-law} will be verified numerically 
in the benchmark problems considered below. 
From a computational perspective, the tangent FTLE possesses an 
important practical advantage: while curvature is a second-derivative quantity 
whose computation becomes highly sensitive near breakup at filament tips, the tangent 
FTLE is obtained from the evolution of the first-derivative tangent field and can 
be naturally continued through topological changes by the \mts\ reconstruction procedure.
Consequently, the tangent FTLE provides a more robust diagnostic for tracking 
filamentation and subsequent breakup.

\subsection{Lagrangian observation windows for visualizing filamentation}
\label{subsec:observation-windows}

As we will demonstrate below, one of the principal advantages of the
Lagrangian tracking algorithm is its accuracy and efficiency in chaotic
and multiscale flow regimes. As the complexity of the interface
increases in such regimes, the formation and evolution of individual filaments
(particularly at microscopic scales) become difficult to discern
from global visualizations alone. To facilitate the qualitative
analysis of these processes, we introduce a collection of 
\emph{Lagrangian observation windows} that remain attached to selected material
regions throughout the evolution. Let
\begin{subequations}\label{eq:window-centers}
\begin{equation}
c_j : [0,\tmax] \to \mathbb{R}^2,
\qquad
j=1,\ldots,N_{\mathsf{win}},
\end{equation}
denote a collection of marker trajectories satisfying the Lagrangian
evolution equation
\begin{equation}
\dot{c}_j(t)
=
u(c_j(t),t),
\qquad
c_j(0)=c_j^0, 
\end{equation}
where the initial locations $\{c_j^0\}_{j=1}^{N_{\mathsf{win}}}$
are chosen according to the filamentary regions of interest. 
For each marker trajectory $c_j(t)$, we define the associated
Lagrangian observation window
\begin{equation}
\mathcal W_j(t)
=
[c_j^1(t)-w_j(t),\,
 c_j^1(t)+w_j(t)]
\times
[c_j^2(t)-w_j(t),\,
 c_j^2(t)+w_j(t)],
\label{eq:observation-window}
\end{equation}
\end{subequations}
where $c_j=(c_j^1,c_j^2)$ and $w_j(t)>0$ denotes a prescribed window half-width.
Because the window centers evolve as passive Lagrangian markers, each
observation window remains attached to the same material region
throughout the evolution. In the numerical examples below, the initial
marker locations are selected using the 
\emph{a posteriori} curvature and tangent FTLE diagnostics
introduced in \Cref{subsec:tangent-ftle}, allowing the resulting
observation windows to remain focused on filament-forming regions over
long time intervals.


\section{Two illustrative examples of Lagrangian tracking}
\label{sec:lgr-examples}

We now consider two numerical examples to illustrate the prototypical 
features of the classical Lagrangian tracking algorithm described 
in \Cref{sec:lgr}: its excellent geometric accuracy and computational 
efficiency; and its inability to naturally resolve topological transitions, in agreement 
with observations reported previously in the literature \cite{DuFiFlJiLiLiWu2006}.
At the same time, these examples introduce the tangent-FTLE diagnostic as a
means of quantifying filamentation and identifying regions where 
fine-scale structures emerge. They therefore benchmark the current
implementation while establishing a baseline comparison for the \mts\
methodology developed in subsequent sections.

\subsection{Rotating vortex benchmark}
\label{subsec:rotating-vortex}

The first example is the rotating vortex test, a standard benchmark 
problem for interface advection methods \cite{RiKo1998}. The prescribed incompressible 
velocity field $u$ is defined as 
\begin{equation}
\label{eq:rotating-vortex-velocity}
\begin{aligned}
u_1(x_1,x_2,t)
&=
-\cos\!\left(\tfrac{\pi t}{8}\right)
\sin^2(\pi x_1)\sin(2\pi x_2),
\\
u_2(x_1,x_2,t)
&=
\phantom{-}\cos\!\left(\tfrac{\pi t}{8}\right)
\sin(2\pi x_1)\sin^2(\pi x_2), 
\end{aligned}
\end{equation}
and the initial interface is a circle of radius $0.15$ centered at 
$(0.5, 0.75)$.  The velocity gradient required for the tangent-FTLE computation
\eqref{eq:xi-evolution} is
\begin{equation}
\label{eq:rotating-vortex-gradient}
Du(x_1,x_2,t)
=
\cos\!\left(\tfrac{\pi t}{8}\right)
\begin{pmatrix}
-\pi\sin(2\pi x_1)\sin(2\pi x_2)
&
-2\pi\sin^2(\pi x_1)\cos(2\pi x_2)
\\[1mm]
2\pi\cos(2\pi x_1)\sin^2(\pi x_2)
&
\pi\sin(2\pi x_1)\sin(2\pi x_2)
\end{pmatrix}.
\end{equation}
The strong vortical flow continuously deforms 
the interface into increasingly thin filaments 
until reaching a state of maximum deformation 
at time $t = 4$, and then reverses 
direction and returns the interface toward 
its initial configuration at the final time 
$\tmax = 8$. The test therefore provides a convenient benchmark for 
assessing geometric accuracy, filament resolution, 
and computational efficiency.

We perform a resolution study by dyadic refinement of the fine-scale $h$ using 
the Lagrangian tracking algorithm with parameters in \Cref{tab:notation} chosen as
\begin{equation}\label{eq:rotating-vortex-params}
\tmax = 8, \quad \Delta t = h, \quad  h_0 = 0.04, \quad h = 2^{-p} h_0, \ p = 0,\ldots,5. 
\end{equation}
The numerical solution at the coarsest 
resolution $h = h_0$ is displayed as the blue 
curve in \Cref{fig:LGR-vs-LS} at the three times 
$t = 1, 4, 8$.  At time $t = 4$, the deformed interface is characterized by 
two curved filament-tips at opposite ``ends'' of the vortical filamentary region. 
One end undergoes comparatively weak filament-tip formation, while the opposite end 
develops a much thinner and more strongly compressed filament tip.

\begin{figure}[ht]
\centering
\begin{subfigure}[t]{0.28\linewidth}
  \centering
  \includegraphics[width=0.95\linewidth]{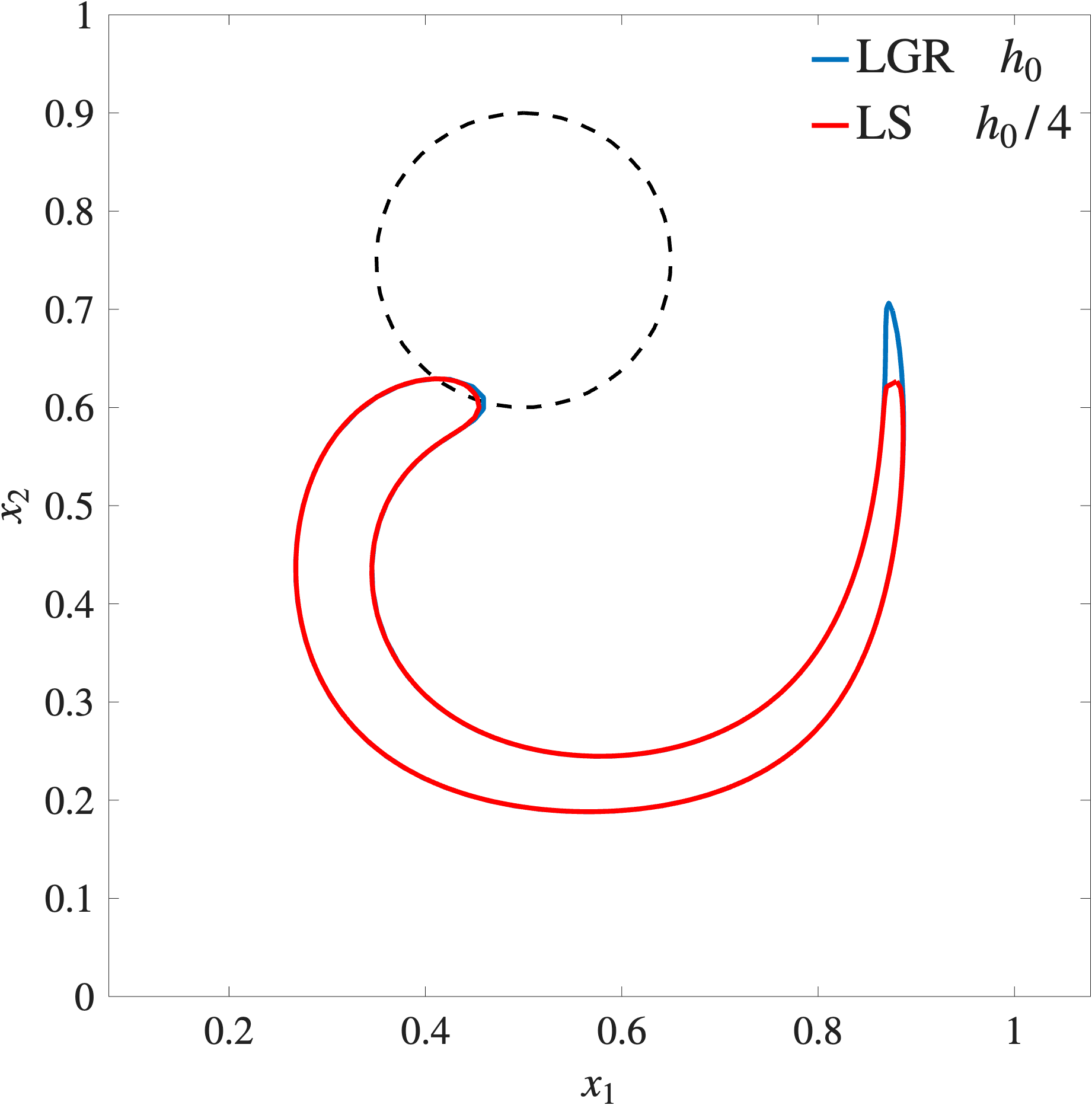}
  \caption{$t = 1$}
  \label{fig:LGR-vs-LS-t1}
\end{subfigure}
\begin{subfigure}[t]{0.28\linewidth}
  \centering
  \includegraphics[width=0.95\linewidth]{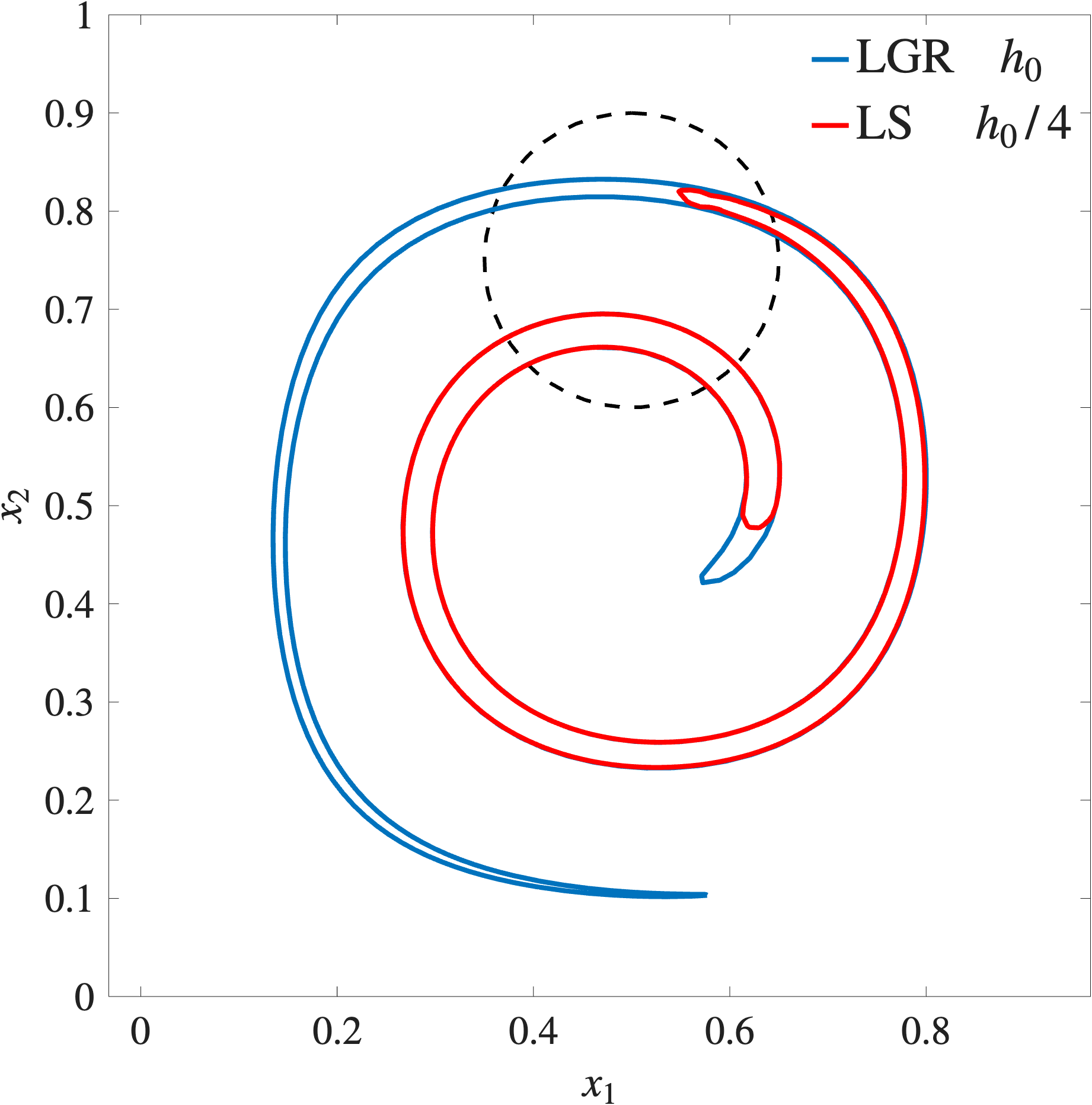}
  \caption{$t=4$}
  \label{fig:LGR-vs-LS-t4}
\end{subfigure}
\begin{subfigure}[t]{0.28\linewidth}
  \centering
  \includegraphics[width=0.95\linewidth]{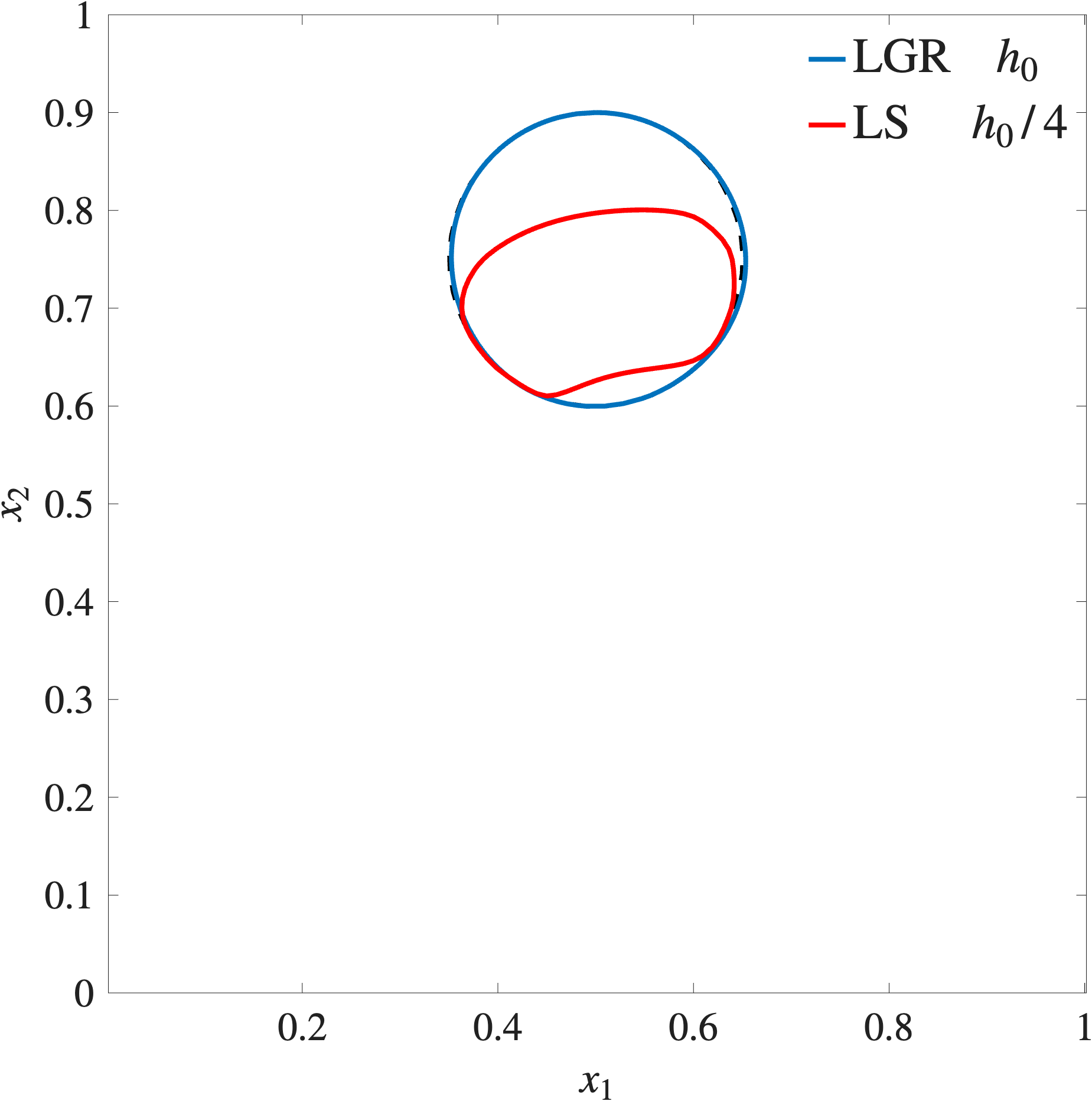}
  \caption{$t=8$}
  \label{fig:LGR-vs-LS-t8}
\end{subfigure}
\caption{
Comparison of classical Lagrangian tracking (blue) and 
a local level-set method (red) for the 
rotating vortex benchmark at times $t = 1, 4, 8$. 
The dashed black curve denotes the initial interface. 
The Lagrangian solution is computed using the coarse 
resolution $h = h_0$, while the level-set solution is 
computed using $h = h_0/4$.
}
\label{fig:LGR-vs-LS}
\end{figure}

\subsubsection{Benchmark comparison with classical Eulerian interface-capturing}

For comparison, we include results obtained using 
a simple Eulerian interface-capturing method, namely, 
a local level-set method \cite{PeMeOsZhKa1999}. In this approach, the signed 
distance function $\phi$ is evolved by the transport equation 
within a tube of radius $6h$, with reinitialization 
performed after every time-step to maintain the signed 
distance property.\footnote{This implementation is not intended 
to be representative of modern level-set methods, and is 
included only as a simple baseline comparison for the 
Lagrangian tracking method. More sophisticated methods typically
provide improved accuracy at additional cost \cite{EnFeFeMi2002,HiKo2005,OlKr2005,DeMoPi2008}.} The same sequence of spatial resolutions 
\eqref{eq:rotating-vortex-params} is used for the level-set simulations, 
with the CFL-constrained time-step given by
$\Delta t = {0.25 h} / {\| u \|_\infty}$. 
At the two coarsest resolutions $h= h_0, h_0/2$, the level-set solution 
undergoes catastrophic mass loss and disappears before 
the final time. Consequently, the red curves displayed 
in \Cref{fig:LGR-vs-LS} correspond to the level-set 
solution computed using the coarsest surviving resolution $h = h_0/4$. 
Even with four times the resolution, the level-set solution is 
unable to match the accuracy of the coarser Lagrangian solution.

\subsubsection{Resolution of filament tip and tangent FTLE diagnostics}
The left panel in \Cref{fig:rotating-vortex-LGR} displays the 
$h$-refined sequence of Lagrangian solutions at time $t = 4$, together with two 
successive zoom-ins of the filament tip. 
The first and second zoom-ins correspond to viewing 
windows of size $h_0 \times h_0 / 3$ and 
$h_0/16 \times h_0 / 48$, respectively. 
As the fine-scale $h$ decreases, the Lagrangian solutions resolve the 
large-curvature region near the filament 
tip with increasing geometric fidelity.

\begin{figure}[ht]
\centering
\begin{subfigure}[t]{0.63\linewidth}
  \centering
  \includegraphics[width=0.99\linewidth]{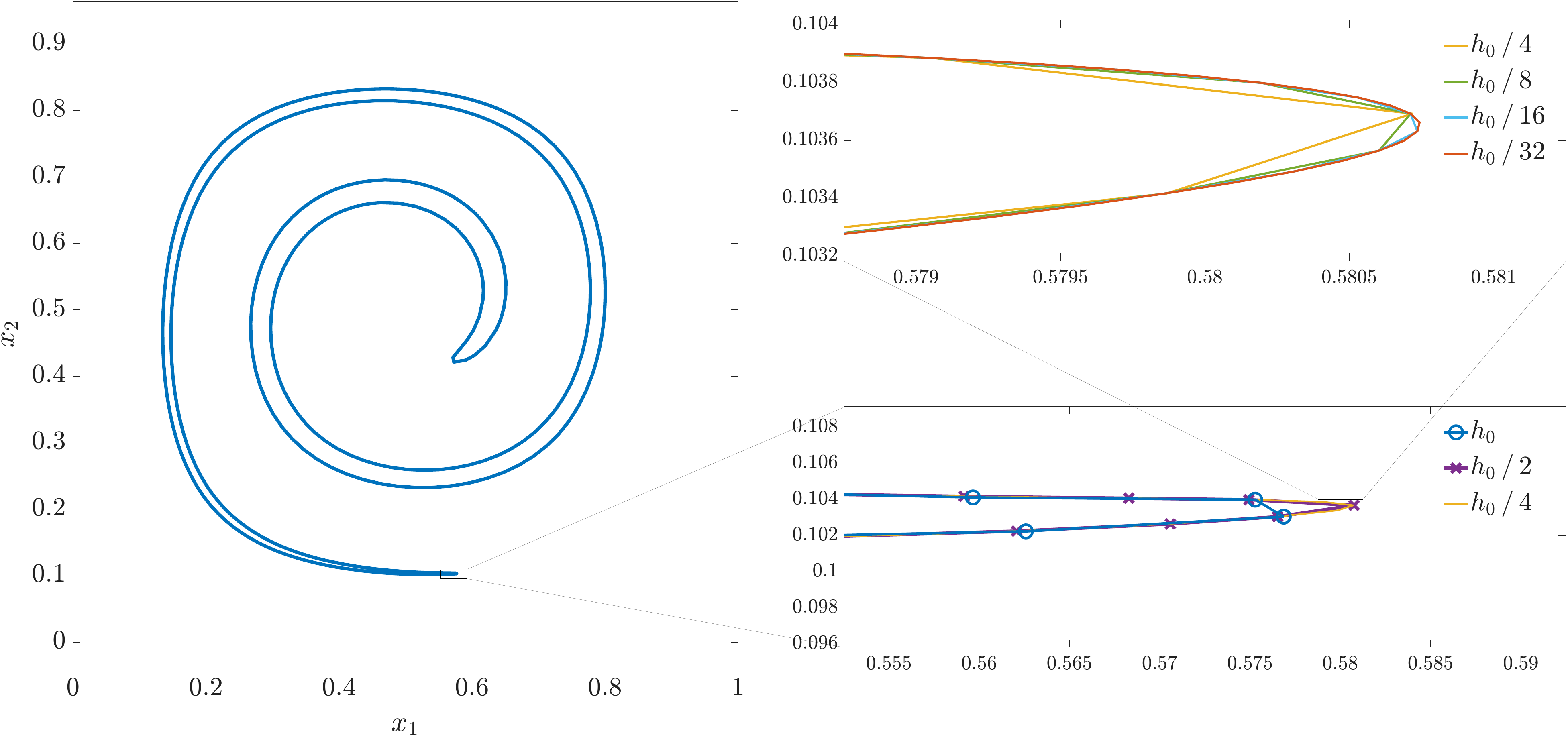}
  \caption{Resolution of the filament tip}
   \label{fig:tip-zoom}
\end{subfigure}
\hspace{0.25em}
\begin{subfigure}[t]{0.34\linewidth}
  \centering
  \includegraphics[width=1.0125\linewidth]{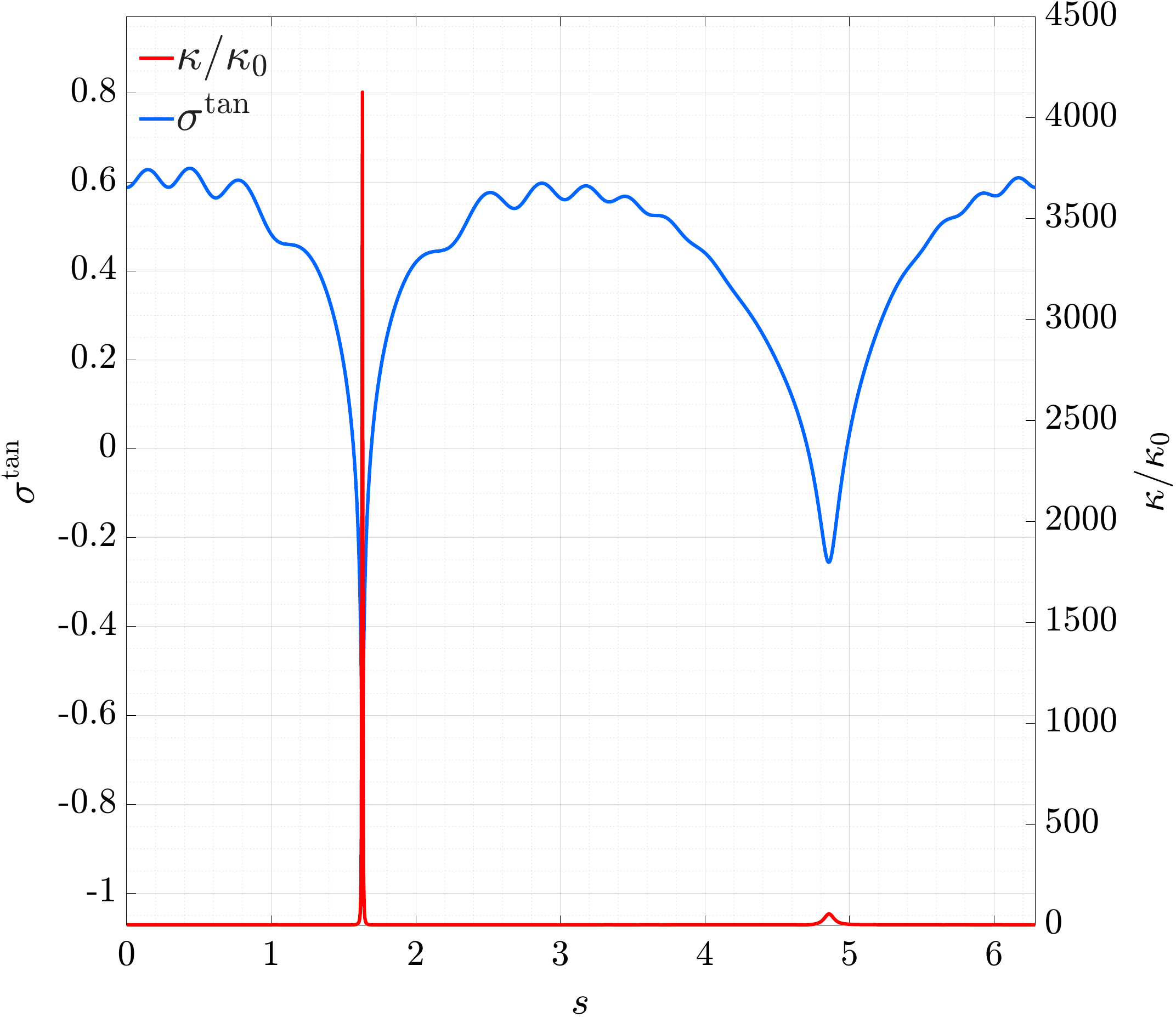}
  \caption{Tangent FTLE and curvature}
  \label{fig:dyadic-refinement}
\end{subfigure}
\caption{
Resolution of the filament tip and adaptive refinement  
for the rotating vortex benchmark at time $t = 4$. 
\textbf{Left:} Sequence of Lagrangian solutions computed 
with decreasing fine-scale $h$, together 
with successive zoom-ins of the filament tip illustrating 
the increasing geometric fidelity of the large-curvature 
region. Relative to the full computational domain 
$[0,1]^2$, the zoomed-in views correspond to windows with aspect ratios 
$h_0  \times h_0/3$ and $h_0/16 \times h_0/48$, respectively.
\textbf{Right:} Tangent FTLE $\sftle(s,t)$ \eqref{ftle} and normalized curvature 
$\kappa(s,t)/\kappa(s,0)$ at the time of maximal deformation 
$t = 4$ for the finest Lagrangian simulation $h = h_0 / 32 = \num{1.25e-3}$. 
The two localized troughs of the tangent FTLE at $s \approx \tfrac{\pi}{2}, \tfrac{3 \pi}{2}$ 
coincide with the curvature peaks at the filament tips.
}
\label{fig:rotating-vortex-LGR}
\end{figure}

The right panel of \Cref{fig:rotating-vortex-LGR} displays the tangent FTLE
\eqref{ftle} at time $t=4$ in blue together with the normalized curvature
$\kappa(s,\tmax)/\kappa(s,0)$, where $\kappa$ is defined
in \eqref{eq:lgr-geom}, in red. Two localized troughs in the tangent FTLE are
observed near $s\approx\pi/2$ and $s\approx3\pi/2$, corresponding respectively
to the strong and weak filamentation at the two ends of the deformed interface.
The sharper FTLE trough near $s\approx\pi/2$ is
accompanied by a correspondingly larger curvature peak, while the broader FTLE
trough near $s\approx3\pi/2$ is associated with a weaker curvature response. 
As expected, the level-set error in \Cref{fig:LGR-vs-LS-t8} is most pronounced near the
strongly filamenting region at $s\approx\pi/2$, while a weaker but still visible
error is also observed near the weaker filamenting region at
$s\approx3\pi/2$.

To verify the stretch--curvature relation \eqref{stretch-curve-law},
\Cref{fig:rotating-vortex-scaling} displays the accumulated tangent
stretching $t\,\sftle(s,4)=\log\lftle(s,4)$ and logarithmic curvature
$\log\kappa(s,4)$ at the time of maximum deformation.
The first two panels plot these quantities against the
parameter $s$, while the third shows the phase plot in
the $(\log\kappa, \,t\sftle)$-plane. Two distinct branches are visible
in the phase plot, corresponding to the strong and weak filamenting
tips identified in \Cref{fig:rotating-vortex-LGR}. Both
branches exhibit the predicted $\kappa^{-1/3}$ scaling, with the
strongly filamenting tip producing the longer branch.

\begin{figure}[ht]
\centering
\begin{subfigure}[t]{0.28\linewidth}
  \centering
  \includegraphics[width=0.9\linewidth]{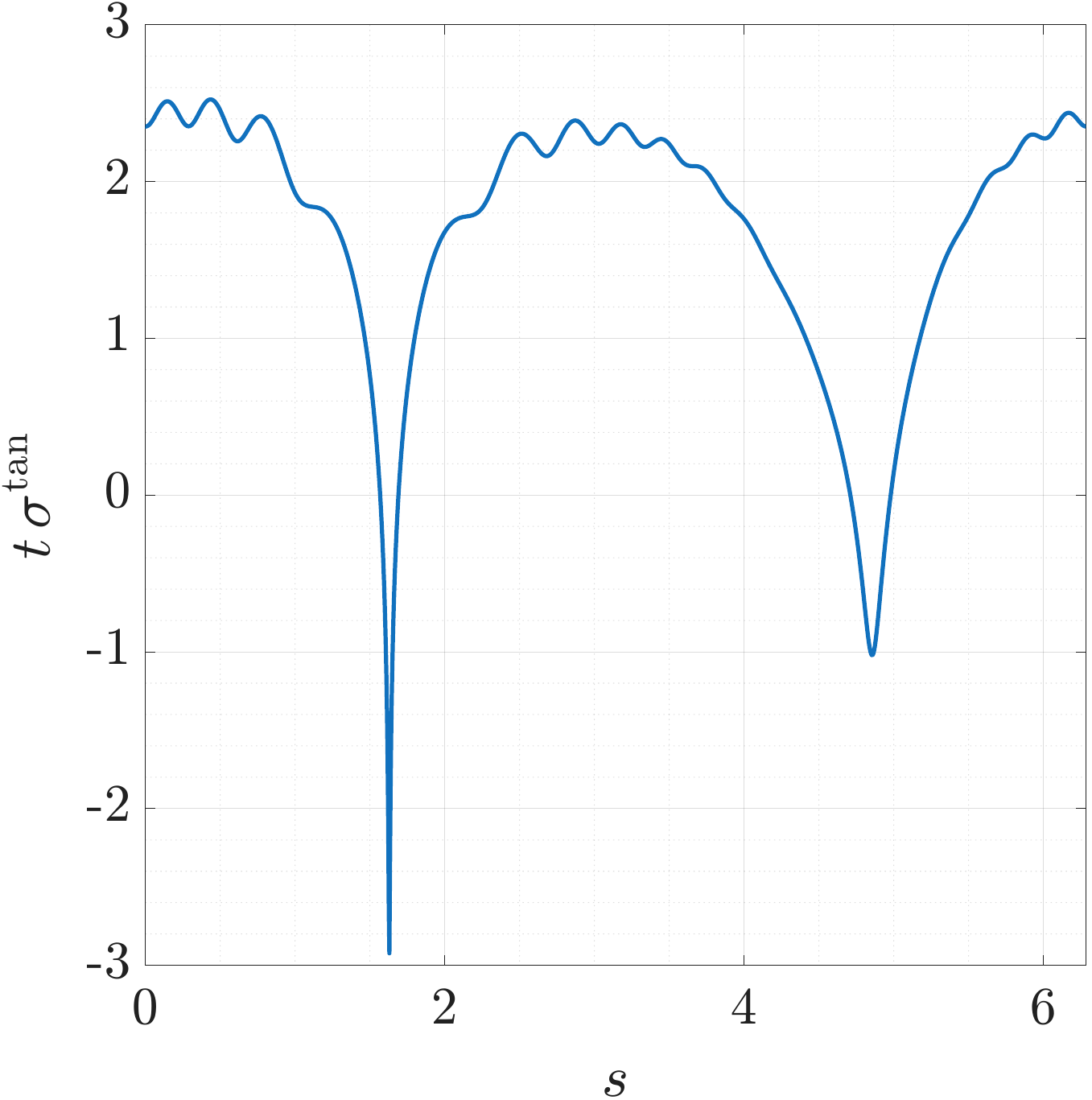}
  \caption{$(s,\,t\sftle(s,4))$}
  \label{fig:rotating-vortex-scaling-a}
\end{subfigure}
\hspace{1em}
\begin{subfigure}[t]{0.28\linewidth}
  \centering
  \includegraphics[width=0.9\linewidth]{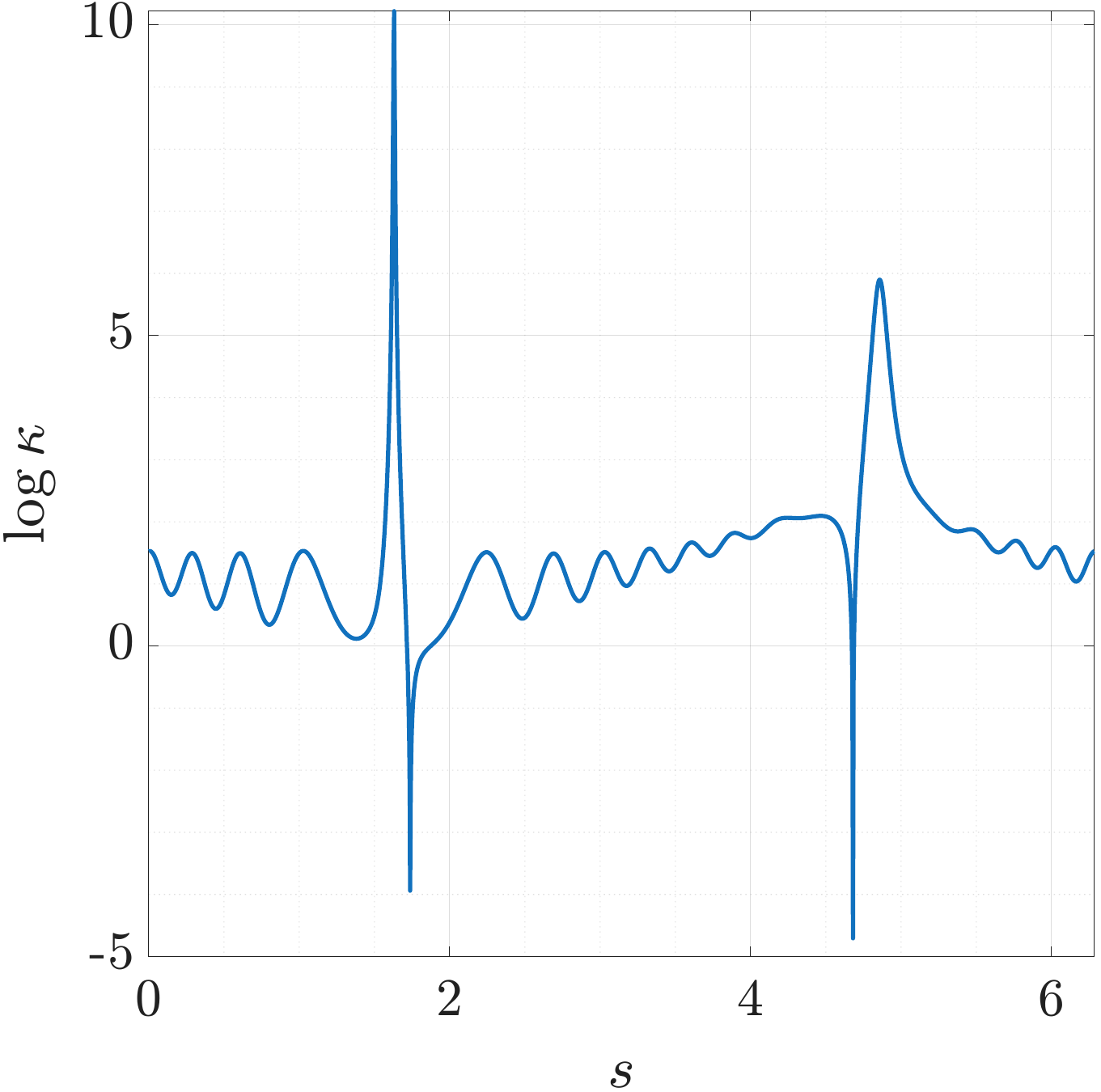}
  \caption{$(s, \,\log\kappa(s,4))$}
  \label{fig:rotating-vortex-scaling-b}
\end{subfigure}
\hspace{1em}
\begin{subfigure}[t]{0.28\linewidth}
  \centering
  \includegraphics[width=0.9\linewidth]{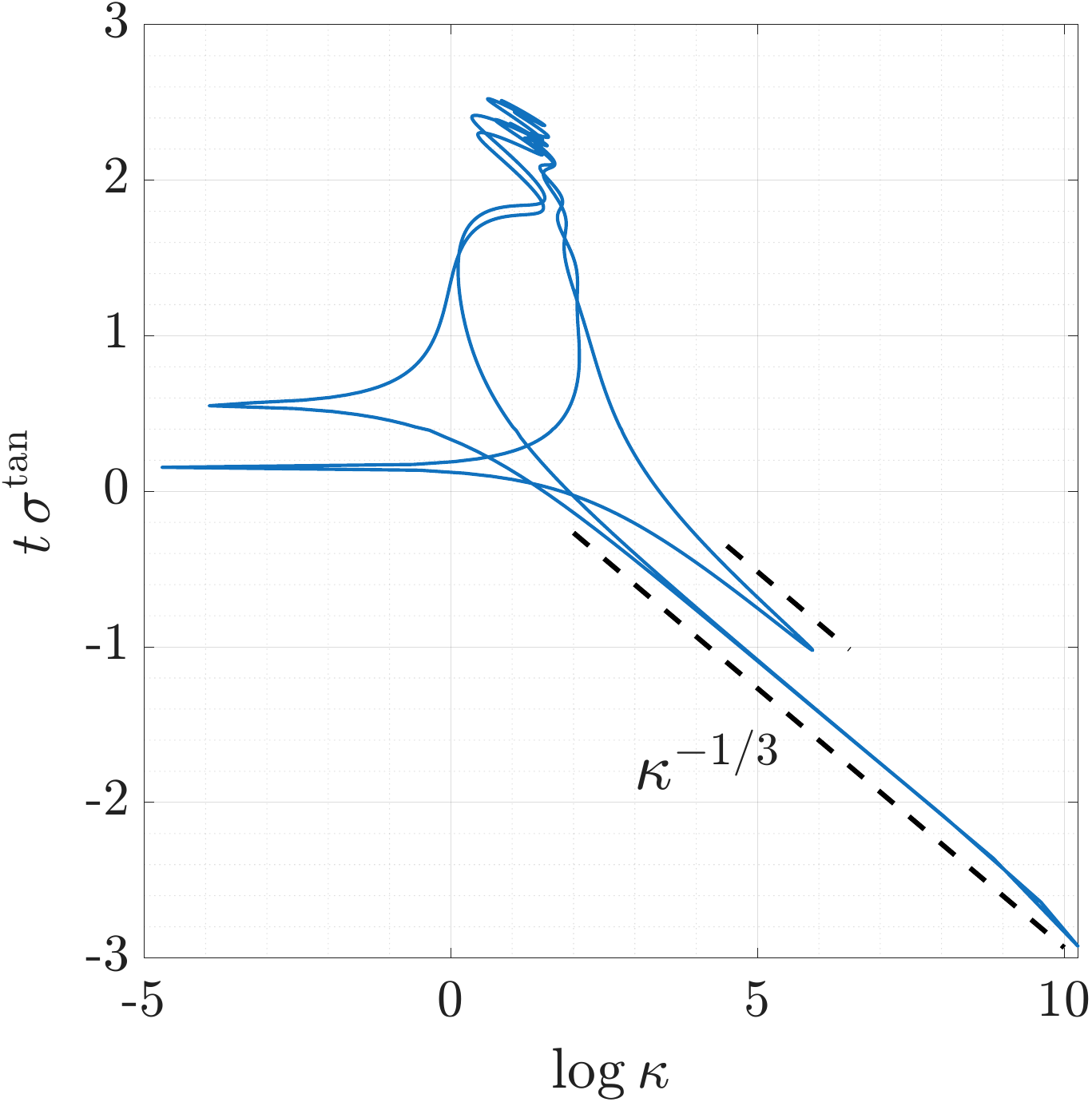}
  \caption{$(\log\kappa(s,4), \,t\sftle(s,4))$}
  \label{fig:rotating-vortex-scaling-c}
\end{subfigure}
\caption{
Numerical verification of the stretch--curvature relation
\eqref{stretch-curve-law} for the rotating-vortex benchmark at the
time of maximum deformation $t=4$.
The dashed lines in panel (c) denote the predicted
$\kappa^{-1/3}$ scaling at the filamenting tips.
}
\label{fig:rotating-vortex-scaling}
\end{figure}

\subsubsection{Convergence and scaling under $h$-refinement}

Finally, the accuracy and computational efficiency of 
the Lagrangian and level-set methods are compared 
quantitatively in \Cref{fig:LGR-vs-LS_diagnostics}. 
For the sequence of resolutions $h = 2^{-p} h_0$, we compute the diagnostics
\begin{equation}
\label{eq:area-errors}
|A_0 - A_*|,
\qquad
|A_K - A_0|,
\qquad
T_{\mathsf{cpu}},
\end{equation}
where $A_*$ denotes the exact initial area enclosed by the interface, $A_0$ the 
numerical initial area, $A_K$ the numerical area at the 
final time $t = \tmax$, and $T_{\mathsf{cpu}}$ the total 
computational runtime. The error $|A_K - A_0|$ will be referred to henceforth as the 
\emph{numerical reversal error}. 
 The results are displayed in 
\Cref{fig:LGR-vs-LS_diagnostics}. While both methods 
produce numerical reversal errors that converge with approximately  
second-order accuracy, the Lagrangian errors are 
three orders of magnitude smaller than the corresponding 
level-set errors at the same resolution. 
Similarly, although both methods exhibit the expected 
$\mathcal{O}(h^{-2})$ scaling in computational cost, the 
Lagrangian method is approximately two orders of magnitude 
faster over the range of resolutions considered here.

\begin{figure}[ht]
\centering
\begin{subfigure}[t]{0.28\linewidth}
  \centering
  \includegraphics[width=0.95\linewidth]{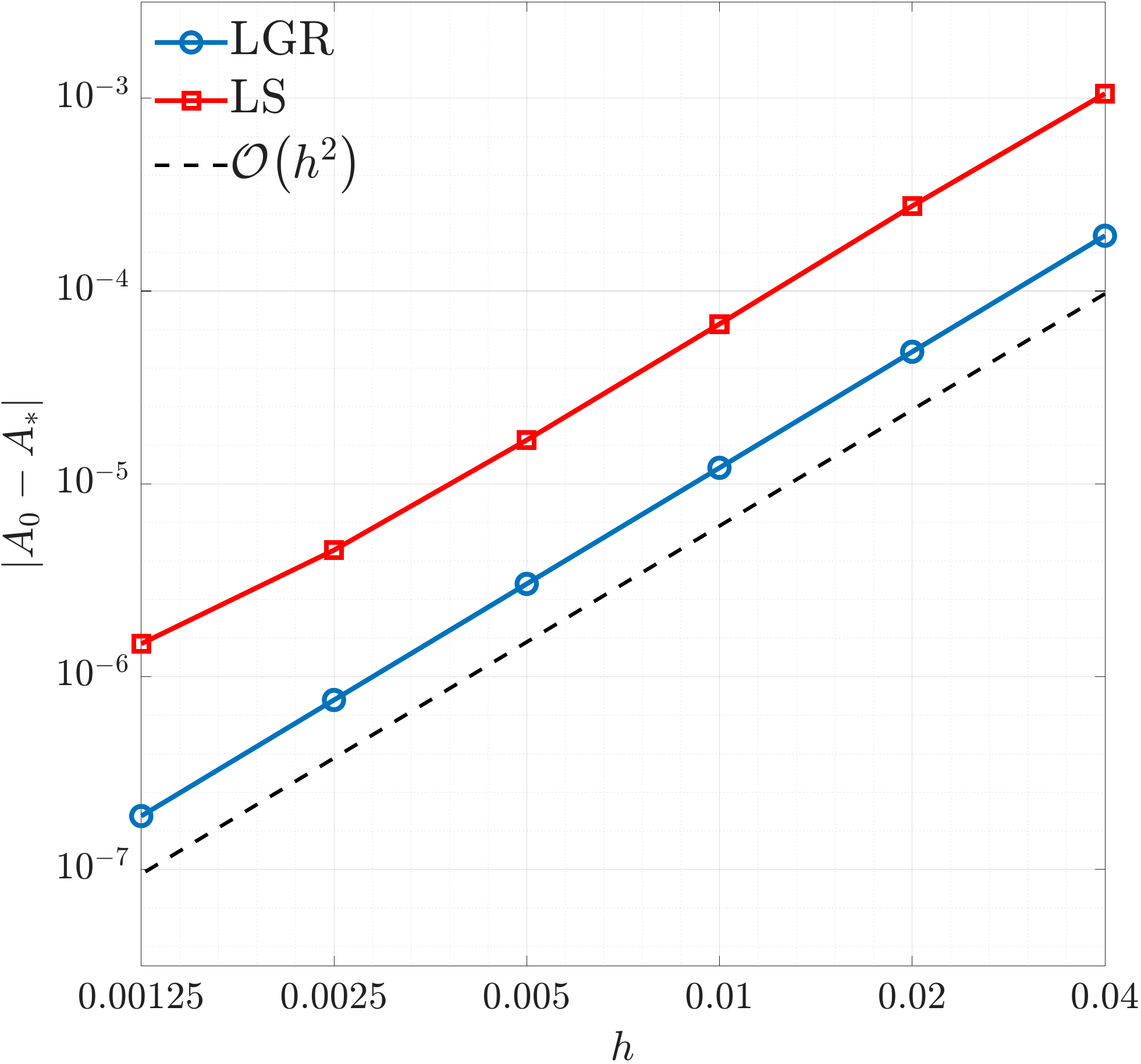}
  \caption{$|A_0 - A_*|$}
  \label{fig:LGR-vs-LS_initial-area-error}
\end{subfigure}
\hspace{1em}
\begin{subfigure}[t]{0.28\linewidth}
  \centering
  \includegraphics[width=0.95\linewidth]{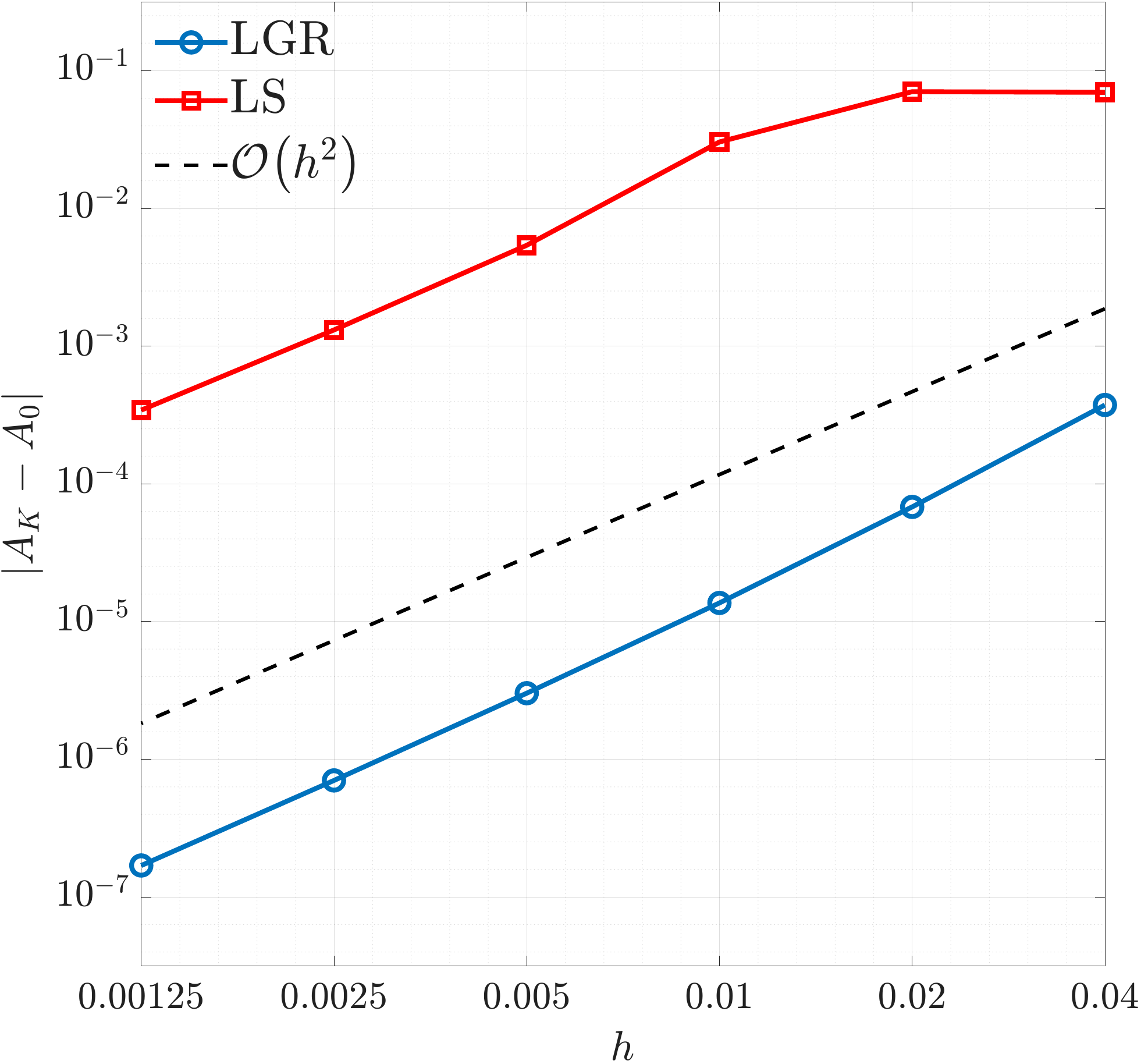}
  \caption{$|A_K - A_0|$}
  \label{fig:LGR-vs-LS_area-error}
\end{subfigure}
\hspace{1em}
\begin{subfigure}[t]{0.28\linewidth}
  \centering
  \includegraphics[width=0.95\linewidth]{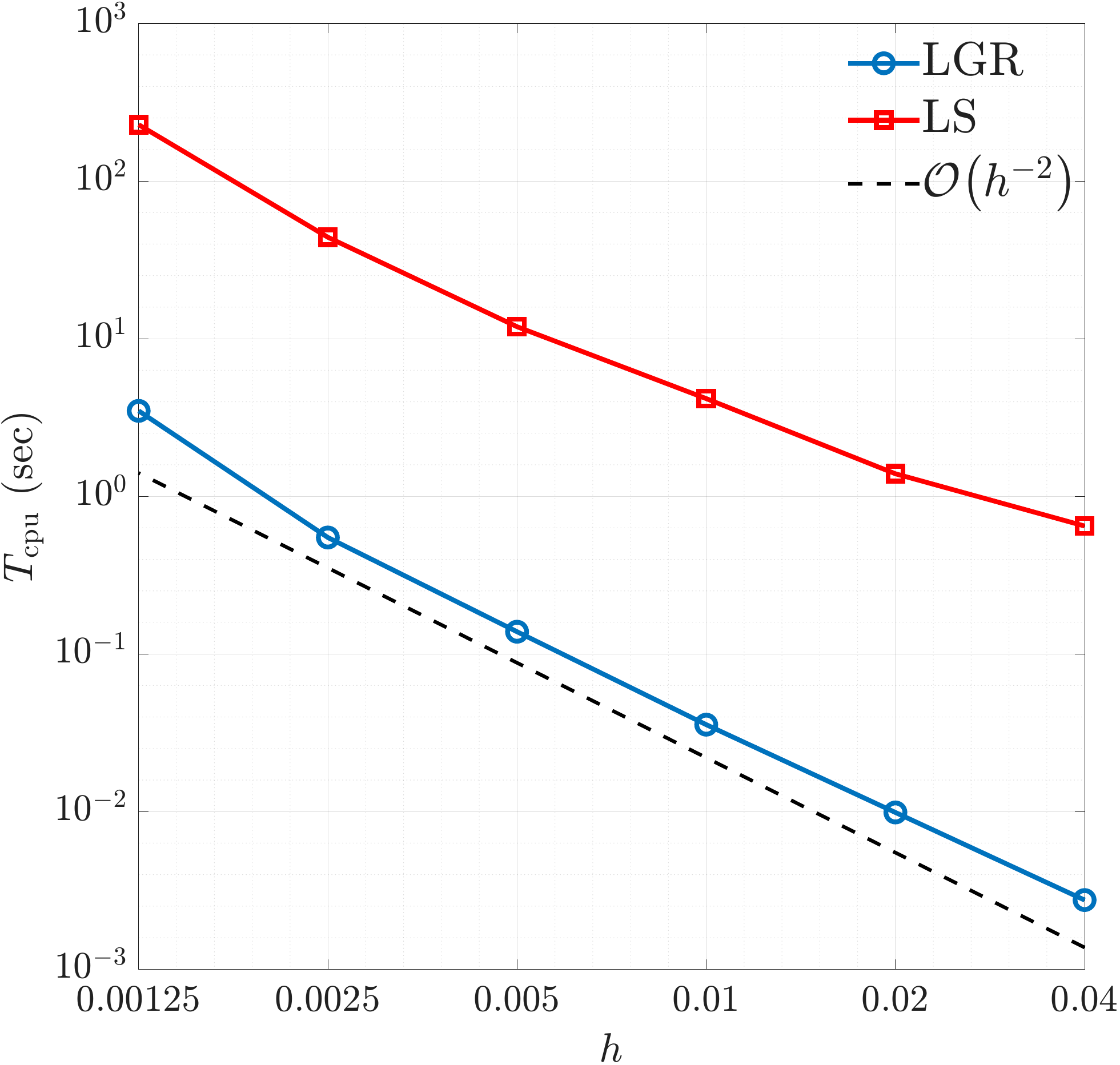}
  \caption{$T_{\mathsf{cpu}}$}
  \label{fig:LGR-vs-LS_runtime}
\end{subfigure}
\caption{
Quantitative comparison of the Lagrangian interface 
tracking and local level-set methods for the rotating 
vortex benchmark. The diagnostics \eqref{eq:area-errors} 
are displayed as functions of the 
characteristic length scale $h$. 
\textbf{Left:} Initial area error $|A_0 - A_*|$. 
\textbf{Center:} Numerical reversal error $|A_K - A_0|$ at 
$t = \tmax$. 
\textbf{Right:} Total computational runtime 
$T_{\mathsf{cpu}}$.
}
\label{fig:LGR-vs-LS_diagnostics}
\end{figure}

\subsection{Filamentation by alternating nonlinear shears}
\label{subsec:alternating-shear-LGR}

While the rotating-vortex benchmark provides a useful test of accuracy,
the results of \Cref{subsec:rotating-vortex} show that it does not
generate the extreme multiscale filamentary structures needed to fully
stress methods for tracking or capturing the interface. 
To address this limitation, \citet{AhSh2009} introduced
more demanding filamentation tests for assessing the ability of MOF
methods to preserve thin interfacial structures. 
These tests have subsequently been adopted and further developed in later studies
using filament-aware MOF methods \cite{JeSuSh2015,HePhXi2023,HeLiPhXi2024}. 
One such filamentation test, the $\mathcal{S}$-flow, will be considered later in \Cref{subsec:s-flow}.

In the same spirit, we consider a nonlinear \emph{alternating-shear} test that 
generates substantially greater interface complexity than the filamentation tests 
comprising the standard MOF test suite. 
In the chaotic-advection literature, alternating-shear maps have been used 
extensively to study the generation of fine-scale structure through repeated stretching and folding. 
\citet{Pierrehumbert1991,Pierrehumbert1994}, for instance, employed randomized 
alternating-shear deformations to investigate tracer microstructure generated by 
chaotic mixing in Rossby--Haurwitz wave dynamics, where filamentation is ultimately arrested by diffusion.
Here, we study the analogous problem for material interfaces, replacing the 
Batchelor-scale diffusive regularization of tracer filaments by the microscale topological 
regularization of tracked interfaces. To this end, we introduce a deterministic and 
reversible alternating-shear benchmark that retains the strong filamentation of 
chaotic-advection models while providing a reproducible setting for numerical verification.

For each integer $m\ge0$, we define the
$m$-th alternating-shear map by the composition
\begin{subequations}\label{eq:alternating-shear}
\begin{equation}
F_m
=
S_{m,2} \circ S_{m,1},
\end{equation}
where $S_{m,1}$ and $S_{m,2}$ are the horizontal and vertical shears
\begin{align}
S_{m,1}(x_1,x_2)
&=
\left(
x_1 + \lambda \sin(kx_2+\varphi_m),
\,
x_2
\right), \\
S_{m,2}(x_1,x_2)
&=
\left(
x_1,
\,
x_2 + \lambda\sin(kx_1+\psi_m)
\right),
\end{align}
the phase shifts $\{ \varphi_m, \psi_m \}$ are given by
\begin{equation}\label{phase-shifts}
\varphi_m 
=
2\pi
\left(
\sqrt{2}(m+1)
+
0.173(m+1)^2
\right) 
\qquad 
\text{and} 
\qquad 
\psi_m
=
2\pi
\left(
\sqrt{3}(m+1)
+
0.217(m+1)^2
\right), 
\end{equation}
and $\{ \lambda, k \}$ are the amplitude and frequency of the sinusoidal deformations. We use the values
\begin{equation}
\lambda = 0.2 \qquad \text{and} \qquad k = 4
\end{equation}
in the numerical experiment below.  
The inverse map is obtained by undoing the two shears:
\begin{equation}\label{eq:alternating-shear-inverse}
F_m^{-1}
=
S_{m,1}^{-1}
\circ
S_{m,2}^{-1},
\end{equation}
where
\begin{align}
S_{m,1}^{-1}(x_1,x_2)
&=
\left(
x_1 - \lambda \sin(kx_2+\varphi_m),
\,
x_2
\right), \\
S_{m,2}^{-1}(x_1,x_2)
&=
\left(
x_1,
\,
x_2 - \lambda\sin(kx_1+\psi_m)
\right).
\end{align}
\end{subequations}

\subsubsection{$C^1$-in-time embedding of the discrete dynamics}

For Lagrangian tracking \eqref{eq:lgr-evolution}, we consider a
long-time simulation on the interval $[0,\tmax]$ with $\tmax=30$,
generated by a smooth divergence-free velocity field which coincides with 
the alternating-shear maps at integer times. The flow consists
of a forward deformation phase on $0 \le t \le \tmax/2$ followed by a
reversal phase on $\tmax/2 \le t \le \tmax$, allowing the final
configuration to be compared directly with the initial interface.

For $t\in[m,m+1]$, let
\[
\tau=t-m,
\qquad
0\le\tau\le1,
\]
and let $f,g:[0,1]\to\mathbb R$ denote compactly supported temporal
weights satisfying
\begin{subequations}\label{eq:alternating-shear-weights}
\begin{equation}
f(\tau)=0,
\quad
\tfrac12\le\tau\le1,
\qquad
\text{and}
\qquad
g(\tau)=0,
\quad
0\le\tau\le\tfrac12,
\end{equation}
with the normalization conditions
\begin{equation}
\int_0^{\frac12} f(\tau)\,\mathrm d\tau = 1
\qquad
\text{and}
\qquad 
\int_{\frac12}^{1} g(\tau)\,\mathrm d\tau = 1.
\end{equation}
The minimal-degree $C^1$ polynomials satisfying these conditions are
\begin{equation}
f(\tau)
=
\begin{cases}
960\,\tau^2\left(\frac12-\tau\right)^2,
& 0 \le \tau \le \frac12, \\[1mm]
0,
& \frac12 < \tau \le 1,
\end{cases}
\qquad
\text{and}
\qquad
g(\tau)
=
\begin{cases}
0,
& 0 \le \tau < \frac12, \\[1mm]
960\,(\tau-\tfrac12)^2(1-\tau)^2,
& \frac12 \le \tau \le 1.
\end{cases}
\end{equation}
\end{subequations}

For $0 \le t < \tmax/2$, let
\begin{subequations}\label{eq:alternating-shear-velocity}
\begin{equation}
m=\lfloor t\rfloor,
\qquad
\tau_f=t-m,
\end{equation}
while for $\tmax/2 \le t \le \tmax$, let
\begin{equation}
q
=
\left\lfloor
t-\tfrac{\tmax}{2}
\right\rfloor,
\qquad
m
=
\tfrac{\tmax}{2}-1-q,
\qquad
\tau_r
=
t-\tfrac{\tmax}{2}-q.
\end{equation}
We then define the velocity field
\begin{equation}
u(x,t)
=
\begin{cases}
\begin{pmatrix}
\lambda f(\tau_f)\sin(kx_2+\varphi_m) \\[1mm]
\lambda g(\tau_f)\sin(kx_1+\psi_m)
\end{pmatrix},
&
0 \le t < \tmax/2,
\\[2em]
\begin{pmatrix}
-\lambda g(\tau_r)\sin(kx_2+\varphi_m) \\[1mm]
-\lambda f(\tau_r)\sin(kx_1+\psi_m)
\end{pmatrix},
&
\tmax/2 \le t \le \tmax.
\end{cases}
\end{equation}
\end{subequations}
During the forward phase, the first half of each interval applies the
horizontal shear $S_{m,1}$, while the second half applies the vertical
shear $S_{m,2}$. Consequently, the flow coincides at integer times with
the compositions
\begin{subequations}\label{shear-identity}
\begin{equation}
\gamma_\alpha(s,m)
=
(F_{m-1}\circ\cdots\circ F_0)
\bigl(\gamma_\alpha(s,0)\bigr),
\qquad
m=1,\ldots,\tmax/2.
\end{equation}
During reversal, the first half of each interval applies
$S_{m,2}^{-1}$, while the second half applies $S_{m,1}^{-1}$, so that
\begin{equation}
\gamma_\alpha(s,\tmax)
=
\gamma_\alpha(s,0).
\end{equation}
\end{subequations}
The velocity gradient required for the tangent-FTLE computation
\eqref{eq:xi-evolution} is
\begin{equation}
Du(x,t)
=
\begin{cases}
\begin{pmatrix}
0
&
\lambda k f(\tau_f)\cos(kx_2+\varphi_m)
\\[1mm]
\lambda k g(\tau_f)\cos(kx_1+\psi_m)
&
0
\end{pmatrix},
&
0 \le t < \tmax/2,
\\[2em]
\begin{pmatrix}
0
&
-\lambda k g(\tau_r)\cos(kx_2+\varphi_m)
\\[1mm]
-\lambda k f(\tau_r)\cos(kx_1+\psi_m)
&
0
\end{pmatrix},
&
\tmax/2 \le t \le \tmax.
\end{cases}
\end{equation}

\subsubsection{Qualitative and quantitative analysis of filamentation}

We compute a sequence of simulations using the Lagrangian tracking
algorithm with
\begin{equation}\label{eq:SF-params}
\tmax = 30,
\quad
h = 2^{-p} h_0,
\quad
\Delta t = 2^{-p}\Delta t_0,
\quad
p = 0,\ldots,5,
\qquad
\text{with}
\quad 
h_0 = 0.06,
\quad
\Delta t_0 = \num{8.0e-3}.
\end{equation}
The solution computed with the intermediate resolution $h=h_0/8$ is shown in
\Cref{fig:alternating-shear-zoom} at the time of maximal deformation,
$t=15$. The repeated alternating-shear mechanism generates thin filamentary 
structures spanning a broad range of scales. In light of the results of 
\Cref{subsec:rotating-vortex}, obtaining a solution of acceptable 
quality with the baseline level-set method would be prohibitively 
expensive on the current computational platform.

The tangent FTLE \eqref{ftle} at $t=15$ is displayed as the blue curve
in \Cref{fig:alternating-shear-FTLE}. Compared with the
rotating vortex benchmark, the alternating-shear
flow exhibits significantly more intricate multiscale structure,
reflecting the repeated generation and interaction of filamentary
regions throughout the evolution. For comparison, the curvature
$\kappa(s,\tmax)$ of the final interface is displayed as the red
curve. The peaks of the curvature align closely with sharp troughs of
the tangent FTLE at the material locations
\begin{equation}\label{zoom-window-init}
(\theta_1,\theta_2,\theta_3,\theta_4)
=
(2.876025,\,
3.451420,\,
4.324865,\,
4.812920).
\end{equation}

\begin{figure}[ht]
\centering
\begin{subfigure}[t]{0.6\linewidth}
  \centering
  \includegraphics[width=\linewidth]{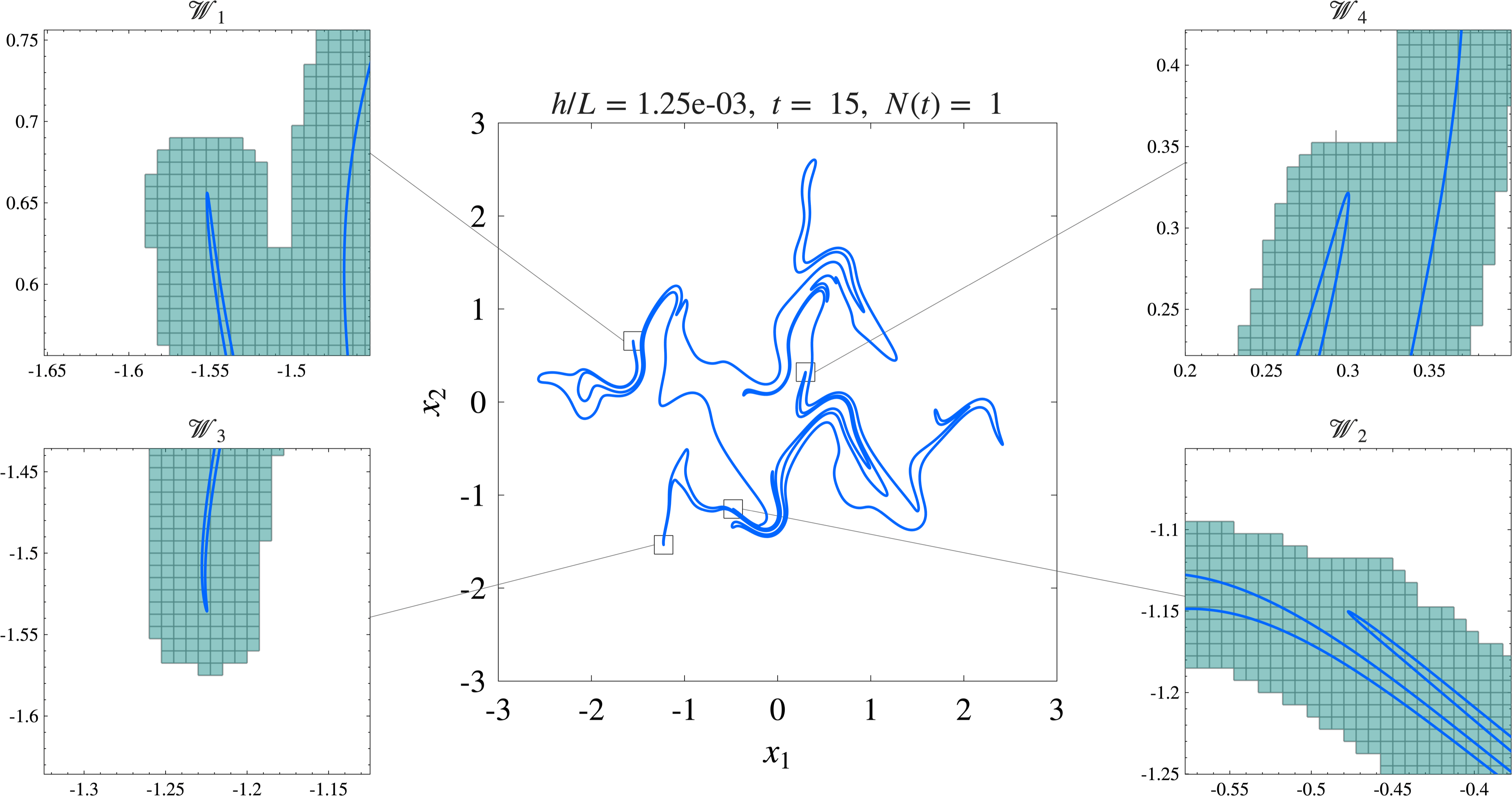}
  \caption{Alternating-shear filamentation at $t=15$}
  \label{fig:alternating-shear-zoom}
\end{subfigure}
\hspace{0.5em}
\begin{subfigure}[t]{0.35\linewidth}
  \centering
  \includegraphics[width=0.9925\linewidth]{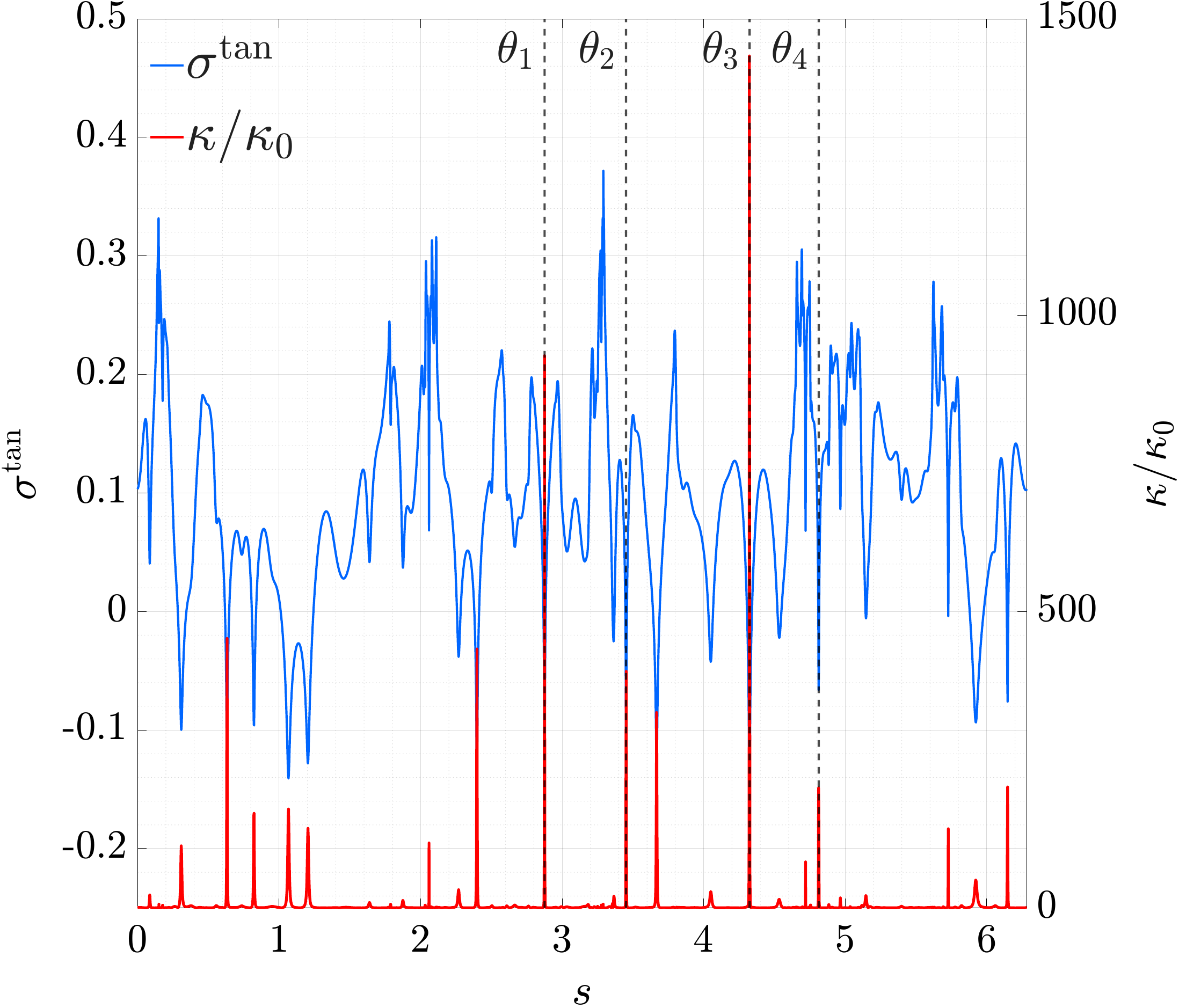}
  \caption{Tangent FTLE and curvature}
  \label{fig:alternating-shear-FTLE}
\end{subfigure}
\caption{
Filamentation generated by repeated application of the alternating-shear
deformation \eqref{eq:alternating-shear}. Results are shown for the
classical Lagrangian tracking method with parameters
\eqref{eq:SF-params}, using the intermediate resolution $h=h_0/8$. 
\textbf{Left:} Adaptively refined Lagrangian solution at the time of maximal
deformation, $t=15$, together with a sequence of localized Lagrangian 
observation windows
illustrating the formation of thin filamentary structures.
For reference, each window also displays a local Cartesian 
grid patch at the prescribed microscale resolution $h$.
The repeated shearing mechanism generates thin filaments 
with sub-microscale geometry which are captured by the Lagrangian solution with high fidelity.
A video of the complete evolution is available at \cite{Ramani2026}.
\textbf{Right:} Tangential FTLE \eqref{ftle} (blue), together with the curvature of the 
final deformed interface (red). 
The marked locations $\theta_1,\ldots,\theta_4$ correspond to sharp 
localized troughs of the tangent FTLE, and are used to track the 
filament-forming regions shown in the left panel.
}
\label{fig:alternating-shear}
\end{figure}

To quantify the accuracy and computational cost of the classical 
Lagrangian tracking method, we consider the recovery of the initial 
interface at the final time $t=\tmax$. The left panel of \Cref{fig:alternating-shear-LGR_diagnostics} 
compares the initial and final interfaces for the intermediate resolution $h=h_0/8$. 
The two interfaces are visually indistinguishable, indicating that the method maintains excellent accuracy 
despite the severe deformation observed at intermediate times. The center and right panels of
\Cref{fig:alternating-shear-LGR_diagnostics} provide quantitative
diagnostics. The numerical reversal error $|A_K-A_0|$ exhibits
second-order convergence under mesh refinement, while the computational
runtime scales as $\mathcal O(h^{-2})$.

\begin{figure}[ht]
\centering
\begin{subfigure}[t]{0.25\linewidth}
  \centering
  \includegraphics[width=0.95\linewidth]{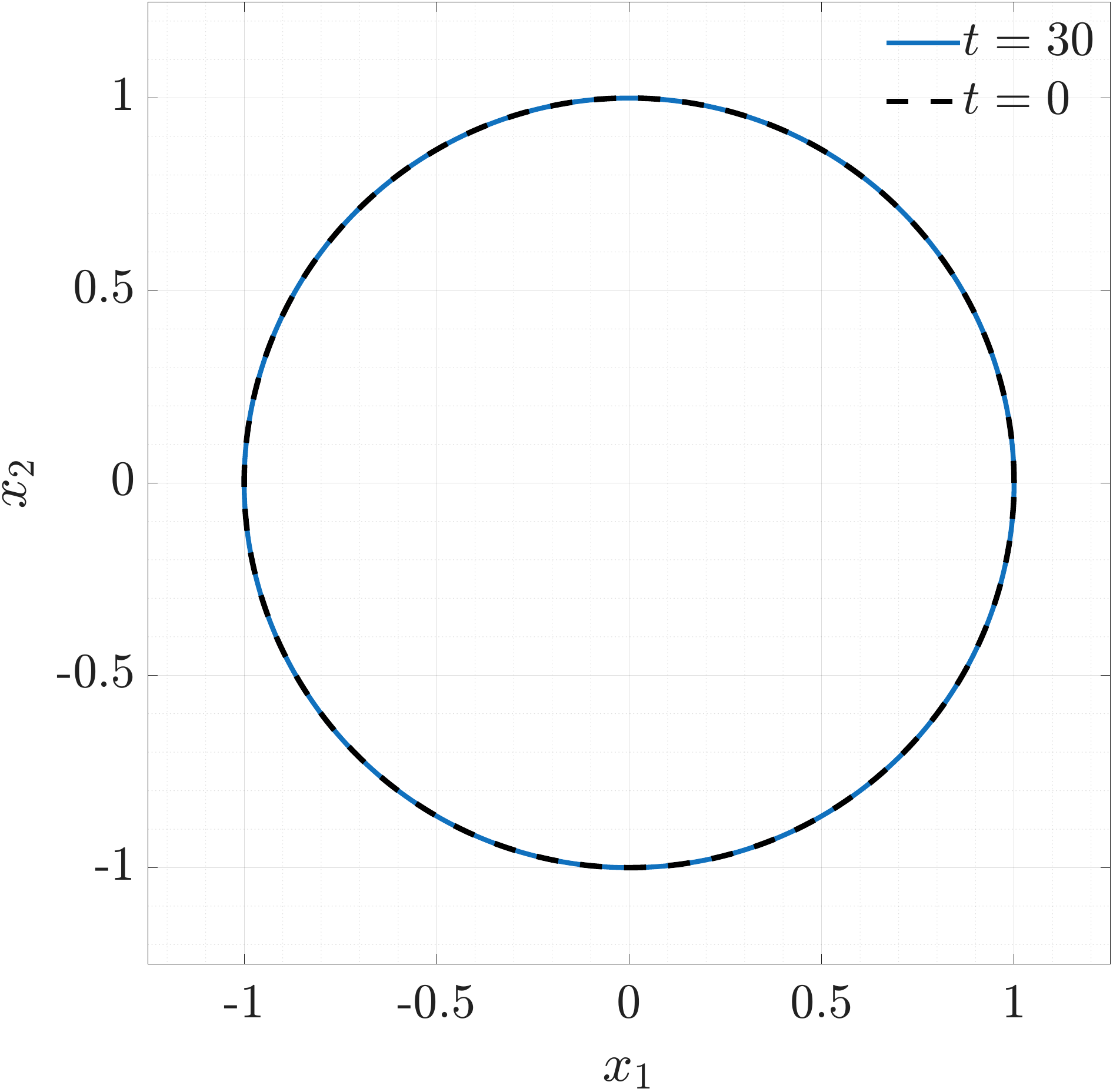}
  \caption{$\Gamma(s,t)$ at $t = 0$ and $t=\tmax$}
  \label{fig:alternating-shear-LGR_diagnostics-final-interface}
\end{subfigure}
\hspace{1em}
\begin{subfigure}[t]{0.25\linewidth}
  \centering
  \includegraphics[width=\linewidth]{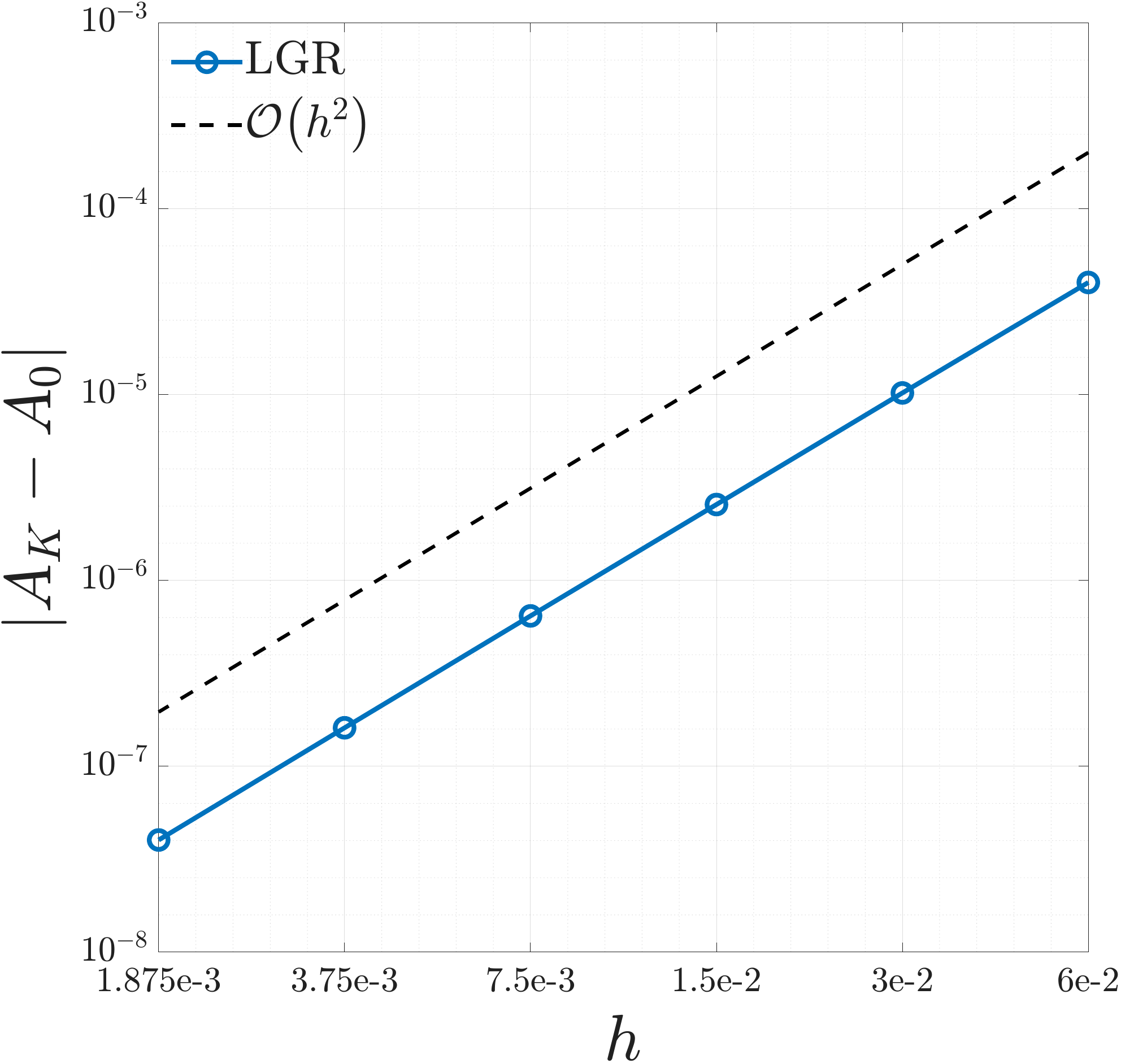}
  \caption{$|A_K-A_0|$}
  \label{fig:alternating-shear-LGR_diagnostics-error}
\end{subfigure}
\hspace{1em}
\begin{subfigure}[t]{0.25\linewidth}
  \centering
  \includegraphics[width=\linewidth]{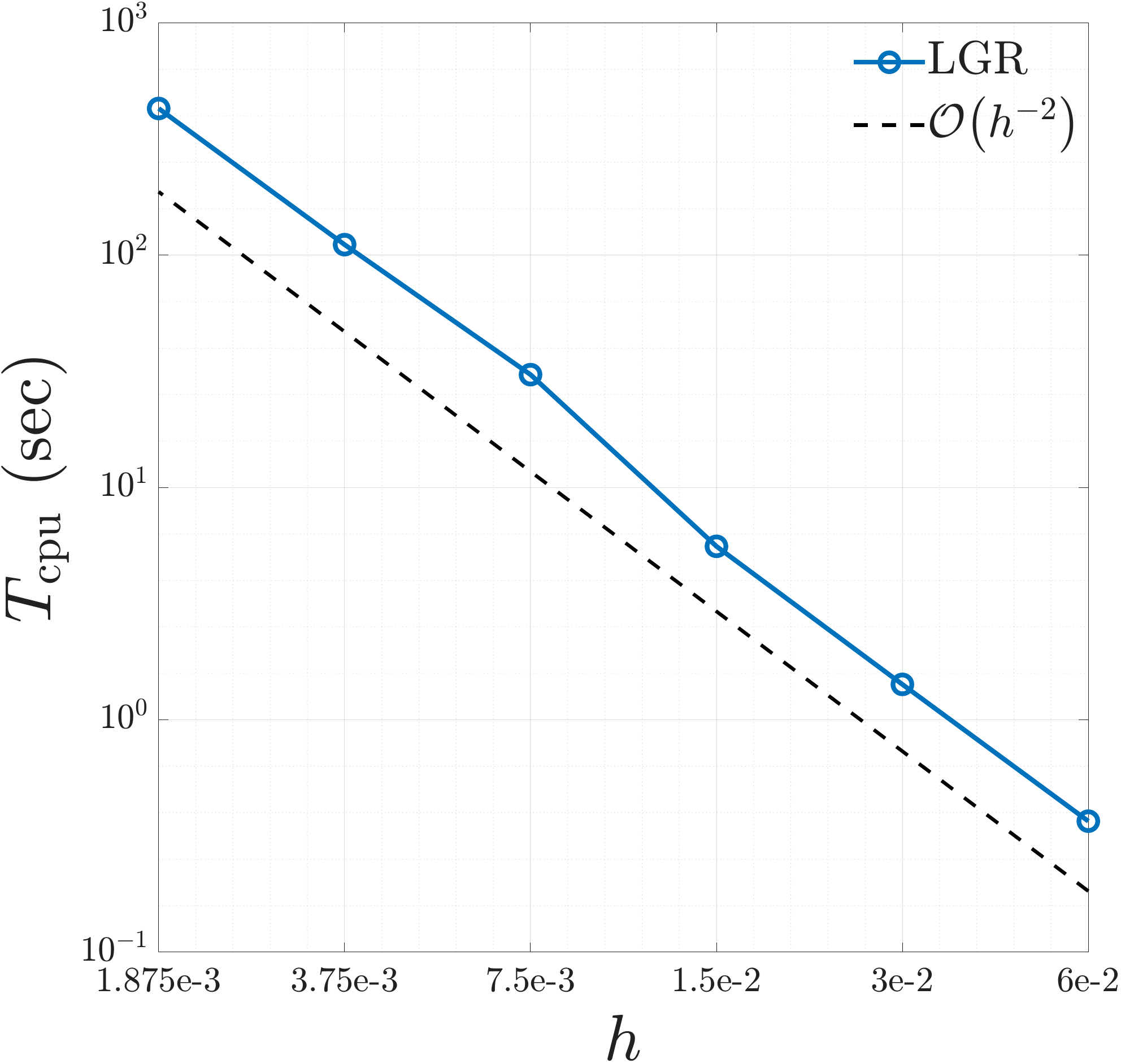}
  \caption{$T_{\mathsf{cpu}}$}
  \label{fig:alternating-shear-LGR_diagnostics-runtime}
\end{subfigure}
\caption{
Convergence study for the nonlinear alternating-shear test using the
classical Lagrangian tracking method and the parameter sequence
\eqref{eq:SF-params}.
\textbf{Left:}
Initial ($t=0$) and final ($t=\tmax$) interfaces for the intermediate
resolution $h=h_0/8$.
\textbf{Center:}
Numerical reversal error $|A_K-A_0|$ as a function of the fine-scale
resolution $h$.
\textbf{Right:}
Computational runtime as a function of $h$.
}
\label{fig:alternating-shear-LGR_diagnostics}
\end{figure}

We now return to the filament-forming regions highlighted in
\Cref{fig:alternating-shear}. In agreement with the filamentation
studies of \citet{Pierrehumbert1991,Pierrehumbert1994}, the underlying
interface exhibits exponential length growth,
$|\Gamma(t)| \sim |\Gamma(0)|e^{0.146t}$. Since the enclosed area is
conserved, this growth implies a corresponding decrease in
the characteristic transverse scale of the filamentary structure. To
track this process qualitatively, we use the Lagrangian observation windows
\eqref{eq:window-centers}, with
\begin{equation}\label{zoom-window-params}
N_{\mathsf{win}} = 4,
\qquad
w_j = 0.1,
\qquad
c_j(0) = (\cos\theta_j,\sin\theta_j), 
\qquad 
\text{and } \theta_j \text{ given by \eqref{zoom-window-init}}. 
\end{equation}
From the localized windows in \Cref{fig:alternating-shear-zoom},
we observe that the filaments in the
$\mathcal{W}_1$, $\mathcal{W}_2$, and $\mathcal{W}_3$ windows fall below
the prescribed fine scale $h$, while the filament in the $\mathcal{W}_4$
window approaches this threshold.
The most severe filamentation occurs in $\mathcal{W}_3$, followed by
$\mathcal{W}_1$; both are associated with pronounced tangent-FTLE troughs
and large curvature concentrations in \Cref{fig:alternating-shear-FTLE}.
By contrast, the $\mathcal{W}_2$ and $\mathcal{W}_4$ windows exhibit weaker
filamentation, consistent with their shallower tangent-FTLE troughs and
smaller curvature concentrations. The corresponding phase-space plots in
\Cref{fig:alternating-shear-LGR_filamentation} show that the filament-tip
regions in all four windows are consistent with the
$\kappa^{-1/3}$ stretch--curvature scaling \eqref{stretch-curve-law}.

\begin{figure}[ht]
\centering
\begin{subfigure}[t]{0.24\linewidth}
  \centering
  \includegraphics[width=0.95\linewidth]{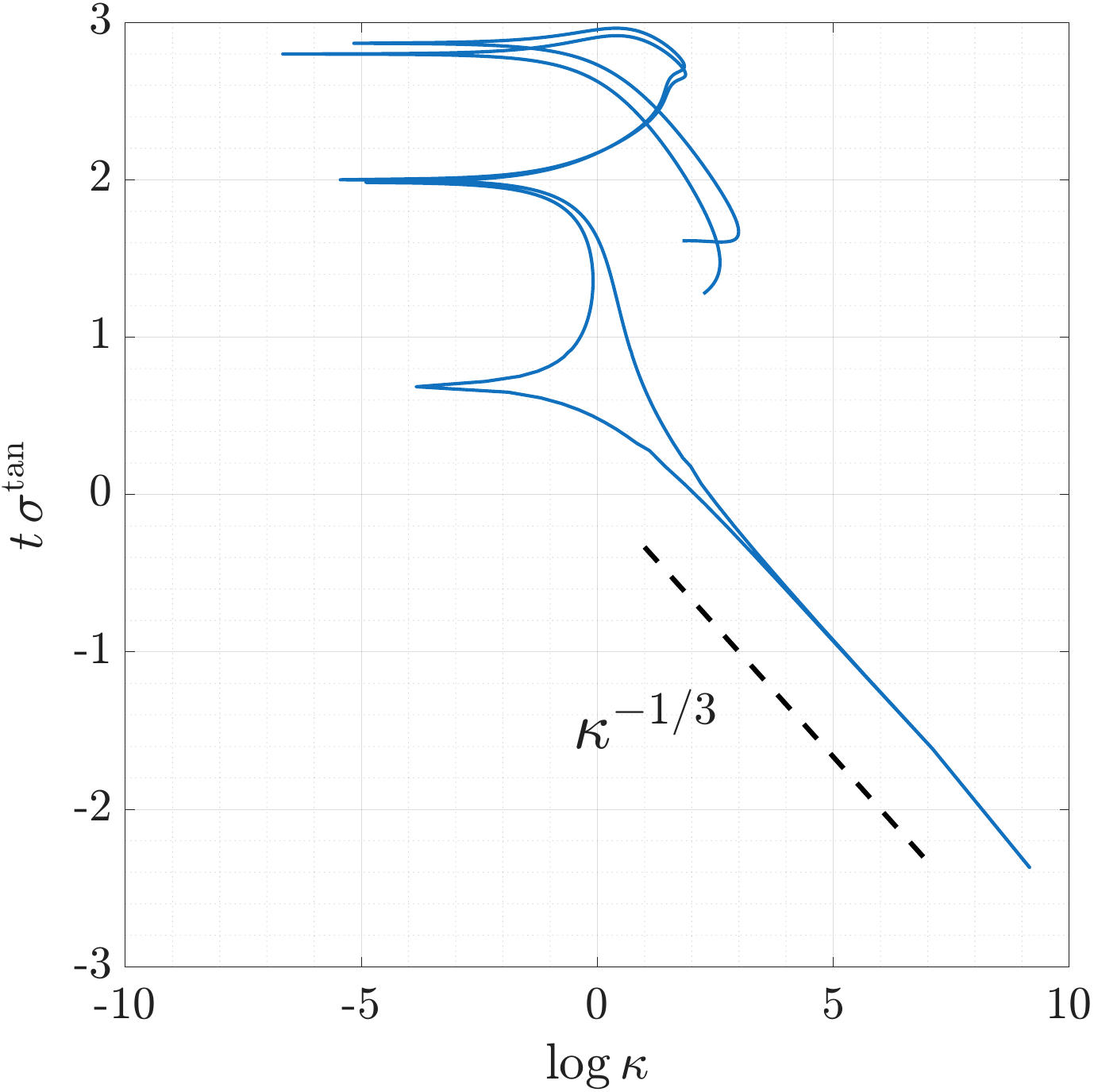}
  \caption{$\mathcal{W}_1$}
  \label{fig:alternating-shear-LGR_filamentation_scaling1}
\end{subfigure}
\begin{subfigure}[t]{0.24\linewidth}
  \centering
  \includegraphics[width=0.95\linewidth]{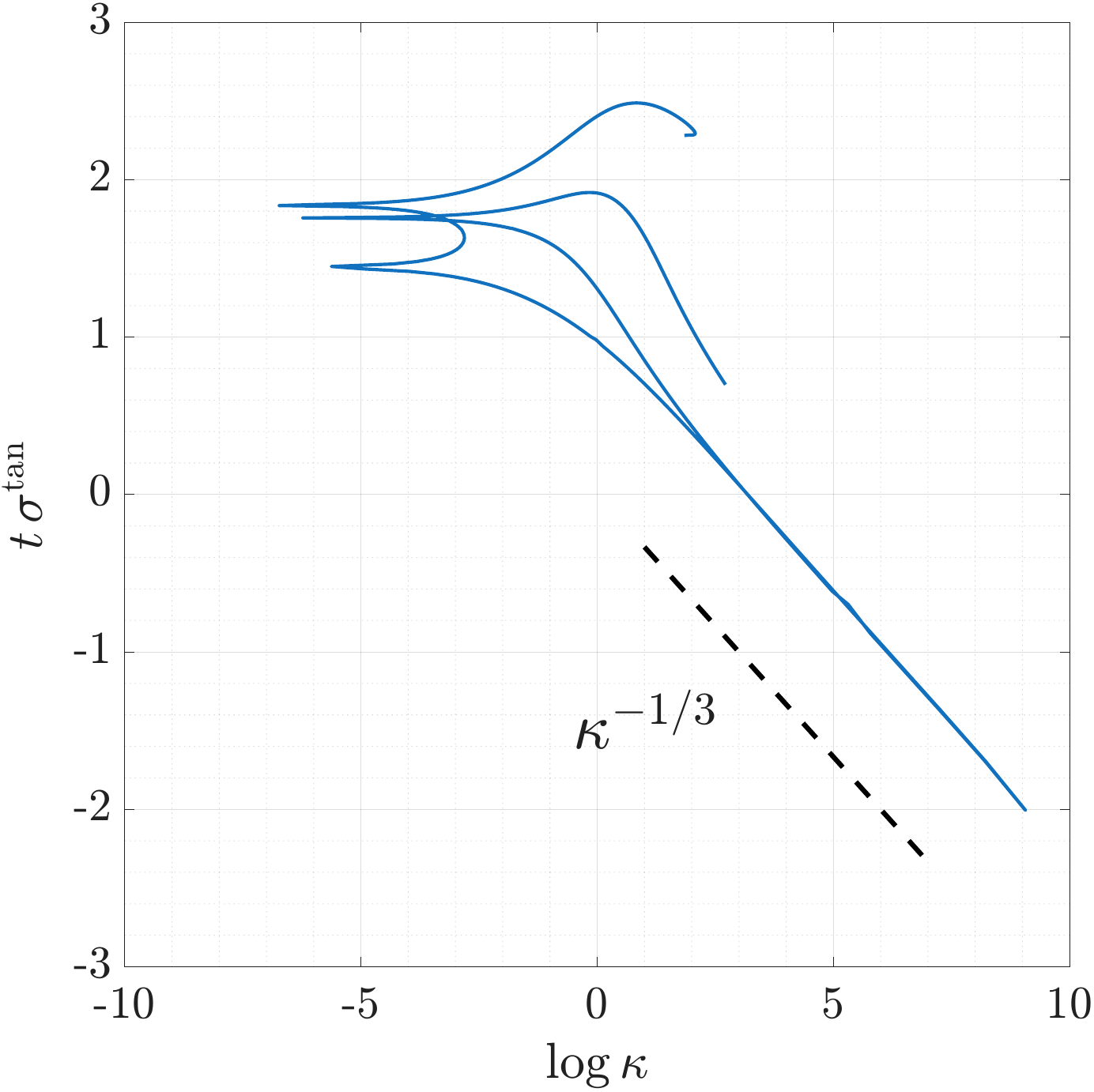}
  \caption{$\mathcal{W}_2$}
  \label{fig:alternating-shear-LGR_filamentation_scaling2}
\end{subfigure}
\begin{subfigure}[t]{0.24\linewidth}
  \centering
  \includegraphics[width=0.95\linewidth]{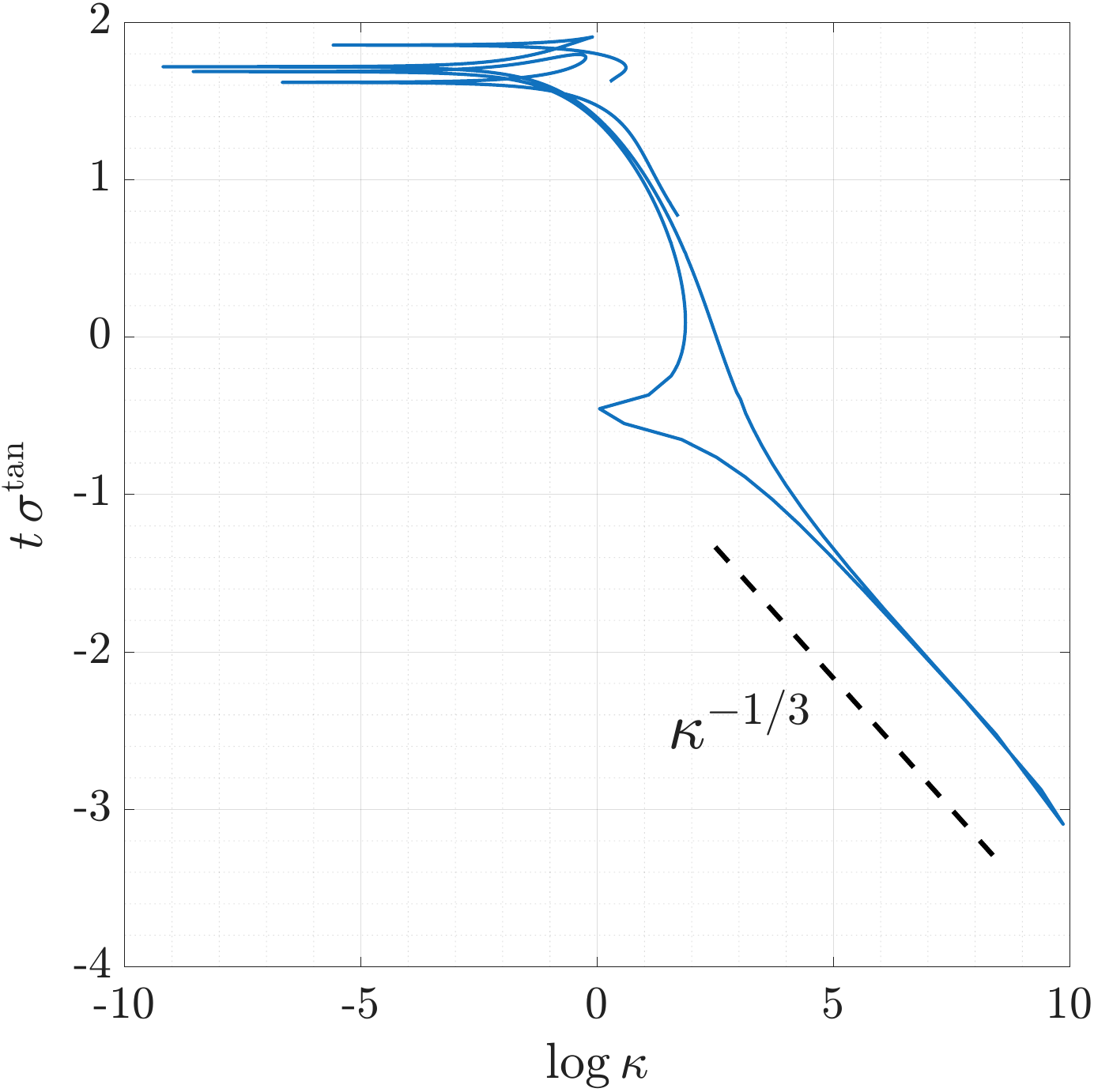}
  \caption{$\mathcal{W}_3$}
  \label{fig:alternating-shear-LGR_filamentation_scaling3}
\end{subfigure}
\begin{subfigure}[t]{0.24\linewidth}
  \centering
  \includegraphics[width=0.95\linewidth]{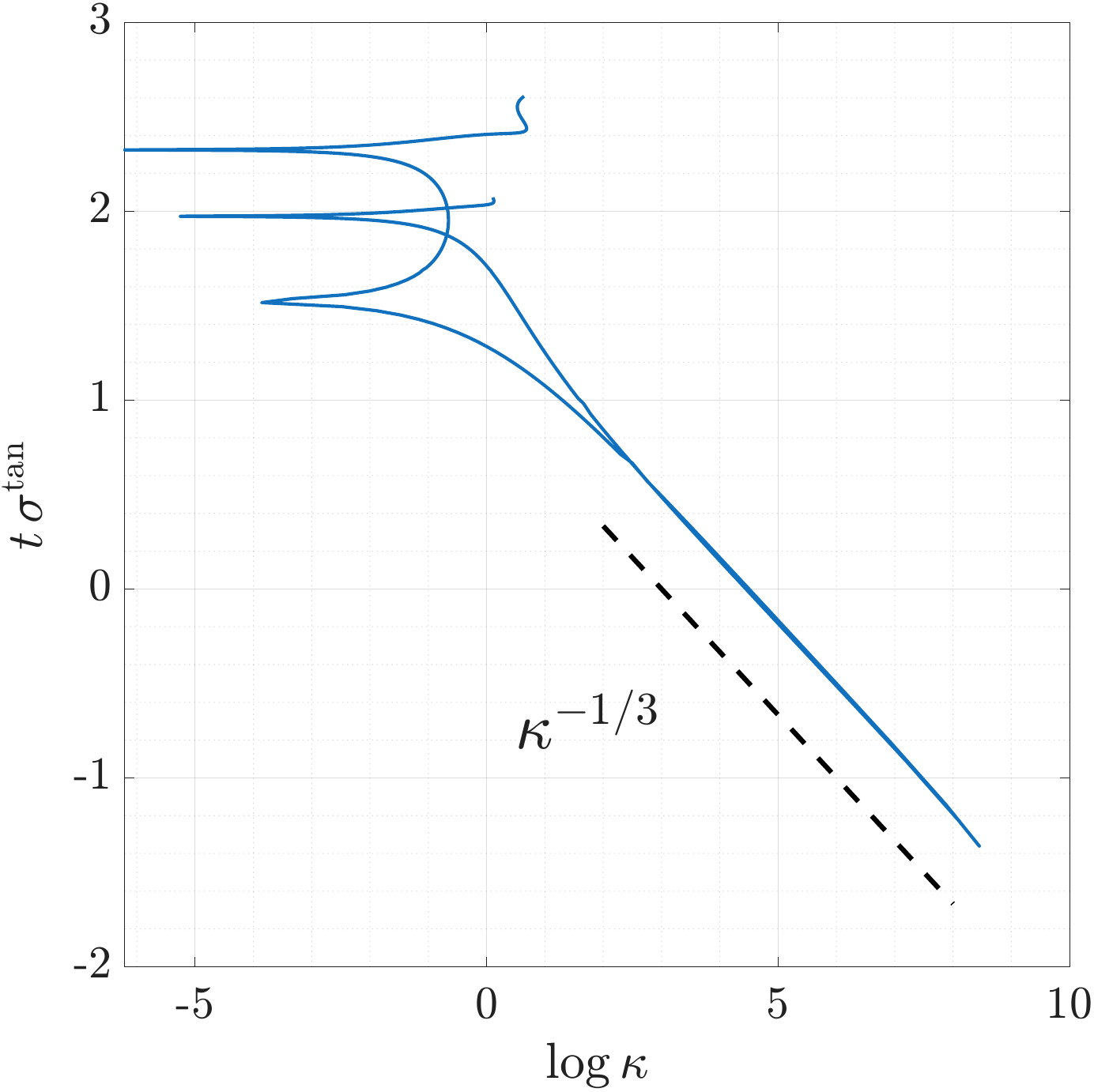}
  \caption{$\mathcal{W}_4$}
  \label{fig:alternating-shear-LGR_filamentation_scaling4}
\end{subfigure}
\caption{
Accumulated tangent stretching $t \sftle$ and
logarithmic curvature $\log\kappa$ in the observation
windows $\mathcal{W}_1,\ldots,\mathcal{W}_4$ for the nonlinear
alternating-shear test. The dashed line denotes the
$\kappa^{-1/3}$ scaling \eqref{stretch-curve-law}.
}
\label{fig:alternating-shear-LGR_filamentation}
\end{figure}

For this problem, a natural large-scale deformation length is
given by the diameter of the box containing the interface at final time
$t=\tmax$, which we approximate here by $L=6$. Consequently, the
scale-separation ratio for the intermediate resolution $h = h_0 / 8$ is $h/L=1.25\times10^{-3}$.
The excellent geometric fidelity and computational efficiency of classical Lagrangian 
tracking are accompanied by a fundamental topological rigidity of the 
representation, forcing thin filaments to persist even when they fall below microscopic scales. 
The \mts\ framework developed in the following sections 
overcomes this limitation by introducing controlled topological 
reconstruction at the microscale resolution.

\section{Capturing topological transitions by an Eulerian snapshot}
\label{sec:extraction}

In this section, we describe the first stage 
of the core \mts\ algorithm: \emph{extracting} a topologically-corrected 
interface family from an \emph{Eulerian snapshot}. 
Several front-tracking methods have employed auxiliary Eulerian
representations to detect or resolve topological transitions
\cite{GlGrLiTa2000,DuFiFlJiLiLiWu2006,BoLiGlLi2011,ShYoJu2011}, exploiting the fact
that Eulerian representations naturally encode connectivity while
Lagrangian representations provide highly accurate geometric
information. The present approach lies in the category of methods that
construct an auxiliary Eulerian field from a Lagrangian interface and
recover a topologically corrected interface from an isocontour of that field. 
In the \mts\ algorithm, the Eulerian field is taken to be a signed-distance function, as in
\cite{CeRo2005,CeRoSiVi2010,ShJu2009,WoThGrTu2009,Muller2009}, although alternative choices,
such as smoothed indicator functions, are also possible
\cite{ShJu2002,SiSh2007,GeGoDeWa2025}.

Suppose that the family of interfaces  $\{\Gamma_{\alpha}(t)\}$
has been evolved over the time interval $(T_{m-1},T_m)$
by the Lagrangian tracking algorithm, and 
the \mts\ algorithm is then applied at time $t=T_m$. 
We denote the \emph{pre-processed} Lagrangian interface family by
\begin{equation}\label{gamma-lgr}
\left\{ \Gammalgr_\alpha \right\}_{\alpha = 1}^{\Nlgr}, 
\qquad 
\text{with}
\quad
\Gammalgr_\alpha \coloneqq \Gamma_\alpha ( T_m^-), 
\quad
\zlgr_\alpha(s) \coloneqq \gamma_\alpha (s,T_m^-), 
\quad 
\text{and} 
\quad
\Nlgr \coloneqq N(T_m^-). 
\end{equation}
The interface family $\{ \Gammalgr_\alpha \}$ is first converted into an 
Eulerian representation through the \emph{signed-distance function} (SDF), 
\begin{equation}\label{sdf}
\phi(x) \coloneqq \pm \min_\alpha \mathrm{dist}\!\left(x,\Gammalgr_\alpha\right),
\end{equation}
defined on a local fine-scale Cartesian grid with cell size $h$. 
The zero-level set 
\begin{subequations}\label{gamma-eul}
\begin{equation}
\Sigma_\phi \coloneqq \{ x \in \mathbb{R}^2 : \phi(x) = 0 \}
\end{equation}
defines the (Eulerian) \emph{extracted}  interface family, 
\begin{equation}
\left\{ \Gammaeul_\beta \right\}_{\beta = 1}^{\Neul}  \coloneqq \mathrm{Conn} (\Sigma_\phi), 
\qquad 
\text{with} 
\quad
\Neul \coloneqq
\#\,\mathrm{Conn}(\Sigma_\phi). 
\end{equation}
\end{subequations}
Here, $\mathrm{Conn}(A)$ denotes the collection of connected components
of a set $A$, while $\#A$ denotes its cardinality, so that 
$\Neul$ is the number of connected components. 
In the remainder of this section, we describe an efficient 
narrow-band method to compute the SDF \eqref{sdf}, along with 
the simple marching-squares procedure used to extract the 
Eulerian interface family \eqref{gamma-eul}.

\subsection{Efficient construction of the signed-distance function}
\label{subsec:sdf}

Let $\Omegaref \subset \mathbb{R}^2$ denote a fixed reference domain 
containing the Lagrangian interface family $\{ \Gammalgr_\alpha \}$, and 
let  $h_0$ denote the mesh spacing of a coarse Cartesian grid on $\Omegaref$. 
The fine-scale $h$ is defined by dyadic refinement of $h_0$,
\begin{equation}\label{h0}
h = 2^{-p} h_0, \quad p \ge 0, \qquad \text{with} \  h_0 > 0 \ \text{prescribed.} 
\end{equation} 
With the dyadic refinement \eqref{h0}, each coarse grid cell is subdivided 
into $2^p \times 2^p$ refined cells of size $h$. 

Given the pre-processed interface family 
$\{ \Gammalgr_\alpha \}_{\alpha=1}^{\Nlgr}$, 
the SDF is constructed only within the narrow tubular neighborhood
\begin{equation}
\mathcal{N}_h
\coloneqq
\left\{
x \in \mathbb{R}^2 :
\min_\alpha
\mathrm{dist}(x,\Gammalgr_\alpha)
\le 3 h
\right\}. 
\end{equation}
Let $\mathcal{G}_h \subset \mathcal{N}_h$ denote the refined Cartesian grid nodes, 
and let $N_{\mathsf{grid}} \coloneqq \# \mathcal{G}_h$. 
The total number of polygonal interface segments is denoted by 
$N_{\mathsf{seg}} \coloneqq \sum_{\alpha=1}^{\Nlgr} N_\alpha$.
For each node $x_i \in \mathcal{G}_h$, 
the unsigned distance is computed by
\begin{equation}
d_i
\coloneqq
\min_\alpha
\mathrm{dist}(x_i,\Gammalgr_\alpha), 
\end{equation}
and the sign is assigned by ray-casting via
\begin{equation}
n_i
\coloneqq
\#
\Bigg(
R_i
\cap
\bigcup_{\alpha=1}^{\Nlgr}
\Gammalgr_\alpha
\Bigg),
\qquad
R_i
\coloneqq
\{x_i+s(1,0):s>0\},
\qquad
\phi_i
\coloneqq
(-1)^{n_i} d_i.
\end{equation}
A direct brute-force implementation computes
\begin{equation}
d_i
=
\min_{\alpha=1,\ldots,\Nlgr}
\;
\min_{j=1,\ldots,N_\alpha}
\mathrm{dist}(x_i,S_{\alpha,j}),
\qquad
i=1,\ldots,N_{\mathsf{grid}},
\end{equation}
requiring $\mathcal{O} \big( N_{\mathsf{grid}} \times N_{\mathsf{seg}} \big)$ 
point-to-segment distance evaluations. 
Similarly, the ray-casting procedure requires testing each horizontal ray 
against all interface segments, yielding the same computational complexity. 

To obtain a scalable algorithm, we reverse the order of the computation: rather than testing every 
grid node against every interface segment, each segment deposits distance information only to 
nearby refined-grid nodes. The same locality principle is used for ray-casting by binning 
candidate crossing segments by horizontal grid row. 
A complete description of the method is given in \Cref{alg:sdf} in \Cref{appendix:aux-algs}. 

Under the adaptive refinement criterion described in
\Cref{subsec:adaptive_refinement}, the polygonal segments 
satisfy $|S_{\alpha,j}| = \mathcal{O}(h)$, so that the $3h$-neighborhood of 
each segment intersects only $\mathcal{O}(1)$ refined cells. 
Consequently, each segment deposits
distance information to only $\mathcal{O}(1)$ refined-grid nodes,
yielding a total unsigned-distance construction cost of $\mathcal{O}(N_{\mathsf{seg}})$.
The same localization principle is used in the ray-casting procedure.
Each segment is inserted only into the horizontal refined-grid rows
intersecting its vertical extent. Since the segment length is
$\mathcal{O}(h)$, each segment intersects only $\mathcal{O}(1)$
rows, so the row-binning stage also costs $\mathcal{O}(N_{\mathsf{seg}})$.
Moreover, under the narrow-band construction and adaptive refinement,
the number of candidate crossing segments associated with each row
remains uniformly bounded independently of resolution. Therefore,
each ray-casting query requires only $\mathcal{O}(1)$ segment
intersection tests, giving a total sign-construction cost of
$\mathcal{O}(N_{\mathsf{grid}})$. 
Hence the proposed localized SDF construction has complexity 
$\mathcal{O}
\big(
N_{\mathsf{seg}}
+
N_{\mathsf{grid}}
\big)$.

\subsection{Interface extraction from the zero-level set}
\label{subsec:zero_level_extraction}

With the SDF in hand, the extraction of the 
topologically-corrected interface family 
$\{ \Gammaeul_\beta \}_{\beta=1}^{\Neul}$ from the zero-level set 
proceeds in three stages:
\[
\renewcommand{\arraystretch}{1.3}
\begin{array}{c}
\text{Identify connected components}
\;\rightarrow\;
\text{Trace component curves}
\;\rightarrow\;
\text{Refine \& filter}
\end{array}
\]
The first stage of the extraction process identifies the connected 
components of $\Sigma_\phi$ (and therefore also $\Neul$, i.e., the number of extracted curves) 
through a standard graph-based marching-squares construction \cite{LoCl1998,WoThGrTu2009,Muller2009,GeGoDeWa2025}. 
The second stage converts each graph component into an ordered polygonal curve.
For each component label $\beta$, we select a starting vertex in the corresponding
connected component of $G_\phi$ and traverse the graph by repeatedly moving to
the neighboring vertex that is not the previous vertex. The physical coordinates
of the visited vertices define an extracted polygonal curve 
$\Gammaeul_\beta=\{\zeul_{\beta,1},\ldots,\zeul_{\beta,N_\beta}\}$, which is 
subsequently processed in the third stage of the extraction pipeline by a refinement and small-scale filtering procedure. 
The complete extraction pipeline, given in 
Algorithms \ref{alg:identify}, \ref{alg:trace}, and \ref{alg:ref-and-filter} in \Cref{appendix:aux-algs}, therefore yields 
a topologically-corrected Eulerian interface family
$\{\Gammaeul_\beta\}_{\beta=1}^{\Neul}$
associated with the original Lagrangian interface family
$\{\Gammalgr_\alpha\}_{\alpha=1}^{\Nlgr}$.
In the next stage of the \mts\ algorithm, these two interface families are related through a 
graph-theoretic adjacency topology, which forms the basis for both topological transition 
detection and subsequent surgical rebuild.

\section{Building the adjacency topology and classifying topological events}
\label{sec:adjacency}

At this stage of the algorithm, the pre-processed 
Lagrangian family $\{\Gammalgr_\alpha\}_{\alpha=1}^{\Nlgr}$
and the extracted Eulerian family 
$\{\Gammaeul_\beta\}_{\beta=1}^{\Neul}$ generally do not match geometrically or topologically. 
The central goal of this section is to infer the \emph{adjacency topology} between the two interface families from their
underlying geometric configuration, and thereby flag topological transitions.

The inferred adjacency topology is represented by a bipartite graph $G$,
following the general philosophy of parent--child connectivity
constructions used in \cite{GlGrLiTa2000,DuFiFlJiLiLiWu2006,BoLiGlLi2011,GeGoDeWa2025}.
The bipartite graph is denoted by
\begin{subequations}
\begin{equation}
G=(\Vlgr, \Veul, E), 
\end{equation}
with vertex sets
\begin{equation}
\Vlgr
=
\{\Gammalgr_\alpha\}_{\alpha=1}^{\Nlgr}
\quad
\text{and}
\quad
\Veul
=
\{\Gammaeul_\beta\}_{\beta=1}^{\Neul}.
\end{equation}
The edge set
\begin{equation}
E
\subset
\Vlgr
\times
\Veul
\end{equation}
\end{subequations}
encodes adjacency relations between the 
Lagrangian  and  Eulerian interfaces.
In practice, the graph may be represented by a binary adjacency matrix
\begin{equation}\label{adj-def}
\Adj
\in
\{0,1\}^{\Nlgr\times\Neul}, 
\qquad 
\text{with}
\quad 
\Adj_{\alpha\beta}=1 
\iff 
\left( \Gammalgr_\alpha, \Gammaeul_\beta \right) \in E. 
\end{equation}
The \mts\ algorithm computes the adjacency matrix $\Adj$
through a robust pipeline consisting of candidate
adjacency construction followed by geometric pruning.
A candidate graph is first constructed
from local geometric proximity information, after which
fine-scale geometric compatibility criteria are used to remove
spurious edges and produce the final adjacency topology. 

Then, the connected components of 
the bipartite graph are extracted to define independent local 
topological events $\{ \Event_1,\ldots,\Event_{\Nevent} \}$, which are classified as:
\begin{itemize}
\item \emph{trivial} events, corresponding to one-to-one relations;
\item \emph{vanishing} events, corresponding to interfaces with no descendants;
\item or \emph{branching} events, corresponding to one-to-many or many-to-one adjacencies;
\end{itemize}
Topology change is then flagged by the presence of any nontrivial event. 

\subsection{Candidate adjacency topology with geometric pruning}
\label{subsec:adjacency_construction}

In this subsection, we construct a preliminary adjacency matrix
\begin{equation}\label{candidate-adjacency}
\widetilde{\Adj}
\in
\{0,1\}^{\Nlgr\times\Neul},
\end{equation}
which encodes parent--child relations between the 
Lagrangian and  Eulerian interfaces.  
The construction combines geometric proximity tests with successive
pruning stages designed to remove geometrically inadmissible attachments,
followed by a final parent-existence enforcement step.
The individual stages of the construction are described below.

\begin{enumerate}\setlength\itemsep{1em}

\item \emph{Candidate adjacency by $\lgr$-$\eul$ proximity:} 
We first construct a candidate adjacency relation using pairwise
geometric proximity between
$\{\Gammalgr_\alpha\}_{\alpha=1}^{N_\lgr}$ and
$\{\Gammaeul_\beta\}_{\beta=1}^{N_\eul}$.
For each pair $(\Gammalgr_\alpha,\Gammaeul_\beta)$, we define the minimum
separation distance
\begin{subequations}\label{cand-adjacency}
\begin{equation}
d_{\alpha\beta}
\coloneqq
\inf_{s,\bar s}
\left|
\zlgr_\alpha(s)
-
\zeul_\beta(\bar s)
\right|.
\end{equation}
The matrix
$\widetilde\Adj$ is then defined by
\begin{equation}
\widetilde\Adj_{\alpha\beta}
=
1
\quad\Longleftrightarrow\quad
d_{\alpha\beta}
\le
h.
\label{eq:candidate-adjacency}
\end{equation}
\end{subequations}

\item \emph{Merge pruning by $\lgr$--$\lgr$ proximity:} 
Candidate multi-parent attachments necessarily correspond to local merge events.
By construction, a merge can occur only when the participating Lagrangian interfaces come within distance $h$ of one another.  
Consequently, whenever an extracted interface possesses 
multiple candidate parents, the corresponding parent interfaces are required to lie within distance $h$ of one another.
Suppose that $\Gammaeul_\beta$ possesses multiple candidate parents, so that
$\widetilde\Adj_{\alpha_1\beta}=\widetilde\Adj_{\alpha_2\beta}=1$
for some $\alpha_1\neq\alpha_2$.
For each such pair $(\alpha_1, \alpha_2)$, we define the $\lgr$-$\lgr$ separation distance
\begin{subequations}\label{lgr-lgr-prune}
\begin{equation}
d_{\alpha_1\alpha_2}^{\lgr}
\coloneqq
\inf_{s,\bar s}
\left|
\zlgr_{\alpha_1}(s)
-
\zlgr_{\alpha_2}(\bar s)
\right|.
\end{equation}
Whenever $d_{\alpha_1 \alpha_2}^{\lgr}>h$,
the attachment associated with the larger $\lgr$-$\eul$ separation distance is removed according to
\begin{equation}
d_{\alpha_1\alpha_2}^{\lgr}>h
 \ 
 \implies
 \ 
\begin{cases}
\widetilde\Adj_{\alpha_1\beta}=0,
&
\text{if } d_{\alpha_1\beta}\ge d_{\alpha_2\beta},
\\
\widetilde\Adj_{\alpha_2\beta}=0,
&
\text{if } d_{\alpha_1\beta}<d_{\alpha_2\beta}.
\end{cases}
\end{equation}
This pruning procedure is repeated until all remaining multi-parent
attachments satisfy $d_{\alpha_1\alpha_2}^{\lgr}\le h$.
\end{subequations}

\item \emph{Local geometric overlap pruning:}
The remaining multi-parent attachments are further pruned using a local
geometric overlap criterion.  For each $\Gammalgr_\alpha$, define the node-ball covering
\begin{subequations}\label{overlap-pruning}
\begin{equation}
\mathcal{B}_\alpha
\coloneqq
\bigcup_i
B_{h/4} \big(\zlgr_{\alpha,i} \big), 
\qquad 
\text{where} 
\quad 
B_r(x)
\coloneqq
\{
y\in\mathbb{R}^2
:
|x-y|\le r
\}.
\end{equation}
Whenever the extracted interface fails to intersect the node-ball covering,
the corresponding attachment is removed:
\begin{equation}
\Gammaeul_\beta
\cap
\mathcal{B}_\alpha
=
\varnothing
\implies
\widetilde\Adj_{\alpha\beta}=0.
\end{equation}
\end{subequations}

\item \emph{Parent existence enforcement:}
Finally, every extracted Eulerian interface is required to be adjacent 
to at least one pre-processed Lagrangian interface. 
More precisely, whenever
\begin{subequations}\label{parent-enforce}
\begin{equation}
\sum_{\alpha=1}^{N_\lgr}
\widetilde\Adj_{\alpha\beta}
=
0,
\end{equation}
the nearest Lagrangian interface is attached:
\begin{equation}
\widetilde\Adj_{\alpha^\ast\beta}=1,
\qquad
\alpha^\ast
\coloneqq
\argmin_\alpha d_{\alpha\beta}.
\end{equation}
\end{subequations}

\end{enumerate}

\begin{remark}
To accelerate the geometric queries arising throughout the 
adjacency topology construction, each interface is first associated with a
padded axis-aligned bounding box \cite{GeGoDeWa2025}.  For a curve
$\Gamma\subset\mathbb{R}^2$, we define
\begin{subequations}
\begin{equation}
B_h(\Gamma)
\coloneqq
[x^1_{\min} - h, x^1_{\max} + h]
\times
[x^2_{\min} - h, x^2_{\max} + h],
\qquad
x^i_{\min}
\coloneqq
\min_{{x}\in\Gamma} x^i,
\quad
x^i_{\max}
\coloneqq
\max_{{x}\in\Gamma} x^i.
\end{equation}
These bounding boxes are subsequently used as cheap geometric
rejection criteria in the pairwise proximity and pruning stages. 
Specifically, the candidate $\lgr$--$\eul$ pair
$(\Gammalgr_\alpha,\Gammaeul_\beta)$ is discarded whenever
\begin{equation}
B_h(\Gammalgr_\alpha)
\cap
B_h(\Gammaeul_\beta)
=
\varnothing, 
\end{equation}
and, for surviving pairs, the geometric proximity query is localized
to the intersection region
\begin{equation}
B_h(\Gammalgr_\alpha)
\cap
B_h(\Gammaeul_\beta).
\end{equation}
\end{subequations}
\end{remark}

\subsection{Topology-change detection and event classification}
\label{subsec:event_decomposition}

We now use the inferred adjacency topology to decompose the global 
topology into independent local events, classify the resulting event structure, and 
thereby detect topological changes.
The preliminary adjacency matrix $\widetilde\Adj$
constructed in \Cref{subsec:adjacency_construction}
defines a bipartite graph $G$, whose connected components 
define independent local topological \emph{events} $\Event_q$.

We denote the connected components of $G$ by
\begin{subequations}\label{events}
\begin{equation}
\Event_q
=
\left(
\Vlgr_q,
\Veul_q,
E_q
\right),
\qquad
q=1,\ldots,\Nevent,
\end{equation}
where
$\Vlgr_q\subset\Vlgr$,
$\Veul_q\subset\Veul$, and
\[
E_q
=
E
\cap
\left(
\Vlgr_q
\times
\Veul_q
\right).
\]
Equivalently,
\begin{equation}\label{eq:event-components}
\Nevent
\coloneqq
\#\,\mathrm{CC}(G),
\qquad
\{
\Event_q : q=1,\ldots,\Nevent
\}
\coloneqq
\mathrm{CC}(G), 
\end{equation}
where $\mathrm{CC}(G)$ denotes the collection of connected components
of the bipartite adjacency graph $G$.  
For bookkeeping, we also define the index sets associated with each event $\Event_q$, 
\begin{equation}
I_q
\coloneqq
\{
\alpha :
\Gammalgr_\alpha\in\Vlgr_q
\}
\qquad
\text{and}
\qquad
J_q
\coloneqq
\{
\beta :
\Gammaeul_\beta\in\Veul_q
\}.
\end{equation}
\end{subequations}
In practice, the connected components are extracted by a 
breadth-first search, initialized from each unvisited Lagrangian-interface
vertex and alternating between adjacent Eulerian-interface and
Lagrangian-interface vertices until no new vertices are reached.

The labels of the Lagrangian and Eulerian interfaces
may be permuted so that the rows and columns associated with each event
appear contiguously. 
Specifically, the connected components of the bipartite adjacency graph
are traversed sequentially, and the associated Lagrangian-interface and
Eulerian-interface labels are reordered in the order they are discovered.
The canonical adjacency matrix is then defined by
\begin{equation}\label{canonical-form}
\Adj
\coloneqq
P_{\lgr}
\,
\widetilde\Adj
\,
P_{\eul}^{\top},
\end{equation}
where
$P_{\lgr}$
and
$P_{\eul}$
denote the row and column permutation matrices corresponding to this relabelling.
With the adjacency matrix in canonical form, the event
collection $\{\Event_1,\ldots,\Event_{\Nevent}\}$ is ordered so that the rows and columns 
associated with each event occupy contiguous index ranges in $\Adj$.

Each event $\Event_q$ is classified according to the cardinalities of
its associated index sets $(I_q,J_q)$:
\begin{equation}\label{event-classify}
\Event_q \ 
\text{ is } \ 
\begin{cases}
\text{\emph{trivial}}, \quad 
&
\text{if } |I_q|=|J_q|=1,
\\
\text{\emph{vanishing}}, \quad
&
\text{if } |I_q|=1
\text{ and }
|J_q|=0,
\\
\text{\emph{branching}}, \quad
&
\text{otherwise}.
\end{cases}
\end{equation}
Representative examples of the different event types are illustrated in
\Cref{fig:event-types}. Topology change is detected whenever at least one event is
non-\emph{trivial}, which is equivalent to the condition
\begin{equation}
\exists\,\alpha
\
\text{such that}
\
\sum_{\beta=1}^{N_\eul}
\Adj_{\alpha\beta}\neq 1,
\qquad
\text{or}
\qquad
\exists\,\beta
\ 
\text{such that}
\
\sum_{\alpha=1}^{N_\lgr}
\Adj_{\alpha\beta}\neq 1.
\end{equation}
If all events are either \emph{trivial} or \emph{vanishing}, then no interface surgery is required.  
In this case, any \emph{vanishing} interfaces are removed, the \mts\
algorithm is terminated, and the
Lagrangian evolution proceeds using the surviving interface family.
Conversely, if at least one \emph{branching} event is present, then the algorithm proceeds to the
surgical reconstruction stage, described in the next section.

\begin{figure}[t]
\centering

\begin{tikzpicture}[
    vertex/.style={
        circle,
        thick,
        minimum size=6mm,
        inner sep=0pt,
        text height=1.5ex,
        text depth=.25ex,
        font=\footnotesize
    },
    old/.style={
        vertex,
        draw=red,
        fill=red!20
    },
    new/.style={
        vertex,
        draw=green!60!black,
        fill=green!20
    },
    edge/.style={
        thick,
        black
    },
    label/.style={
        font=\small\itshape,
        anchor=base
    },
    matrixlabel/.style={
        font=\small,
        anchor=center,
        minimum height=8mm,
        inner sep=0pt
    }
]

\def\ymat{-1.15}
\def\ylab{-2.00}

\def\xtriv{0}
\def\xvan{3}
\def\xmerge{6}
\def\xsplit{9}
\def\xbranch{12}

\node[old] (t-old1) at ($( \xtriv,0 )+(-0.75,0)$) {$\alpha$};
\node[new] (t-new1) at ($( \xtriv,0 )+( 0.75,0)$) {$\beta$};

\draw[edge] (t-old1) -- (t-new1);

\node[matrixlabel] at (\xtriv,\ymat)
{$A_q=\begin{bmatrix}1\end{bmatrix}$};

\node[label] at (\xtriv,\ylab)
{trivial};

\node[old] (v-old1) at ($( \xvan,0 )+(-0.75,0)$) {$\alpha$};

\node[matrixlabel] at (\xvan,\ymat)
{$A_q=\begin{bmatrix}0\end{bmatrix}$};

\node[label] at (\xvan,\ylab)
{vanishing};

\node[old] (m-old1) at ($( \xmerge,0 )+(-0.75,0.8)$) {$\alpha_1$};
\node[old] (m-old2) at ($( \xmerge,0 )+(-0.75,0.0)$) {$\alpha_2$};

\node[new] (m-new1) at ($( \xmerge,0 )+(0.75,0.4)$) {$\beta$};

\draw[edge] (m-old1) -- (m-new1);
\draw[edge] (m-old2) -- (m-new1);

\node[matrixlabel] at (\xmerge,\ymat)
{$A_q=\begin{bmatrix}1\\1\end{bmatrix}$};

\node[label] at (\xmerge,\ylab)
{merge};

\node[old] (s-old1) at ($( \xsplit,0 )+(-0.75,0.4)$) {$\alpha$};

\node[new] (s-new1) at ($( \xsplit,0 )+(0.75,0.8)$) {$\beta_1$};
\node[new] (s-new2) at ($( \xsplit,0 )+(0.75,0.0)$) {$\beta_2$};

\draw[edge] (s-old1) -- (s-new1);
\draw[edge] (s-old1) -- (s-new2);

\node[matrixlabel] at (\xsplit,\ymat)
{$A_q=\begin{bmatrix}1&1\end{bmatrix}$};

\node[label] at (\xsplit,\ylab)
{split};

\node[old] (c-old1) at ($( \xbranch,0 )+(-0.75,0.8)$) {$\alpha_1$};
\node[old] (c-old2) at ($( \xbranch,0 )+(-0.75,0.0)$) {$\alpha_2$};

\node[new] (c-new1) at ($( \xbranch,0 )+(0.75,0.8)$) {$\beta_1$};
\node[new] (c-new2) at ($( \xbranch,0 )+(0.75,0.0)$) {$\beta_2$};

\draw[edge] (c-old1) -- (c-new1);
\draw[edge] (c-old1) -- (c-new2);
\draw[edge] (c-old2) -- (c-new2);

\node[matrixlabel] at (\xbranch,\ymat)
{$A_q=\begin{bmatrix}1&1\\0&1\end{bmatrix}$};

\node[label] at (\xbranch,\ylab)
{branching};

\end{tikzpicture}

\caption{
Representative local topological events and their associated
adjacency submatrices $A_q$. Red vertices denote pre-processed
Lagrangian interfaces, green vertices denote extracted Eulerian
components, and black edges denote adjacency.  The \emph{trivial} and
\emph{vanishing} events do not require surgical reconstruction, whereas
the \emph{merge}, \emph{split}, and more general \emph{branching}
events trigger the interface surgery.
}
\label{fig:event-types}

\end{figure}
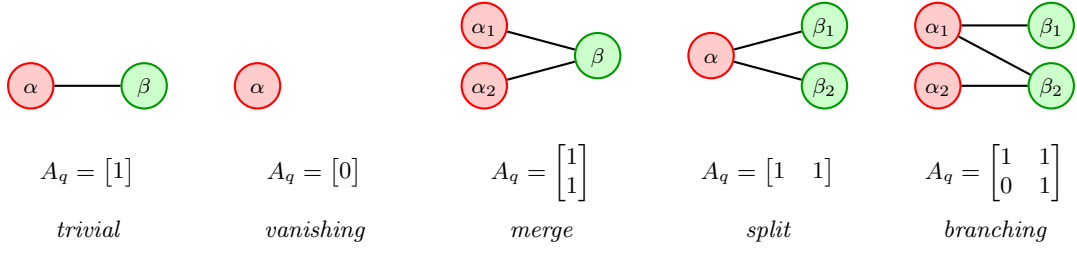

\section{Interface surgery using topological defect measures}
\label{sec:surgery}

The third and final stage of the core \mts\ algorithm is \emph{interface surgery}. 
The objective of this stage is to surgically stitch together
the Lagrangian family $\{ \Gammalgr_\alpha \}_{\alpha=1}^{\Nlgr}$
and extracted Eulerian family $\{ \Gammaeul_\beta \}_{\beta=1}^{\Neul}$,
using the local event decomposition
\eqref{events}--\eqref{event-classify},
to produce a surgically reconstructed post-transition interface family
\begin{equation}\label{gamma-post}
\left\{
\Gamma_\alpha(T_m^+)
\right\}_{\alpha=1}^{N(T_m^+)}, 
\qquad 
\text{with} 
\quad 
N(T_m^+) \coloneqq \Neul
\quad 
\text{and}
\quad 
\Gamma_\alpha(T_m^+) = \{ \gamma_\alpha(s,T_m^+) : s \in \mathbb{T} \}. 
\end{equation}
The surgically reconstructed interface family \eqref{gamma-post} is subsequently passed back to the
Lagrangian tracking algorithm and evolved forward until the
next \mts\ time $T_{m+1}$.

For each $q = 1,\ldots,\Nevent$, the reconstruction procedure is 
performed independently on the local event $\Event_q$
defined by the event decomposition \eqref{events}--\eqref{event-classify}. 
We therefore suppress the event index $q$ throughout the
remainder of this section and write
\[
\Event \equiv \Event_q,
\qquad
\Vlgr \equiv \Vlgr_q,
\qquad
\Veul \equiv \Veul_q,
\qquad
E \equiv E_q,
\qquad
I \equiv I_q,
\qquad
J \equiv J_q.
\]
As noted in \Cref{subsec:event_decomposition}, interface surgery 
is only required for the \emph{branching} events, for which 
neither the pre-processed Lagrangian interfaces
nor the extracted Eulerian interfaces are individually suitable as the
post-transition interface family. 
The Lagrangian interfaces maintain accuracy
away from the topological transition, but possess the incorrect topology.
Conversely, the extracted Eulerian interfaces possess the correct 
topology, but generally exhibit lower accuracy away from the transition region. 
The key objective of interface surgery is to identify the localized
regions in which topological inconsistencies occur and restrict
surgical modifications to those regions alone.
At a high level, the procedure consists of:
\begin{enumerate}
\item constructing localized \emph{topological defect densities}
$(\mu,\eta)$, together with the associated \emph{defect measure}
$\mathcal D$, its regularized approximation $\Defect$, and the
resulting localized \emph{surgical region} $\Surg$, which identify the
portions of the interface families responsible for topological
inconsistencies;

\item decomposing the Lagrangian and Eulerian interface families into
compatible \emph{surgical pieces}
$\{\Plgr,\Peul\}$, and smoothly and consistently stitching the resulting
pieces together to produce the topologically correct post-transition
interface family \eqref{gamma-post}.
\end{enumerate}
In the remainder of this section, we describe this surgical
reconstruction procedure in detail.

\begin{remark}[Direct-replacement variant]
\label{rem:direct-replacement}
Instead of performing surgical reconstruction for a \emph{branching}
event $\Event=(\Vlgr,\Veul,E)$, one may replace the post-transition
interfaces associated with this event by the extracted Eulerian
interfaces in $\Veul$. That is, for each $\beta \in J$, we set
\begin{equation}\label{direct-replacement}
\Gamma_\beta(T_m^+)
=
\{
\gamma^\eul_\beta(s)
:
s \in \mathbb T
\},
\qquad
\text{together with} 
\quad
N(T_m^+) \coloneqq \Neul .
\end{equation}
This yields a simpler \emph{direct-replacement} variant of the \mts\
algorithm, analogous to the \emph{global} topological reconstruction
algorithms of \cite{ShJu2002,SiSh2007,Muller2009,ShYoJu2011,PaLoChScPoZa2024,GeGoDeWa2025}. As we will demonstrate below in
\Cref{subsec:need-for-surgery}, the direct-replacement \mts\ variant \eqref{direct-replacement},
while simpler, produces nontrivial geometric errors away from
topological transitions due to the global Eulerian reconstruction. The
surgical approach developed in this section eliminates such errors.
\end{remark}

\subsection{Identifying topological defects through closest-point degeneracies}
\label{subsec:defect_measures}

The localized surgical regions are identified through
event-local topological defect densities constructed from
bidirectional closest-point projections between the pre-processed
Lagrangian and extracted Eulerian interface families.
Classical approaches for identifying topological inconsistencies typically
rely on grid-intersection tests
\cite{GlGrLiTa2000,DuFiFlJiLiLiWu2006,WoThGrTu2009,BoLiGlLi2011,HeKaStEtYaWo2024}.
More closely related to the present work are geometric approaches based
on closest-point degeneracies and local interface signatures, including
the kink-detection framework of \cite{HeCoLu2022} and the 
signature method of \cite{ChMaPoZa2022}. The defect densities
introduced below may be viewed as a closest-point degeneracy detector
adapted to the event-local Lagrangian--Eulerian correspondence
structure constructed in \Cref{sec:adjacency}.

\subsubsection{Topological defect densities}

For each \emph{branching} event $\Event$, we define
binary defect densities
\begin{equation}\label{one-sided-defect-densities}
\mu_\alpha : \mathbb{T} \to \{0,1\},
\quad
\alpha \in I,
\qquad
\text{and}
\qquad
\eta_\beta : \mathbb{T} \to \{0,1\},
\quad
\beta \in J,
\end{equation}
where $\mu_\alpha(s)$ and $\eta_\beta(s)$ indicate whether the
corresponding interface points
$\zlgr_\alpha(s)$ and $\zeul_\beta(s)$
are identified as topologically defective.
The ``Lagrangian'' defect density $\mu_\alpha$ is constructed from a pair
of closest-point degeneracy indicators associated with the extracted
Eulerian interface family.
The ``Eulerian'' defect density $\eta_\beta$ is defined analogously by
reversing the roles of the Lagrangian and Eulerian interface families in
the construction below.

We first introduce the closest-label map and closest-point projection
associated with the extracted Eulerian interface family:
\begin{subequations}\label{eq:lgr-correspondence}
\begin{align}
\Lambda_\alpha(s)
&\coloneqq
\argmin_{\beta\in J}
\inf_{\bar s \in\mathbb{T}}
\left|
\zlgr_\alpha(s)
-
\zeul_\beta(\bar s)
\right|,
\label{eq:lgr-closest-label}
\\[0.25em]
\Pi_\alpha(s)
&\coloneqq
\argmin_{z\in\Gammaeul}
|\zlgr_\alpha(s) - z|.
\label{eq:lgr-closest-projection}
\end{align}
\end{subequations}
Here \eqref{eq:lgr-closest-projection} should be understood in the sense
that $\Pi_\alpha(s)$ denotes a selected closest point, since the
minimizer need not be unique.
The defect density is then defined as the union of two binary indicator
functions,
\begin{equation}
\mu_\alpha(s)
=
\max
\{
\mu_\alpha^1 (s),
\mu_\alpha^2 (s)
\}, 
\end{equation}
where the individual indicators $\mu_\alpha^1, \mu_\alpha^2 : \mathbb{T} \to \{0,1\}$ are defined as follows:
\begin{enumerate}\setlength\itemsep{1em}
\item
The \emph{label-jump indicator}
$\mu_\alpha^1$
marks neighborhoods of discontinuities in the closest-label map
$\Lambda_\alpha$.
More precisely, denoting the jump set of $\Lambda_\alpha$ by 
\begin{subequations}\label{eq:mu1-defect}
\begin{equation}
\mathcal{J}
\coloneqq
\{
s\in\mathbb{T}
:
\Lambda_\alpha
\text{ is discontinuous at } s
\},
\end{equation}
we define
\begin{equation}
\mu_\alpha^1(s)
=
\mathbbm{1}_{\left\{
\left|
\zlgr_\alpha(s)
-
\zlgr_\alpha(s_*)
\right|
\leq
h 
\,
: 
\,
s_*\in\mathcal{J}
\right\}}, 
\end{equation}
\end{subequations}
where $\mathbbm{1}_A$ denotes the indicator function on the set $A$. 
Discontinuities in the closest-label map indicate locations where nearby
points on the same Lagrangian interface are associated with different
Eulerian interfaces, signaling a breakdown of the discrete correspondence
structure near \emph{branching} events.

\item
The \emph{projection-distortion indicator}
$\mu_\alpha^2$
marks points where the closest-point projection from the Lagrangian
interface family onto the Eulerian interface family exhibits large
local expansion along the interface.
For each $s\in\mathbb{T}$, let
$\Gammaeul_{\beta(s)}$
denote the Eulerian interface selected by the closest-label map
$\Lambda_\alpha$, and let
$\bar s(s)\in\mathbb{T}$
denote a parameter value satisfying
\begin{subequations}\label{eq:mu2-defect}
\begin{equation}
\Pi_\alpha(s)
=
\zeul_{\beta(s)}
\bigl(
\bar s(s)
\bigr).
\end{equation}
The projected arclength parameter $\sigma(s)$ 
along the Eulerian interface $\Gammaeul_{\beta(s)}$ is then given by
\begin{equation}
\sigma(s)
=
\int_0^{\bar s(s)}
\left|
\frac{\mathrm{d}}{\mathrm{d}r}
\zeul_{\beta(s)}(r)
\right|
\,\mathrm{d}r. 
\end{equation}
The local projection-distortion factor is then defined by
\begin{equation}
\rho_\alpha(s)
\coloneqq
\left|
\frac{\mathrm{d}\sigma}{\mathrm{d}s}(s)
\right|,
\end{equation}
and the projection-distortion indicator by
\begin{equation}
\mu_\alpha^2(s)
=
\mathbbm{1}_{\{
\rho_\alpha(s)
>
\rho_{\max}
\}}, 
\qquad \rho_{\max} = 10. 
\end{equation}
\end{subequations}
When the Lagrangian and Eulerian interface families are locally
well-aligned, the closest-point projection behaves approximately like
an arclength-preserving reparametrization, so that
$\rho_\alpha \approx 1$.
Large values of $\rho_\alpha$ therefore indicate strong local
distortion of the closest-point structure along the interface.

\end{enumerate}

The indicators $\{ \mu_\alpha^1, \eta_\beta^1 \}$ and
$\{ \mu_\alpha^2, \eta_\beta^2 \}$ capture complementary
failure modes of the $\lgr$-$\eul$ correspondence structure.
The label-jump indicators
$\{ \mu_\alpha^1, \eta_\beta^1 \}$
detect discrete changes in the associated dual-family interface label,
but are insensitive to geometric distortion within a single
dual-family correspondence.
Conversely, the projection-distortion indicators
$\{ \mu_\alpha^2, \eta_\beta^2 \}$
detect strong local stretching of the closest-point projection even
when the associated dual-family interface label remains unchanged.

To define the topological defect densities $(\mu,\eta)$, the one-sided defect densities
\eqref{one-sided-defect-densities} are propagated through
their closest-point projections onto the dual interface family,
thereby symmetrizing the defect construction between the
Lagrangian and Eulerian interface families.
Introducing the disjoint unions
$\Gammalgr \coloneqq \bigsqcup_{\alpha\in I}\Gammalgr_\alpha$
and
$\Gammaeul \coloneqq \bigsqcup_{\beta\in J}\Gammaeul_\beta$,
we define the defect densities
$\mu : \Gammalgr \to \{0,1\}$ and
$\eta : \Gammaeul \to \{0,1\}$ by
\begin{subequations}\label{defect-densities}
\begin{equation}
\mu(\zlgr_\alpha(s))
=
\max
\left\{
\mu_\alpha(s),
\,
\eta_\beta(\bar s)
:
\Pi_\alpha(s)=\zeul_\beta(\bar s),
\,
\beta\in J
\right\},
\qquad
\alpha\in I,
\end{equation}
and
\begin{equation}
\eta(\zeul_\beta(s))
=
\max
\left\{
\eta_\beta(s),
\,
\mu_\alpha(\bar s)
:
\Pi_\beta(s)=\zlgr_\alpha(\bar s),
\,
\alpha\in I
\right\},
\qquad
\beta\in J, 
\end{equation}
\end{subequations}
respectively.

\subsubsection{Defect measures and the surgical region}

The interface-supported defect densities $(\mu,\eta)$ induce an
associated defect measure
\begin{subequations}\label{defect-measure}
\begin{equation}
\mathcal{D}
\in
\mathcal{M}(\mathbb{R}^2),
\end{equation}
where $\mathcal{M}(\mathbb{R}^2)$ denotes the space of finite Radon
measures on $\mathbb{R}^2$.
The measure $\mathcal{D}$ is characterized by its pairing with smooth
compactly supported test functions
$\varphi \in C_c(\mathbb{R}^2)$:
\begin{equation}
\langle
\mathcal{D},
\varphi
\rangle
 \, = \, 
\sum_{\alpha\in I}
\int_{\mathbb{T}}
\mu(\zlgr_\alpha(s))
\,
\varphi(\zlgr_\alpha(s))
\left|
\zlgr'_\alpha(s)
\right|
\,\mathrm{d}s
 \ + \ 
\sum_{\beta\in J}
\int_{\mathbb{T}}
\eta(\zeul_\beta(s))
\,
\varphi(\zeul_\beta(s))
\left|
\zeul'_\beta(s)
\right|
\,\mathrm{d}s.
\end{equation}
\end{subequations}
In the numerical algorithm, the singular defect measure
$\mathcal{D}$ is replaced by the $h$-scale regularization
\begin{subequations}\label{defect-measure-regularized}
\begin{equation}
\Defect
\coloneqq
\psi_h * \mathcal{D},
\end{equation}
where
$\psi_h : \mathbb{R}^2 \to [0,\infty)$
is the compactly supported cosine mollifier
\begin{equation}
\psi_h(x)
=
\begin{cases}
\tfrac12
\left(
1+
\cos\left(
\tfrac{\pi |x|}{h}
\right)
\right),
&
|x|<h,
\\[0.25em]
0,
&
|x|\ge h. 
\end{cases}
\end{equation}
\end{subequations}
The regularized defect measure $\Defect$ provides an Eulerian
representation of the defective portions of the Lagrangian and
Eulerian interface families. 

The associated \emph{surgical region} is then defined as the super-level set
\begin{equation}\label{eq:surgical-region}
\Surg
\coloneqq
\{
x \in \mathbb{R}^2
:
\Defect(x) > \tau
\},
\end{equation}
where $\tau>0$ is a prescribed threshold. In the numerical examples
presented below, we use $\tau=0.1$. The surgical region localizes the
portion of the computational domain requiring topological
reconstruction, while leaving the remainder of the interface family
unchanged.

The construction of the defect densities, defect measures, and surgical region
is illustrated for a simple \emph{merge} event in \Cref{fig:defect}.
The left panel displays the Lagrangian and Eulerian interfaces,
while the center panel shows the interface-supported
defect densities $(\mu,\eta)$ and the singular defect measure
$\mathcal D$ that they induce.
The right panel shows the (regularized) Eulerian defect measure
$\Defect$, with the associated defect ``hotspot'' identifying the
surgical region $\Surg$ used in the subsequent reconstruction.

\begin{figure}[ht]
\centering
\begin{subfigure}[t]{0.275\linewidth}
  \centering
  \includegraphics[width=0.85\linewidth]{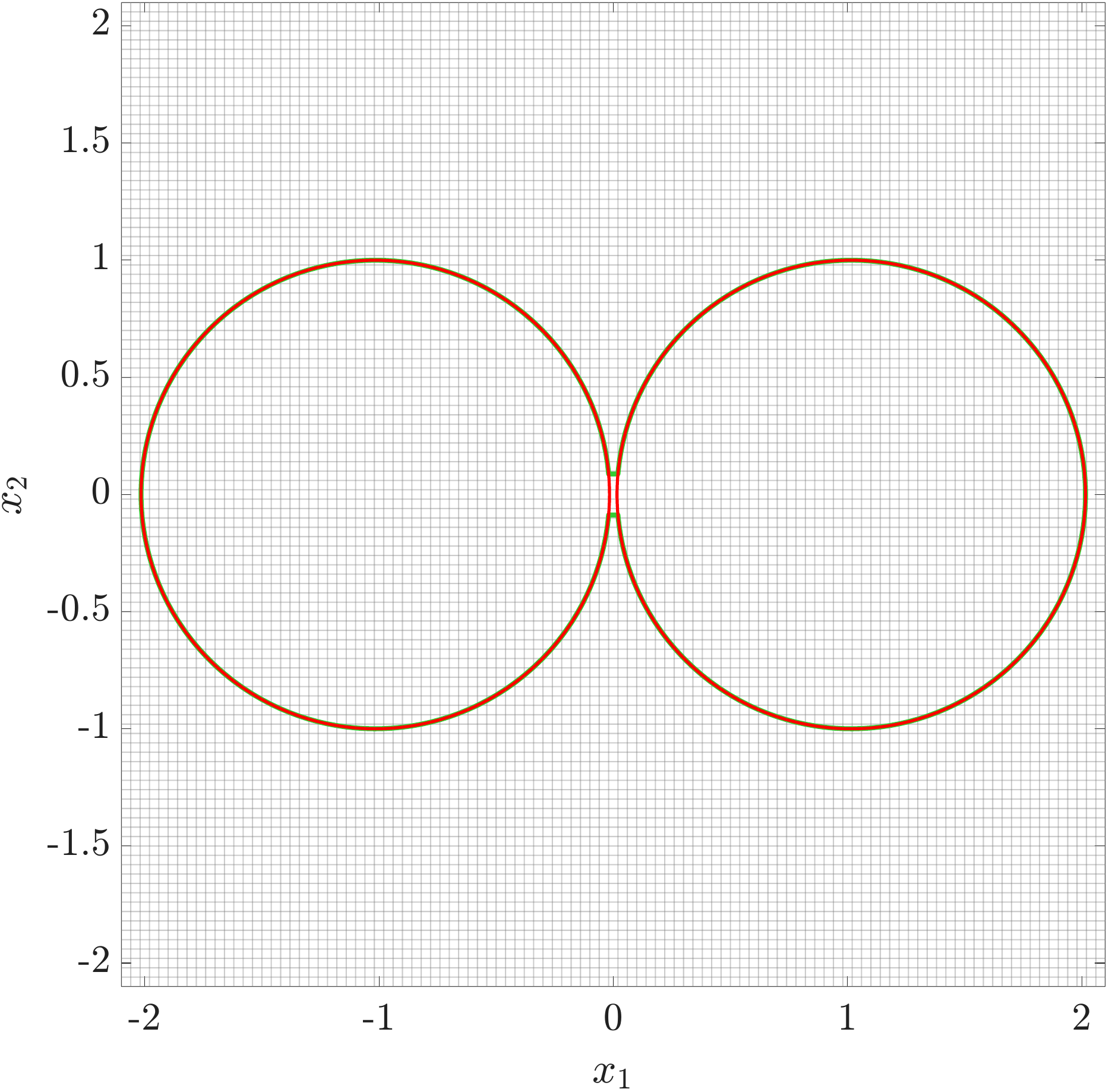}
  \caption{Interfaces $\Gammalgr_\alpha$ and $\Gammaeul_\beta$}
  \label{fig:defect-interface}
\end{subfigure}
\hspace{1em}
\begin{subfigure}[t]{0.275\linewidth}
  \centering
  \includegraphics[width=0.85\linewidth]{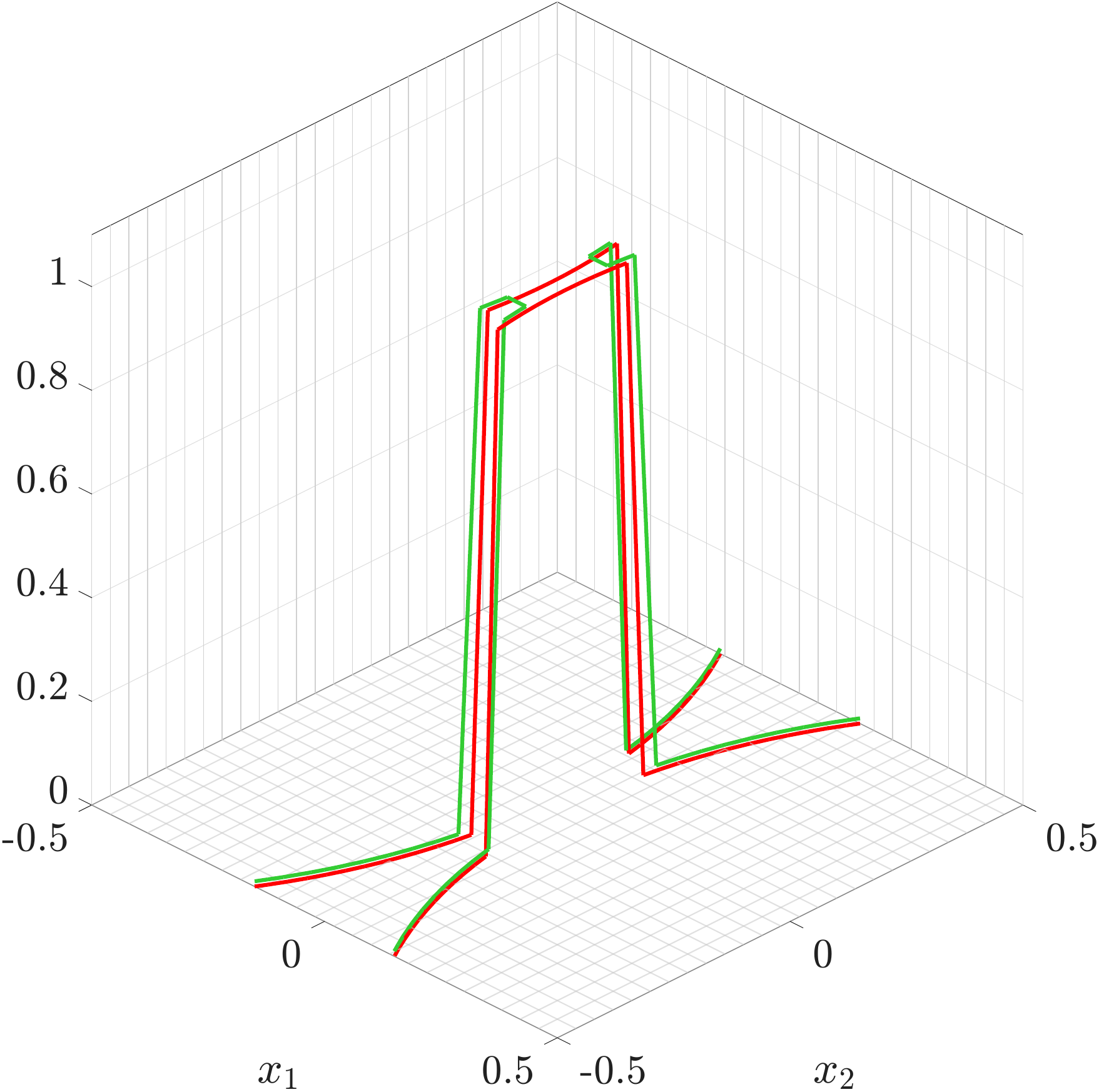}
  \caption{Defect densities $(\mu,\eta)$}
  \label{fig:defect-density}
\end{subfigure}
\hspace{1em}
\begin{subfigure}[t]{0.275\linewidth}
  \centering
  \includegraphics[width=0.935\linewidth]{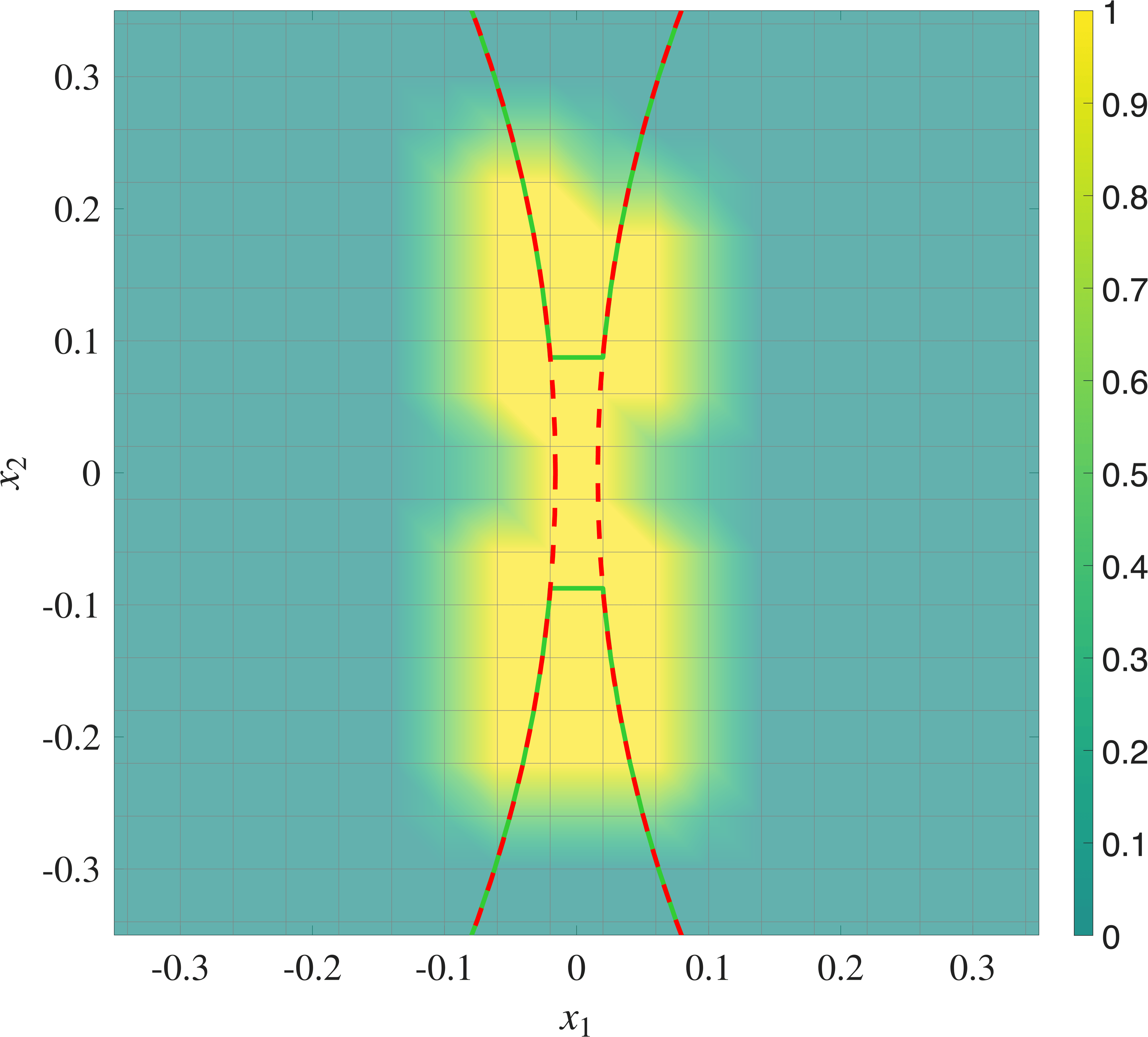}
  \caption{Regularized defect measure $\Defect$}
  \label{fig:defect-measure}
\end{subfigure}
\caption{
Construction of the topological defect densities, defect
measures, and surgical region for a representative \emph{merge} event.
\textbf{Left:}
Pre-processed Lagrangian interfaces (red) and extracted Eulerian interfaces (green).
\textbf{Center:}
Interface-supported defect densities $(\mu,\eta)$ defined by the
one-sided defect indicators
\eqref{eq:mu1-defect}--\eqref{eq:mu2-defect}
and the symmetrized projection construction
\eqref{defect-densities}, shown in a zoomed neighborhood of the \emph{merge}.
\textbf{Right:}
Regularized defect measure $\Defect$
defined by the mollification
\eqref{defect-measure-regularized}, shown in a zoomed neighborhood of
the \emph{merge}. The localized defect ``hotspot'' identifies the
surgical region $\Surg$.
}
\label{fig:defect}
\end{figure}

\subsection{Stitching surgical pieces into reconstructed interfaces}
\label{subsec:localized_surgery}

In this subsection, we describe how the surgical region $\Surg$
computed in the previous subsection is used to decompose the
interfaces into localized \emph{surgical pieces}, which are
subsequently stitched together using a graph-theoretic reconstruction
algorithm to form the post-transition interface family.
Surgical treatments of topological change have appeared in a variety
of front-tracking methods \cite{DuFiFlJiLiLiWu2006,WoThGrTu2009,BoLiGlLi2011,HeKaStEtYaWo2024}.
The distinguishing feature of the present approach is that the
localization is driven by the defect measure $\Defect$, rather than by
grid-intersection tests
\cite{DuFiFlJiLiLiWu2006,WoThGrTu2009,BoLiGlLi2011,HeKaStEtYaWo2024},
yielding an event-local surgical region adapted to the geometry of the
topological inconsistency.

The surgical decomposition treats the two interface families differently.
The retained Lagrangian pieces are the connected components of the
Lagrangian family \emph{outside} the surgical region, while the retained 
Eulerian pieces are the connected components of the Eulerian family 
\emph{inside} the surgical region:
\begin{subequations}\label{eq:surgical-pieces}
\begin{alignat}{5}
\Plgr
&\coloneqq
\left\{
\plgr_i
\right\}_{i=1}^{\Npiece}
&&\coloneqq
\mathrm{Conn}
\left(
\Gammalgr \cap \Surg^c
\right),
\qquad
&
\Npiece
&\coloneqq
\#
\mathrm{Conn}
\left(
\Gammalgr \cap \Surg^c
\right),
\\
\Peul
&\coloneqq
\left\{
\peul_j
\right\}_{j=1}^{\Mpiece}
&&\coloneqq
\mathrm{Conn}
\left(
\Gammaeul \cap \Surg
\right),
\qquad
&
\Mpiece
&\coloneqq
\#
\mathrm{Conn}
\left(
\Gammaeul \cap \Surg
\right).
\end{alignat}
\end{subequations}

\subsubsection{Graph-theoretic formulation of surgical reconnection}
\label{subsubsec:endpoint_graph}

Each surgical piece is an open curve segment with two endpoints, so that the
reconstruction problem may be reformulated entirely in terms of endpoint
connectivity.
Here and below, if $p$ is an open curve segment, then
$\partial p$ denotes the set consisting of its two endpoints.
We denote the collections of Lagrangian and Eulerian endpoints by
\begin{subequations}\label{eq:surgical-endpoints}
\begin{alignat}{6}
\Vpiece
&\coloneqq
\left\{
\vpiece_i
\right\}_{i=1}^{2\Npiece},
\qquad
&
\vpiece_{2i-1},
\vpiece_{2i}
&\in
\partial\plgr_i,
\qquad
&
i
&=1,\ldots,\Npiece,
\\
\Wpiece
&\coloneqq
\left\{
\wpiece_j
\right\}_{j=1}^{2\Mpiece},
\qquad
&
\wpiece_{2j-1},
\wpiece_{2j}
&\in
\partial\peul_j,
\qquad
&
j
&=1,\ldots,\Mpiece.
\end{alignat}
The complete endpoint set is then given by
\begin{equation}\label{eq:all-endpoints}
\widehat{V}
\coloneqq
\Vpiece
\cup
\Wpiece .
\end{equation}
\end{subequations}

The connectivity already present within the surgical pieces is encoded
by the \emph{intra-piece edge set}
\begin{subequations}\label{edge-pieces-surgery}
\begin{equation}\label{eq:intra-piece-edges}
E_{\mathrm{intra}}
\coloneqq
\left\{
(\vpiece_{2i-1},\vpiece_{2i})
\right\}_{i=1}^{\Npiece}
\cup
\left\{
(\wpiece_{2j-1},\wpiece_{2j})
\right\}_{j=1}^{\Mpiece}.
\end{equation}
Each edge in $E_{\mathrm{intra}}$ connects the two endpoints of a
single surgical piece. 
The missing connectivity is supplied by an \emph{inter-piece edge set}
\begin{equation}
E_{\mathrm{inter}}
\subset
\Vpiece\times\Wpiece,
\end{equation}
\end{subequations}
which connect Lagrangian endpoints to Eulerian endpoints.
Thus, the surgical reconstruction problem reduces to determining the
edge set $E_{\mathrm{inter}}$.
The construction of $E_{\mathrm{inter}}$ is governed by a collection of
local geometric admissibility criteria, which are described in the next
subsection.
These criteria are designed so that every endpoint in $\widehat V$ is
incident to exactly one inter-piece edge.
Consequently, $E_{\mathrm{inter}}$ defines a perfect matching between
the endpoint sets $\Vpiece$ and $\Wpiece$.

The complete \emph{interface surgery graph} is then defined by
\begin{equation}\label{eq:final-surgery-graph}
\widehat G
\coloneqq
(\widehat V,\widehat E),
\qquad
\widehat E
\coloneqq
E_{\mathrm{intra}}
\cup
E_{\mathrm{inter}}.
\end{equation}
By construction, every vertex of $\widehat G$ is incident to exactly
one intra-piece edge and exactly one inter-piece edge, 
so that $\widehat G$ is a $2$-regular graph.
It follows that $\widehat G$ admits the decomposition
\begin{equation}\label{eq:cycle-decomposition}
\widehat G
=
C_1
\sqcup
\cdots
\sqcup
C_K,
\end{equation}
where each $C_k$ is a simple cycle.
The cycles $\{C_k\}_{k=1}^{K}$ define a collection of closed
reconstructed interfaces, yielding the post-transition interface family 
\eqref{gamma-post}.

The graph-theoretic formulation above naturally leads to a global
combinatorial reconstruction problem, in which one seeks a perfect
matching satisfying additional global constraints, such as a prescribed
number of reconstructed interfaces.
Rather than solving this global matching problem directly, the present
implementation of the \mts\ algorithm constructs the inter-piece edge set
$E_{\mathrm{inter}}$ using a greedy cycle-tracing procedure based
on local geometric admissibility criteria.

\subsubsection{Greedy construction of inter-piece connectivity}
\label{subsubsec:greedy_matching}

Starting from an Eulerian surgical piece, admissible Lagrangian and
Eulerian pieces are appended alternately to a partially reconstructed
chain.
Each selected continuation determines an inter-piece edge in
$E_{\mathrm{inter}}$, and the tracing procedure continues until the
chain closes to form a cycle.
The resulting cycle defines a reconstructed interface, after which the
procedure is repeated until all surgical pieces have been incorporated
into reconstructed interfaces.

The reconstruction may be viewed as an inductive continuation process.
Given a partially reconstructed chain and its current terminal piece
$p$, the next step consists of selecting a surgical piece $q$ that
provides a geometrically compatible continuation of the chain.
For an oriented surgical piece $p$, we denote its exit endpoint and
outgoing unit tangent vector by $x_{\mathrm{ext}}(p)$ and $\tau_{\mathrm{ext}}(p)$, 
respectively.
Similarly, for a candidate continuation $q$, we denote its entry
endpoint by $x_{\mathrm{ent}}(q)$.
The endpoint separation between $p$ and $q$ is 
\begin{subequations}\label{eq:greedy-continuation}
\begin{equation}
r(p,q)
\coloneqq
\left|
x_{\mathrm{ent}}(q)
-
x_{\mathrm{ext}}(p)
\right|,
\end{equation}
while the directional alignment is measured by
\begin{equation}
a(p,q)
\coloneqq
\frac{
\left(
x_{\mathrm{ent}}(q)
-
x_{\mathrm{ext}}(p)
\right)
\cdot
\tau_{\mathrm{ext}}(p)
}
{
\left|
x_{\mathrm{ent}}(q)
-
x_{\mathrm{ext}}(p)
\right|
}.
\end{equation}
The next piece in the reconstruction is chosen by minimizing the
\emph{continuation functional}
\begin{equation}
q^*
=
\argmin_q
\left\{
\mathcal C(p,q)
:
r(p,q)\le h
\right\},
\qquad
\mathcal C(p,q)
\coloneqq
r(p,q)^2
+
\lambda
\bigl(
1-a(p,q)
\bigr),
\qquad
\lambda
\coloneqq
(4h)^2.
\end{equation}
The locality constraint $r(p,q)\le h$ excludes nonlocal candidates,
and the continuation functional balances endpoint proximity and
tangent continuity across the proposed connection.
\end{subequations}

Each application of \eqref{eq:greedy-continuation} determines a new
inter-piece edge in $E_{\mathrm{inter}}$ and extends the partially
reconstructed chain.
The continuation process is repeated until the evolving chain returns
to its initial surgical piece, thereby forming a closed cycle.
Once a cycle has been completed, its constituent surgical pieces are
removed from further consideration and the tracing procedure is
restarted from another unused Eulerian piece.
Repeating this process produces a collection of reconstructed cycles,
which define the post-transition interface family
\eqref{gamma-post}. The complete cycle-tracing procedure 
used in the numerical implementation is
summarized in \Cref{alg:greedy-event-surgery} in \Cref{appendix:aux-algs}.

\subsubsection{Smoothing of surgically reconstructed interfaces}
\label{subsubsec:untangling}

Because the extracted Eulerian interface segments are obtained through a
grid-based contour extraction procedure, the surgically reconstructed
interfaces generally inherit locally faceted geometric artifacts.
To improve the geometric quality of the reconstructed interfaces, we
apply a smoothing procedure to all extracted interface segments.

For a closed extracted loop
$z=\{z_i\}_{i=1}^{N}$,
we introduce the smoothing functional
\begin{subequations}
\begin{equation}
\mathcal J_{\mathrm{sm}}(\widehat z;z)
=
\wfit
\sum_{i=1}^{N}
|\widehat z_i-z_i|^2
+
\wsm
\sum_{i=1}^{N}
|\widehat z_{i-1}-2\widehat z_i+\widehat z_{i+1}|^2.
\end{equation}
The first term penalizes deviation from the extracted geometry, while
the second penalizes oscillations through a discrete second-difference
regularization.
The smoothed loop
$\widetilde z=\{\widetilde z_i\}_{i=1}^{N}$
is defined by
\begin{equation}
\widetilde z
=
\argmin_{\widehat z}
\mathcal J_{\mathrm{sm}}(\widehat z;z).
\label{eq:closed-loop-smoothing}
\end{equation}
The minimizer is obtained by solving the corresponding periodic linear
system for the two coordinate components of $\widetilde z$.
To compensate for the small area loss introduced by the smoothing
operation, the resulting curve is subsequently corrected by a small
normal displacement chosen to recover the original enclosed area.

For extracted subchains appearing within reconstructed interfaces, we
use the same smoothing functional together with an additional
tangent-matching penalty.
Let
$z_{i_1},\ldots,z_{i_2}$
denote an extracted subchain of a reconstructed interface,
with neighboring retained Lagrangian points
$z_{i_1-1}$ and $z_{i_2+1}$ held fixed.
We define
\begin{equation}
\mathcal J_{\mathrm{tan}}(\widehat z)
=
\wtan
\Bigl(
|\widehat z_{i_1}-(z_{i_1-1}+\tau_L)|^2
+
|\widehat z_{i_2}-(z_{i_2+1}-\tau_R)|^2
\Bigr),
\end{equation}
where $\tau_L$ and $\tau_R$ denote the incoming and outgoing tangent
vectors supplied by the neighboring retained Lagrangian segments.
The extracted subchain is then replaced by the minimizer
\begin{equation}
\widetilde z
=
\argmin_{\widehat z}
\Bigl\{
\mathcal J_{\mathrm{sm}}(\widehat z;z)
+
\mathcal J_{\mathrm{tan}}(\widehat z)
\Bigr\}.
\label{smooth-inter-piece}
\end{equation}
\end{subequations}
The additional term weakly enforces tangent continuity across the
reconstructed connection.
The minimizer is again obtained by solving the corresponding linear
system for the extracted subchain.

In all numerical experiments, we use
$\wfit=10^{-5}$,
$\wsm=h^2$,
and
$\wtan=1$.
The complete smoothing procedure is summarized in
\Cref{alg:smooth-interface} in \Cref{appendix:aux-algs}.

\subsection{Diagnostic continuation through topological transitions}

The reconstruction procedure described above determines the
post-transition interface family \eqref{gamma-post}, which is
subsequently evolved by the Lagrangian solver. Associated with the
evolving interface are the tangent FTLE diagnostic and the Lagrangian
observation windows introduced in
\Cref{subsec:tangent-ftle,subsec:observation-windows}, respectively.
Because interface surgery replaces portions of the pre-transition
Lagrangian interface by surgically reconstructed geometry, additional
bookkeeping is required to propagate these diagnostics through
topological transitions.

\subsubsection{Closest-point projection of tangent field data}

The tangent FTLE diagnostic \eqref{ftle} depends on both the tangent
field $\xi_\alpha(s,t)$ and the reference quantity
$\xi_\alpha(s,0)$, which are propagated in the numerical method as
auxiliary Lagrangian state variables. At a topological transition,
interface surgery creates a new post-transition interface family, so
this state must also be transferred to the reconstructed interfaces.
This is analogous to the adaptive refinement and coarsening procedure,
where tangent field data are transferred by interpolation. The key
difference is that, after surgery, interpolation along a single
interface is no longer available. Instead, the transfer is performed
through a closest-point projection onto the pre-transition
Lagrangian interface family.

For each reconstructed interface point $\gamma_\alpha(s,T_m^+)$, we define
\begin{equation}\label{eq:ftle-transfer-projection}
\Pi_\alpha^*(s)
\coloneqq
\argmin_{z\in
\bigcup \Gamma_\beta(T_m^-)}
\left|
\gamma_\alpha(s,T_m^+) - z
\right|.
\end{equation}
Suppose that $\Pi_\alpha^*(s_i)$ lies on the pre-transition
segment joining $\gamma_\beta(s_j,T_m^-)$ and
$\gamma_\beta(s_{j+1},T_m^-)$, so that
\begin{equation}\label{eq:ftle-transfer-segment}
\Pi_\alpha^*(s_i)
=
(1-\theta)\gamma_\beta(s_j,T_m^-)
+
\theta\gamma_\beta(s_{j+1},T_m^-),
\qquad
0\leq \theta \leq 1.
\end{equation}
We then transfer the auxiliary tangent state by linear interpolation:
\begin{subequations}\label{eq:ftle-transfer-state}
\begin{align}
\xi_\alpha(s_i,T_m^+)
&=
(1-\theta) \, \xi_\beta(s_j,T_m^-)
+
\theta \, \xi_\beta(s_{j+1},T_m^-),
\\
|\xi_\alpha(s_i,0)|
&=
(1-\theta) \, |\xi_\beta(s_j,0)|
+
\theta \, |\xi_\beta(s_{j+1},0)|.
\end{align}
\end{subequations}
The post-transition tangent stretch factor and tangent FTLE are then
evaluated from \eqref{ftle}.

\subsubsection{Defect measure correction of observation windows}

The observation windows introduced in
\Cref{subsec:observation-windows}
are attached to Lagrangian marker trajectories and therefore remain
focused on the same material region during a purely Lagrangian
evolution. At a topological transition, however, the associated
filamentary structures may undergo break-up with mass loss, causing the 
window centers to drift away from the regions of
interest. We therefore correct the observation windows using the
regularized defect measure $\mathcal D_h$, recentering each window
according to the local defect distribution and enlarging it when
necessary to fully contain the defect region.

Let $\mathcal W_j(T_m^-)$
denote an observation window immediately before a topological transition. 
If the defect measure vanishes
within $\mathcal W_j(T_m^-)$, the original center and half-width are 
retained: $c_j(T_m^+) = c_j(T_m^-)$ and $w_j(T_m^+) = w_j(T_m^-)$.
Otherwise, starting from the pre-transition window center
$c_j(T_m^-)$, we iteratively update
\begin{equation}\label{eq:window-center-correction}
c_j^{(r+1)}
=
\frac{
\int_{\mathcal W_j^{(r)}}
x\,\mathcal D_h(x)\,dx
}{
\int_{\mathcal W_j^{(r)}}
\mathcal D_h(x)\,dx
},
\end{equation}
where $\mathcal W_j^{(r)}$ denotes the observation window centered at
$c_j^{(r)}$ and $c_j^{(0)} = c_j(T_m^-)$.
The iteration terminates once the change in the window center falls
below a prescribed tolerance, and the final iterate is taken as
$c_j(T_m^+)$.
After recentering, the window half-width $w_j(T_m^+)$ is increased, if
necessary, until the support of $\mathcal D_h$ inside the
corrected window no longer intersects a thin boundary strip of
$\mathcal W_j(T_m^+)$.

\section{Numerical examples}
\label{sec:examples}

\subsection{Code implementation and computational platform}

The \mts\ algorithm is implemented in modern Fortran using
double-precision arithmetic and OpenMP parallelization, while MATLAB is used
for post-processing and visualization.  
All numerical experiments were performed on a MacBook Pro
equipped with an M1 Pro processor (10 CPU cores) and 32\,GB of
memory, using 10 OpenMP threads for all reported timings.
All simulations use the default parameter
values listed in \Cref{tab:notation}.

The code is organized around three principal data structures:
the Lagrangian interface family, the sparse Eulerian level-set
representation, and the topology-processing machinery.
The Lagrangian interfaces are stored as closed parametric curves,
while the Eulerian representation is stored using sparse locally
refined level-set blocks restricted to a narrow neighborhood of the
interface. Topological transitions are managed through an event-local
topology structure that mirrors the mathematical framework developed in
\Cref{sec:adjacency,sec:surgery}. At each \mts\ time $T_m$, this
structure stores the pre-processed Lagrangian interface family
$\{\Gammalgr_\alpha\}_{\alpha=1}^{\Nlgr}$, the extracted Eulerian
interface family $\{\Gammaeul_\beta\}_{\beta=1}^{\Neul}$, the canonical
adjacency matrix $\Adj$, and the event decomposition
$\{\Event_1,\ldots,\Event_{\Nevent}\}$. Each event stores the
associated defect measures, surgical piece families $\Plgr$ and
$\Peul$, and the reconstructed interfaces, allowing topological
transitions to be processed independently and locally.

\subsection{Satellite-interface dynamics in the rotating-vortex benchmark}

We first apply the \mts\ algorithm to the rotating-vortex benchmark
introduced in \Cref{subsec:rotating-vortex}. The velocity field $u$,
initial interface $\gamma_\alpha(s,0)$, final time $\tmax$, time-step $\Delta t$, and 
coarse computational scale $h_0$ are the same as in
\eqref{eq:rotating-vortex-velocity}--\eqref{eq:rotating-vortex-params}. 
The fine-scale resolution is given by the dyadic refinement sequence
$h = 2^{-p}h_0$, for $p=1,\ldots,5$. 
Topological processing is performed throughout the filament-formation
phase at the \mts\ times
\begin{equation}
T_m = 0.2\,m,
\qquad
m=1,\ldots,20, 
\end{equation}
and terminating at the time of maximal deformation $t=4$.

The \mts\ solutions computed with the coarsest resolution
$h=\num{2.e-2}$ are displayed in \Cref{fig:rotating-vortex-h=0.02}
at a representative collection of \mts\ times for which topological
transitions occur. In each subfigure, the left panel displays the
topologically processed interface family (blue), together with a
reference solution (black) computed using the classical Lagrangian
tracking algorithm with the finest resolution
$h=\num{1.25e-3}$. Also shown is a heatmap of the regularized defect
measure $\Defect(x,t)$ defined by
\eqref{defect-measure-regularized}. The bright localized hotspots
correspond to topological defects and determine the associated surgical
region $\Surg$ through \eqref{eq:surgical-region}. The right panel
displays the adjacency graph associated with the corresponding
topological event. Red vertices represent the pre-processed
Lagrangian interfaces, green vertices represent the extracted Eulerian
interfaces, and edges indicate the adjacency relations identified by
the topology extraction procedure.

\begin{figure}[ht]
\centering
\begin{subfigure}[t]{0.32\linewidth}
  \centering
  \includegraphics[width=\linewidth]{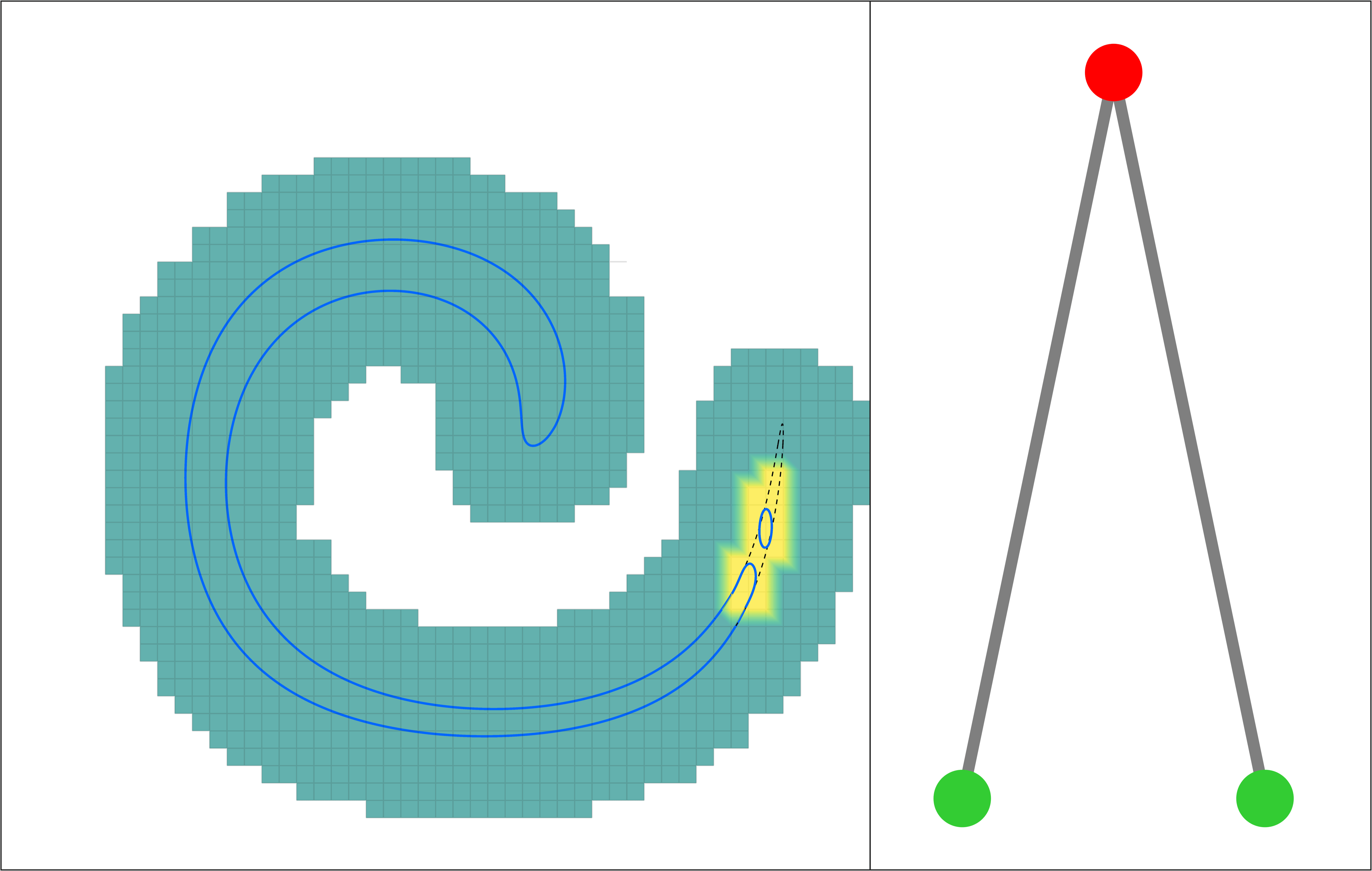}
  \caption{$t=1.4$}
  \label{fig:rotating-vortex-h=0.02_t=1.4}
\end{subfigure}
\begin{subfigure}[t]{0.32\linewidth}
  \centering
  \includegraphics[width=\linewidth]{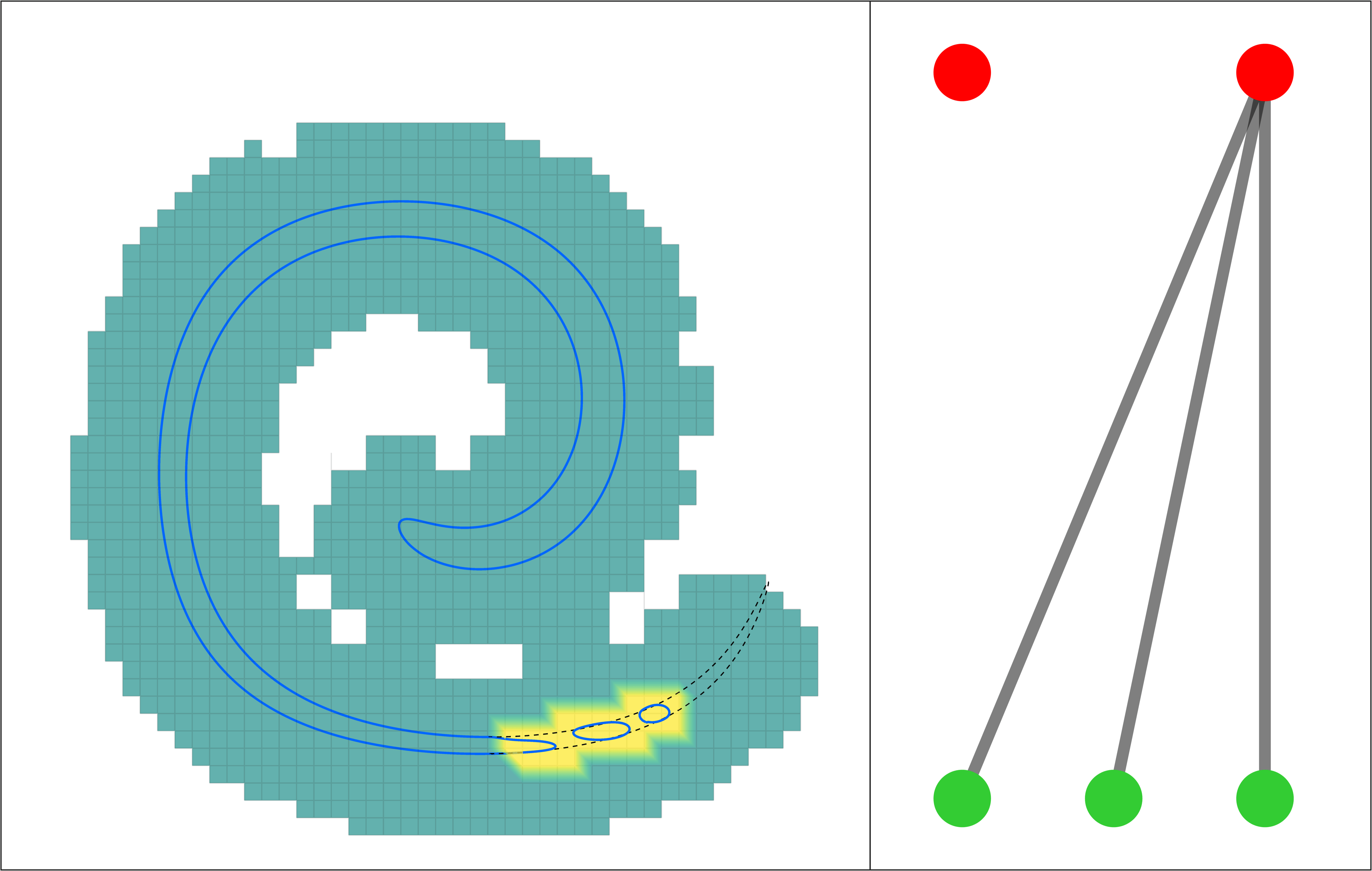}
  \caption{$t=1.8$}
  \label{fig:rotating-vortex-h=0.02_t=1.8}
\end{subfigure}
\begin{subfigure}[t]{0.32\linewidth}
  \centering
  \includegraphics[width=\linewidth]{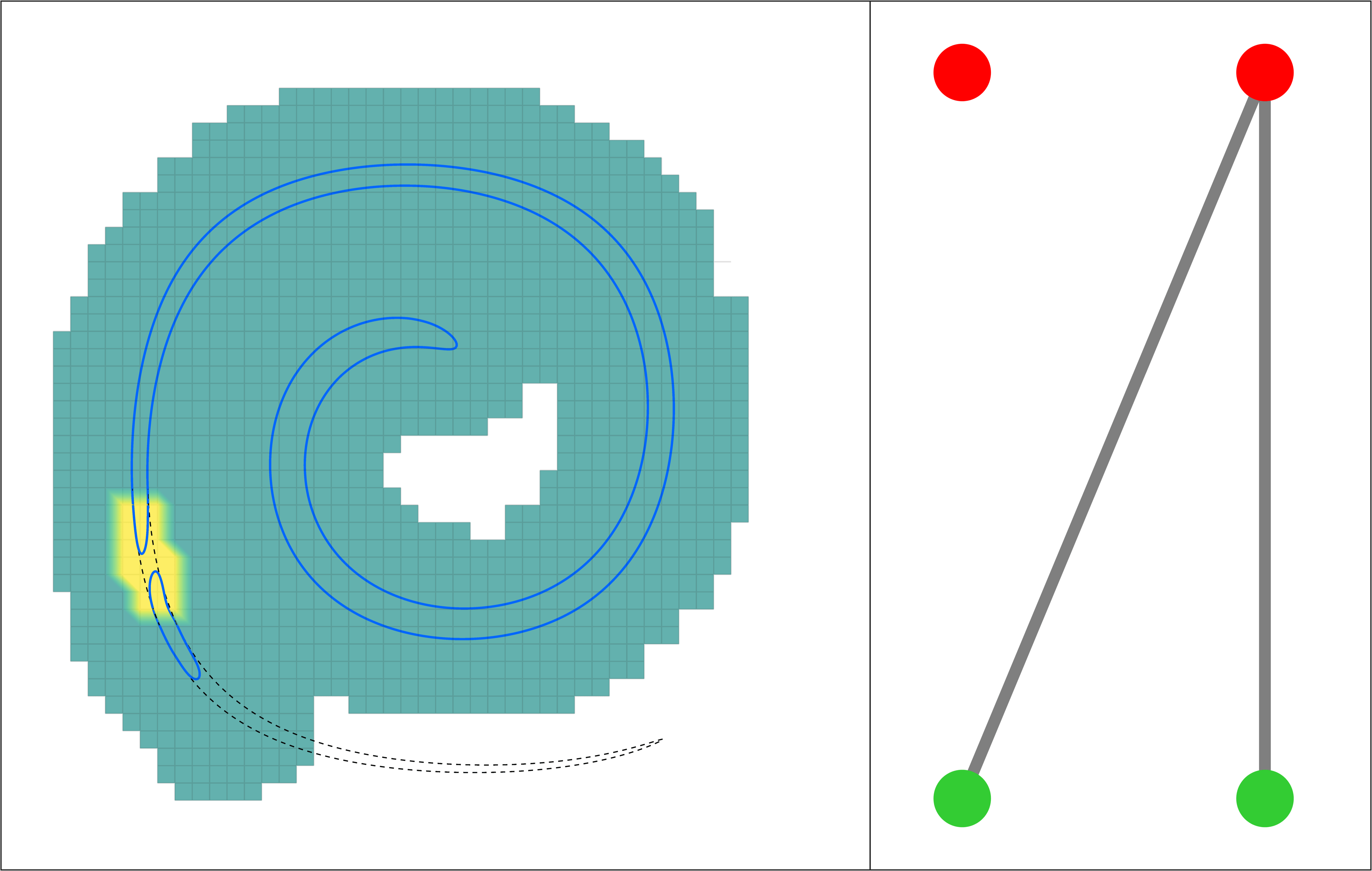}
  \caption{$t=2.6$}
  \label{fig:rotating-vortex-h=0.02_t=2.6}
\end{subfigure}

\vspace{1em}

\begin{subfigure}[t]{0.32\linewidth}
  \centering
  \includegraphics[width=\linewidth]{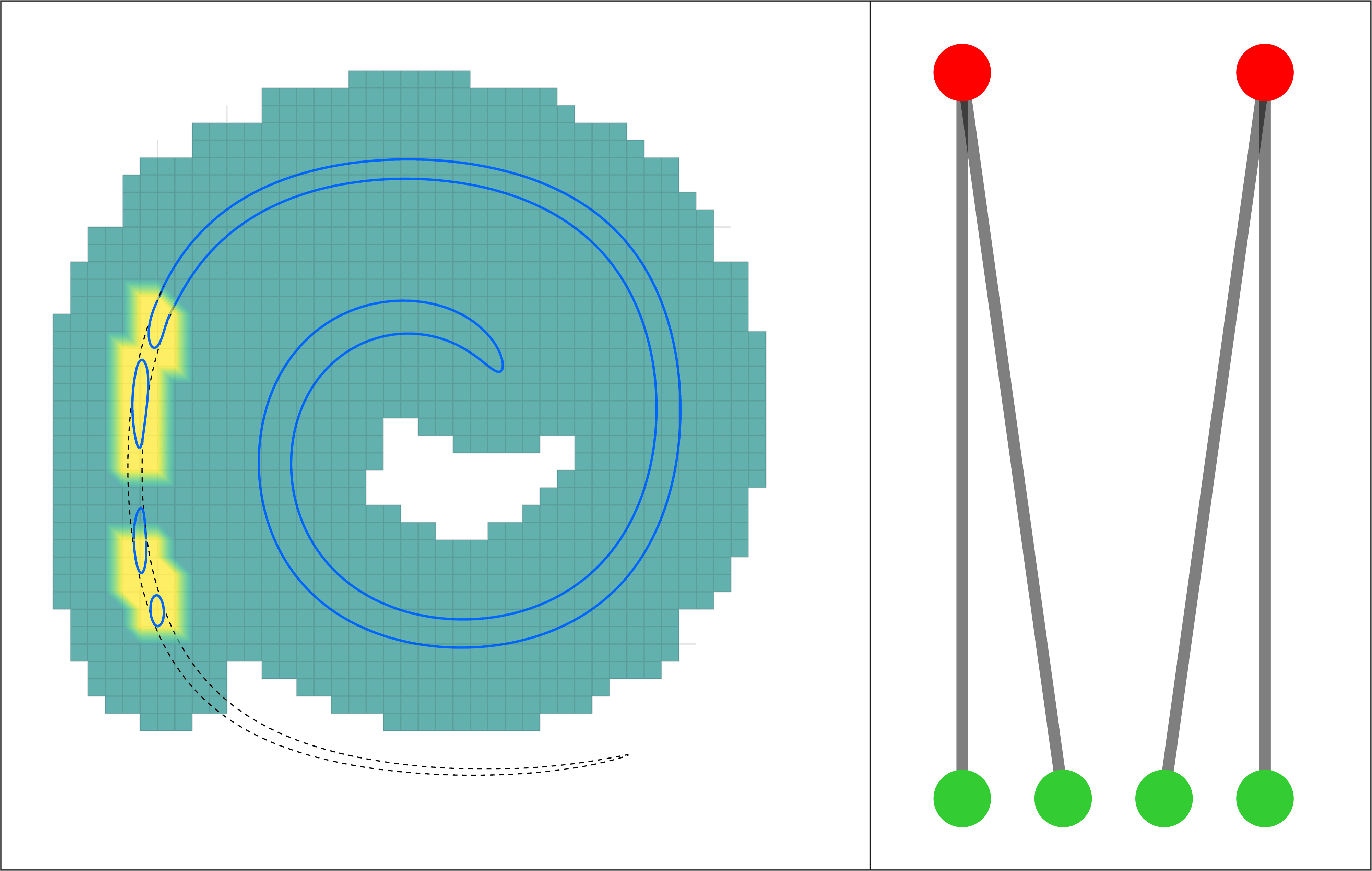}
  \caption{$t=2.8$}
  \label{fig:rotating-vortex-h=0.02_t=2.8}
\end{subfigure}
\begin{subfigure}[t]{0.32\linewidth}
  \centering
  \includegraphics[width=\linewidth]{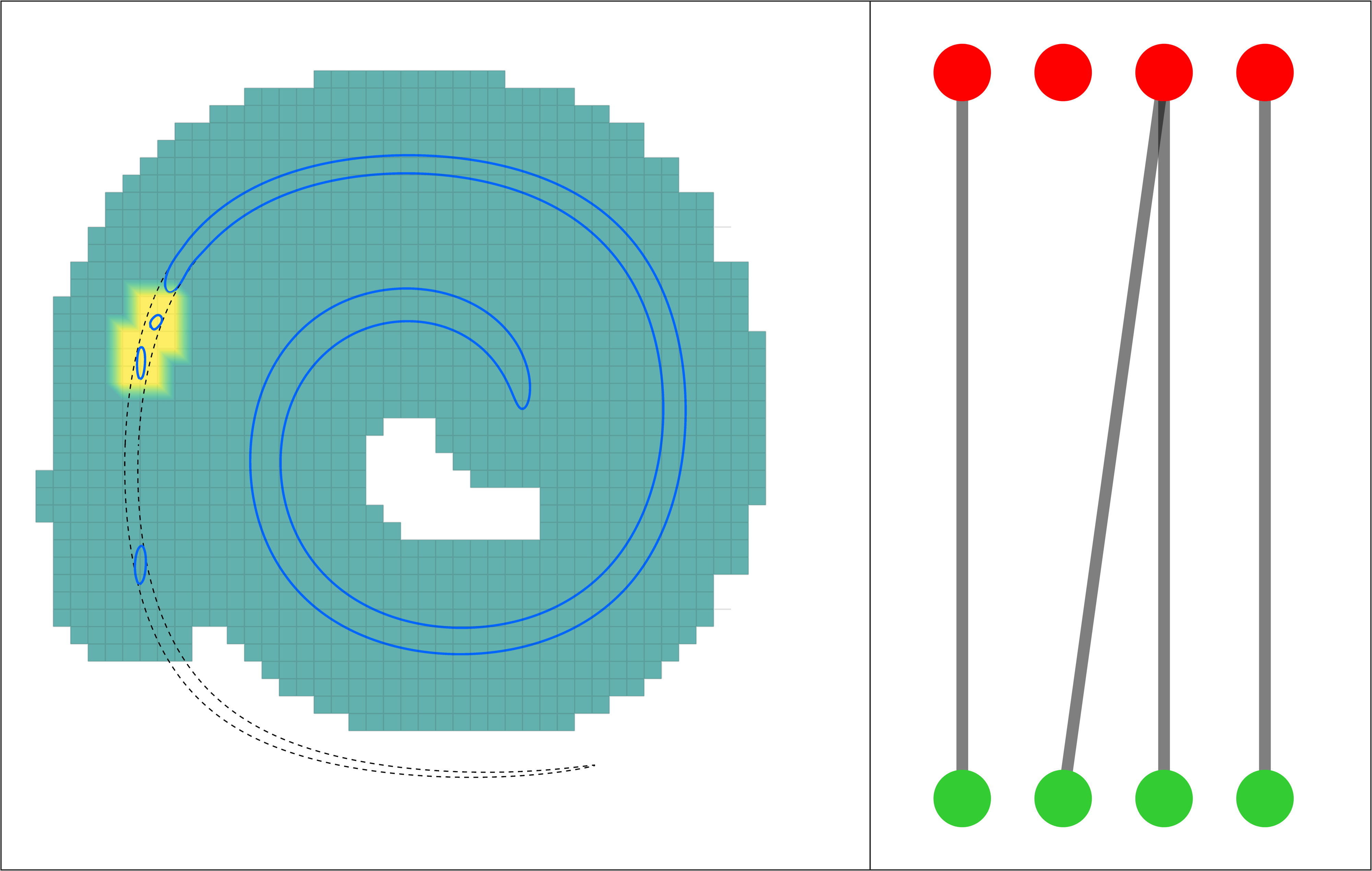}
  \caption{$t=3.0$}
  \label{fig:rotating-vortex-h=0.02_t=3.0}
\end{subfigure}
\begin{subfigure}[t]{0.32\linewidth}
  \centering
  \includegraphics[width=\linewidth]{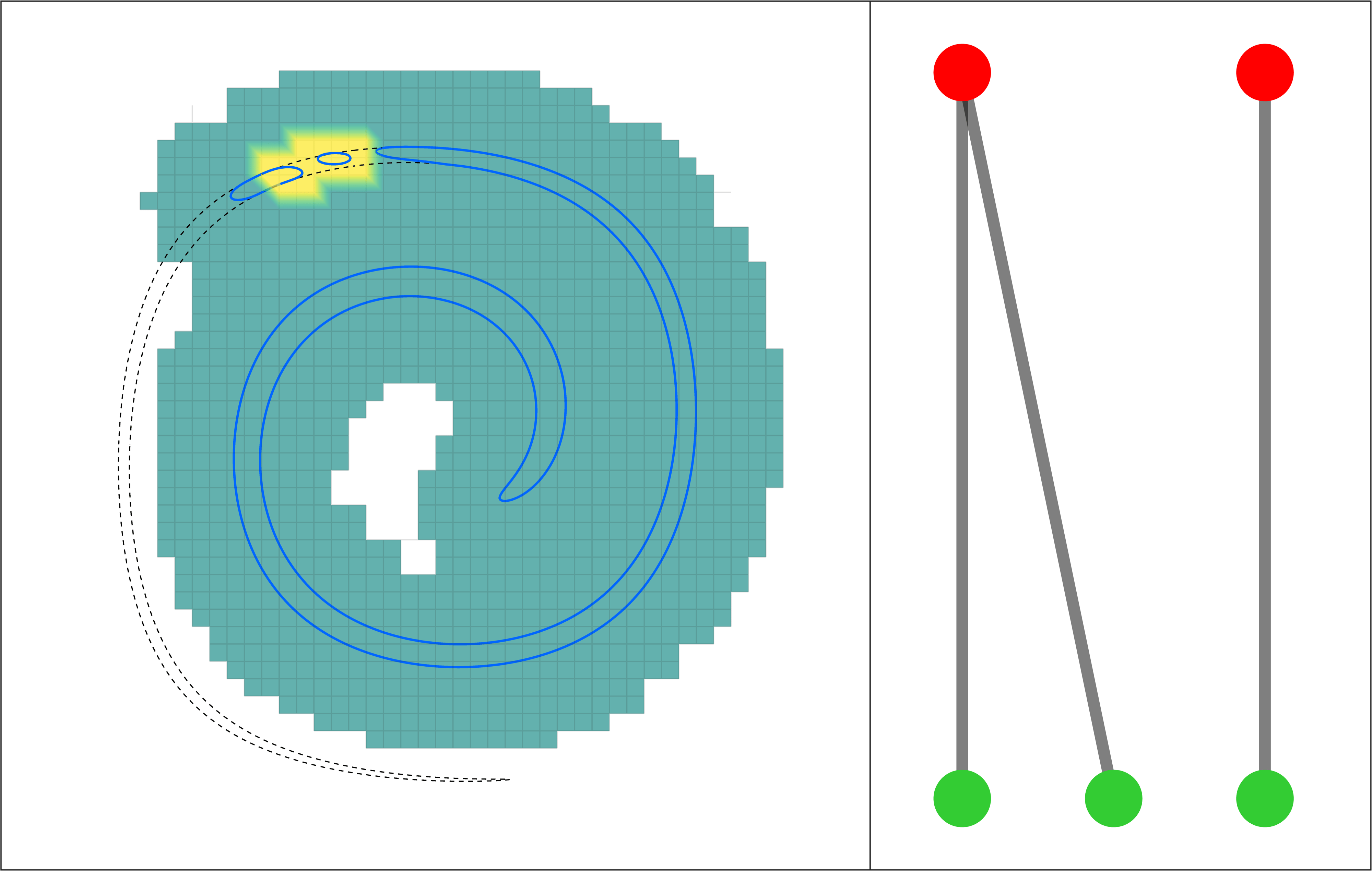}
  \caption{$t=3.8$}
  \label{fig:rotating-vortex-h=0.02_t=3.8}
\end{subfigure}
\caption{
Representative topological transitions for the rotating-vortex
benchmark computed using the \mts\ algorithm with
$h=\num{2.e-2}$.
The left panel in each subfigure displays the \mts\ solution (blue), the 
reference solution computed using classical Lagrangian tracking with
$h=\num{1.25e-3}$ (black), and the regularized defect measure $\Defect(x,t)$. 
The right panel displays the corresponding adjacency graph, with
red vertices denoting pre-processed Lagrangian interfaces, green
vertices denoting extracted Eulerian interfaces, and edges denoting
adjacency relations between the two interface families.
The sequence illustrates repeated \emph{split} and \emph{vanishing}
events generated by the strongly filamenting tip of the interface.
}
\label{fig:rotating-vortex-h=0.02}
\end{figure}

The topological evolution is driven by the strongly filamenting tip
identified in \Cref{subsec:rotating-vortex} by the pronounced trough of
the tangent FTLE near $s=\pi/2$ on the initial interface. As the
filament tip undergoes progressive stretching and compression, it
repeatedly pinches off to form small satellite interfaces, producing a
sequence of one-to-many \emph{split} events. The resulting satellite
interfaces are subsequently stretched and sheared by the flow until
their characteristic length scales fall below the grid scale $h$,
triggering \emph{vanishing} events. 
Qualitatively similar sequences of \emph{split} and \emph{vanishing} 
events have been reported for VOF-based and LCRM methods 
\cite{ShJu2009,ShYoJu2011,ChMaPoZa2022,HePhXi2023,PaLoChScPoZa2024}.
\Cref{fig:rotating-vortex-h=0.02} illustrates this repeated cycle of
satellite-interface formation and disappearance throughout the filament-formation phase. 
Meanwhile, away from the strongly filamenting tip, the \mts\ solution
is indistinguishable from the reference solution, confirming that the topological processing remains
localized to the surgical region and
does not introduce spurious distortions elsewhere along the interface.
This contrasts with VOF-based methods and LCRM, where
the interface geometry may be degraded globally by Eulerian advection
and reconstruction errors \cite{ShJu2009,ShYoJu2011,PaLoChScPoZa2024}.

\subsubsection{Resolution dependence and convergence}
The first panel of \Cref{fig:rotating-vortex-convergence} displays the
\mts\ solution computed with $h = 0.02$ at the final time $\tmax=8$, when the reversible flow
returns the interface to its initial configuration. The dominant error
is the mass loss associated with the strongly filamenting tip. 
As the fine-scale $h$ is decreased
(\Cref{fig:rotating-vortex-h=0.01_t=8.0,fig:rotating-vortex-h=0.005_t=8.0,fig:rotating-vortex-h=0.0025_t=8.0}),
progressively smaller filamentary structures are resolved by the
underlying Lagrangian representation. Consequently, fewer
\emph{split} and \emph{vanishing} events are triggered, leading to a
corresponding reduction in mass loss. In the finest simulation
$h=\num{1.25e-3}$, all dynamically relevant length scales remain
resolved throughout the evolution and no topological events are
detected. Thus, the fine-scale parameter $h$ provides direct control
over the degree of topological processing, with the \mts\ algorithm
reducing to classical Lagrangian tracking in the limit that all
filamentary scales are resolved.

\begin{figure}[ht]
\centering
\begin{subfigure}[t]{0.24\linewidth}
  \centering
  \includegraphics[width=0.95\linewidth]{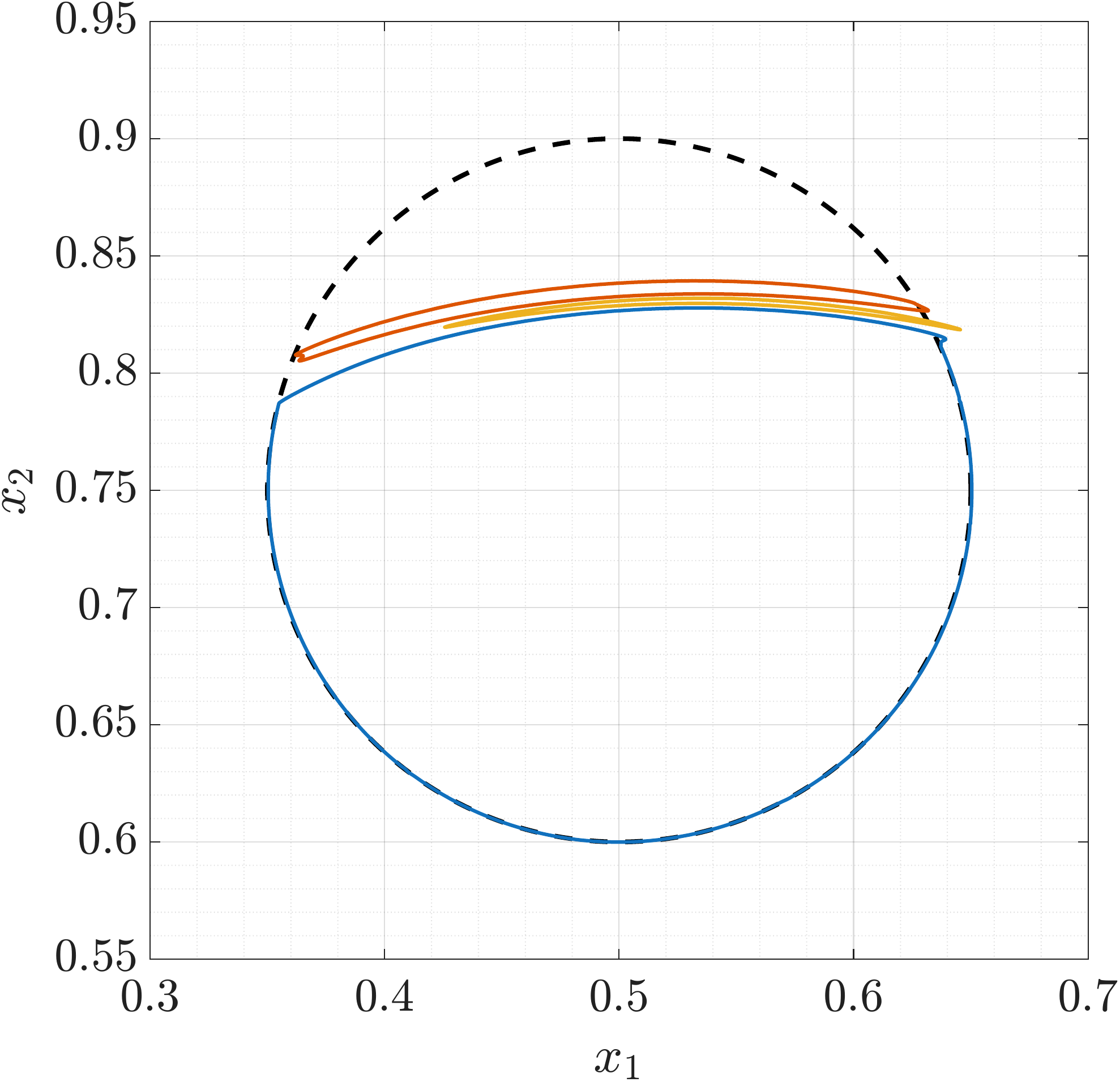}
  \caption{$h = \num{2.e-2}$}
  \label{fig:rotating-vortex-h=0.02_t=8.0}
\end{subfigure}
\begin{subfigure}[t]{0.24\linewidth}
  \centering
  \includegraphics[width=0.95\linewidth]{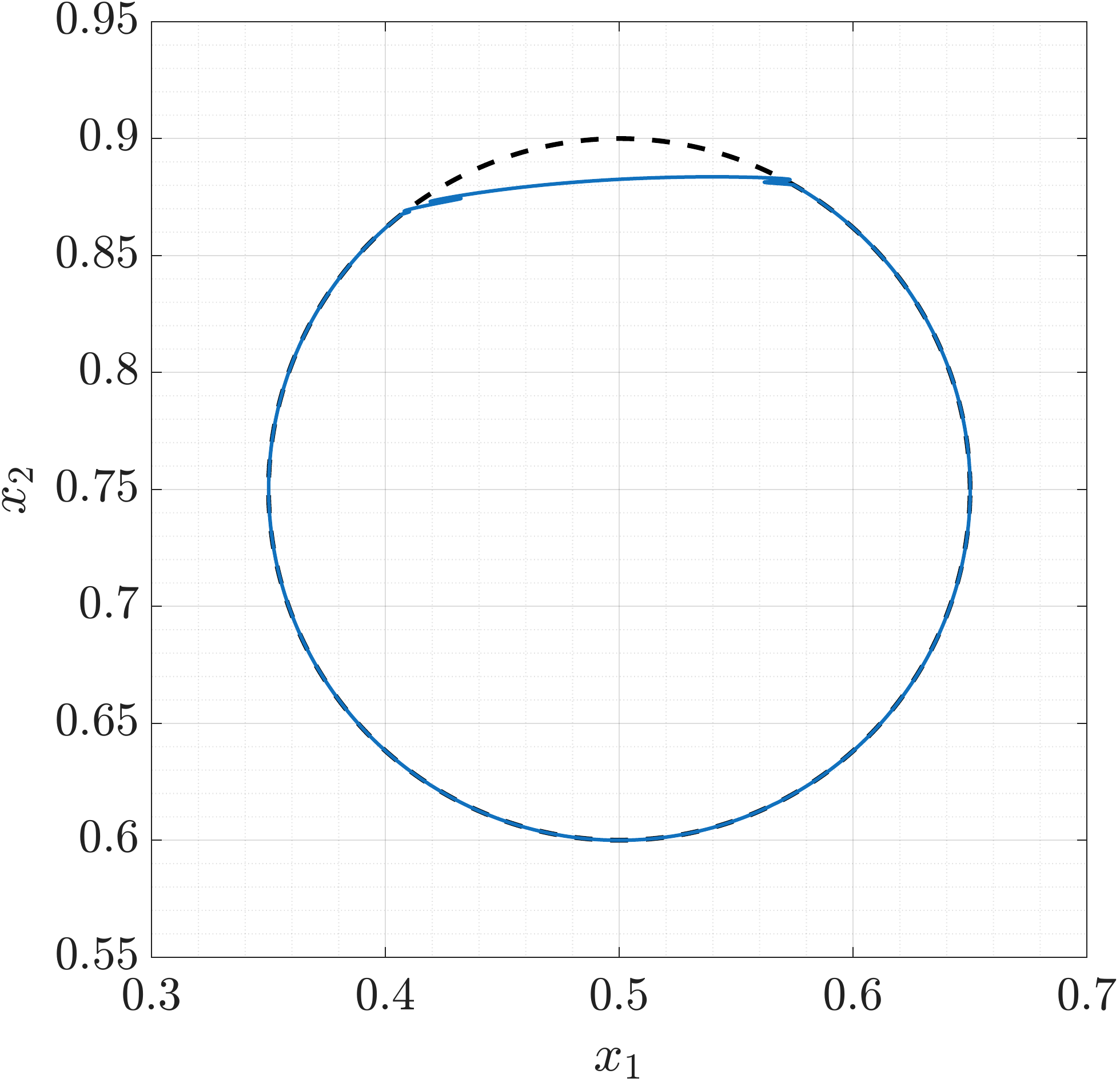}
  \caption{$h = \num{1.e-2}$}
  \label{fig:rotating-vortex-h=0.01_t=8.0}
\end{subfigure}
\begin{subfigure}[t]{0.24\linewidth}
  \centering
  \includegraphics[width=0.95\linewidth]{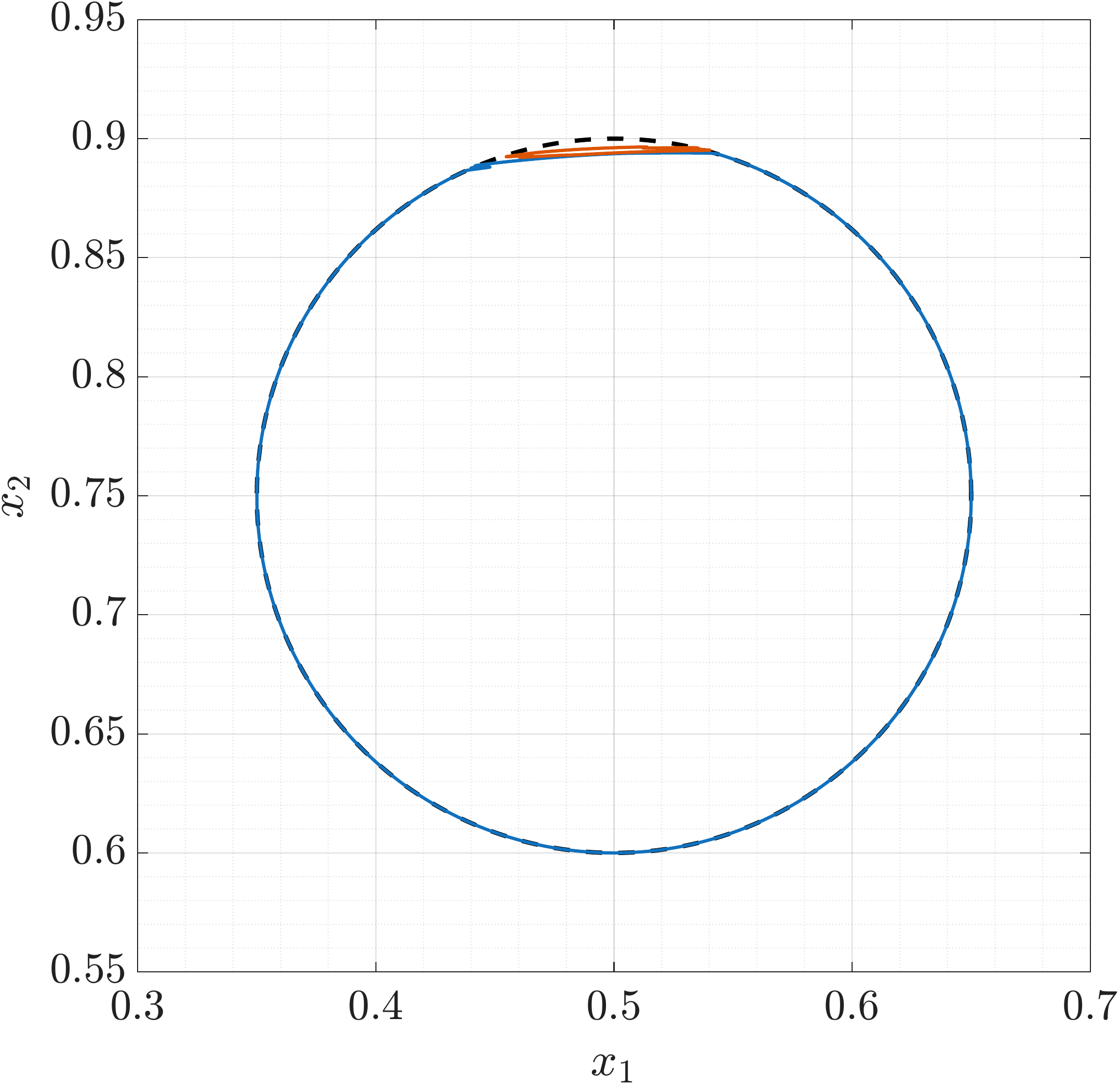}
  \caption{$h = \num{5.e-3}$}
  \label{fig:rotating-vortex-h=0.005_t=8.0}
\end{subfigure}
\begin{subfigure}[t]{0.24\linewidth}
  \centering
  \includegraphics[width=0.95\linewidth]{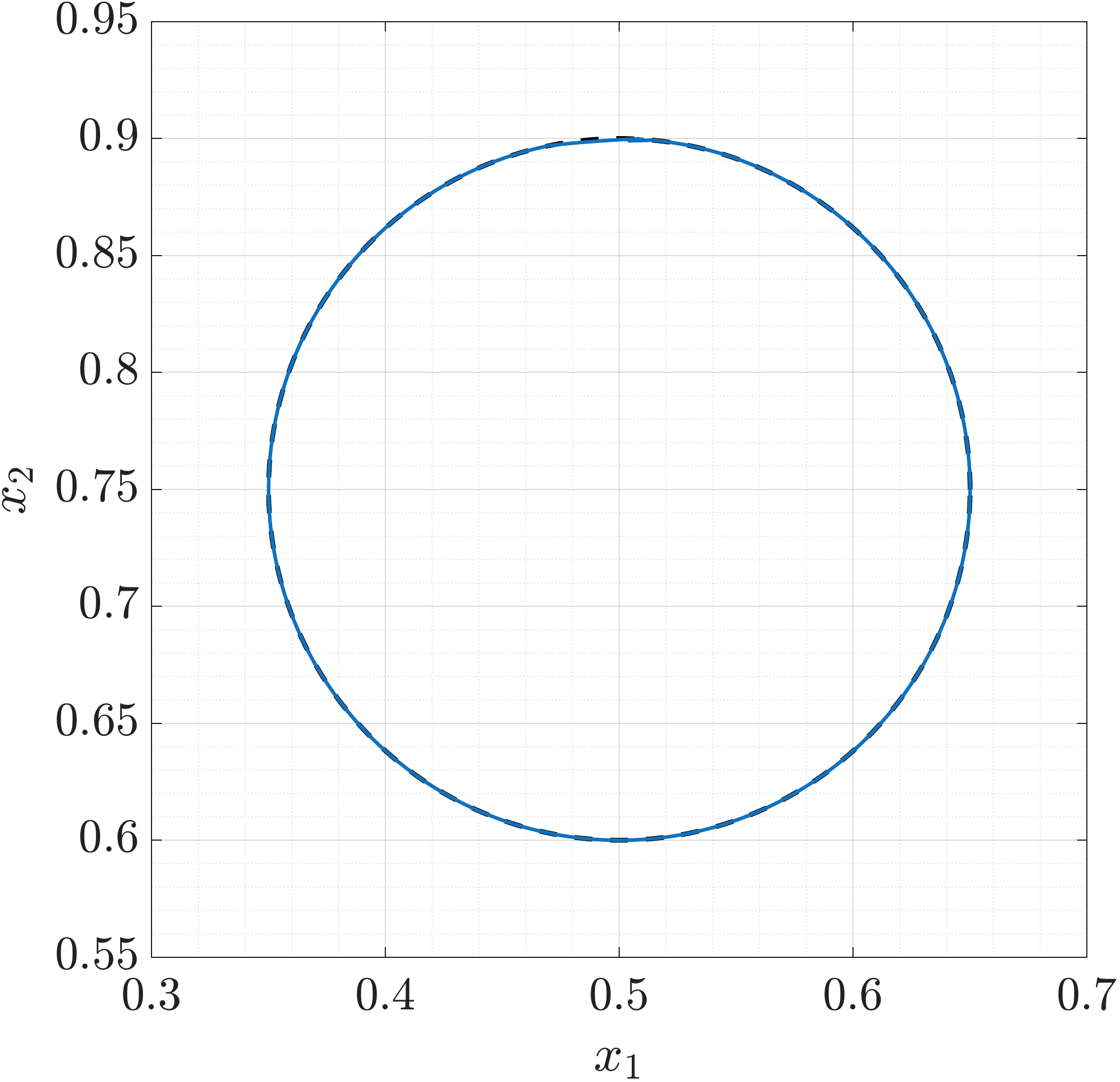}
  \caption{$h = \num{2.5e-3}$}
  \label{fig:rotating-vortex-h=0.0025_t=8.0}
\end{subfigure}
\caption{
\mts\ solutions for the rotating-vortex benchmark at
$\tmax=8$ for decreasing values of the fine-scale parameter $h$.
The colored curves denote the \mts\ interface family, while the black
curve denotes the reference solution computed using classical
Lagrangian tracking with $h=\num{1.25e-3}$. As $h$ decreases,
progressively smaller filamentary structures are resolved, resulting in
fewer split and vanishing events and correspondingly less mass loss.
}
\label{fig:rotating-vortex-convergence}
\end{figure}

The reduction in mass loss observed in
\Cref{fig:rotating-vortex-convergence} is confirmed quantitatively by
the results in \Cref{tab:rotating-vortex-runtime}. Meanwhile, the runtime statistics
show that the cost of the \mts\ algorithm remains comparable to that of the
underlying Lagrangian tracking algorithm, accounting for roughly
$20$--$70\%$ of the total runtime over the resolutions considered.
Since resolving progressively smaller filamentary structures reduces the frequency of
topological events and the associated surgery operations, the relative cost of
the \mts\ algorithm decreases as $h$ is reduced. For this problem, the dominant
contribution to the \mts\ overhead is the Eulerian extraction stage, which typically
accounts for roughly $50$--$80\%$ of the \mts\ runtime. By contrast, adjacency-topology
construction contributes only a few percent of the \mts\ cost and is therefore
negligible. The surgery stage accounts for approximately $15$--$45\%$ of the
\mts\ runtime at coarse and intermediate resolutions where repeated topological
transitions occur, but its relative cost decreases as the fine scale $h$ is reduced,
ultimately vanishing at the finest resolution where no topological events occur.

\begin{table}[ht]
\centering
\scriptsize
\caption{
Accuracy and runtime statistics for solutions to the rotating-vortex
benchmark computed using the Lagrangian tracking\,+\,\mts\
topology-processing algorithm.
The table reports the numerical reversal error $|A_K-A_0|$, together with wall-clock
runtimes for the Lagrangian tracking component, the core \mts\
algorithm, and the overall computation.
}
\label{tab:rotating-vortex-runtime}
\begin{tabular}{l@{\hspace{2em}}ccccc}
\toprule
$h$
& \num{2e-2}
& \num{1e-2}
& \num{5e-3}
& \num{2.5e-3}
& \num{1.25e-3}
\\
\midrule

$|A_K-A_0|$
& \num{1.35e-2}
& \num{1.83e-3}
& \num{2.70e-4}
& \num{2.00e-5}
& \num{1.92e-8}
\\
Order
& --
& 2.88
& 2.76
& 3.76
& 10.02
\\
\midrule

\LGR\ runtime (s)
& \num{8.00e-2}
& \num{3.38e-1}
& 1.75
& 5.73
& 7.86
\\
Order
& --
& 2.08
& 2.37
& 1.71
& 0.46
\\
\midrule

\mts\ runtime (s)
& \num{1.86e-1}
& \num{2.84e-1}
& \num{8.31e-1}
& 1.35
& 2.72
\\
Order
& --
& 0.61
& 1.55
& 0.70
& 1.01
\\
\midrule

Total runtime (s)
& \num{2.66e-1}
& \num{6.22e-1}
& 2.58
& 7.07
& 10.6
\\
Order
& --
& 1.23
& 2.05
& 1.45
& 0.58
\\

\bottomrule
\end{tabular}
\end{table}

While direct comparison is complicated by differing error
metrics,\footnote{The errors reported in \cite{JeSuSh2015,HePhXi2023}
are Eulerian volume-fraction errors, whereas
\Cref{tab:rotating-vortex-runtime} reports the numerical reversal error
$|A_K-A_0|$.} the results compare favorably with those reported in
Table~1 of \cite{JeSuSh2015} and Table~2 of \cite{HePhXi2023}, which
exhibit larger errors and substantially larger runtimes at comparable
effective resolutions. Although such comparisons should be interpreted with caution, the 
combination of the quantitative results and the qualitative comparisons presented here 
strongly suggests that the localized, event-driven nature of the \mts\ algorithm enables 
substantially higher accuracy at significantly lower computational cost than the MOF filament-aware 
interface-capturing approaches \cite{JeSuSh2015,HePhXi2023}.

\subsubsection{Tangent FTLE through topological transitions}
\Cref{fig:rotating-vortex-ftle} displays the tangent FTLE diagnostic 
obtained through the projection \eqref{eq:ftle-transfer-projection}--\eqref{eq:ftle-transfer-state}. 
For visualization, the tangent FTLE is plotted directly over the computed interface family in 
physical space.\footnote{Because topological transitions alter the structure of the interface 
family, the tangent FTLE no longer admits a natural global parameter-space representation. 
Plotting the diagnostic directly on the interface geometry avoids this ambiguity.}
These results should be contrasted with the classical Lagrangian computation shown in 
\Cref{fig:dyadic-refinement}, where the tangent FTLE exhibits two localized troughs corresponding 
to the strong and weak filamenting tips. The trough associated with the weakly 
filamenting tip remains essentially unchanged in the \mts\ simulations. 
By contrast, the highly localized trough associated with the strongly filamenting tip becomes progressively 
attenuated as the interface is propagated through successive split and vanishing events. 
This attenuation reflects the removal of fine filamentary scales by the topological processing. 
As the fine-scale $h$ is decreased, a smaller range of filamentary structures is affected by 
surgery, and the attenuation correspondingly weakens. In the limit that all dynamically relevant scales 
are resolved, the \mts\ tangent FTLE converges to the classical Lagrangian diagnostic.

\begin{figure}[ht]
\centering
\begin{subfigure}[t]{0.19\linewidth}
  \centering
  \includegraphics[width=\linewidth]{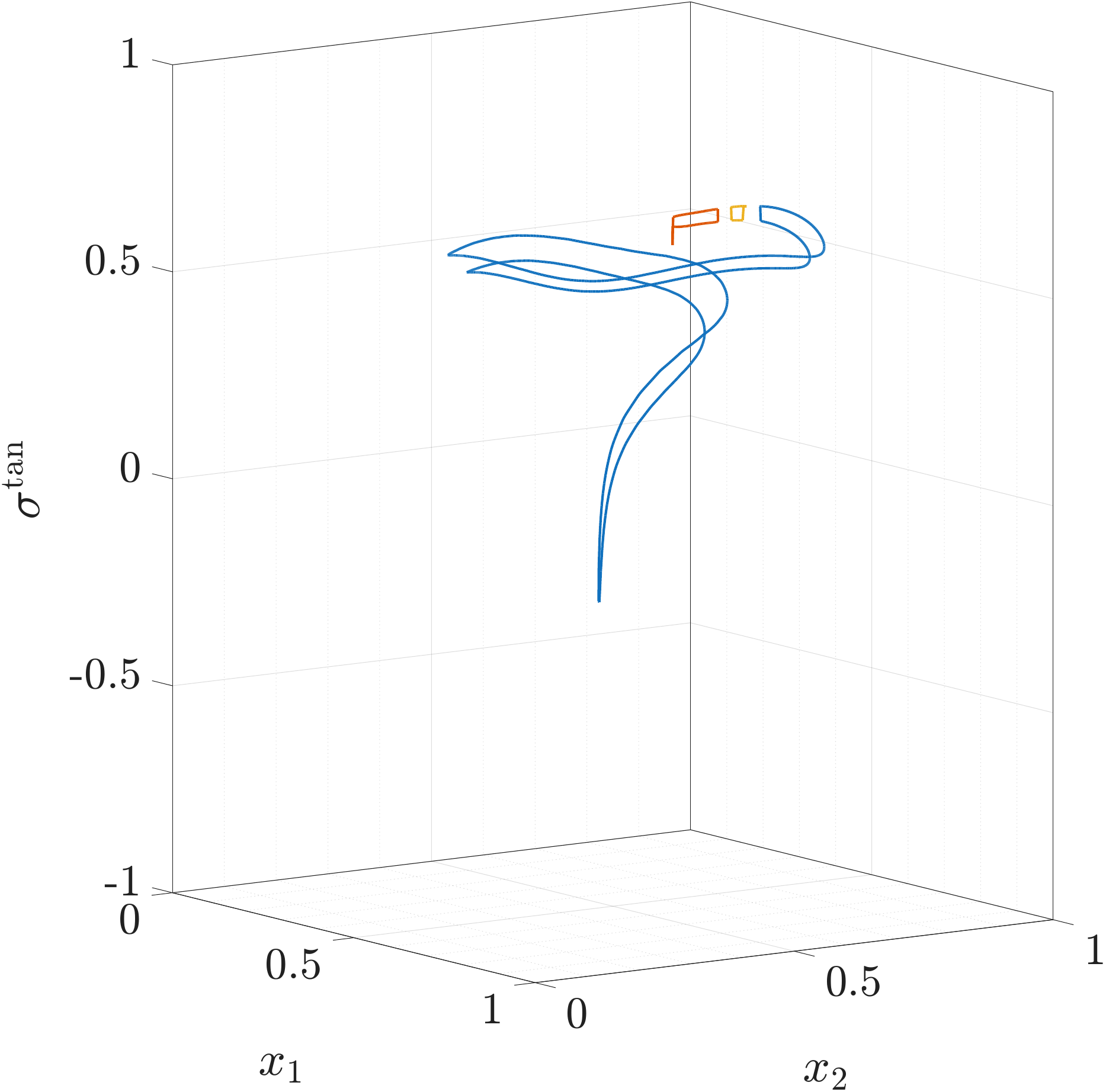}
  \caption{$h=\num{2.e-2}$}
  \label{fig:rotating-vortex-ftle-h=0.02}
\end{subfigure}
\begin{subfigure}[t]{0.19\linewidth}
  \centering
  \includegraphics[width=\linewidth]{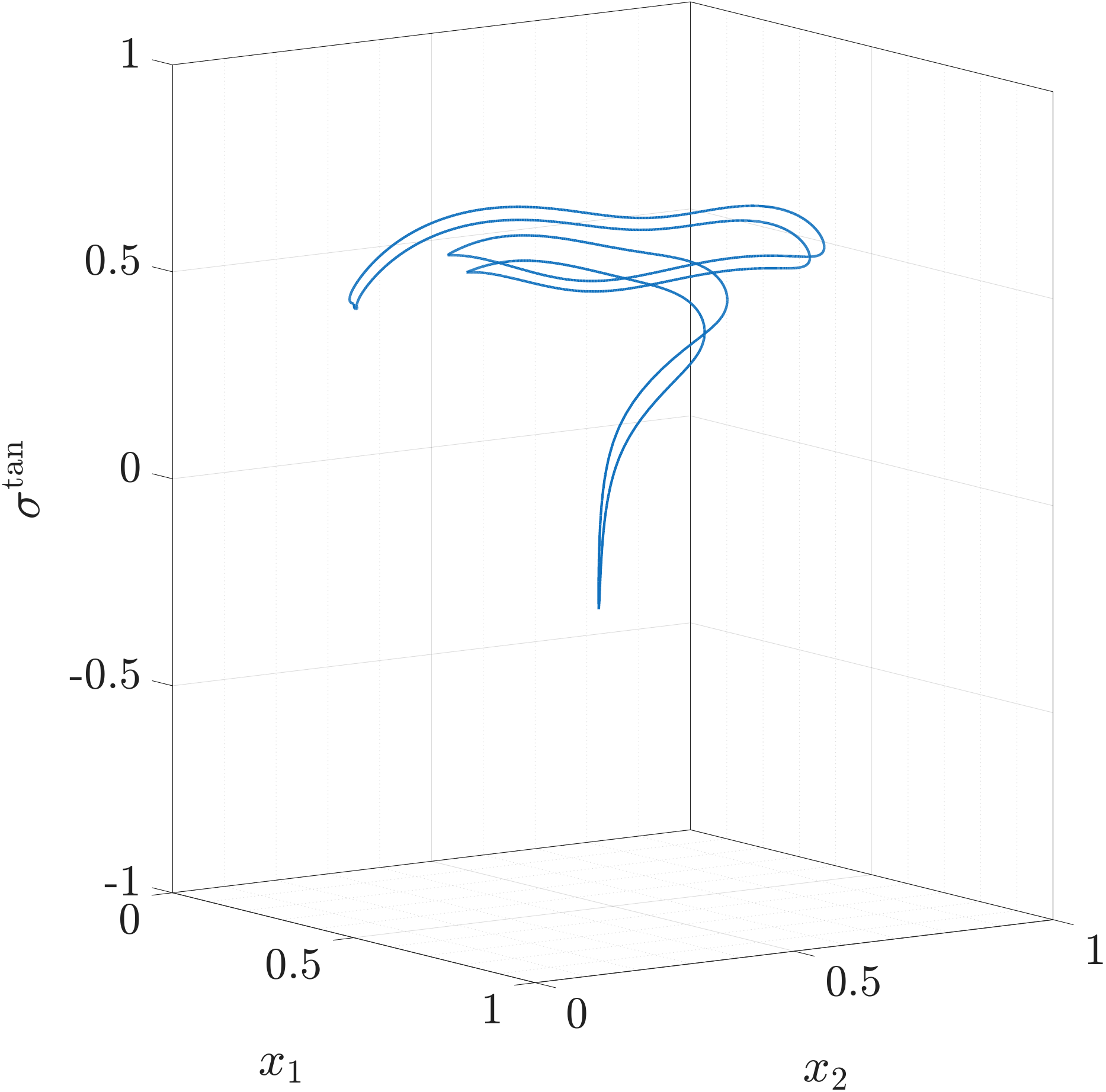}
  \caption{$h=\num{1.e-2}$}
  \label{fig:rotating-vortex-ftle-h=0.01}
\end{subfigure}
\begin{subfigure}[t]{0.19\linewidth}
  \centering
  \includegraphics[width=\linewidth]{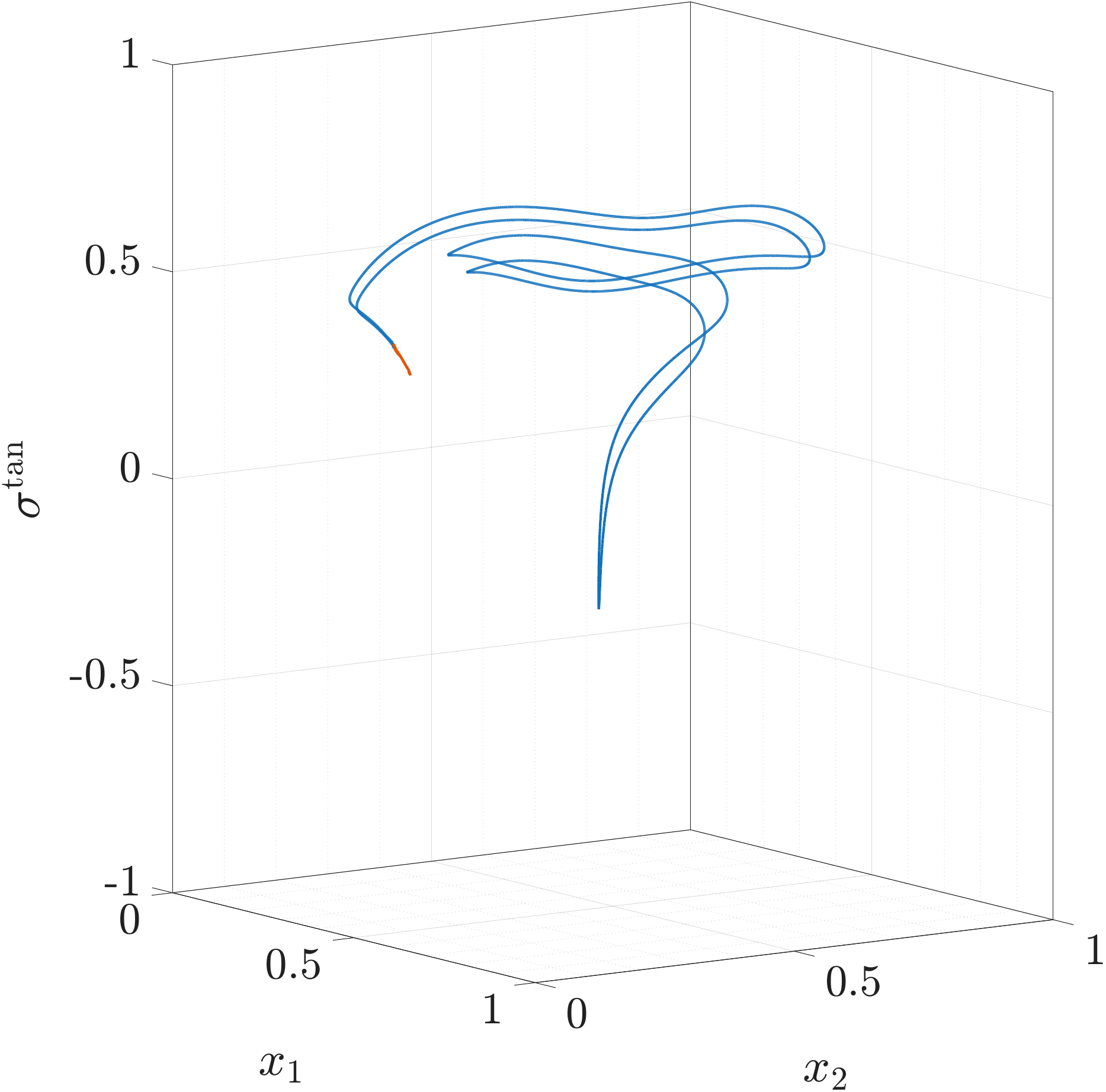}
  \caption{$h=\num{5.e-3}$}
  \label{fig:rotating-vortex-ftle-h=0.005}
\end{subfigure}
\begin{subfigure}[t]{0.19\linewidth}
  \centering
  \includegraphics[width=\linewidth]{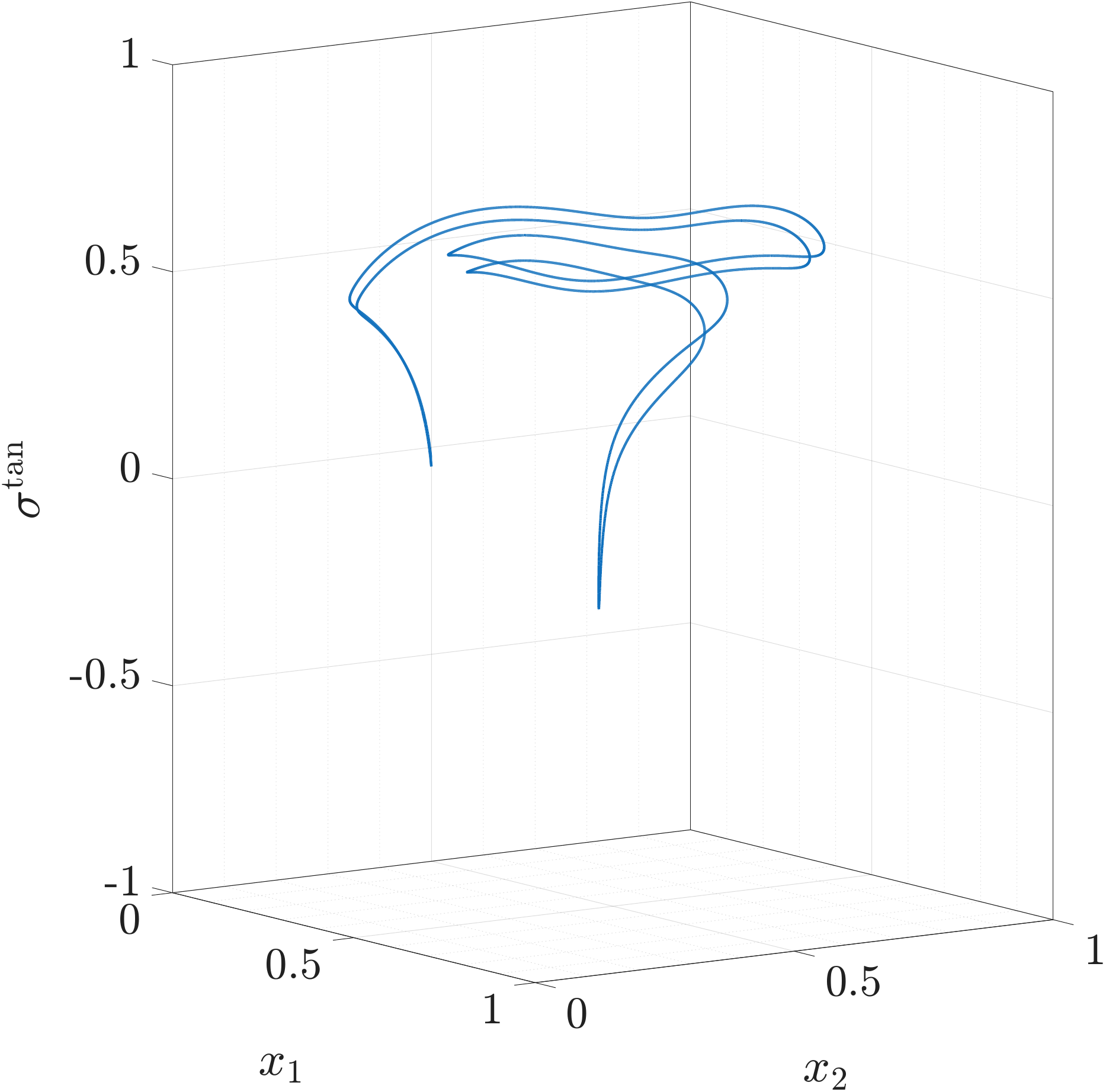}
  \caption{$h=\num{2.5e-3}$}
  \label{fig:rotating-vortex-ftle-h=0.0025}
\end{subfigure}
\begin{subfigure}[t]{0.19\linewidth}
  \centering
  \includegraphics[width=\linewidth]{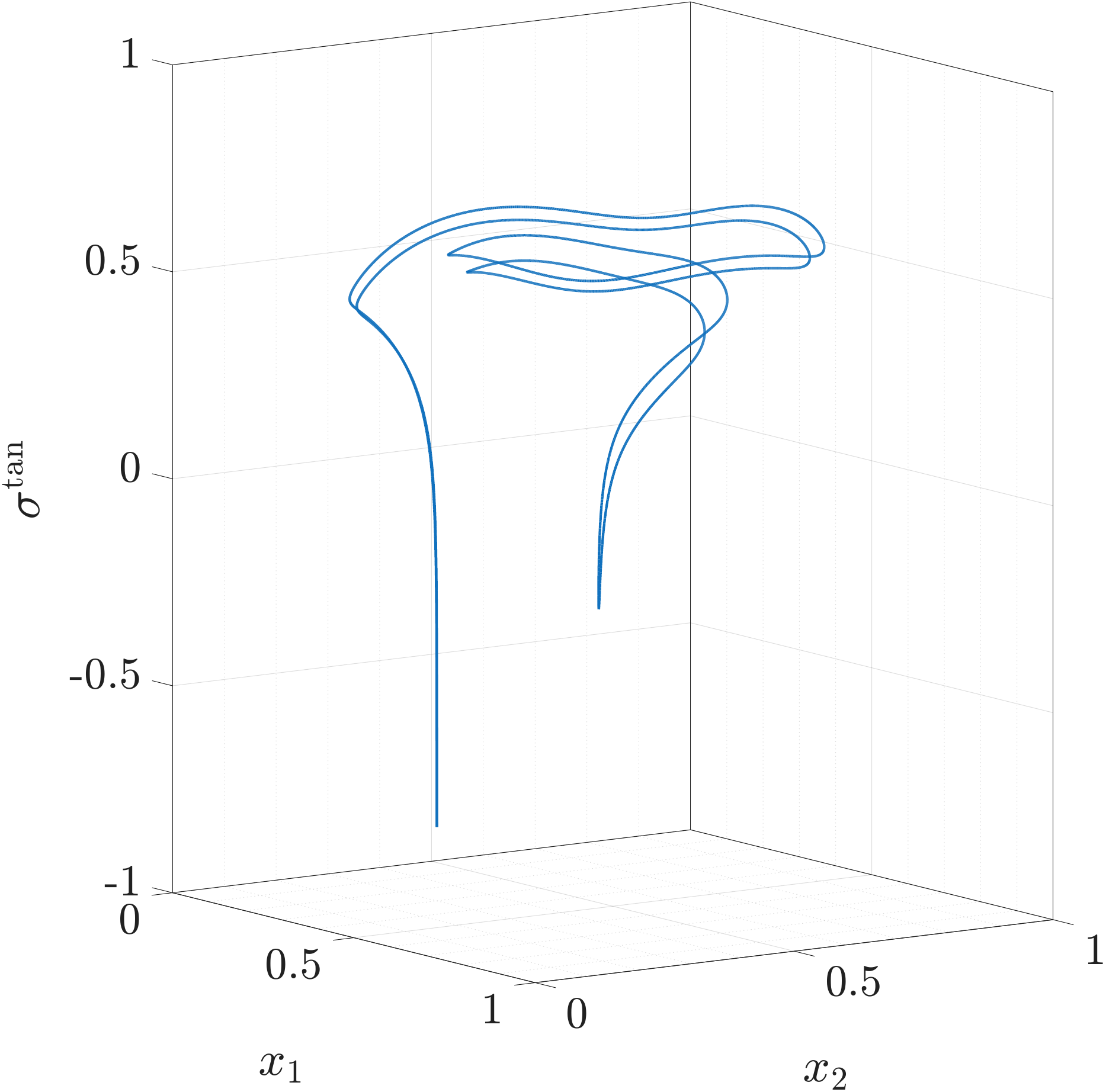}
  \caption{$h=\num{1.25e-3}$}
  \label{fig:rotating-vortex-ftle-h=0.00125}
\end{subfigure}
\caption{
Tangent FTLE $\sftle(s,t)$ for the \mts\ solutions of the
rotating-vortex benchmark at the time of maximal deformation $t=4$ for
decreasing values of the fine-scale parameter $h$. The diagnostic is
visualized directly on the computed interface family. As $h$
decreases, the trough associated with the strongly filamenting tip
becomes increasingly pronounced, approaching the classical Lagrangian
diagnostic shown in \Cref{fig:dyadic-refinement}, while the trough
associated with the weakly filamenting tip remains essentially unchanged.
}
\label{fig:rotating-vortex-ftle}
\end{figure}

The geometric consequences of the topological processing may be
understood through the stretch--curvature relation \eqref{stretch-curve-law}. 
\Cref{fig:rotating-vortex-scaling-MTS} displays the corresponding phase plot 
in a neighborhood of the strongly filamenting tip. 
Since the $\kappa^{-1/3}$ branch is generated by the filament
tip itself, removal of filamentary structure below the prescribed
microscale $h$ by the \mts\ algorithm produces a corresponding
truncation of the branch. At coarse resolutions, only a short segment
of the branch remains.As $h$ is decreased, less of the filament tip is removed by the
topological processing, and the branch extends toward the classical
Lagrangian result. The \mts\ phase plots, particularly at coarse resolutions, are somewhat 
noisy away from the principal $\kappa^{-1/3}$ branch. 
This behavior reflects the sensitivity of curvature, a second-derivative quantity, to the 
small geometric errors introduced in the surgical regions, where the interface is 
reconstructed by stitching together Lagrangian and Eulerian interface segments. 
By contrast, the tangent FTLE remains smooth, as shown 
in \Cref{fig:rotating-vortex-ftle}, confirming its robustness
as a diagnostic of filamentation through topological change.

\begin{figure}[ht]
\centering
\begin{subfigure}[t]{0.19\linewidth}
  \centering
  \includegraphics[width=0.995\linewidth]{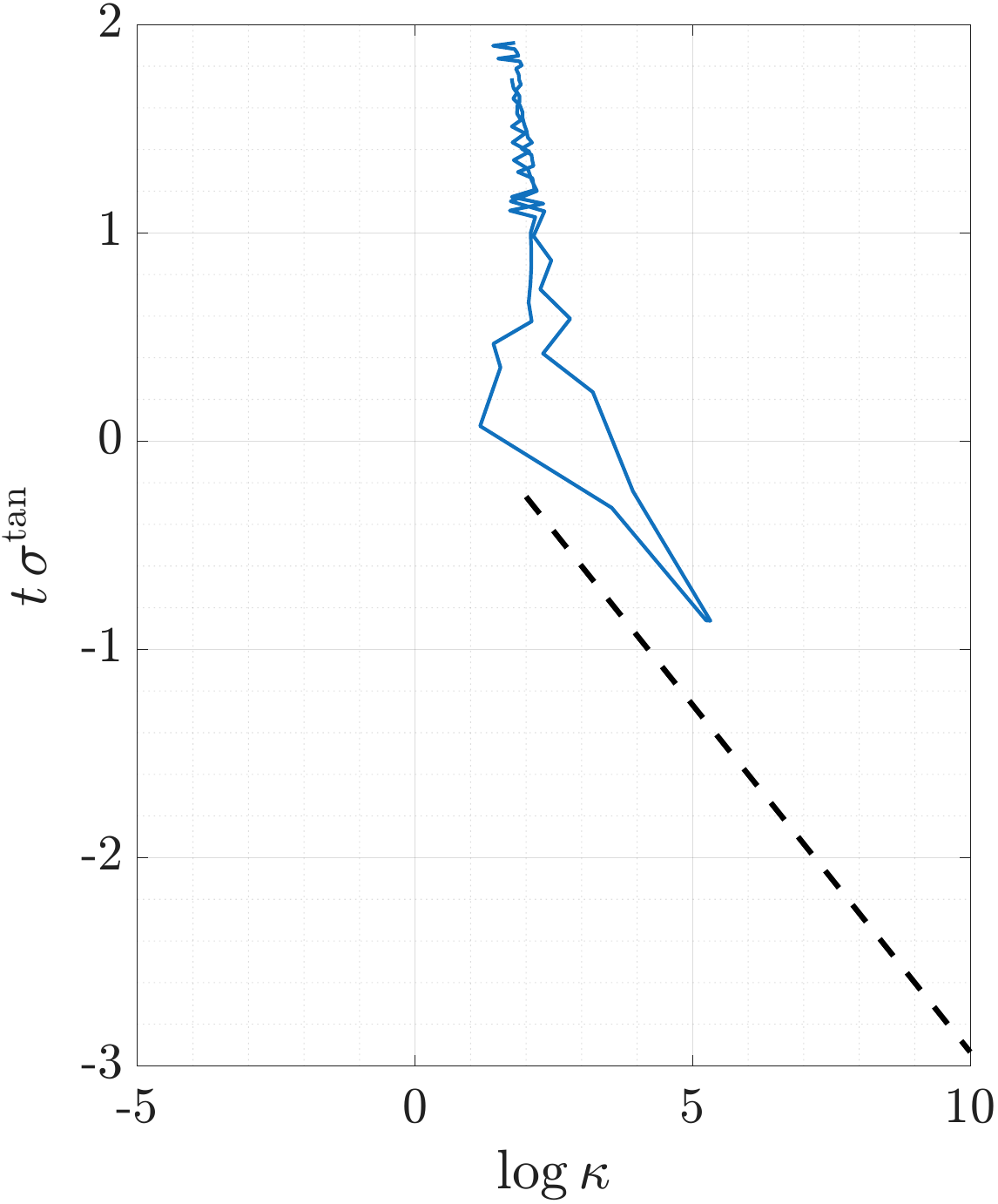}
  \caption{$h=\num{2.e-2}$}
  \label{fig:rotating-vortex-scaling-h=0.02}
\end{subfigure}
\begin{subfigure}[t]{0.19\linewidth}
  \centering
  \includegraphics[width=\linewidth]{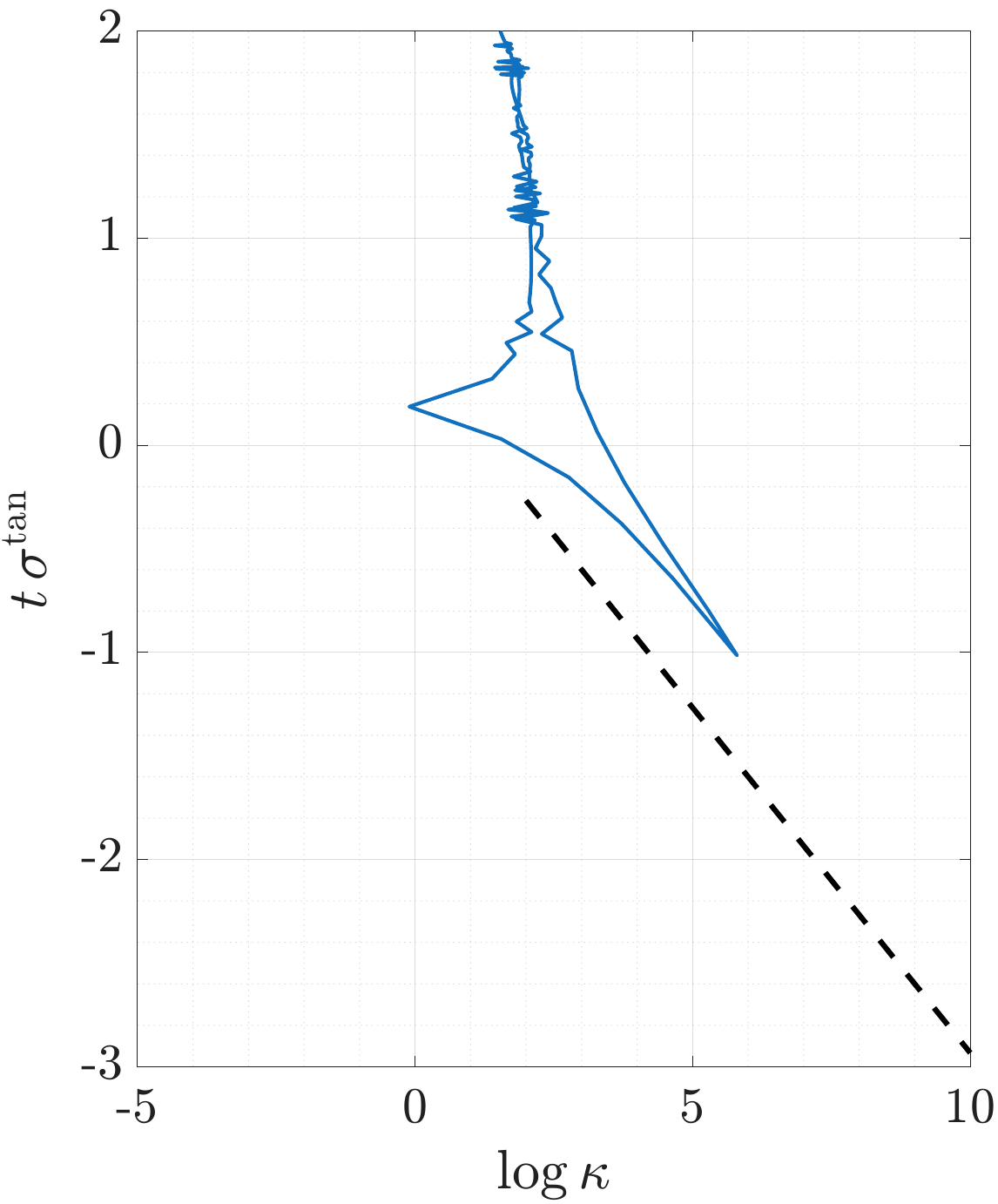}
  \caption{$h=\num{1.e-2}$}
  \label{fig:rotating-vortex-scaling-h=0.01}
\end{subfigure}
\begin{subfigure}[t]{0.19\linewidth}
  \centering
  \includegraphics[width=0.995\linewidth]{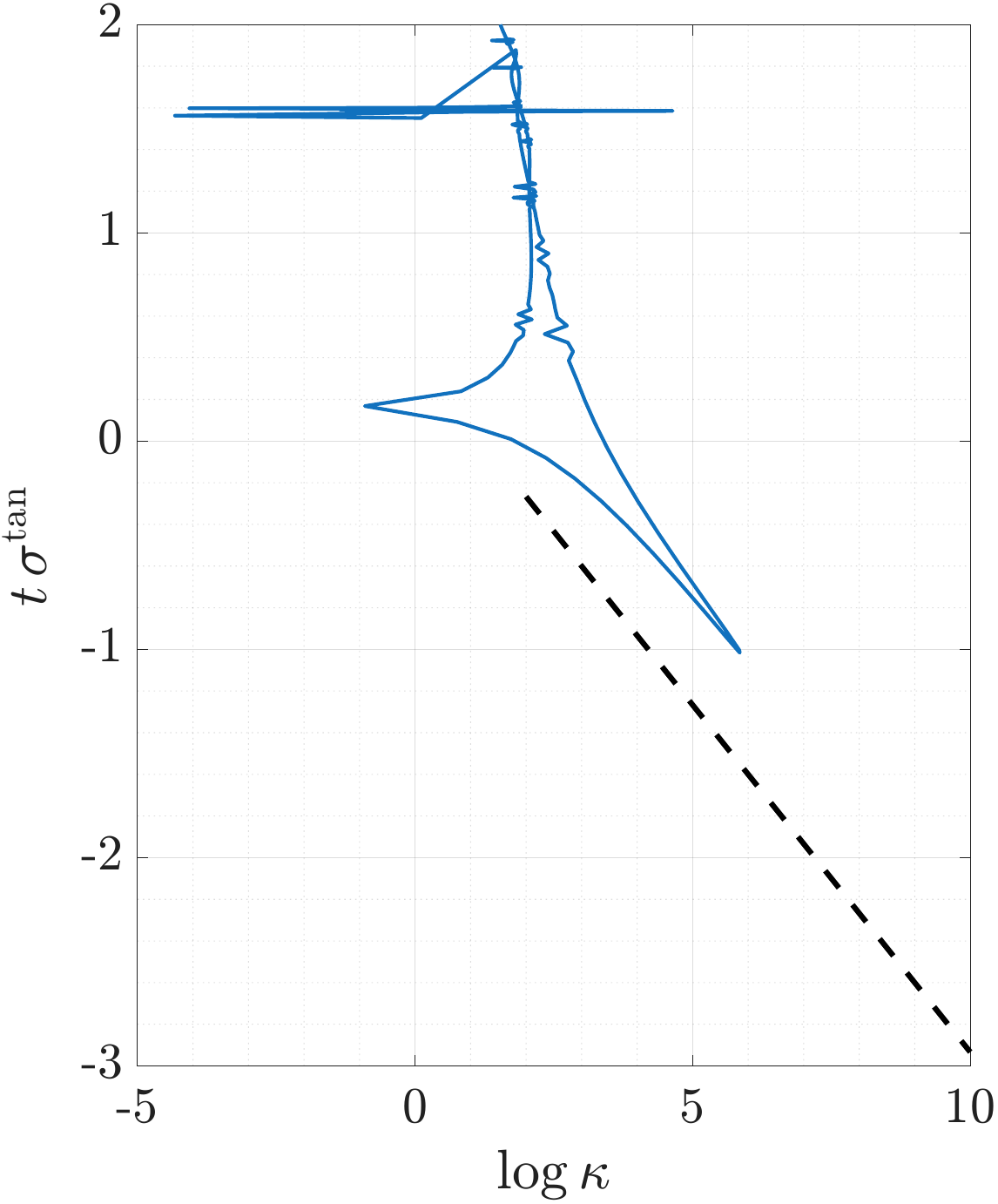}
  \caption{$h=\num{5.e-3}$}
  \label{fig:rotating-vortex-scaling-h=0.005}
\end{subfigure}
\begin{subfigure}[t]{0.19\linewidth}
  \centering
  \includegraphics[width=\linewidth]{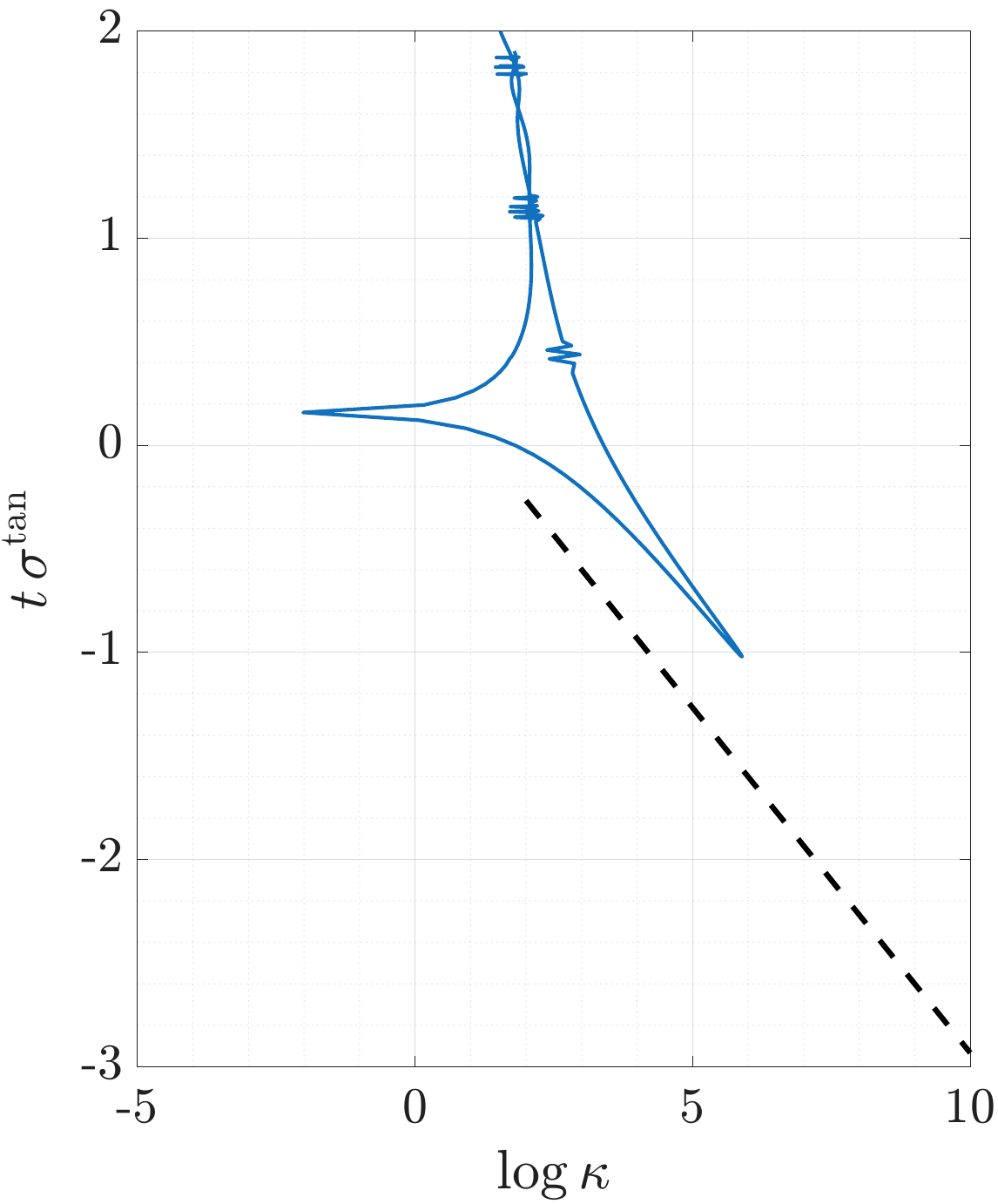}
  \caption{$h=\num{2.5e-3}$}
  \label{fig:rotating-vortex-scaling-h=0.0025}
\end{subfigure}
\begin{subfigure}[t]{0.19\linewidth}
  \centering
  \includegraphics[width=0.995\linewidth]{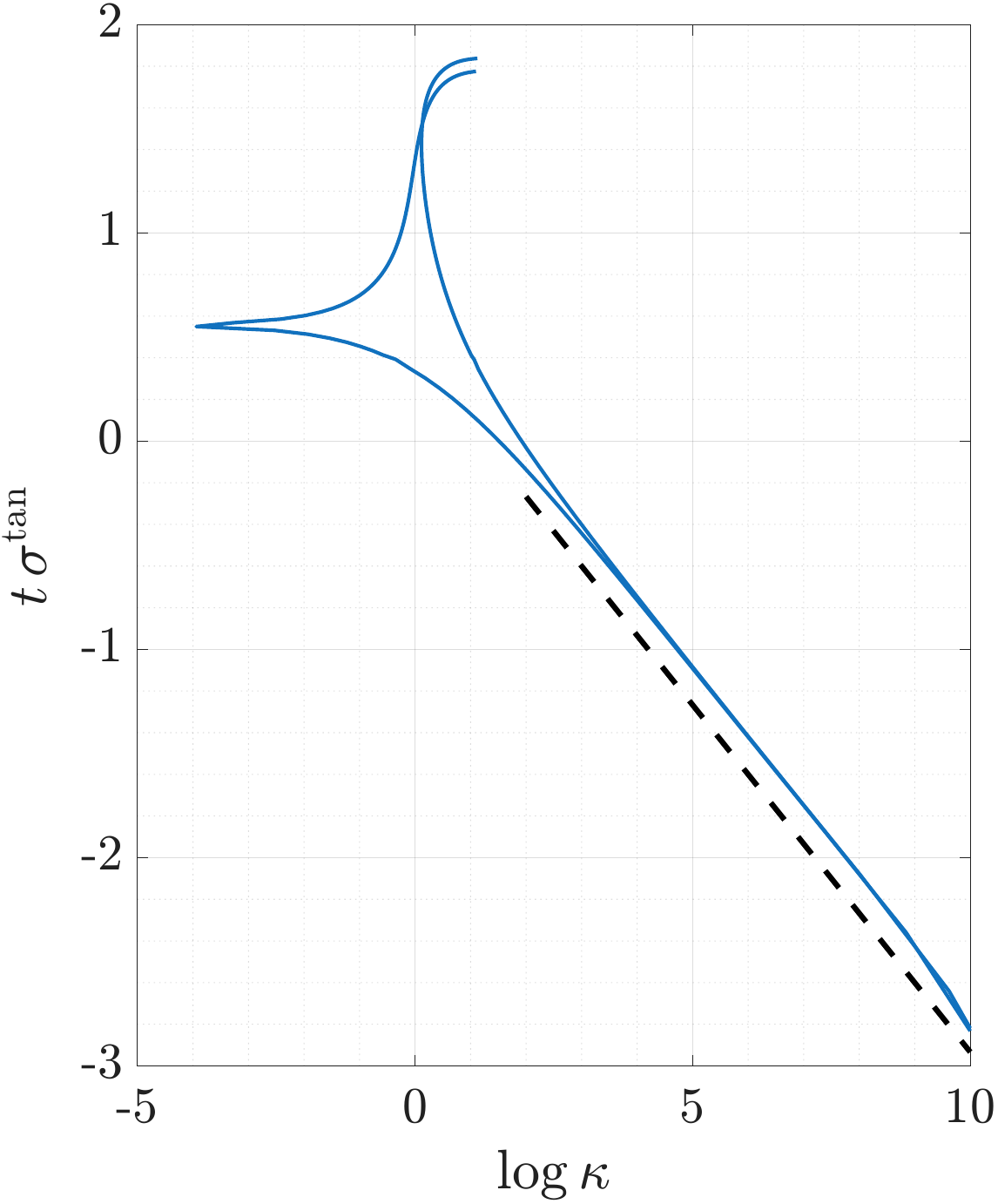}
  \caption{$h=\num{1.25e-3}$}
  \label{fig:rotating-vortex-scaling-h=0.00125}
\end{subfigure}
\caption{
Stretch--curvature phase plots near the strongly filamenting tip for
the $h$-refined \mts\ solutions to the rotating-vortex benchmark. The
$\kappa^{-1/3}$ branch provides a geometric signature of the filament
tip. As the fine-scale $h$ is decreased, a larger portion of the branch is recovered.
}
\label{fig:rotating-vortex-scaling-MTS}
\end{figure}

\subsection{The $\mathcal{S}$-flow benchmark}
\label{subsec:s-flow}

The second example is the $\mathcal{S}$-flow benchmark, originally introduced in
\cite{AhSh2009} and subsequently studied in
\cite{JeSuSh2015,HePhXi2023,HeLiPhXi2024}.
The prescribed incompressible velocity field $u$ is defined by
\begin{equation}
\label{eq:s-flow-velocity}
u(x_1,x_2)
=
\tfrac14
\left( 
(4x_1-2)+(4x_2-2)^3, \,
-(4x_2-2)-(4x_1-2)^3
\right)^\top,
\end{equation}
and the initial interface is a circle of radius $0.25$ centered at
$(0.5,0.5)$. The velocity gradient required for the tangent-FTLE
computation \eqref{eq:xi-evolution} is
\begin{equation}
\label{eq:s-flow-gradient}
Du(x_1,x_2)
=
\begin{pmatrix}
1 & 3(4x_2-2)^2 \\[1mm]
-3(4x_1-2)^2 & -1
\end{pmatrix}.
\end{equation}
The solution is evolved to the final time $\tmax=3$.
During the evolution, the initially circular interface is deformed
into a characteristic S-shaped configuration containing thin
filamentary regions and sharp geometric features.

We apply the \mts\ algorithm using the spatial resolutions
$h=0.02$ and $h=0.01$, with corresponding time steps
$\Delta t = 0.01$ and $\Delta t = 0.005$.
Topological processing is performed at the \mts\ times
\[
T_m = 0.1\,m,
\quad
m=1,\ldots,30.
\]
The \mts\ solutions are displayed in \Cref{fig:s-flow}.
At the coarser resolution, several satellite-interface formation events
are observed in the strongly filamenting regions of the flow.
Notably, the filament-breakup process develops symmetrically near the
two terminal filament tips, reflecting the corresponding symmetry in
the tangent-FTLE field discussed below.
As the resolution is increased, these pinch-off events occur at
progressively smaller scales, and the computed interface converges
qualitatively toward the reference solution shown in
\Cref{fig:s-flow_reference_t=3.0}.

\begin{figure}[ht]
\centering
\begin{subfigure}[t]{0.28\linewidth}
  \centering
  \includegraphics[width=\linewidth]{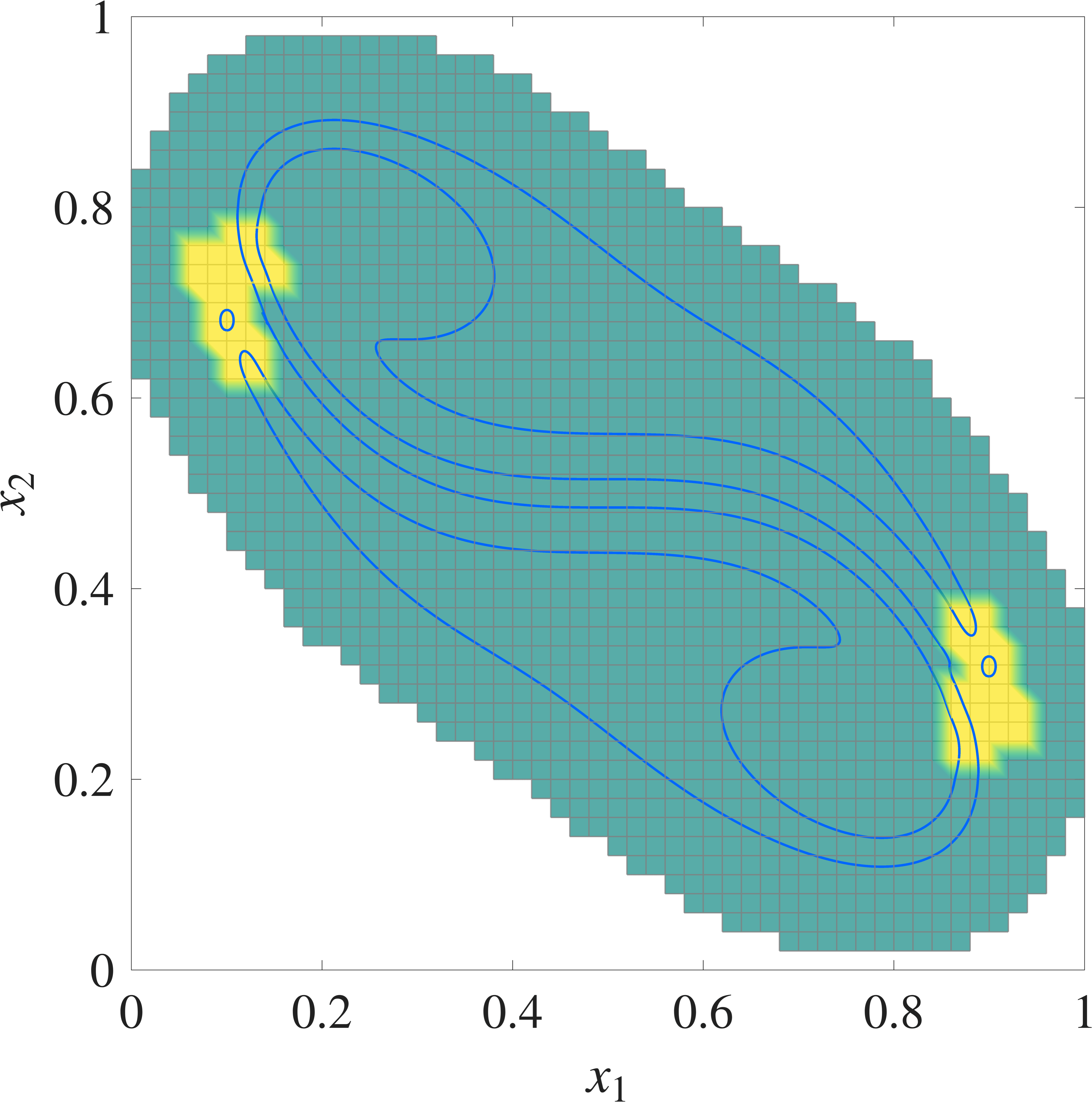}
  \caption{$h=0.02$, $t=2.8$}
  \label{fig:s-flow_h=0.02_t=2.8}
\end{subfigure}
\hspace{1em}
\begin{subfigure}[t]{0.28\linewidth}
  \centering
  \includegraphics[width=\linewidth]{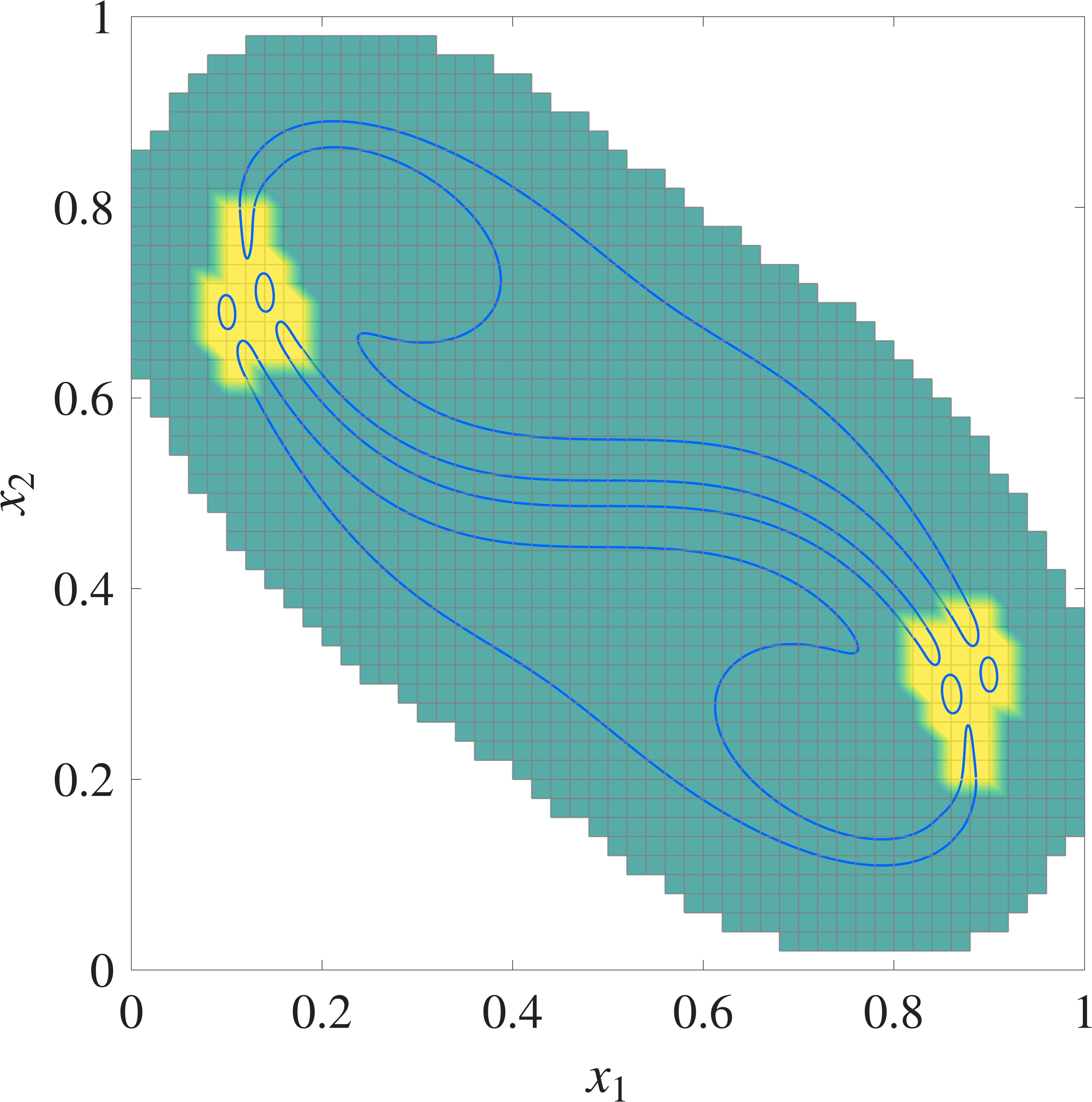}
  \caption{$h=0.02$, $t=2.9$}
  \label{fig:s-flow_h=0.02_t=2.9}
\end{subfigure}
\hspace{1em}
\begin{subfigure}[t]{0.28\linewidth}
  \centering
  \includegraphics[width=\linewidth]{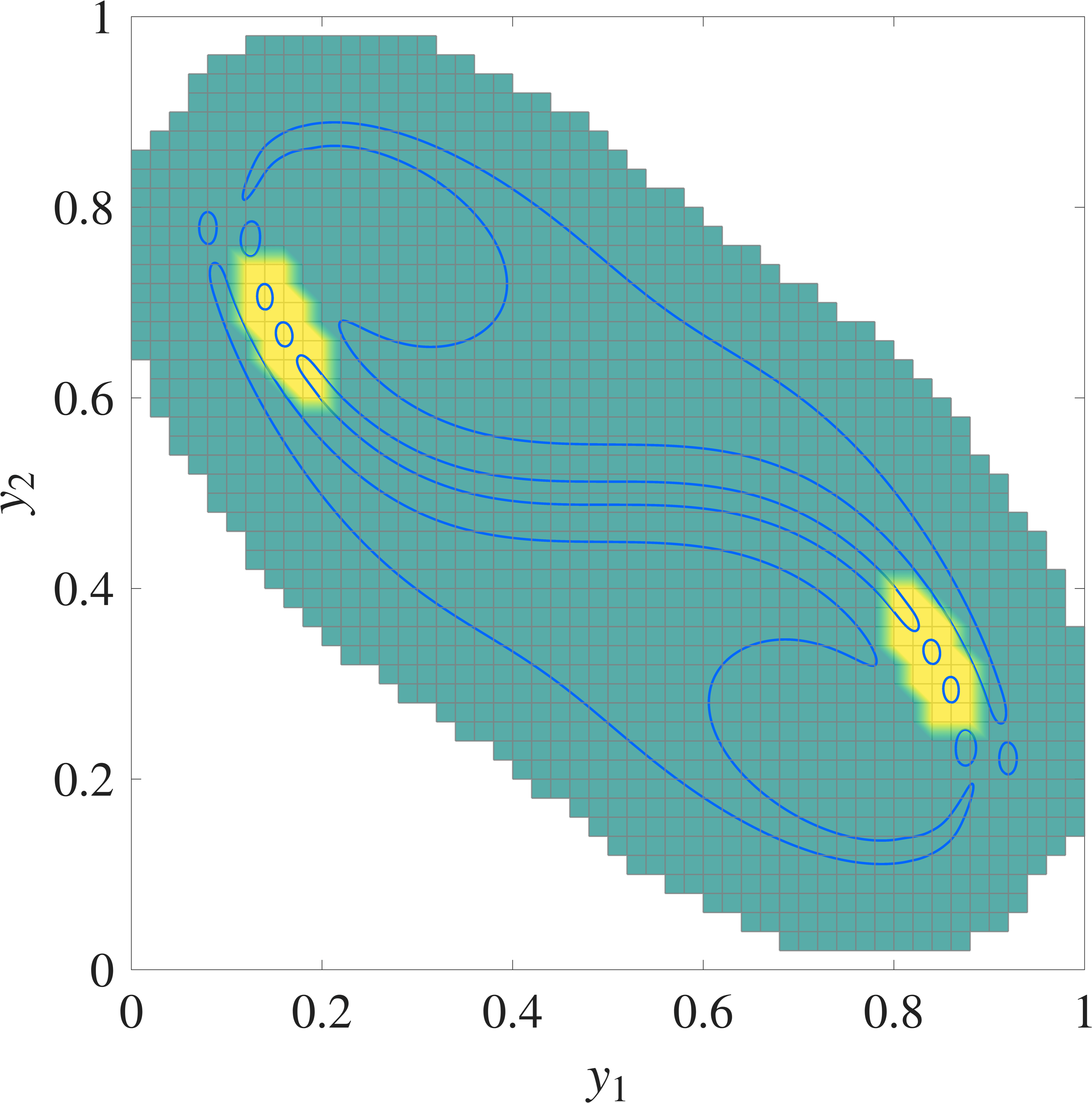}
  \caption{$h=0.02$, $t=3.0$}
  \label{fig:s-flow_h=0.02_t=3.0}
\end{subfigure}

\vspace{1em}

\begin{subfigure}[t]{0.28\linewidth}
  \centering
  \includegraphics[width=\linewidth]{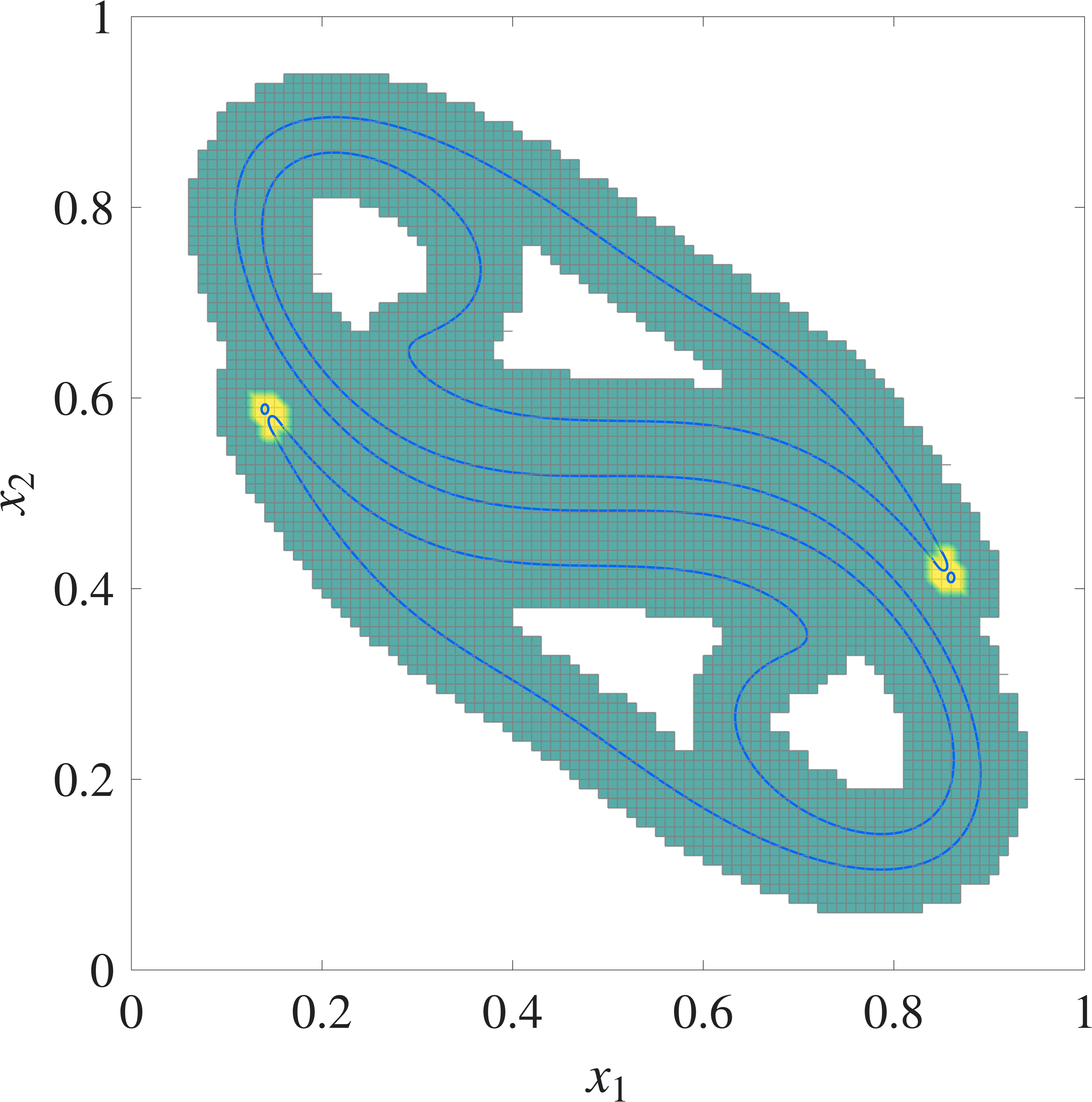}
  \caption{$h=0.01$, $t=2.6$}
  \label{fig:s-flow_h=0.01_t=2.6}
\end{subfigure}
\hspace{1em}
\begin{subfigure}[t]{0.28\linewidth}
  \centering
  \includegraphics[width=\linewidth]{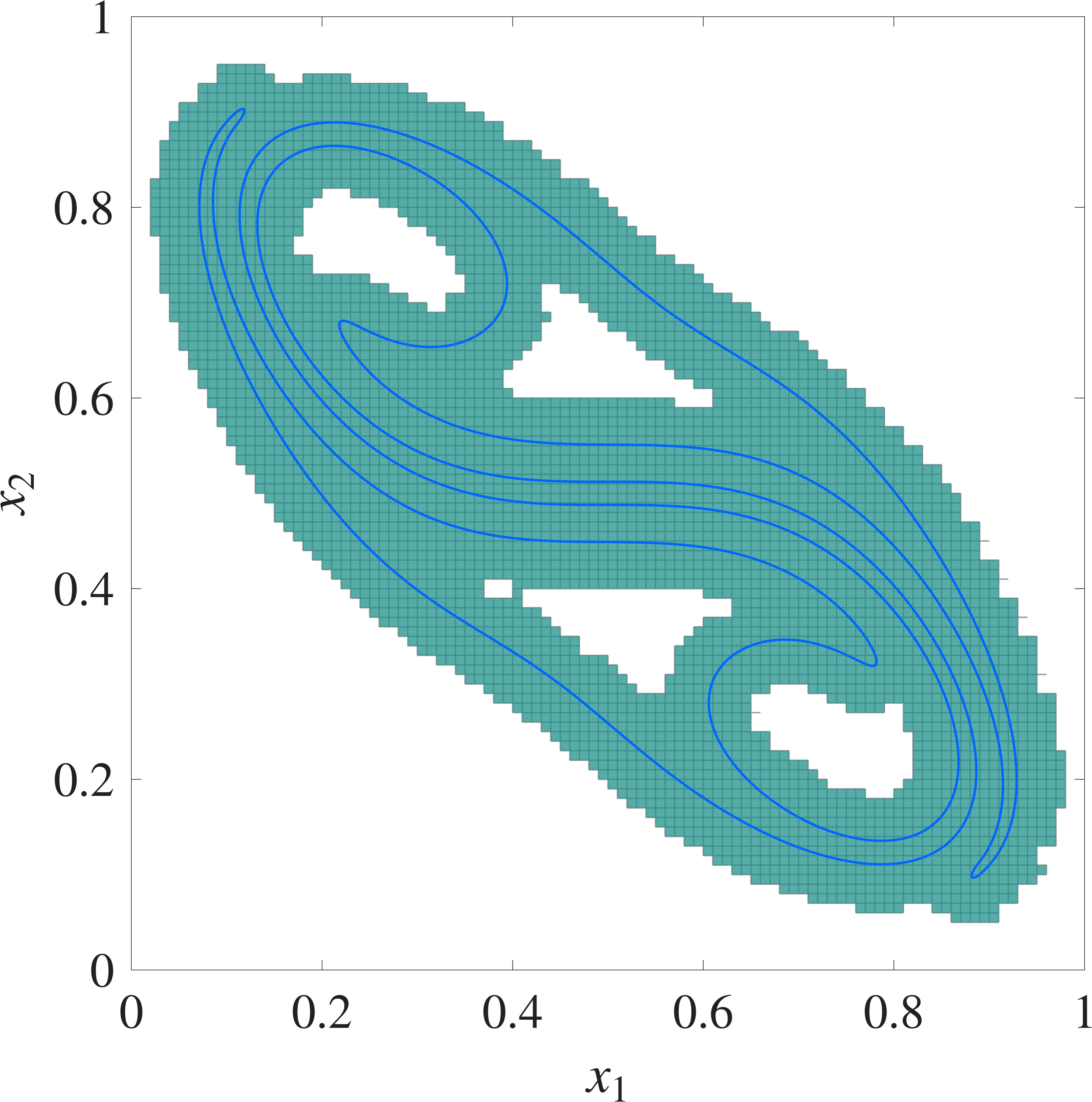}
  \caption{$h=0.01$, $t=3.0$}
  \label{fig:s-flow_h=0.01_t=3.0}
\end{subfigure}
\hspace{1em}
\begin{subfigure}[t]{0.28\linewidth}
  \centering
  \includegraphics[width=\linewidth]{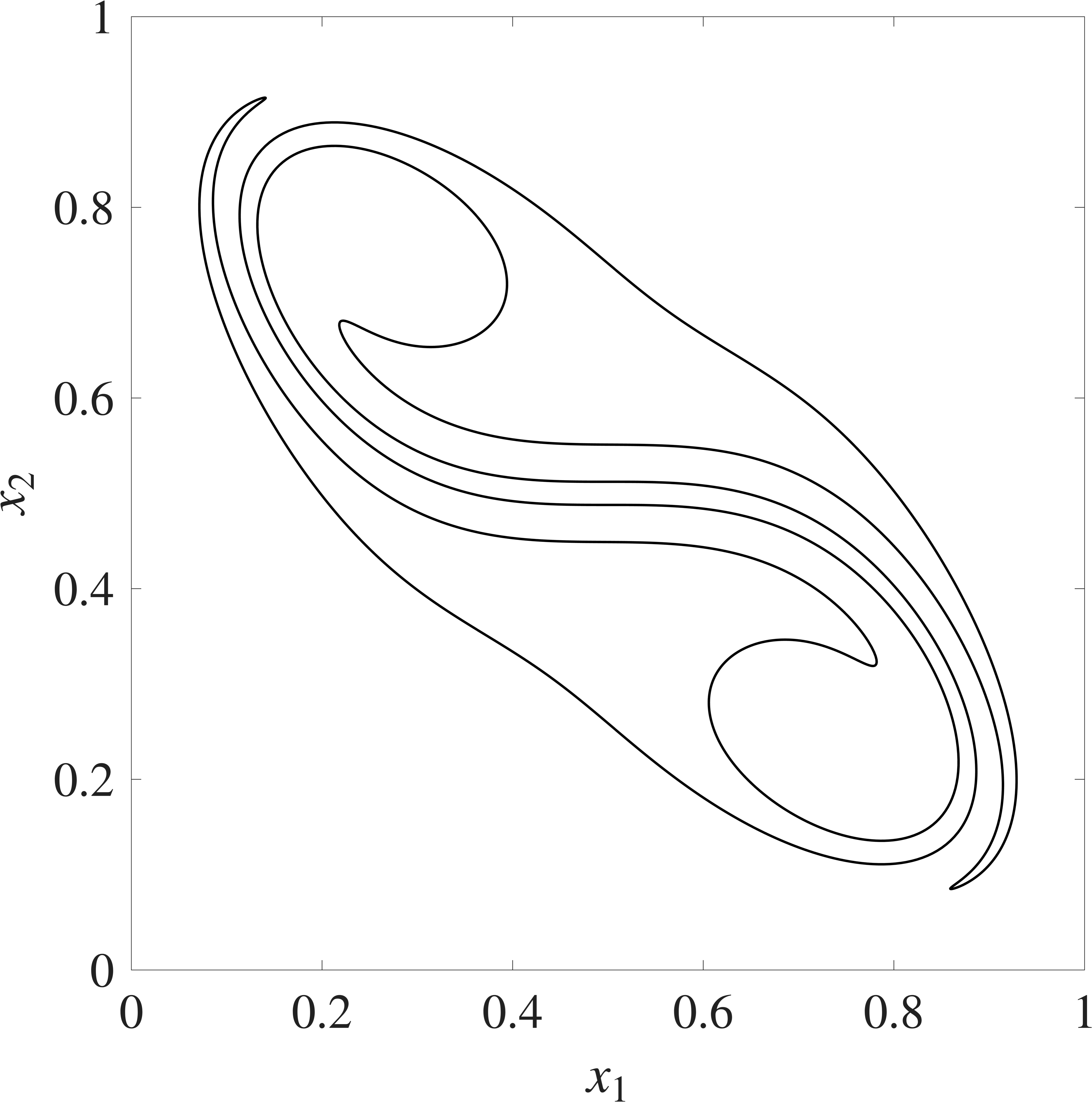}
  \caption{reference $t=3.0$}
  \label{fig:s-flow_reference_t=3.0}
\end{subfigure}
\caption{
$\mathcal{S}$-flow benchmark computed using the \mts\ algorithm.
The top row shows the solution with $h=0.02$ at times
$t=2.8$, $2.9$, and $3.0$.
The bottom row shows the solution with $h=0.01$ at times
$t=2.6$ and $3.0$, together with a reference solution at $t=3.0$.
Satellite-interface formation events observed in the coarser
computation are confined to progressively smaller scales as the
resolution is increased.
}
\label{fig:s-flow}
\end{figure}

The tangent FTLE computed by classical Lagrangian tracking
is displayed in \Cref{fig:s-flow-ftle_reference}.
Consistent with the symmetric filament-breakup dynamics observed in
\Cref{fig:s-flow}, the diagnostic exhibits paired trough structures
localized near the two terminal filament tips.
Comparison with the tangent-FTLE fields for the rotating-vortex and
nonlinear alternating-shear benchmarks shown in
\Cref{fig:dyadic-refinement,fig:alternating-shear-FTLE}
demonstrates that the $\mathcal{S}$-flow dynamics occupy an intermediate
regime of multiscale filamentation complexity.

A consistent qualitative relationship between filamentation
and the tangent FTLE is observed: localized troughs in the reference
tangent-FTLE field characterize the filamenting regions where breakup
occurs in the corresponding \mts\ computations.
The \mts\ algorithm continues the tangent-FTLE diagnostic through
topological transitions by transferring the tangent variation field
between the pre- and post-surgery interfaces using the closest-point projection.
As the resolution is increased, the reconstructed tangent-FTLE fields
shown in \Cref{fig:s-flow-ftle_h=0.02,fig:s-flow-ftle_h=0.01}
converge qualitatively toward the reference diagnostic computed by
classical Lagrangian tracking, shown in \Cref{fig:s-flow-ftle_reference}.

\begin{figure}[ht]
\centering
\begin{subfigure}[t]{0.28\linewidth}
  \centering
  \includegraphics[width=\linewidth]{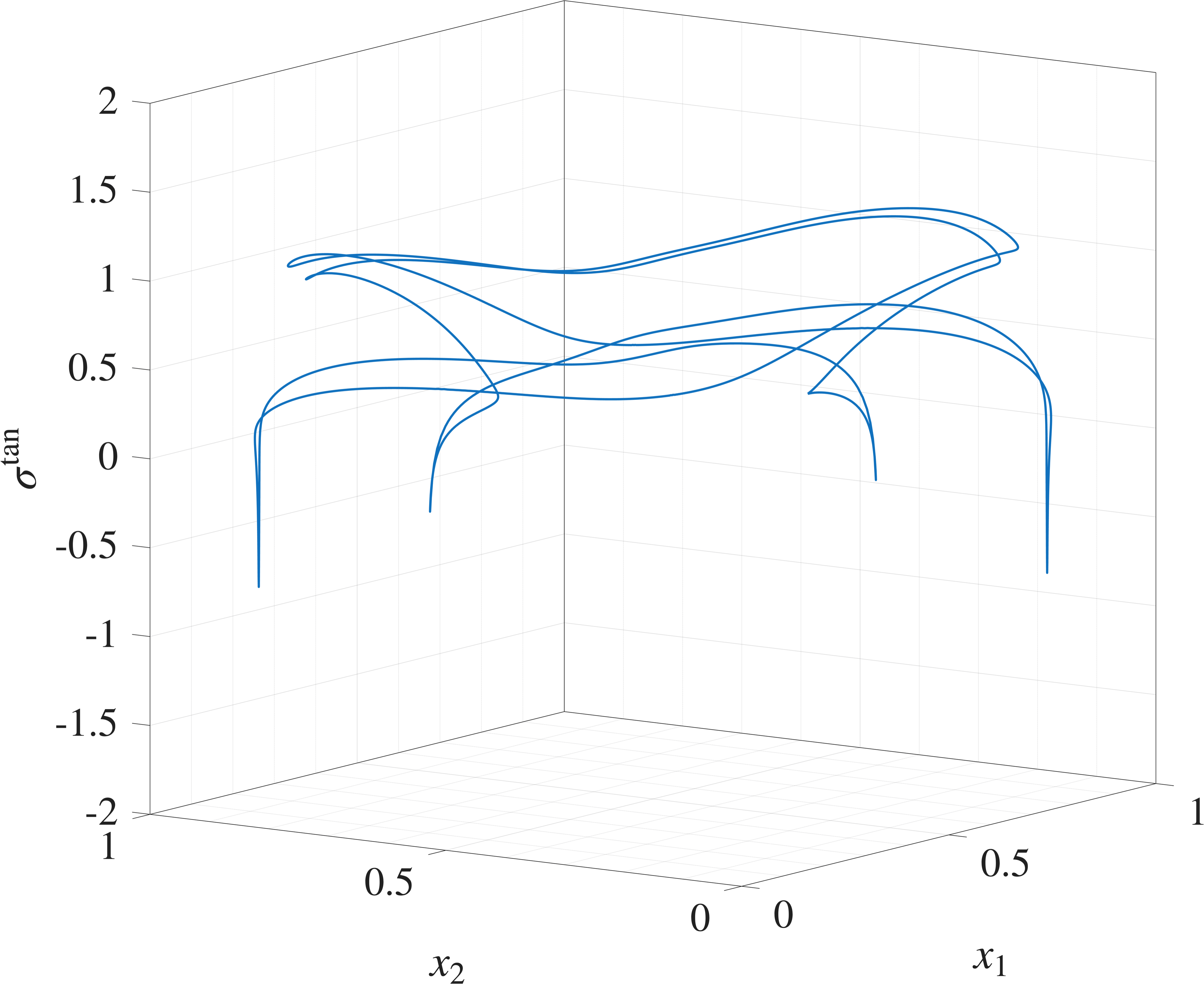}
  \caption{reference}
  \label{fig:s-flow-ftle_reference}
\end{subfigure}
\hspace{1em}
\begin{subfigure}[t]{0.28\linewidth}
  \centering
  \includegraphics[width=\linewidth]{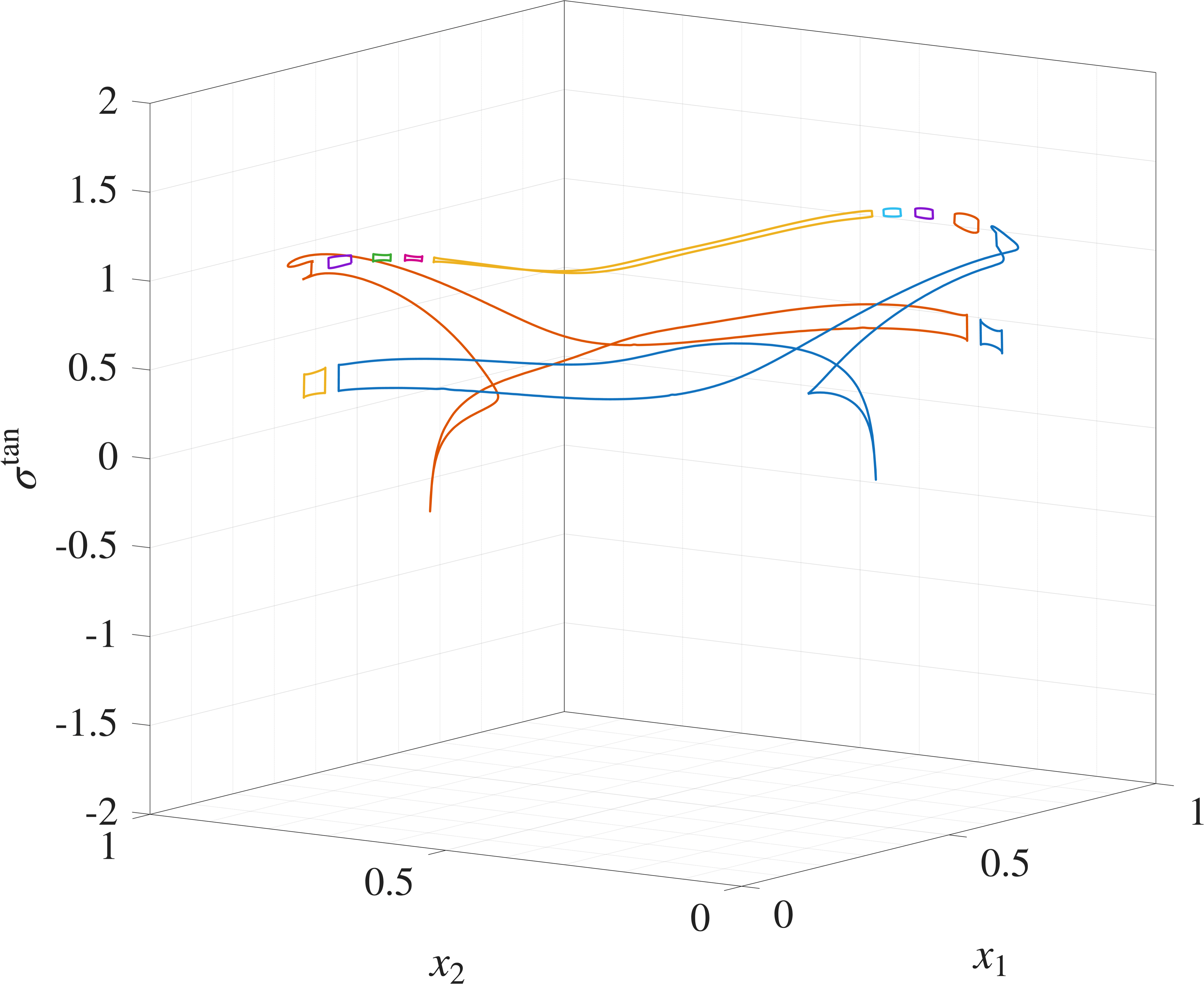}
  \caption{$h=0.02$}
  \label{fig:s-flow-ftle_h=0.02}
\end{subfigure}
\hspace{1em}
\begin{subfigure}[t]{0.28\linewidth}
  \centering
  \includegraphics[width=\linewidth]{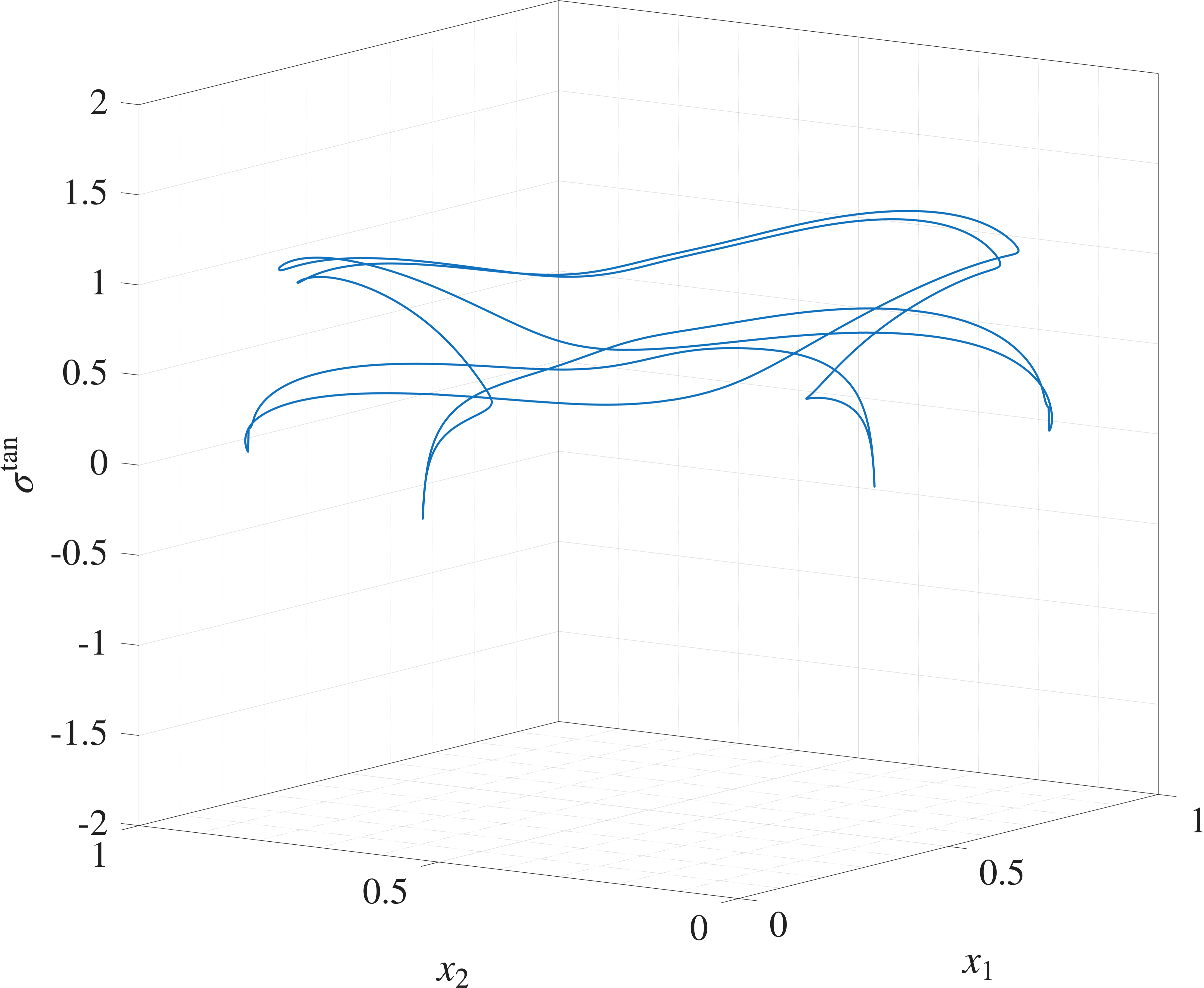}
  \caption{$h=0.01$}
  \label{fig:s-flow-ftle_h=0.01}
\end{subfigure}
\caption{
Tangent FTLE diagnostics for the $\mathcal{S}$-flow benchmark.
The left panel shows the reference tangent-FTLE field computed using
classical Lagrangian tracking, while the center and right panels show
the corresponding tangent-FTLE fields computed using the 
\mts\ algorithm with $h=0.02$ and $h=0.01$, respectively.
The \mts\ algorithm continues the tangent-FTLE diagnostic through
topological transitions by transferring the tangent variation field
between the pre- and post-surgery interfaces using the closest-point projection.
As the resolution is increased, the reconstructed tangent-FTLE field
converges qualitatively toward the reference diagnostic.
}
\label{fig:s-flow-ftle}
\end{figure}

\subsection{Microscale filament breakup by nonlinear alternating-shearing}

We now assess the efficacy of the \mts\ algorithm on the more
challenging nonlinear alternating-shear test introduced in
\Cref{subsec:alternating-shear-LGR}. The velocity field $u$, initial
interface $\gamma_\alpha(s,0)$, final time $\tmax$, time-step
$\Delta t$, and coarse computational scale $h_0$ are identical to those
used in \Cref{subsec:alternating-shear-LGR}. The fine-scale resolution
is given by the dyadic refinement sequence
\begin{equation}
h = 2^{-p} h_0,
\qquad
p=1,\ldots,5,
\end{equation}
and topological processing is performed throughout the filament-formation
phase at the \mts\ times
\begin{equation}
T_m = m,
\qquad
m=1,\ldots,15.
\end{equation}
As demonstrated in \Cref{subsec:alternating-shear-LGR}, the repeated
alternating-shear deformation generates multiscale filamentary structures, with breakup occurring near material
locations identified by localized troughs in the tangent FTLE (\Cref{fig:alternating-shear-FTLE}). The
resulting sequence of repeated one-to-many satellite-formation events provides a
stringent test of the robustness and scalability of the \mts\
topology-processing algorithm.

The \mts\ interface evolution is shown in
\Cref{fig:alternating-shear-topology-sequence} for the 
microscale-resolving case $h/L=\num{1.25e-3}$, where $L=6$. Since no topological
events are detected for $t<10$, only the times
$t=10,\ldots,15$ are displayed. The left panels show the evolving
\mts\ interface representation together with the localized defect
measure and surgical region, while the right panels show the
corresponding adjacency graphs. 
As anticipated by the tangent-FTLE analysis of
\Cref{subsec:alternating-shear-LGR}, filament breakup occurs near the
filament tips identified by the FTLE troughs. The most severe breakup occurs within the $\mathcal{W}_1$
and $\mathcal{W}_3$ observation windows, in agreement with the qualitative
conclusions of \Cref{fig:alternating-shear}. 

Conversely, the $\mathcal{W}_2$ and $\mathcal{W}_4$ observation windows associated with the
remaining FTLE troughs identified in
\Cref{fig:alternating-shear-FTLE} exhibit filament-tip pinch-off events
with satellites whose characteristic size falls below the
prescribed sub-microscale resolution (parameters $\Nmin$, $\Lmin$, and $\Amin$ in \Cref{tab:notation}) 
and are therefore filtered during the Eulerian extraction phase; see the 
$\mathcal{W}_4$, $\mathcal{W}_2$, and $\mathcal{W}_4$ observation windows
in \Cref{alternating-shear_h=0.0075_t=11.0},
\Cref{alternating-shear_h=0.0075_t=13.0}, and
\Cref{alternating-shear_h=0.0075_t=14.0}, respectively. 
Nevertheless, the defect measure $\Defect$ identifies the pinch-off
location on the parent filament, and the interface surgery smoothly
reconstructs the filament tip. This behavior may be viewed as a weak
form of geometric regularization acting below the prescribed microscale
resolution, in which unresolved satellites are removed while
the larger-scale filament geometry is preserved.

The repeated generation, disappearance, and interaction of satellite
interfaces produces a rich sequence of topological events, with a
maximum of $26$ interfaces generated at $t=14$. Throughout this
evolution, the \mts\ algorithm robustly reconstructs the interface
family through numerous one-to-many \emph{splits} and
\emph{vanishings}, together with a few \emph{merges}.

\setcounter{subfigure}{0}
\begin{figure}[p]
\ContinuedFloat
\centering

\begin{subfigure}[t]{0.7\linewidth}
  \centering
  \includegraphics[width=\linewidth]
  {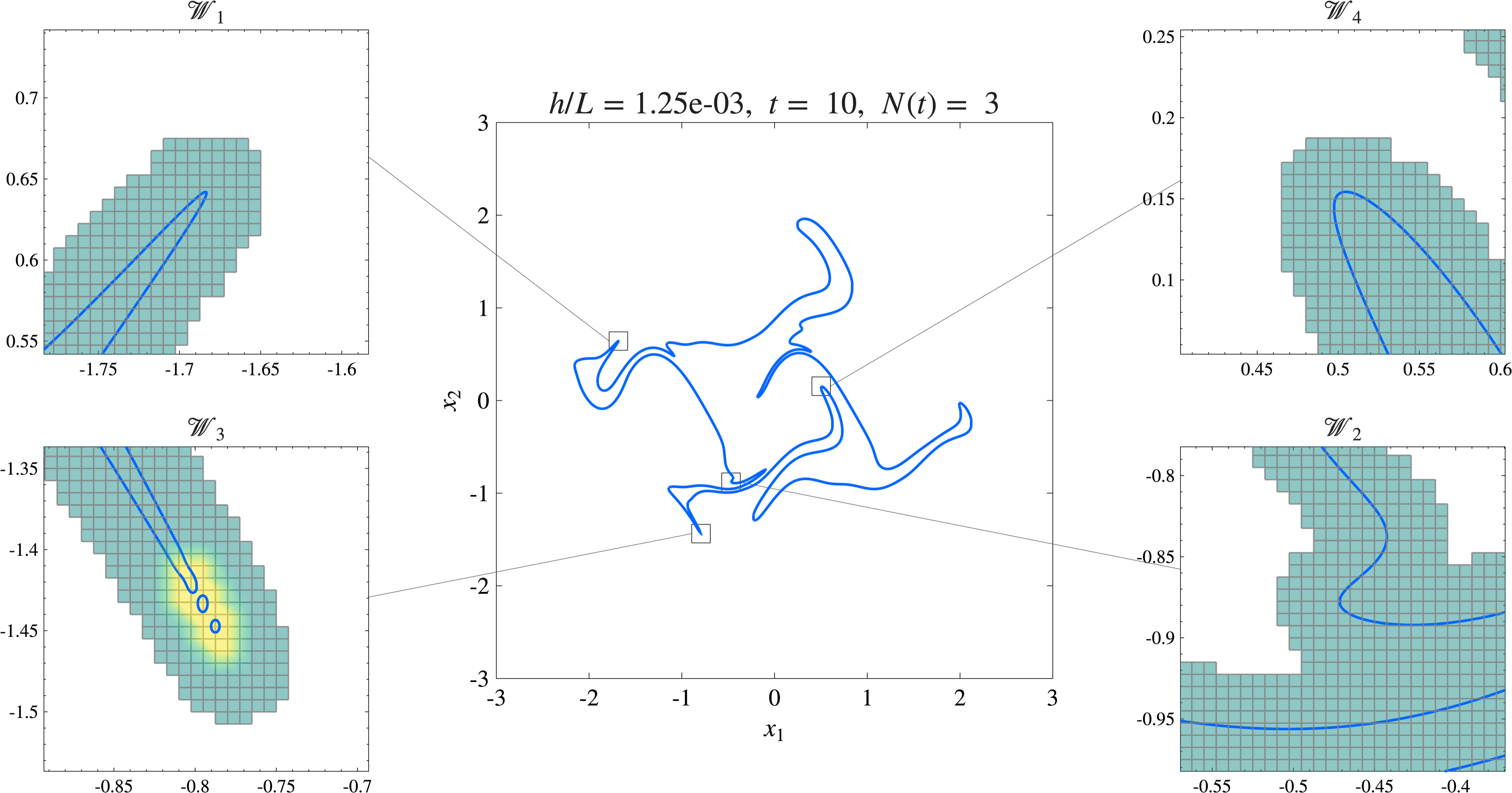}
  \caption{$t=10$}
  \label{alternating-shear_h=0.0075_t=10.0}
\end{subfigure}
\hfill
\begin{subfigure}[t]{0.28\linewidth}
  \centering
  \raisebox{0.75cm}{%
  \includegraphics[width=\linewidth]
  {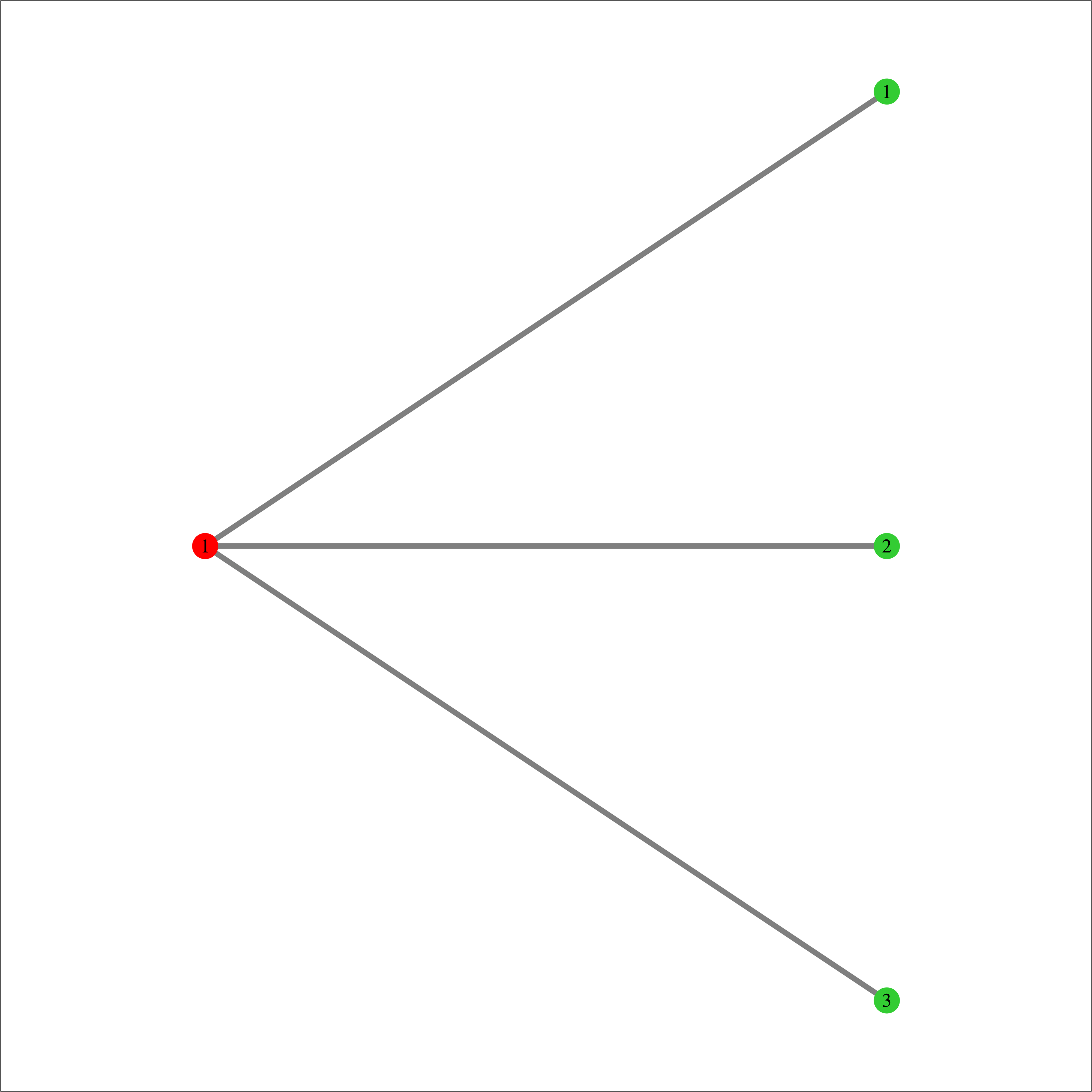}}
\end{subfigure}

\vspace{2em}

\begin{subfigure}[t]{0.7\linewidth}
  \centering
  \includegraphics[width=\linewidth]
  {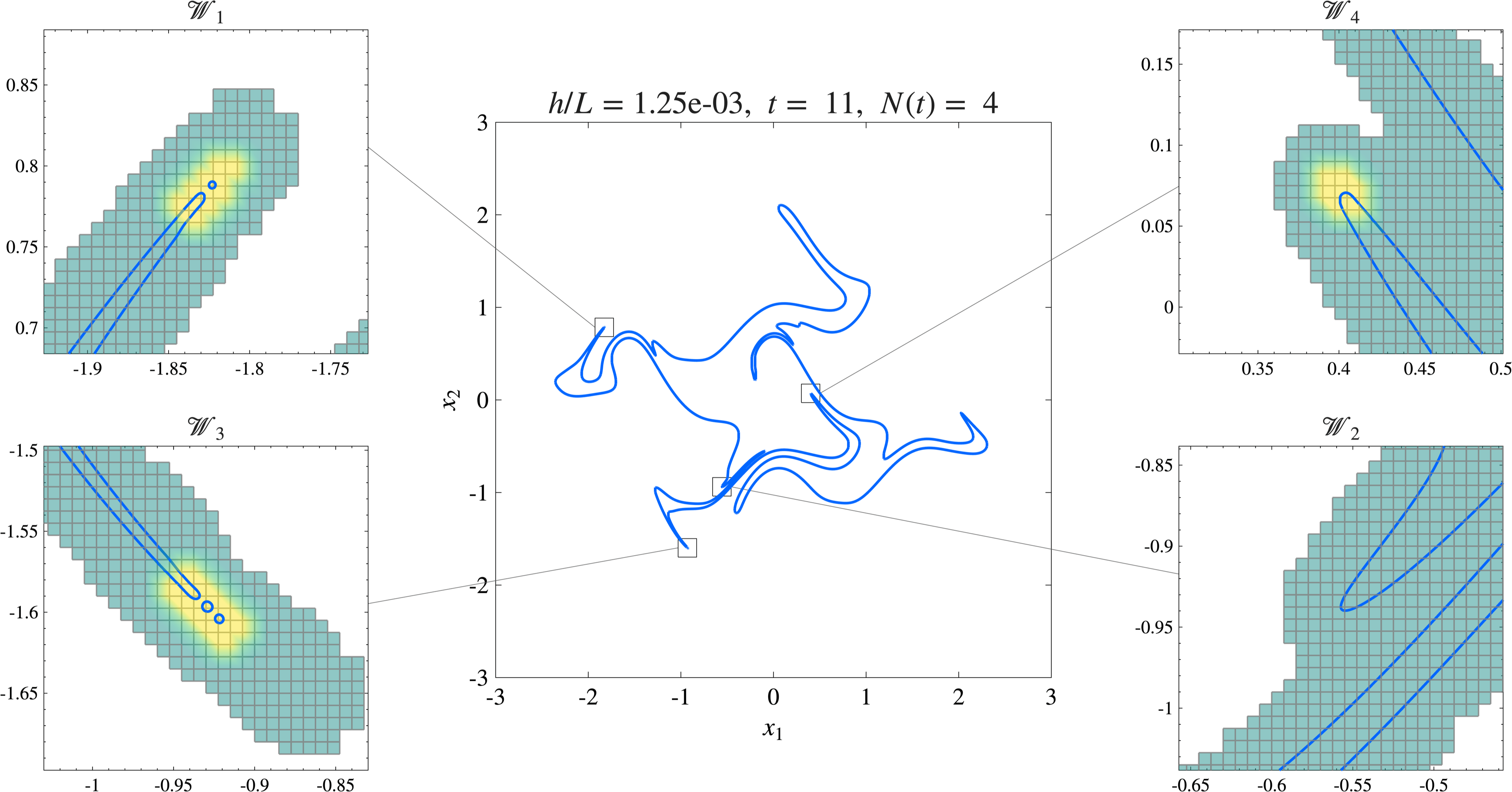}
  \caption{$t=11$}
  \label{alternating-shear_h=0.0075_t=11.0}
\end{subfigure}
\hfill
\begin{subfigure}[t]{0.28\linewidth}
  \centering
  \raisebox{0.75cm}{%
  \includegraphics[width=\linewidth]
  {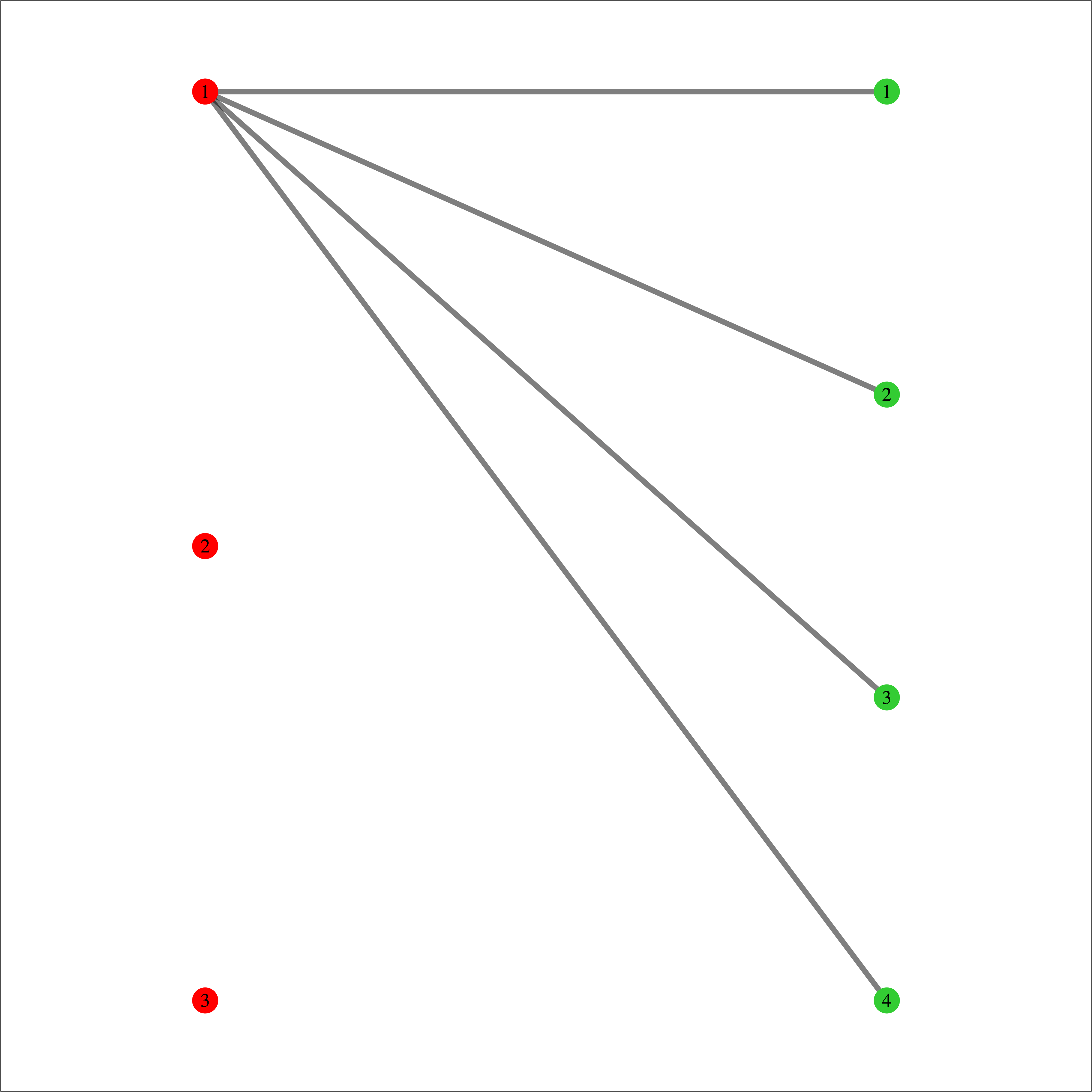}}
\end{subfigure}

\vspace{2em}

\begin{subfigure}[t]{0.7\linewidth}
  \centering
  \includegraphics[width=\linewidth]
  {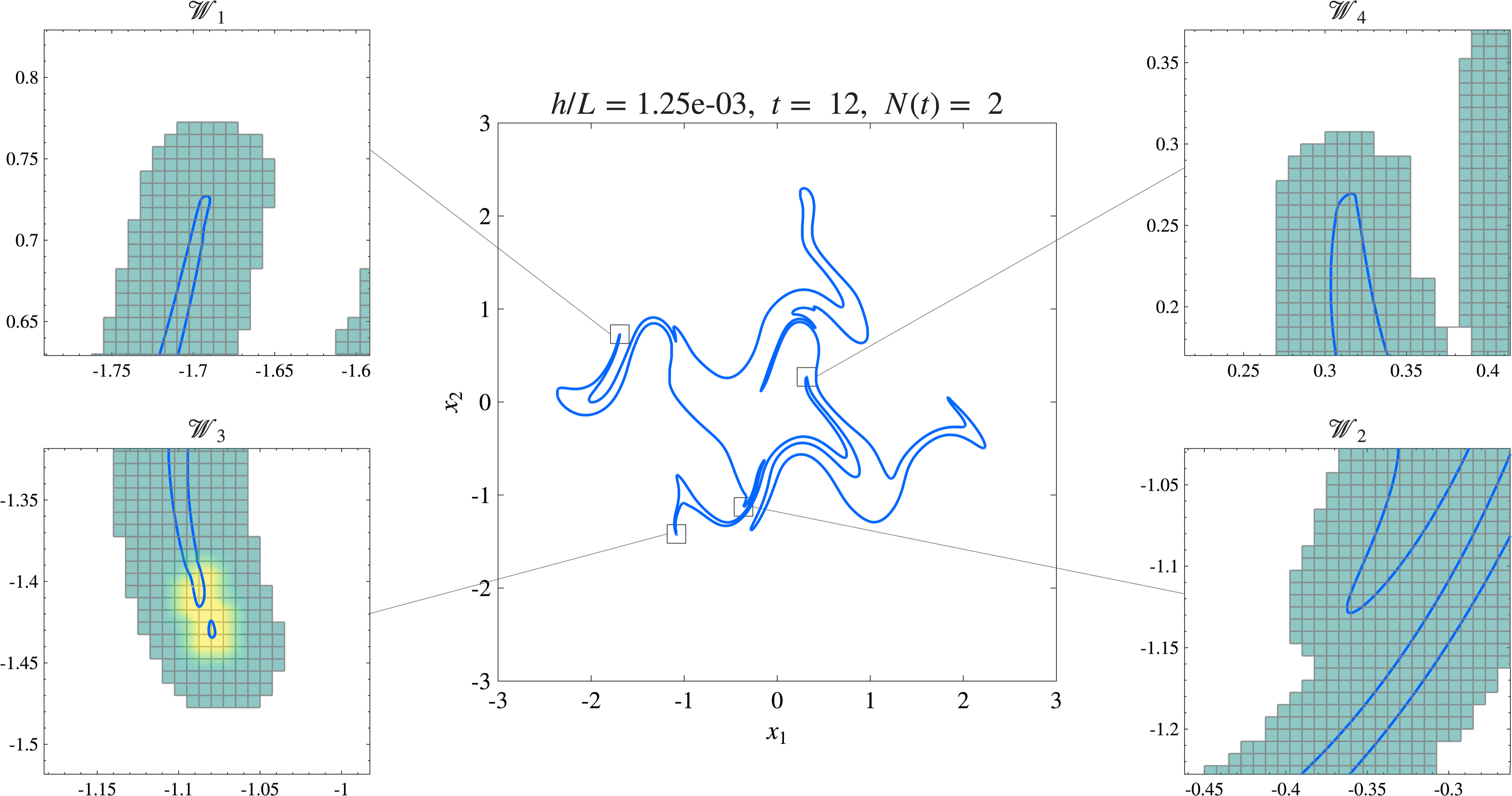}
  \caption{$t=12$}
  \label{alternating-shear_h=0.0075_t=12.0}
\end{subfigure}
\hfill
\begin{subfigure}[t]{0.28\linewidth}
  \centering
  \raisebox{0.75cm}{%
  \includegraphics[width=\linewidth]
  {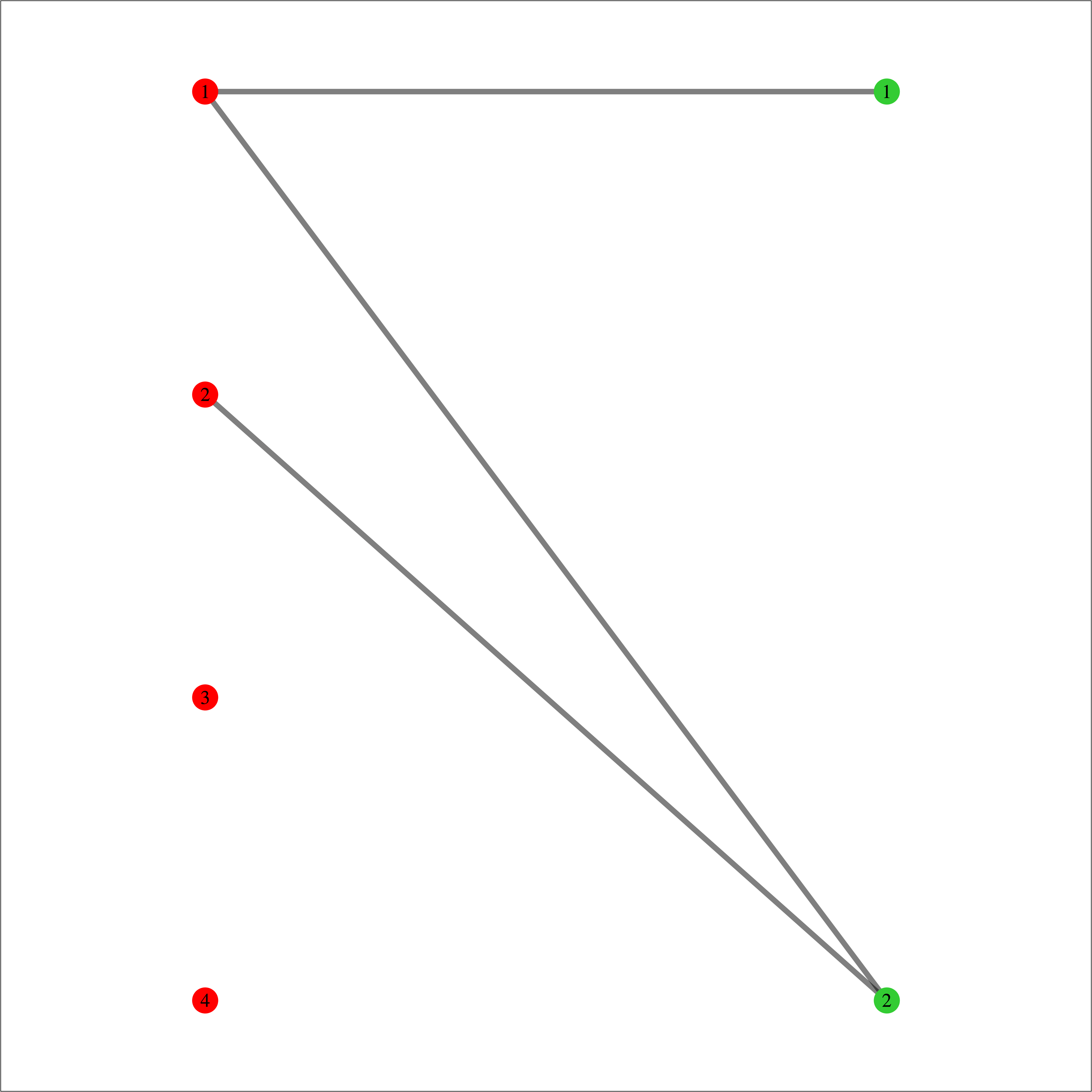}}
\end{subfigure}

\caption{
Topology processing for the alternating-shear test at the
microscale resolving resolution $h/L=\num{1.25e-3}$.
The panels show the sequence of detected topological transitions during
the filament-breakup phase $t=10,\ldots,15$.
In each row, the left panel displays the \mts\ interface
representation together with the localized defect measure used to
identify the surgical region, while the right panel displays the
corresponding adjacency graph. The evolving interface family is robustly reconstructed by the \mts\
algorithm through the resulting sequence of \emph{splits} and \emph{merges}.
A video of the complete evolution is available at \cite{Ramani2026}.
}
\label{fig:alternating-shear-topology-sequence}
\end{figure}

\begin{figure}[p]
\ContinuedFloat
\centering

\begin{subfigure}[t]{0.7\linewidth}
  \centering
  \includegraphics[width=\linewidth]
  {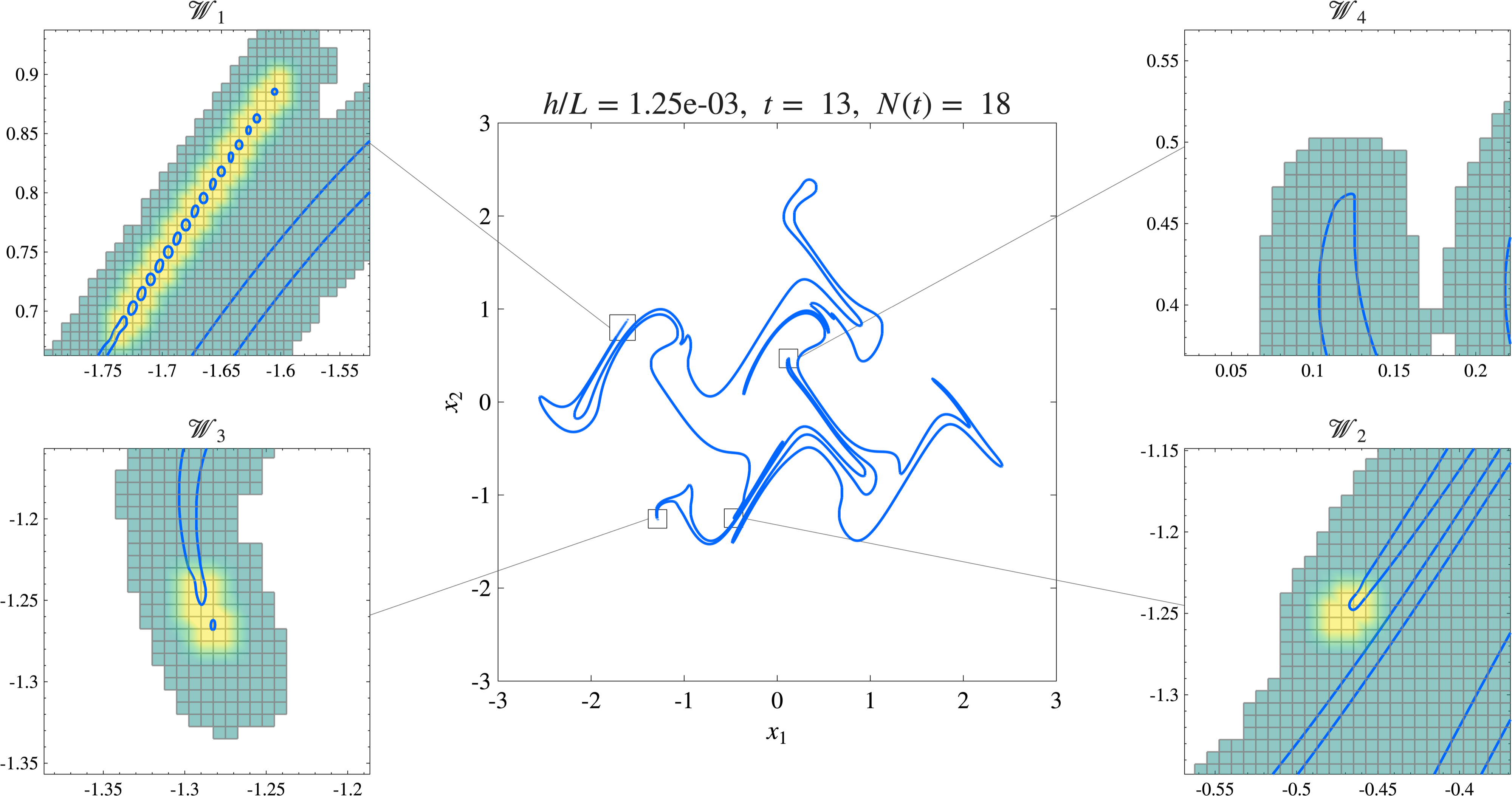}
  \caption{$t=13$}
  \label{alternating-shear_h=0.0075_t=13.0}
\end{subfigure}
\hfill
\begin{subfigure}[t]{0.28\linewidth}
  \centering
  \raisebox{0.75cm}{%
  \includegraphics[width=\linewidth]
  {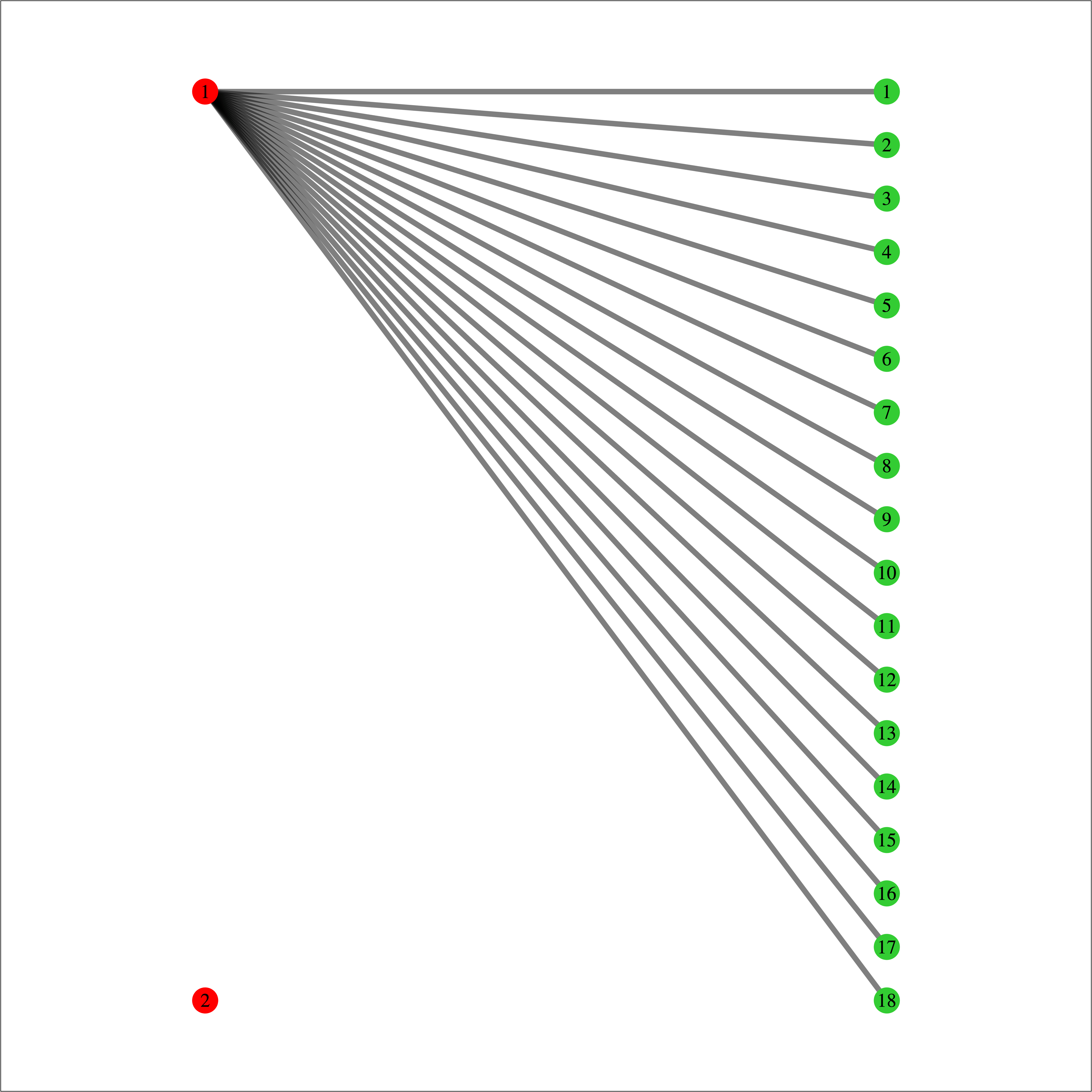}}
\end{subfigure}

\vspace{2em}

\begin{subfigure}[t]{0.7\linewidth}
  \centering
  \includegraphics[width=\linewidth]
  {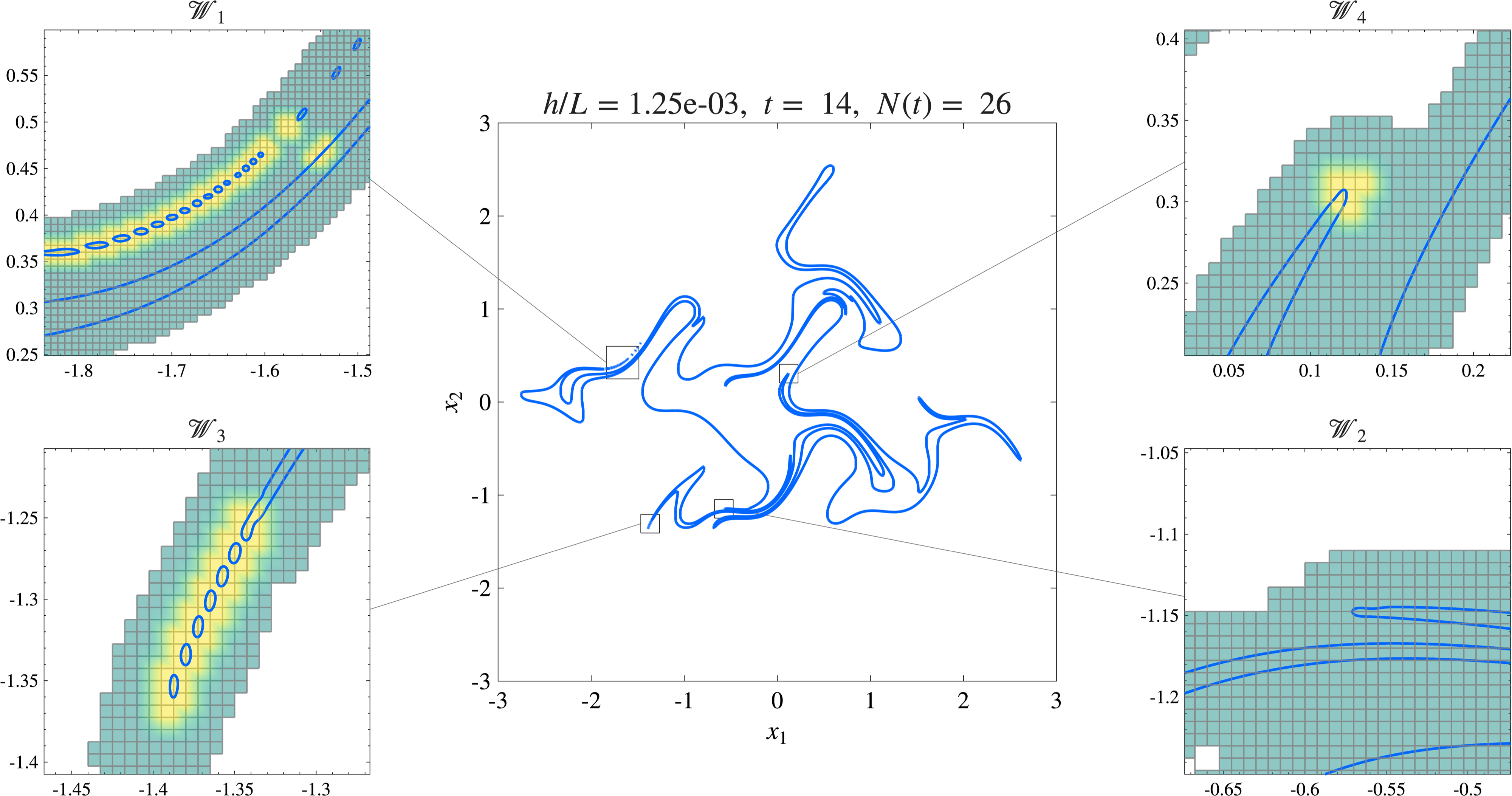}
  \caption{$t=14$}
  \label{alternating-shear_h=0.0075_t=14.0}
\end{subfigure}
\hfill
\begin{subfigure}[t]{0.28\linewidth}
  \centering
  \raisebox{0.75cm}{%
  \includegraphics[width=\linewidth]
  {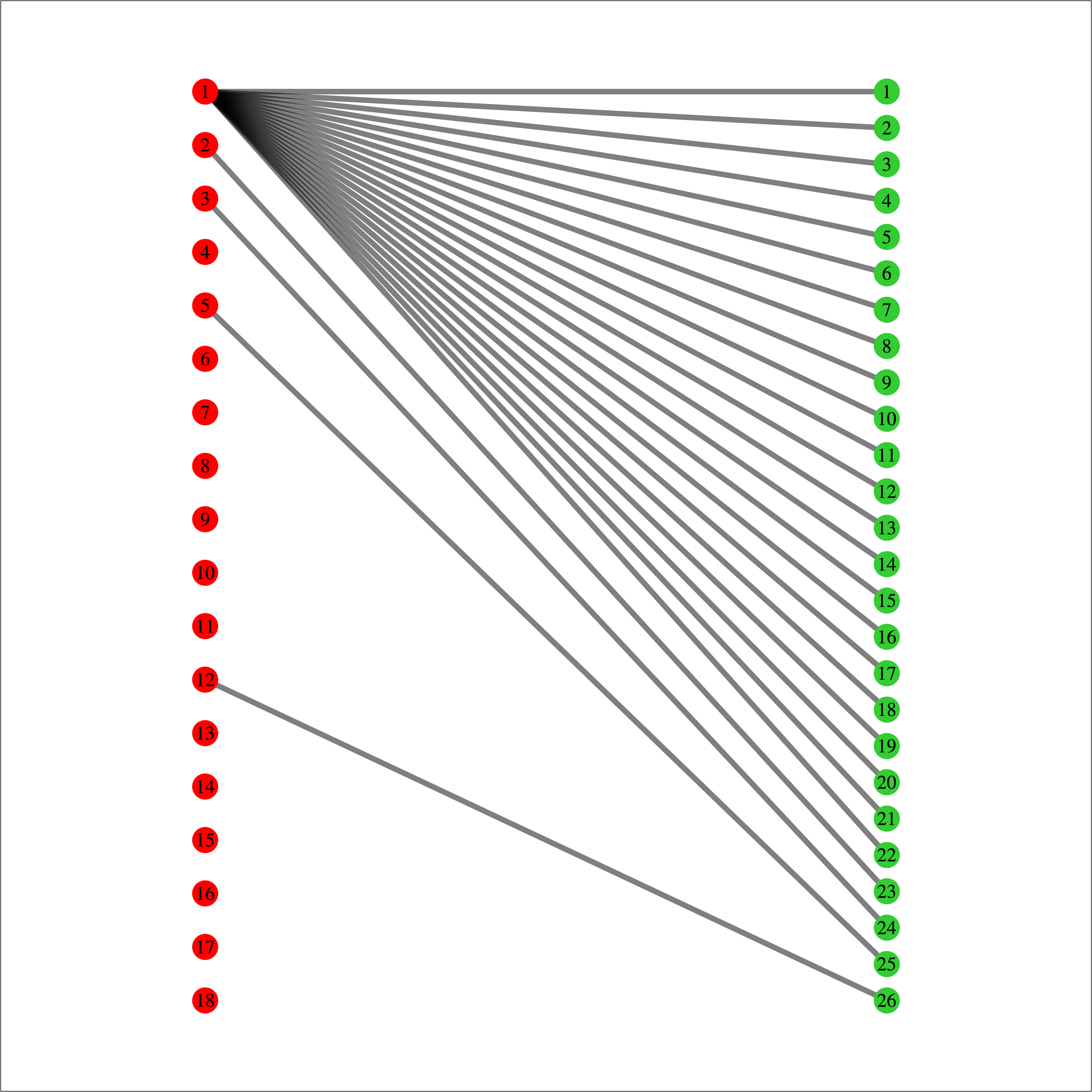}}
\end{subfigure}

\vspace{2em}

\begin{subfigure}[t]{0.7\linewidth}
  \centering
  \includegraphics[width=\linewidth]
  {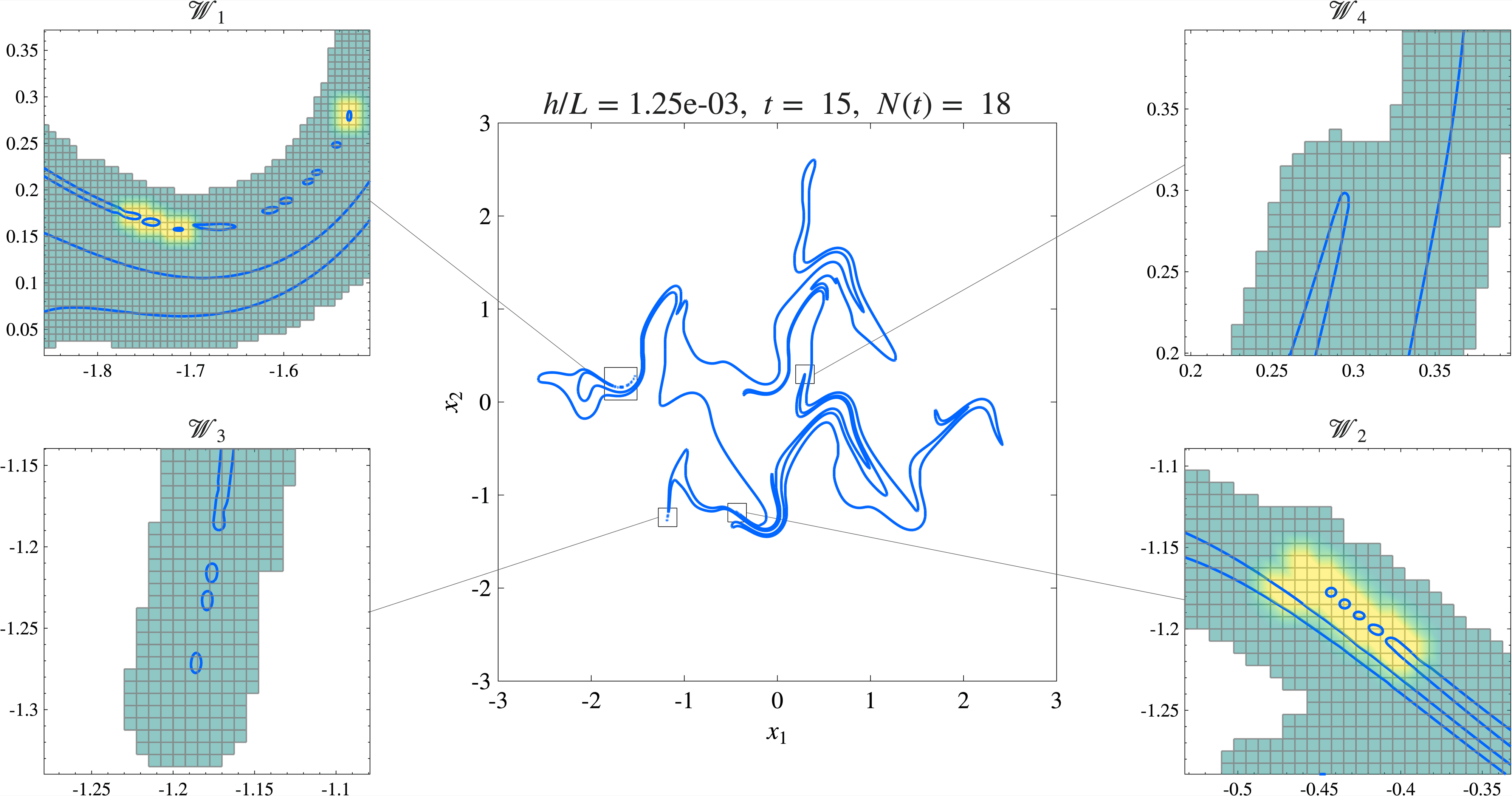}
  \caption{$t=15$}
  \label{alternating-shear_h=0.0075_t=15.0}
\end{subfigure}
\hfill
\begin{subfigure}[t]{0.28\linewidth}
  \centering
  \raisebox{0.75cm}{%
  \includegraphics[width=\linewidth]
  {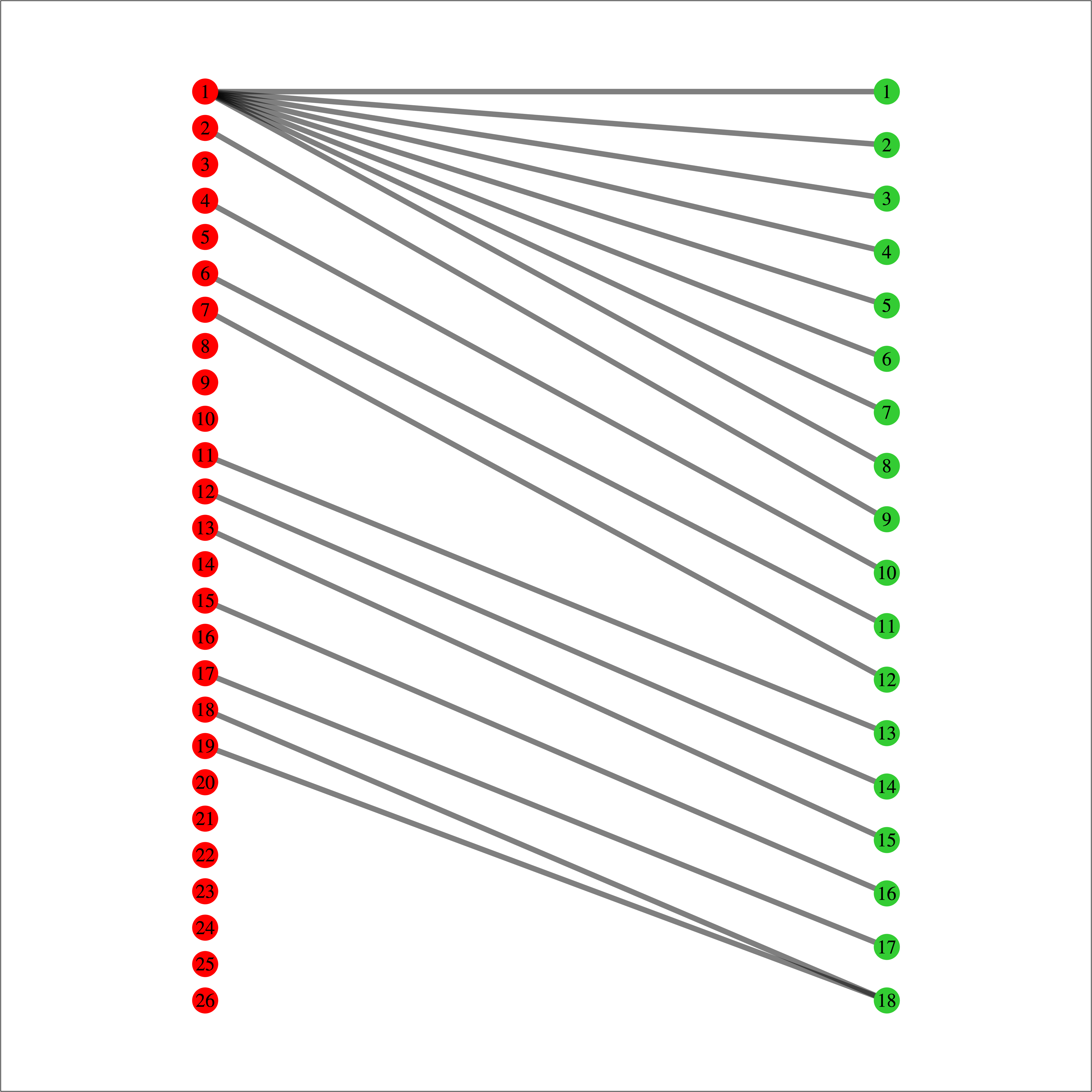}}
\end{subfigure}

\caption[]{
Continued from previous page.
}
\end{figure}

\subsubsection{The need for surgical reconstruction}
\label{subsec:need-for-surgery}

We now use the reversal-to-identity property
\eqref{shear-identity} to compare the
surgical \mts\ algorithm with the simpler
direct-replacement variant \eqref{direct-replacement}. The final
interfaces at $\tmax=30$ computed using the two approaches are shown in
\Cref{fig:alternating-shear-comparison} for the microscale resolving
case $h/L=\num{1.25e-3}$. In the surgical solution
(\Cref{fig:alternating-shear-comparison_surgical}), the satellite
interfaces generated during the filament-breakup phase are mapped back
to the material locations associated with the localized tangent-FTLE
troughs, and the interface is close to indistinguishable from the
reference solution away from these regions. By contrast, the direct-replacement solution
(\Cref{fig:alternating-shear-comparison_non-surgical}) exhibits
larger geometric errors, reflecting the repeated
discarding of accurately tracked Lagrangian geometry through global
Eulerian reconstruction. The final column of \Cref{tab:alternating-shear-runtime} confirms quantitatively
that the direct-replacement variant produces significantly larger
errors at the same resolution.

\begin{figure}[ht]
\centering
\begin{subfigure}[t]{0.67\linewidth}
  \centering
  \includegraphics[width=\linewidth]{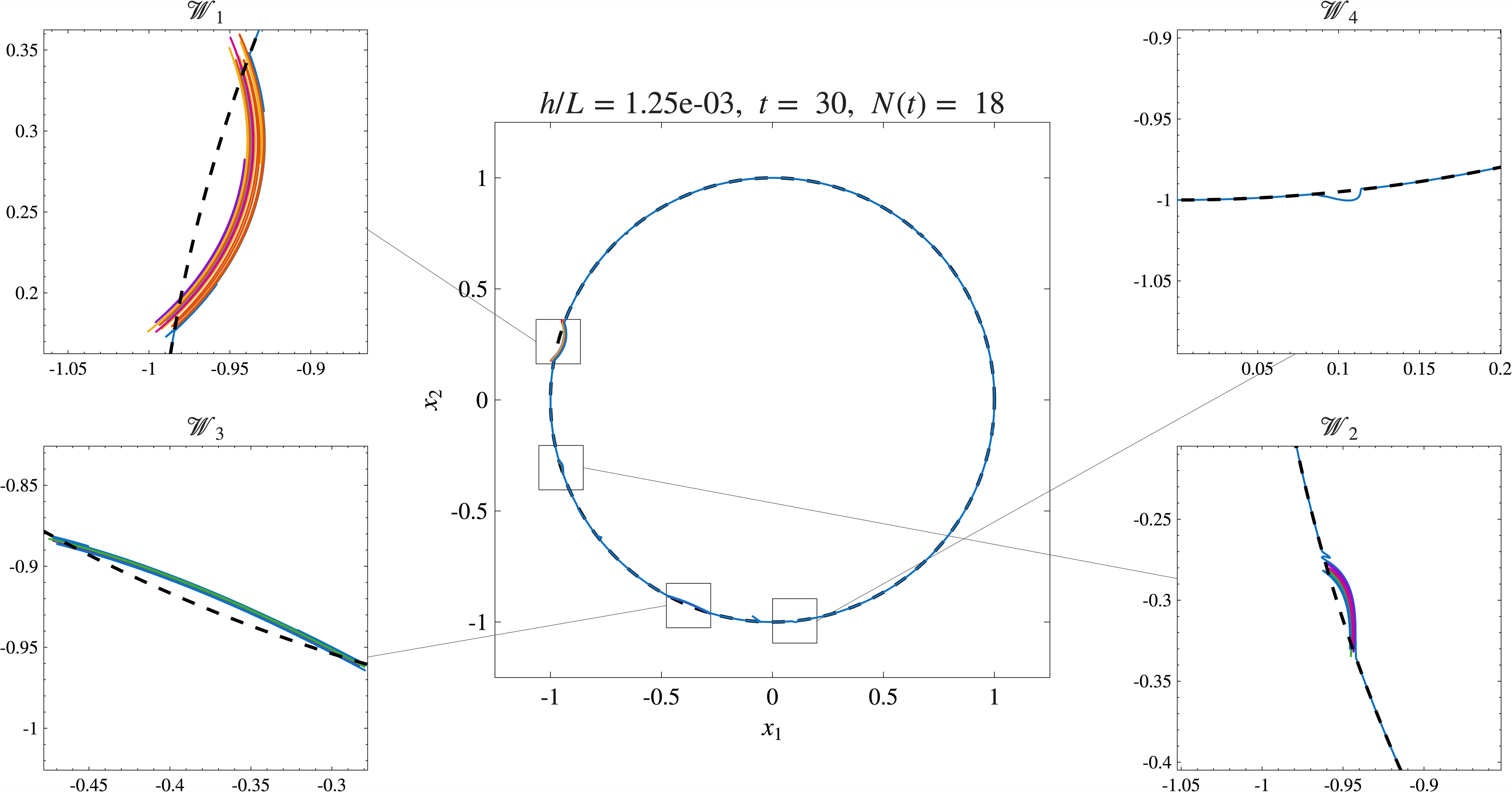}
  \caption{\mts\ solution}
  \label{fig:alternating-shear-comparison_surgical}
\end{subfigure}
\hspace{1em}
\begin{subfigure}[t]{0.285\linewidth}
  \centering
  \raisebox{0.26cm}{%
  \includegraphics[width=\linewidth]{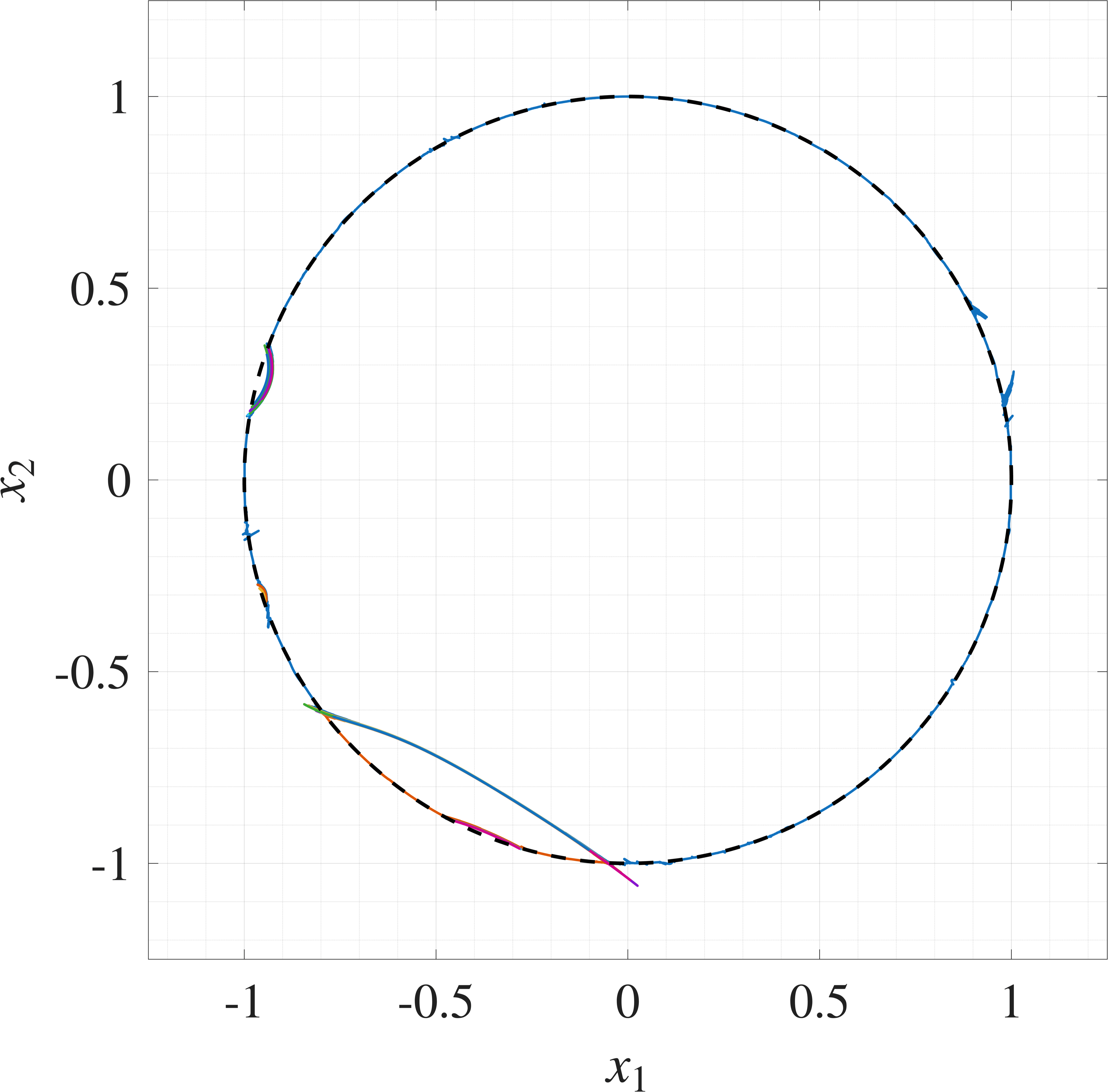}}
  \caption{direct-replacement solution}
  \label{fig:alternating-shear-comparison_non-surgical}
\end{subfigure}
\caption{
Comparison of the surgical \mts\ algorithm and the direct-replacement
variant \eqref{direct-replacement} for the nonlinear alternating-shear
test with microscale resolution $h/L=\num{1.25e-3}$ at the final time $\tmax = 30$.
The left panel displays the surgical \mts\ solution, while the right panel displays the 
direct-replacement solution. Repeated global Eulerian reconstruction 
produces nontrivial geometric errors in the latter.
}
\label{fig:alternating-shear-comparison}
\end{figure}

\subsubsection{Convergence and scalability}

The qualitative convergence of the \mts\ algorithm is illustrated in
\Cref{fig:alternating-shear-convergence}, which compares the final
interfaces at $\tmax=30$ for the dyadic refinement sequence
$h=2^{-p}h_0$. As $h$ is refined, the recovered
interface qualitatively converges toward the exact solution. At the coarser resolutions,
the accumulated effect of unresolved satellite disappearances produces
visible mass loss, and portions of the reconstructed interface are no
longer contained within the initial configuration. 
As the resolution is refined, the \mts\ solution becomes nearly
indistinguishable from the reference solution. Topological transitions
continue to occur, but they become progressively more localized and are
confined to sub-microscale structures.

\begin{figure}[ht]
\centering
\begin{subfigure}[t]{0.24\linewidth}
  \centering
  \includegraphics[width=0.95\linewidth]{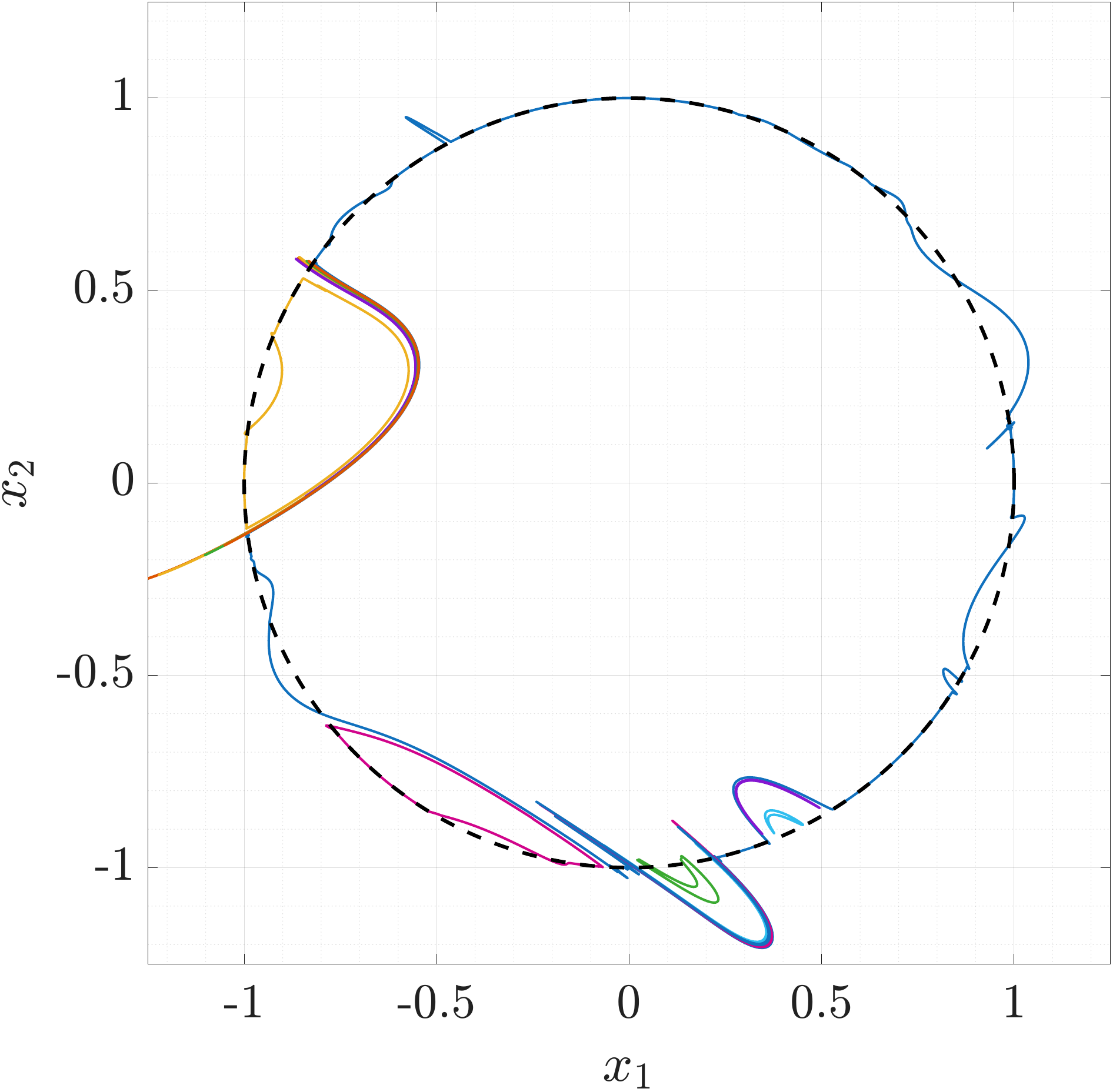}
  \caption{$h/L = \num{5.e-3}$}
  \label{fig:alternating-shear-h=0.03_t=30.0}
\end{subfigure}
\begin{subfigure}[t]{0.24\linewidth}
  \centering
  \includegraphics[width=0.95\linewidth]{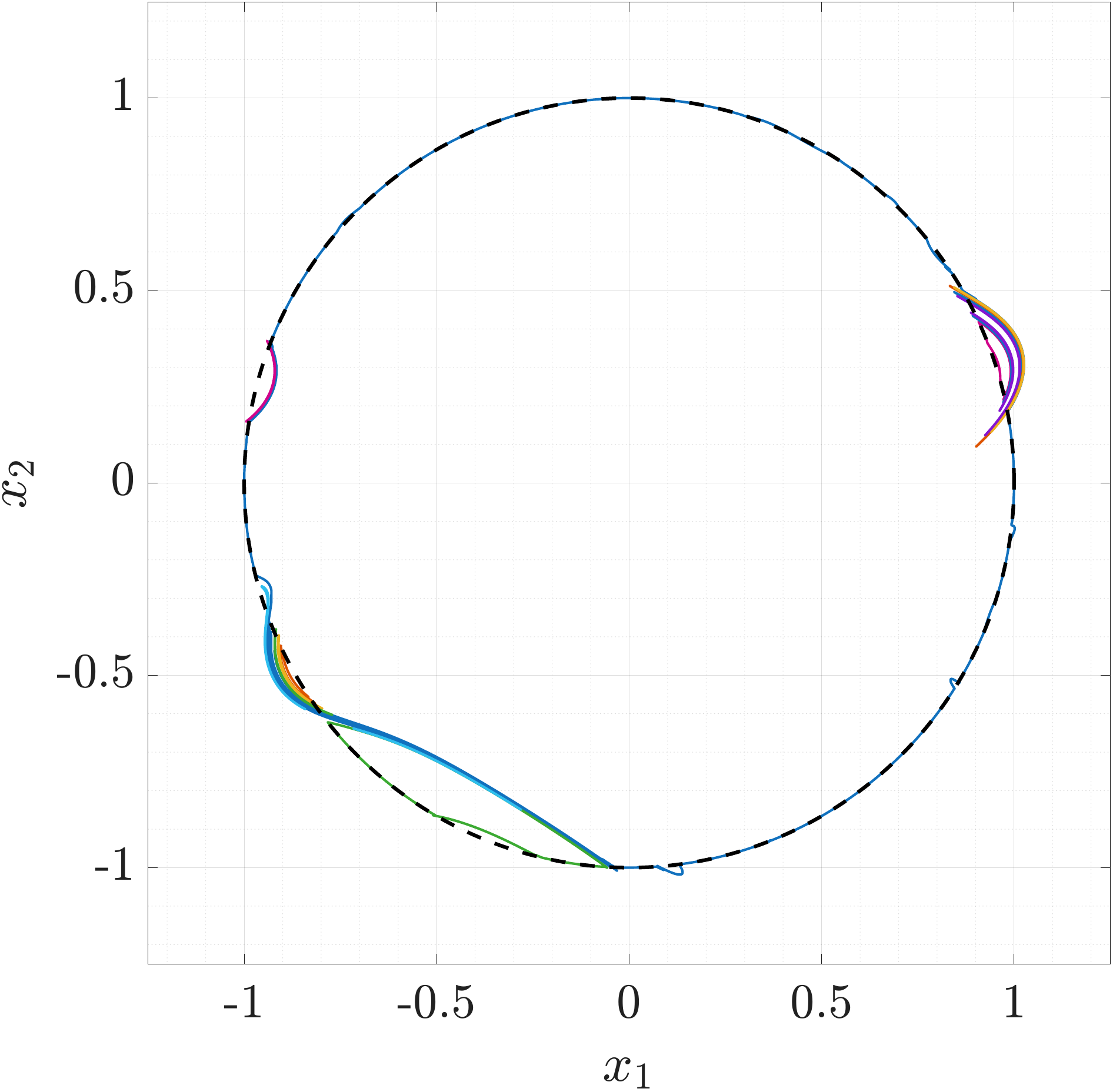}
  \caption{$h/L = \num{2.5e-3}$}
  \label{fig:alternating-shear-h=0.015_t=30.0}
\end{subfigure}
\begin{subfigure}[t]{0.24\linewidth}
  \centering
  \includegraphics[width=0.95\linewidth]{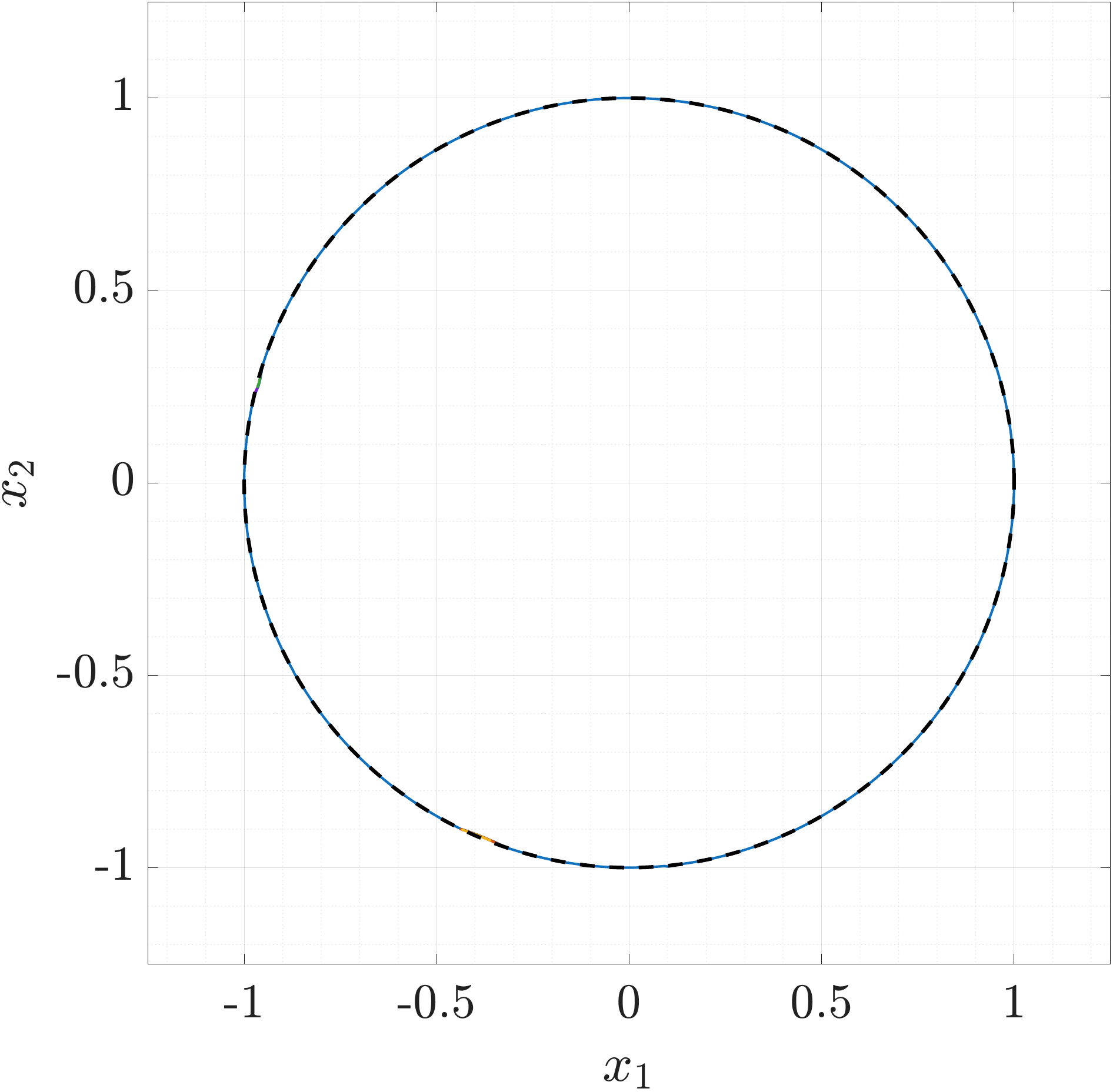}
  \caption{$h/L = \num{6.25e-4}$}
  \label{fig:alternating-shear-h=0.00375_t=30.0}
\end{subfigure}
\begin{subfigure}[t]{0.24\linewidth}
  \centering
  \includegraphics[width=0.95\linewidth]{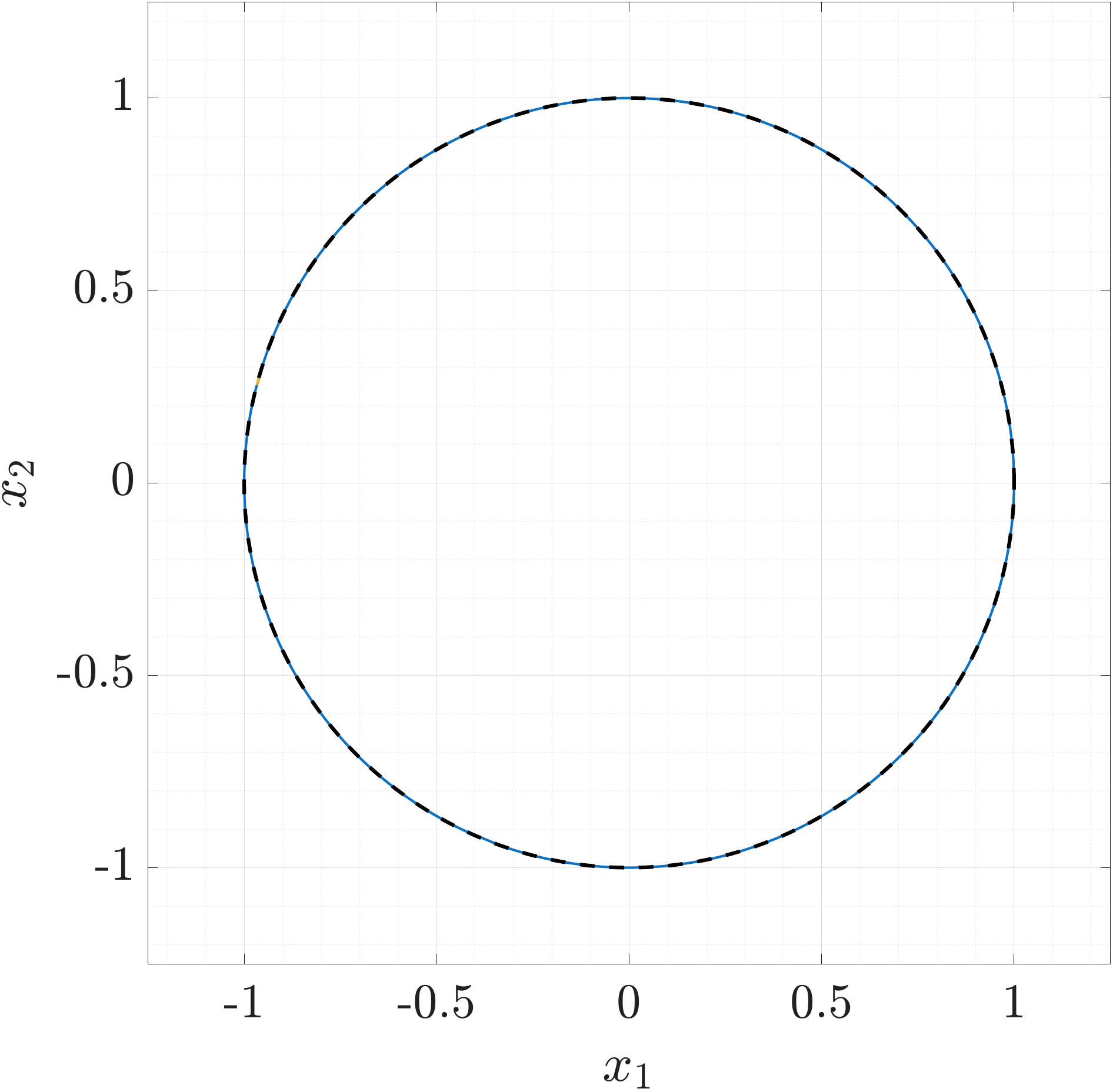}
  \caption{$h/L = \num{3.125e-4}$}
  \label{fig:alternating-shear-h=0.001875_t=30.0}
\end{subfigure}
\caption{
Qualitative convergence of the \mts\ algorithm for the nonlinear
alternating-shear test.
The panels compare the final interfaces at $\tmax=30$ for the dyadic
refinement sequence $h=2^{-p}h_0$. The reference solution is displayed in black. 
}
\label{fig:alternating-shear-convergence}
\end{figure}

Quantitative accuracy and runtime statistics are reported in
\Cref{tab:alternating-shear-runtime}. We first consider the numerical reversal error
$|A_K-A_0|$. Although the error decreases rapidly under refinement, the convergence
is not monotone. In particular, the error increases between the two
coarsest resolutions. This behavior reflects the competing influence of
two distinct error mechanisms. On the one hand, unresolved satellite
disappearances introduce mass loss and therefore decrease the enclosed
area. On the other hand, geometric reconstruction errors may locally
expand portions of the interface and increase the enclosed area.
At the coarsest resolutions these effects partially cancel, making the
area error alone an imperfect measure of geometric accuracy. 
A more informative diagnostic is provided by the symmetric-difference error \cite{JeSuSh2015} defined by
\begin{equation}
E_{\mathrm{sym}}
=
\left|
\Omega_K \,\triangle\, \Omega_0
\right|,
\qquad 
A \, \triangle \, B
=
(A\setminus B)\cup(B\setminus A), 
\end{equation}
where $\Omega_K$ and $\Omega_0$ denote the regions enclosed by the
final and initial interfaces, respectively, and
$A \,\triangle \, B$ denotes the symmetric difference of two sets.
Unlike the area-based reversal error, $E_{\mathrm{sym}}$ measures the total geometric
discrepancy between the two enclosed regions and therefore does not
permit cancellation between mass-loss and over-expansion errors.
Consistent with the qualitative behavior observed in
\Cref{fig:alternating-shear-convergence}, the symmetric-difference error
decreases more uniformly under refinement, with an observed convergence
rate of approximately three.

\begin{table}[ht]
\centering
\scriptsize
\caption{
Accuracy and runtime statistics for solutions to the nonlinear
alternating-shear test computed using the Lagrangian tracking + \mts\
topology-processing algorithm.
The table reports the numerical reversal error $|A_K-A_0|$, the symmetric-difference error
$E_{\mathrm{sym}}$, and wall-clock runtimes for the Lagrangian tracking
component, the \mts\ topology-processing component, and the overall
computation. The final column reports the corresponding
direct-replacement solution computed at 
$h/L=\num{1.25e-3}$, obtained by replacing the Lagrangian interface
family with the extracted Eulerian interface family at each \mts\ time.
}
\label{tab:alternating-shear-runtime}
\begin{tabular}{l@{\hspace{2em}}ccccc|c}
\toprule
$h/L$
& \num{5e-3}
& \num{2.5e-3}
& \num{1.25e-3}
& \num{6.25e-4}
& \num{3.125e-4}
& \num{1.25e-3} 
\\
& & & & & & \emph{direct-replacement} 
\\
\midrule

$|A_K-A_0|$
& \num{6.92e-3}
& \num{1.71e-2}
& \num{4.32e-3}
& \num{4.76e-4}
& \num{7.52e-5}
& \num{6.16e-3}
\\
Order
& --
& -1.30
& 1.98
& 3.18
& 2.66
\\
\midrule

$E_{\mathrm{sym}}$
& \num{1.63e-1}
& \num{4.77e-2}
& \num{5.34e-3}
& \num{5.92e-4}
& \num{7.92e-5}
& \num{9.95e-3}
\\
Order
& --
& 1.77
& 3.16
& 3.17
& 2.90
\\
\midrule

\LGR\ runtime (s)
& \num{4.01e0}
& \num{2.04e1}
& \num{7.34e1}
& \num{2.91e2}
& \num{1.11e3}
& \num{6.85e1}
\\
Order
& --
& 2.35
& 1.85
& 1.99
& 1.93
\\
\midrule

\mts\ runtime (s)
& \num{1.58e0}
& \num{6.60e0}
& \num{2.51e1}
& \num{1.00e2}
& \num{3.73e2}
& \num{5.93e0}
\\
Order
& --
& 2.06
& 1.93
& 2.00
& 1.90
\\
\midrule

Total runtime (s)
& \num{5.59e0}
& \num{2.70e1}
& \num{9.85e1}
& \num{3.92e2}
& \num{1.48e3}
& \num{7.45e1}
\\
Order
& --
& 2.27
& 1.87
& 1.99
& 1.92
\\

\bottomrule
\end{tabular}
\end{table}

The runtime statistics in \Cref{tab:alternating-shear-runtime} show that
both the Lagrangian tracking and \mts\ topology-processing components
scale approximately as $\mathcal O(h^{-2})$. Despite the rich sequence
of \emph{splits}, \emph{vanishings}, and \emph{merges} encountered during the filament-breakup
phase, the topology-processing stage accounts for only about $25\%$ of
the overall computational cost.

\subsubsection{Statistical description of filament breakup}

There are several motivations for studying microscale filamentation 
dynamics from a statistical perspective.

\begin{enumerate}
\item
\emph{Mathematically}, many sharp-interface evolution 
problems become ill-posed through the generation of progressively finer scales 
\cite{Delort1991,DeSz2009,Szekelyhidi2011}. 
This naturally motivates the consideration of statistical solutions, in which the interface 
evolution is described in terms of an ensemble of admissible realizations rather 
than a single deterministic trajectory \cite{DiMa1987,LaMiPa2021}.

\item
From a \emph{physical} perspective, the formation of arbitrarily fine
filamentary structures signals the breakdown of the idealized
sharp-interface description, which is subsequently resolved by the
\mts\ algorithm as filament breakup. 
Similar fragmentation processes arise in ligament
breakup and spray atomization, where the resulting droplet population
is commonly characterized through its size distribution
\cite{Villermaux2007}. An additional physical interpretation comes from 
chaotic-advection studies in which tracer structures below the Batchelor 
scale are removed by diffusion and the resulting tracer microstructure 
is characterized statistically \cite{Pierrehumbert1991,Pierrehumbert1994}. 
Here, we employ an analogous strategy,
replacing the diffusive regularization by \mts\ and evaluating the
persistence of filament-breakup events through ensemble statistics.

\item
The \emph{computational efficiency} of the \mts\ algorithm makes it
feasible to generate large ensembles of microscale-resolving interface
simulations.

\end{enumerate}

To this end, we construct an ensemble of $50$ realizations of the
nonlinear alternating-shear test at the microscale resolution
$h/L=\num{1.25e-3}$. The initial interface is fixed as the unit circle,
while the alternating-shear map is perturbed through
pseudo-random variations of the phase parameters.\footnote{For each ensemble index $k$ and each map index $m$,
the phase shifts \eqref{phase-shifts} are perturbed according to
$\phi_m\mapsto\phi_m+\varepsilon a_{k,m}$ and
$\psi_m\mapsto\psi_m+\varepsilon b_{k,m}$, where
$\varepsilon=0.2$ and
$a_{k,m},b_{k,m}\in[-1,1]$ are pseudo-random phase perturbations.}
We construct two sets of ensembles: a Lagrangian ensemble $\Sigmam$\LGR$_h$ and an 
\mts\ ensemble $\Sigmam$\mts$_h$ generated by the Lagrangian (\LGR) tracking and 
\mts\ algorithms, respectively. 
The complete 50-member \mts\ ensemble required approximately 
24 minutes of wall-clock time, of which approximately 20\% was spent on topology processing.

The first twenty members of each ensemble at the time of maximal deformation, $t=15$, are shown in
\Cref{fig:alternating-shear_ensemble}. 
The large-scale geometry of the interface remains
consistent across the ensemble, and the principal effect of the
ensemble averaging is a coarse-graining of the microscale
filamentary structures. In particular, the ensemble members occupy a
finite envelope whose width may be interpreted as a statistical
``thickening'' of the underlying sharp interface. 
Within the observation windows 
$\mathcal{W}_1$, $\mathcal{W}_2$, and $\mathcal{W}_4$, the 
uncertainty envelope exceeds the characteristic filament thickness, causing the 
corresponding filament-breakup processes to be effectively suppressed 
by the coarse-graining. In the strongest filamenting window 
$\mathcal{W}_3$, however, the filament-breakup process survives the 
uncertainty-induced coarse-graining, and remains \emph{statistically persistent}.

\begin{figure}[ht]
\centering
\begin{subfigure}[t]{0.6\linewidth}
  \centering
  \includegraphics[width=\linewidth]{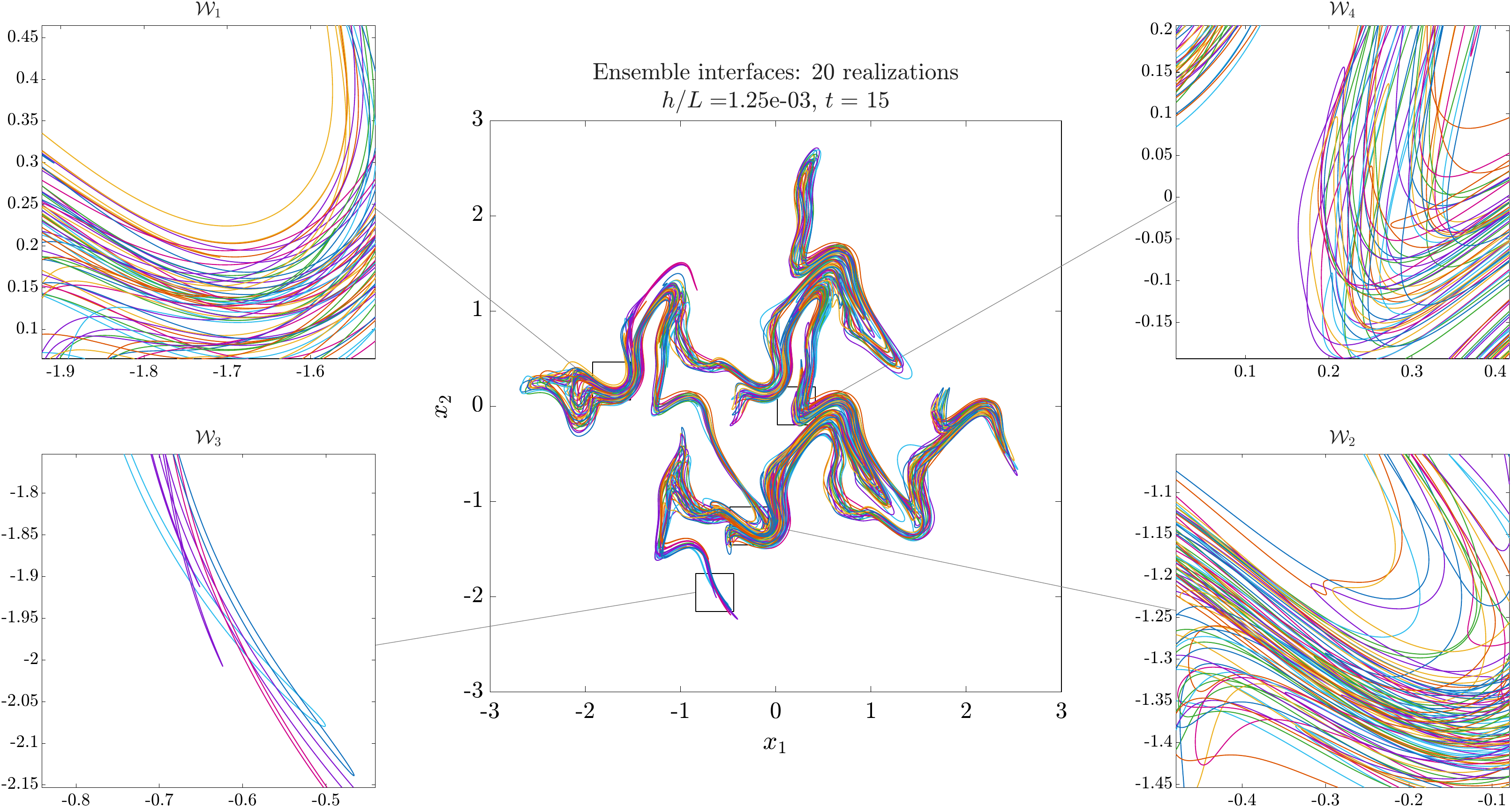}
  \caption{$\Sigmam$\LGR$_h$ ensemble}\vspace{1em}
  \label{fig:alternating-shear_ensemble_LGR}
\end{subfigure}
\begin{subfigure}[t]{0.6\linewidth}
  \centering
  \includegraphics[width=\linewidth]{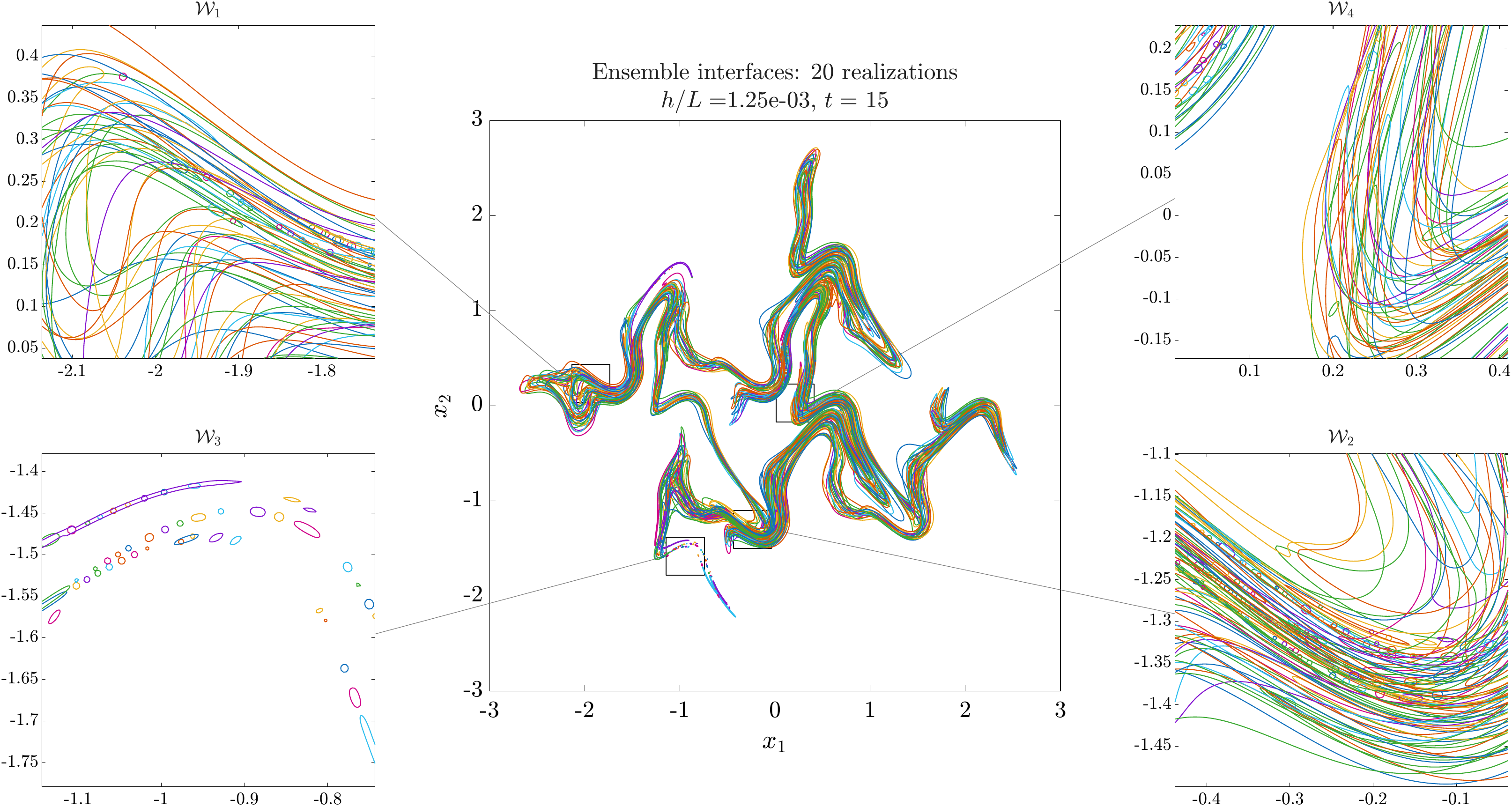}
  \caption{$\Sigmam$\mts$_h$ ensemble}
  \label{fig:alternating-shear_ensemble_MTS}
\end{subfigure}
\caption{
First twenty realizations of the pseudo-random alternating-shear ensembles at the 
time of maximal deformation, $t=15$, computed with microscale resolution $h/L=\num{1.25e-3}$.
A video of the simulation is available at \cite{Ramani2026}. 
The large-scale interface geometry remains consistent across the ensemble, while the
ensemble spread coarse-grains the filamentary structures into a statistically thickened interface.
Within the observation windows $\mathcal{W}_1$, $\mathcal{W}_2$, and
$\mathcal{W}_4$, the uncertainty envelope exceeds the characteristic
filament thickness, causing the corresponding filament-breakup
processes to be suppressed. By contrast, the
breakup process within the strongest filamenting window
$\mathcal{W}_3$ survives the uncertainty-induced coarse-graining and
remains \emph{statistically persistent}.
}
\label{fig:alternating-shear_ensemble}
\end{figure}

To quantify the breakup process, we next consider statistics of the
interfaces generated within the ensemble. The satellite-count
distribution shown in \Cref{fig:droplet_statistics-c} is relatively
broad, indicating that the precise number of breakup fragments remains
sensitive to microscale variations in the filamentation process.
Consequently, individual realizations often exhibit substantially
different topological outcomes. Nevertheless, the rank-ordered area
statistics shown in \Cref{fig:droplet_statistics-a,fig:droplet_statistics-b}
reveal a much more coherent picture. For each realization, the
interface areas are rank ordered and normalized by the total enclosed
area. The resulting area spectrum exhibits a clear separation between
the primary interface and the population of satellite droplets, whose
normalized areas span several orders of magnitude, including scales well
below the prescribed microscale resolution. 
Moreover, the smooth median spectrum and
relatively narrow interquartile range indicate that the breakup process
generates a consistent hierarchy of satellite sizes across the ensemble.
Similar broad size hierarchies arise in ligament-fragmentation and 
spray-atomization processes, where a dominant parent structure produces 
a distributed population of progressively smaller fragments \cite{Villermaux2007}. 
The ability to capture this scale separation highlights a key feature of the 
\mts\ framework: topological resolution is controlled by the prescribed scale $h$, while 
geometric information is retained at significantly smaller scales.

\begin{figure}[ht]
\centering
\begin{subfigure}[t]{0.28\linewidth}
  \centering
  \includegraphics[width=0.95\linewidth]{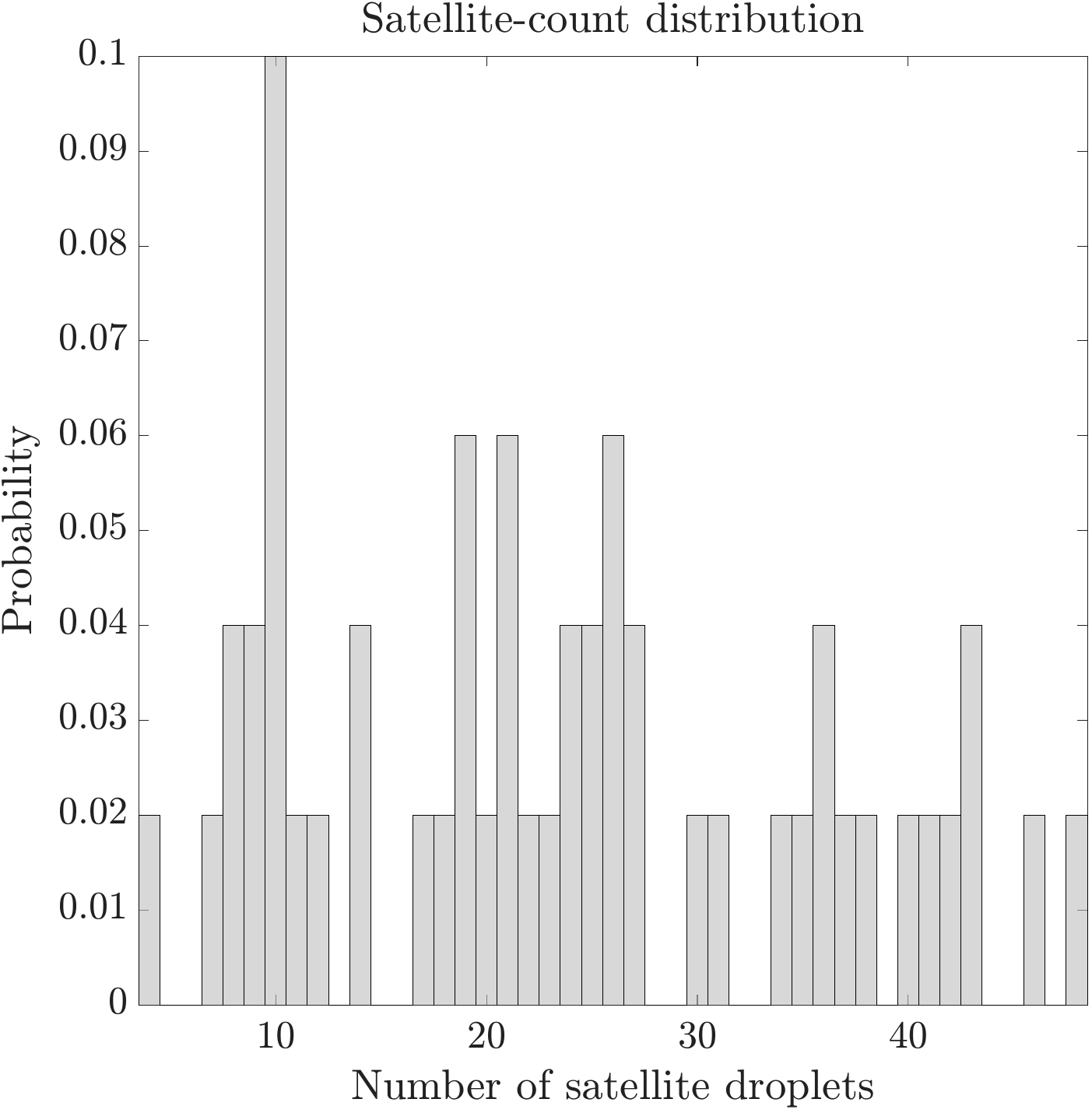}
  \caption{}
  \label{fig:droplet_statistics-c}
\end{subfigure}
\hspace{2em}
\begin{subfigure}[t]{0.28\linewidth}
  \centering
  \includegraphics[width=0.95\linewidth]{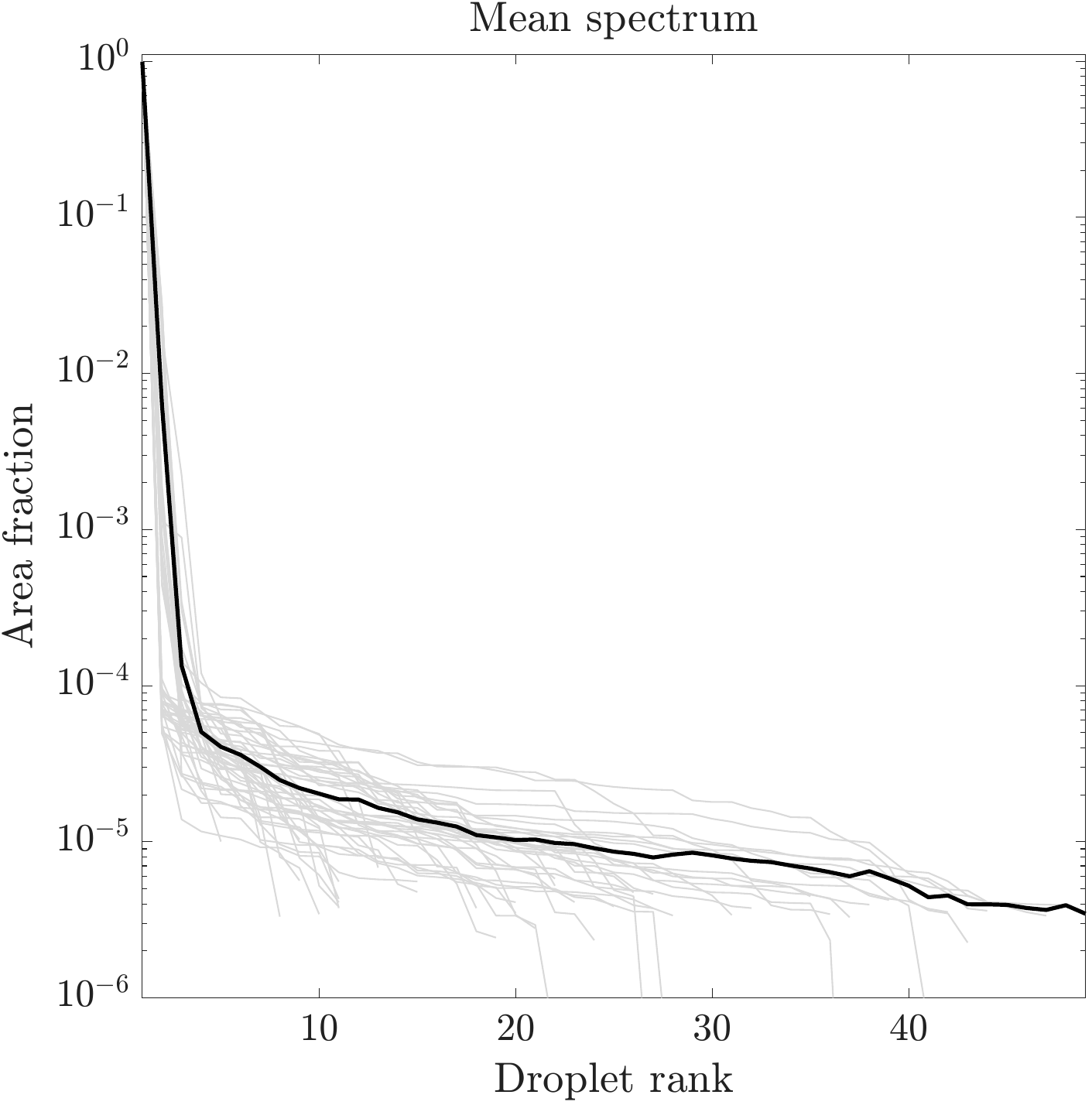}
  \caption{}
  \label{fig:droplet_statistics-b}
\end{subfigure}
\hspace{2em}
\begin{subfigure}[t]{0.28\linewidth}
  \centering
  \includegraphics[width=0.95\linewidth]{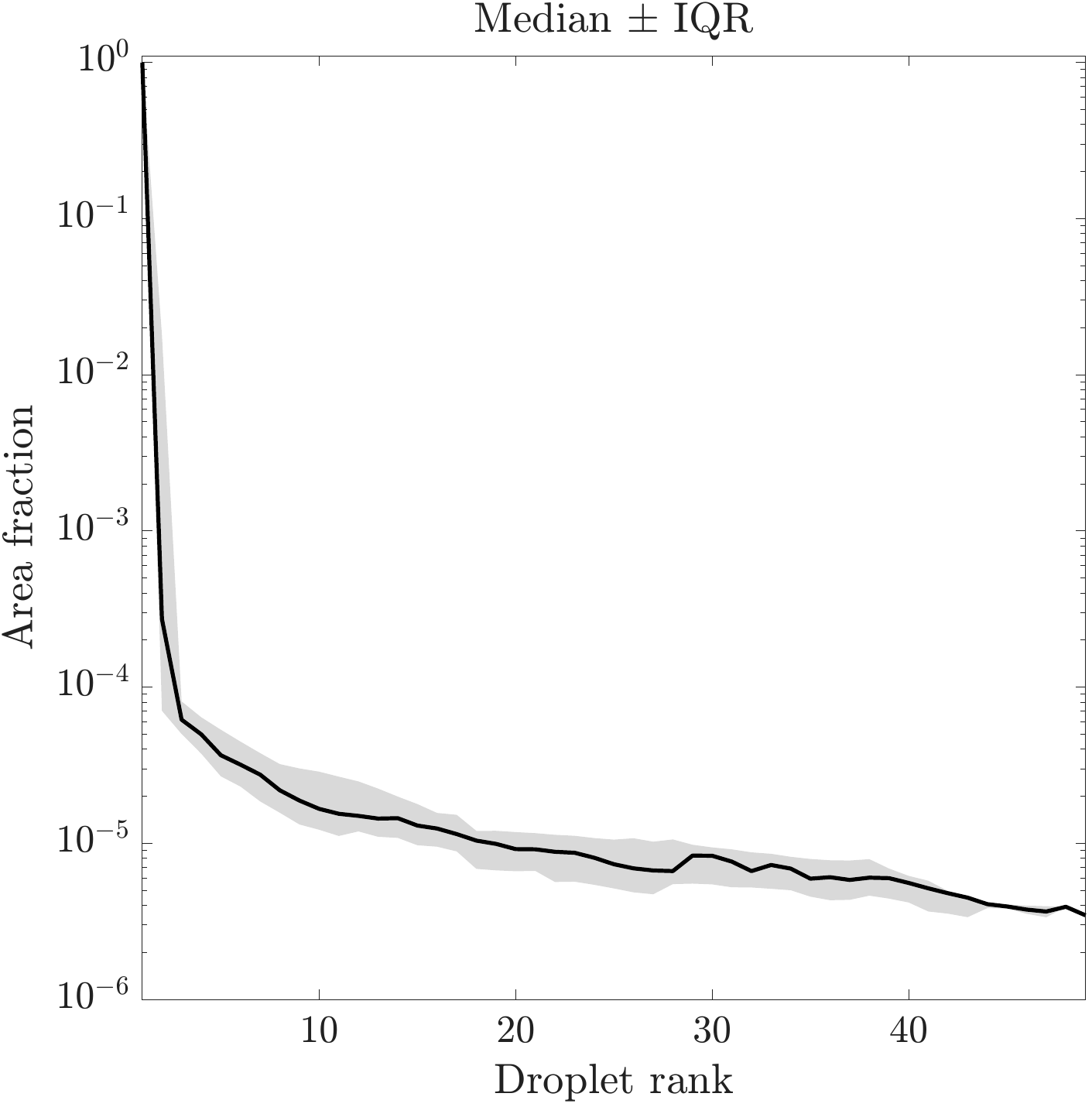}
  \caption{}
  \label{fig:droplet_statistics-a}
\end{subfigure}
\caption{
Droplet statistics for the pseudo-random alternating-shear ensemble at
$t=15$ and microscale resolution $h/L=\num{1.25e-3}$.
\textbf{Left:} Probability distribution of the number of satellite
droplets generated in each realization.
\textbf{Center:} Rank-ordered area spectrum showing the normalized area
fractions of all interface components; gray curves denote individual
realizations and the black curve denotes the ensemble mean.
\textbf{Right:} Median rank-ordered area spectrum together with the
interquartile range.
Although the number of satellite droplets varies substantially across
the ensemble, the rank-ordered area statistics exhibit a consistent
hierarchy of satellite sizes spanning several orders of magnitude,
including scales below the prescribed microscale resolution.
}
\label{fig:droplet_statistics}
\end{figure}

To visualize the coarse-graining induced by the ensemble averaging, we
convert each realization into an Eulerian phase-field $\chi(x,t)$,
defined as a regularized indicator function of the enclosed region:
\begin{subequations}\label{phase-field}
\begin{equation}\label{phase-field-a}
\chi(x,t)
=
\begin{cases}
1, & \phi(x,t)\le -4h,\\[2mm]
0, & \phi(x,t)\ge 4h,\\[2mm]
\tfrac12\left(1-\tfrac{\phi}{4h}
-\tfrac{1}{\pi}\sin\!\left(\pi\tfrac{\phi}{4h}\right)\right),
& |\phi(x,t)|<4h,
\end{cases}
\end{equation}
where $\phi(x,t)$ denotes the signed-distance function. Thus,
$\chi=1$ inside the interface, $\chi=0$ outside, and transitions
smoothly across a narrow band of width $8h$ surrounding the interface.
The ensemble-averaged phase field is then defined by
\begin{equation}\label{phase-field-b}
\bar{\chi}(x,t)
=
\frac{1}{50}
\sum_{k=1}^{50}
\chi_k(x,t),
\end{equation}
\end{subequations}
where $\chi_k$ denotes the phase field associated with the $k$-th
realization. The quantity $\bar{\chi}(x,t)$ may be interpreted as the
fraction of realizations for which the point $x$ lies inside the
enclosed region.

\begin{figure}[H]
\centering
\begin{subfigure}[t]{0.28\linewidth}
  \centering
  \caption{\LGR$_h$}
  \includegraphics[width=0.95\linewidth]{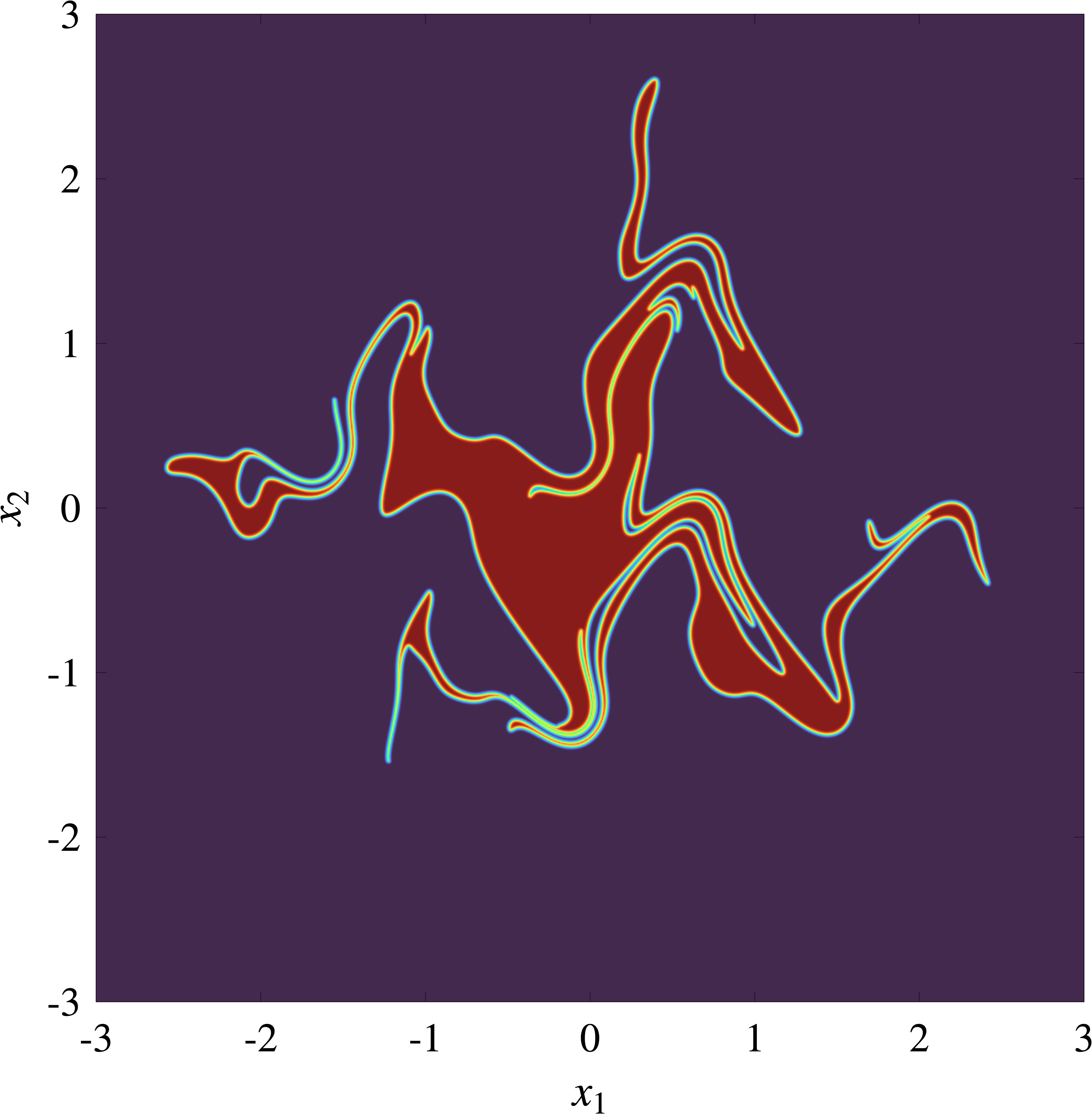}
  \label{fig:alternating-shear_phasefield_LGR}
\end{subfigure}
\hspace{2em}
\begin{subfigure}[t]{0.28\linewidth}
  \centering
  \caption{$\Sigmam$\LGR$_h$}
  \includegraphics[width=0.95\linewidth]{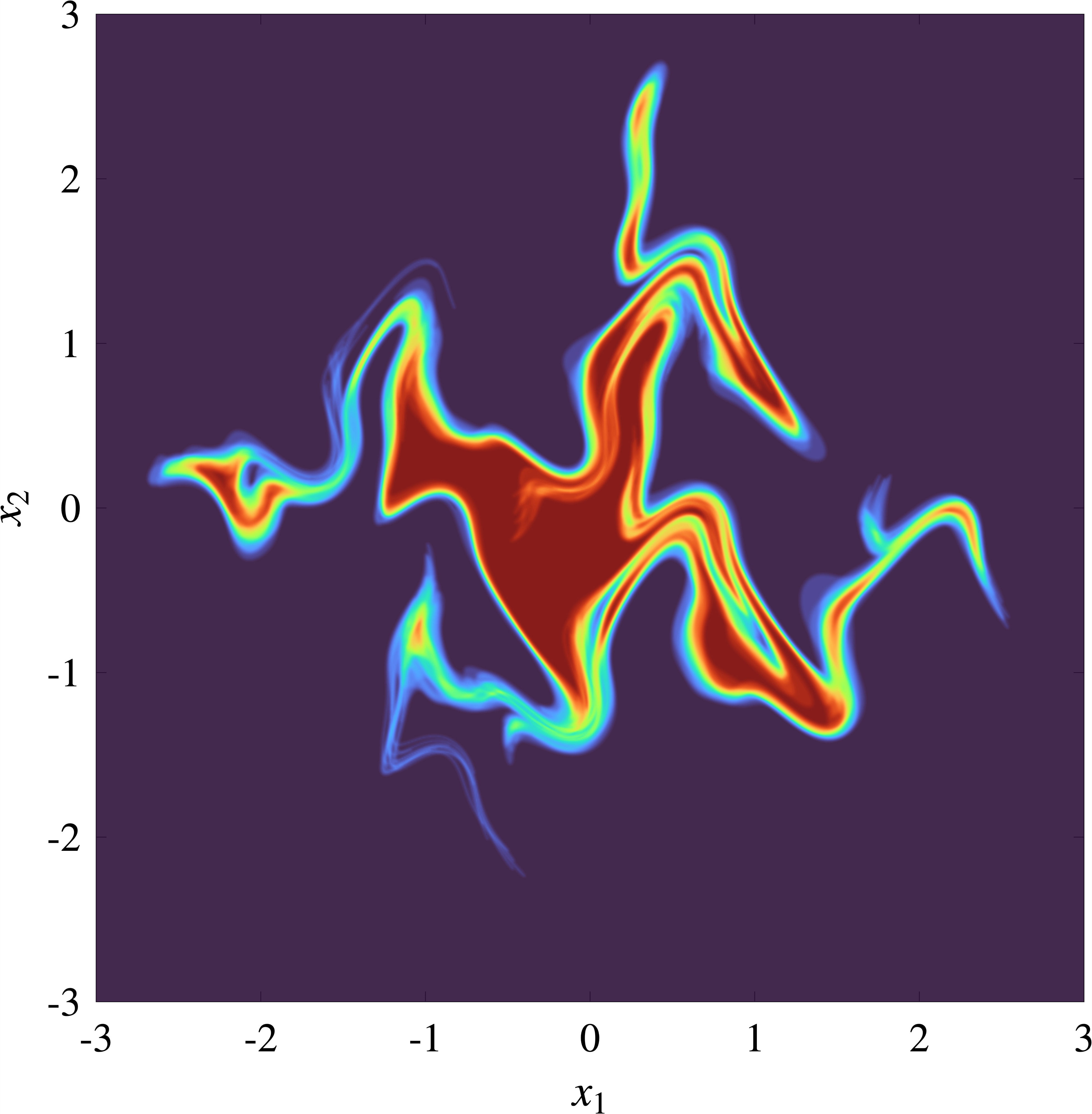}
  \label{fig:alternating-shear_phasefield_LGR-ens}
\end{subfigure}
\hspace{2em}
\begin{subfigure}[t]{0.28\linewidth}
  \centering
  \caption{$\Sigmam$\LGR$_{2h}$}
  \includegraphics[width=0.95\linewidth]{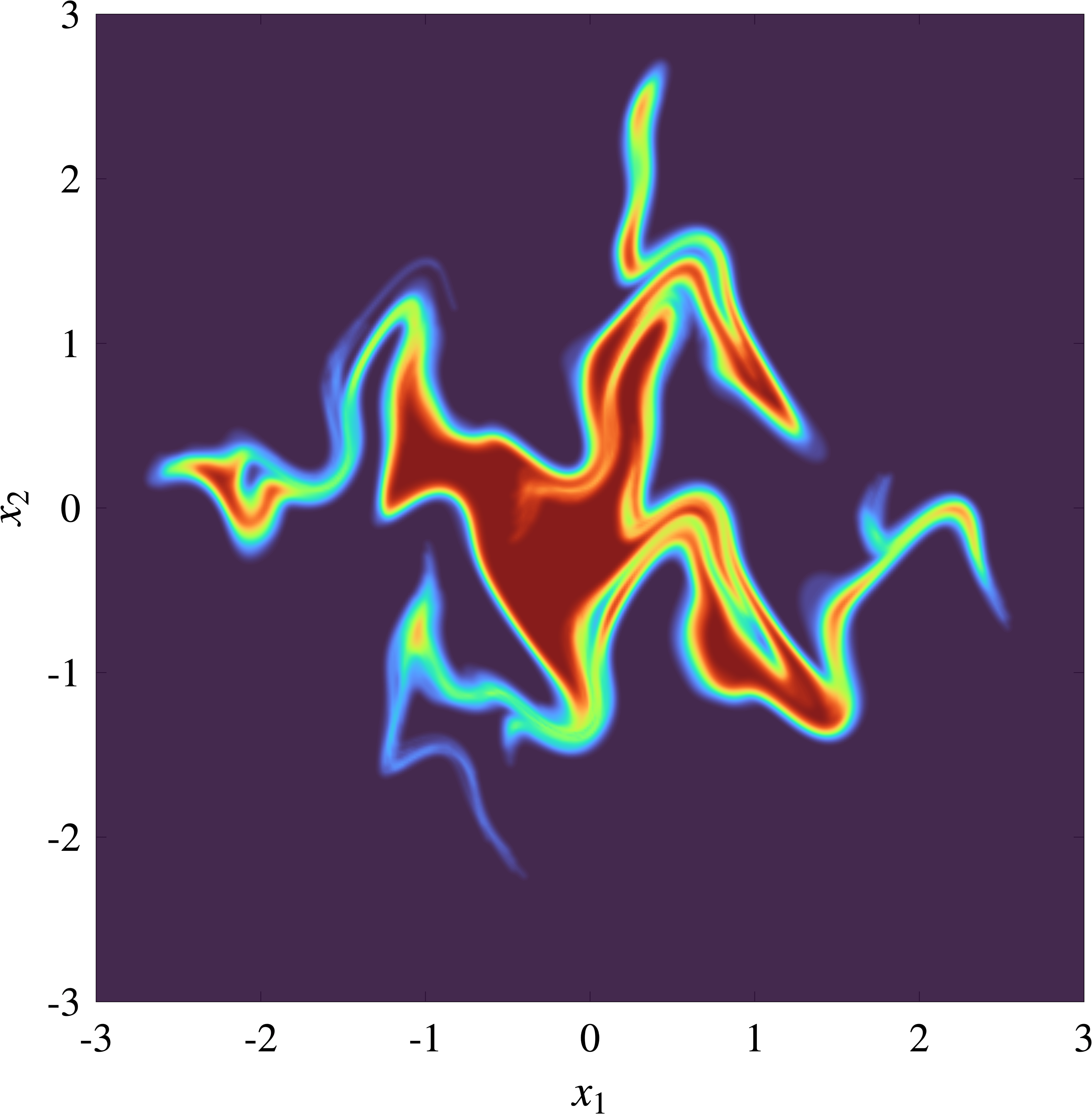}
  \label{fig:alternating-shear_phasefield_LGR-ens2}
\end{subfigure}
\vspace{2em}
\begin{subfigure}[t]{0.28\linewidth}
  \centering
  \caption{\mts$_h$}
  \includegraphics[width=0.95\linewidth]{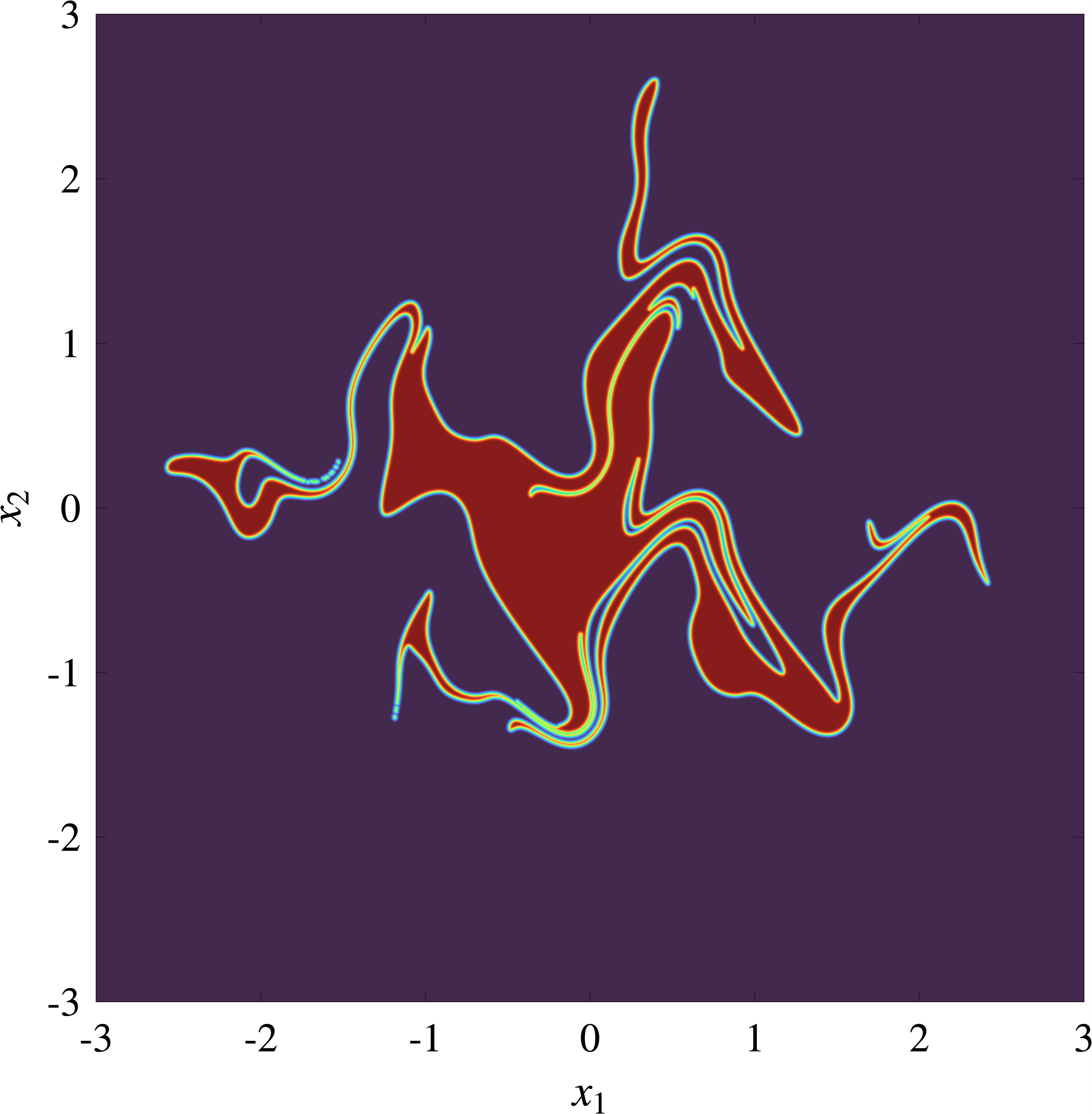}
  \label{fig:alternating-shear_phasefield_MTS}
\end{subfigure}
\hspace{2em}
\begin{subfigure}[t]{0.28\linewidth}
  \centering
  \caption{$\Sigmam$\mts$_h$}
  \includegraphics[width=0.95\linewidth]{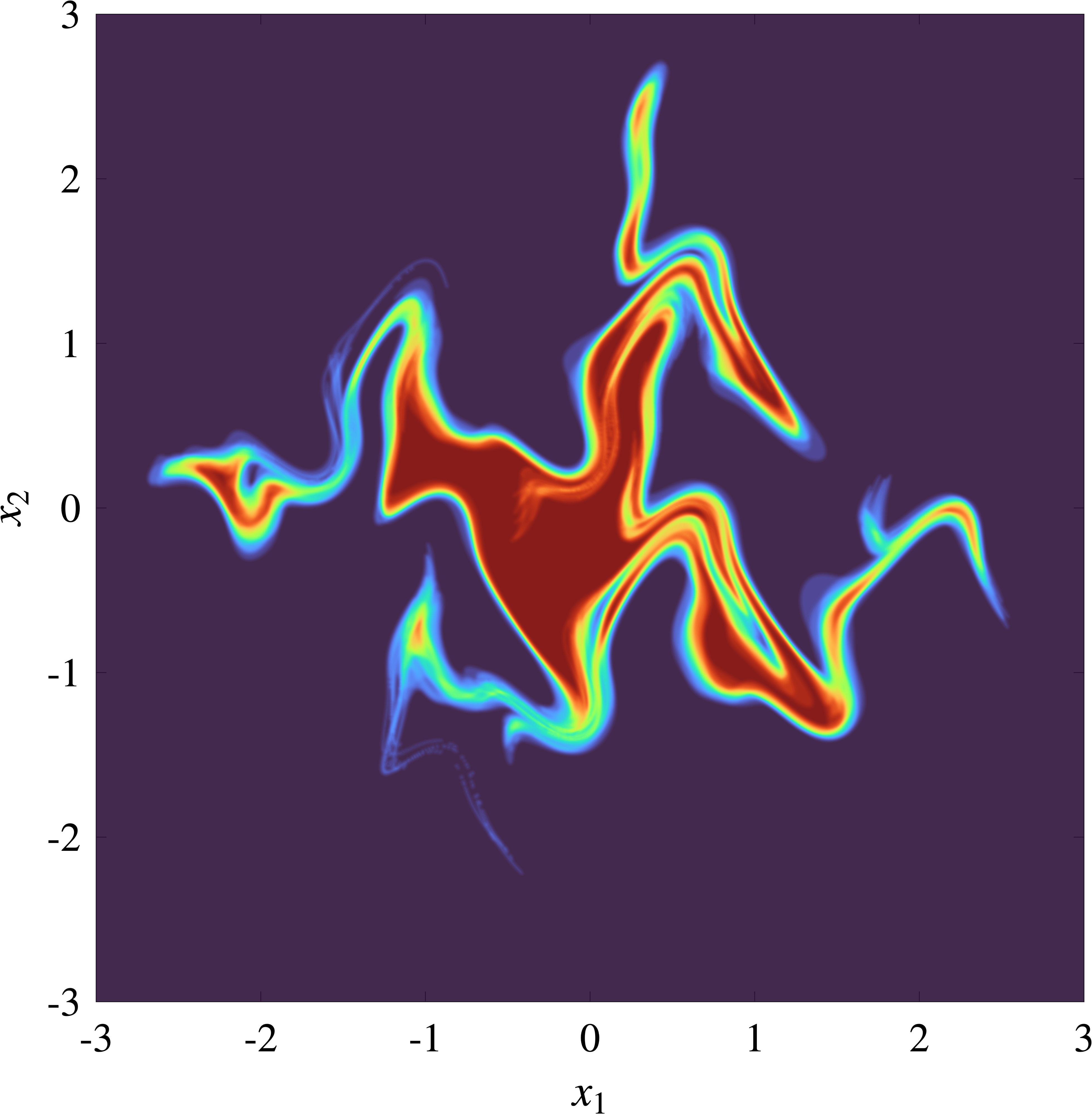}
  \label{fig:alternating-shear_phasefield_MTS-ens}
\end{subfigure}
\hspace{2em}
\begin{subfigure}[t]{0.28\linewidth}
  \centering
  \caption{$\Sigmam$\mts$_{2h}$}
  \includegraphics[width=0.95\linewidth]{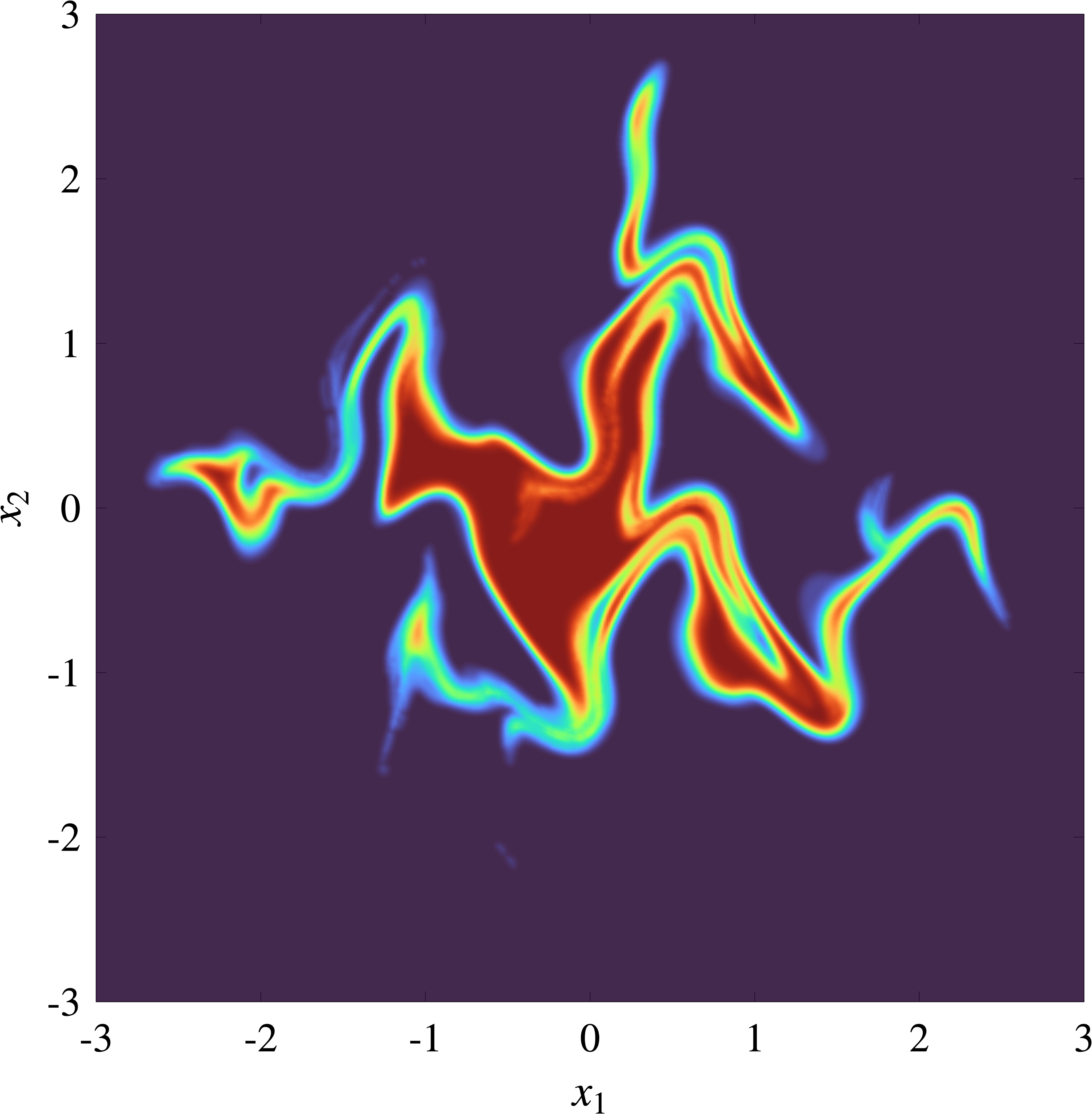}
  \label{fig:alternating-shear_phasefield_MTS-ens2}
\end{subfigure}
\vspace{-2em}
\caption{
Phase-field representation of the pseudo-random alternating-shear
ensemble at the time of maximal deformation, $t=15$.
The left column shows the phase-field $\chi(x,t)$ computed by 
\eqref{phase-field-a} for a single realization. The center and right columns show the
ensemble-averaged fields $\bar{\chi}(x,t)$ computed using \eqref{phase-field-b} at
resolutions $h/L=\num{1.25e-3}$ and $h/L=\num{2.5e-3}$,
respectively. Ensemble averaging converts the sharp-interface
description into a diffuse-interface representation whose thickness
reflects uncertainty in the interface location. 
In the strongest filamenting region, the breakup dynamics 
are statistically persistent. 
}
\label{fig:alternating-shear_phasefield}
\end{figure}

The phase-field comparison in \Cref{fig:alternating-shear_phasefield}
illustrates the diffuse regularization induced by ensemble averaging.
Although the individual \LGR$_h$ and \mts$_h$ realizations differ
visibly near topological events, their ensemble-averaged phase fields
are very similar, indicating that the differences between the two
methods are primarily microscale. In the ensemble average, uncertainty
in the computed interface location generates a diffuse transition region
that absorbs many realization-dependent filament-breakup features. The
strongest filamenting region provides the main exception: at resolution
$h$, the breakup leaves a visible satellite branch in the coarse-grained
phase field. The role of the \mts\ regularization scale $h$ is analogous to that of 
the Batchelor scale in the chaotic-advection studies of 
\cite{Pierrehumbert1991,Pierrehumbert1994}: filamentary structures below this scale are 
removed from the sharp-interface description through topological regularization. 
Ensemble averaging then acts on the resulting family of regularized realizations, suppressing 
weak filament-breakup signatures and retaining only the most robust breakup events.
By contrast, classical Lagrangian tracking is topologically rigid and therefore admits no analogous mechanism. 

At resolution $2h$ the same strongly filamenting branch has essentially
disappeared. Since coarsening simultaneously broadens the diffuse
transition layer and causes topological regularization to occur at a
larger scale, this behavior highlights the strength of the underlying
filamentation process in this region. 
The same competition between
filamentation, topological regularization, and statistical
coarse-graining is also visible near the origin, but with a different
coarse-grained signature: rather than eliminating a coherent filament
branch, the flow fragments into a spray.

\section{Conclusion}
\label{sec:conclusion}

In this work, we introduced a novel Microscale Topological Surgery (\mts) algorithm for 
efficient and accurate topological processing of interfaces tracked by Lagrangian methods. 
The method is based on an Eulerian hybridization that supplies a topological oracle through 
isosurface extraction, together with an adjacency topology inferred from the geometric 
configuration of the pre-processed and extracted interface families. 
A subsequent reconstruction procedure, localized through the use of topological 
defect measures, surgically stitches together the two interface families to maintain 
geometric accuracy through topological transitions. 
Application of the method to a new nonlinear alternating-shear benchmark demonstrates 
its ability to accurately capture complex multiscale filament-breakup processes while 
maintaining $\mathcal{O}(h^2)$ geometric accuracy and optimal $\mathcal{O}(h^{-2})$ computational scaling.
The computational efficiency of the \mts\ framework also enables the construction of large 
ensembles of microscale-resolving simulations; for the alternating-shear benchmark 
considered here, the ensemble-averaged solution reveals a statistically persistent 
filament-breakup process in the strongest filamenting region.

Several directions for future work appear promising. 
One natural direction is a three-dimensional implementation of the \mts\ algorithm. 
Another is a more systematic investigation of the
relationship between \mts\ and
diffuse-interface regularization, particularly in light of the
statistical coarse-graining behavior observed in the present work and
the physical regularization framework proposed in \cite{LoTr1998}.
Related to this, \cite{LoTr1998} envisioned the development of
boundary-integral methods capable of accommodating topological
transitions. To date, such approaches have largely been restricted to
Eulerian formulations or simplified classes of topology changes
\cite{Herrmann2005,GaGrSe2009,BaAnMe2004}. 
Finally, the computational efficiency of the \mts\
algorithm makes it feasible to generate large ensembles of
microscale-resolving interface simulations, suggesting potential
applications in data-driven modeling and statistical learning of
filamentation and breakup dynamics.

\section*{Acknowledgements}
This work was supported by the Mark Kac 
Applied Mathematics Postdoctoral Fellowship at the 
Center for Nonlinear Studies at Los Alamos National Laboratory. 
Los Alamos National Laboratory Report LA-UR-26-25296.


\appendix

\section{Algorithms for interface tracking, capturing, \& surgery}
\label{appendix:aux-algs}

This appendix provides the algorithms used by the
hybrid Lagrangian-tracking and Eulerian-snapshot framework developed in this work. 
In particular, we describe: 
\begin{itemize}
\item the adaptive refinement procedure for the
Lagrangian interface discretization (\Cref{subsec:adaptive_refinement}); 
\item the construction of the localized
signed distance function representation (\Cref{subsec:sdf}); 
\item the graph-based marching-squares
component identification algorithm used for zero-level extraction; the
graph traversal procedure used to convert connected components into
ordered polygonal curves; and the subsequent refinement and filtering
operations applied to the extracted interface family 
(all \Cref{subsec:zero_level_extraction}).
\item the greedy event-local surgical reconstruction algorithm and the
post-processing procedure used to smooth the reconstructed
interfaces (\Cref{subsec:localized_surgery}).
\end{itemize}
The notation and default parameter choices appearing in these 
algorithms are given in \Cref{tab:notation}.


\begin{breakablealgorithm}
\caption{Adaptive refinement of an interface}
\label{alg:air}
\begin{algorithmic}[1]
\Require Interface $\gamma_{\alpha,i}(t)$; auxiliary tangent state
$\{\xi_{\alpha,i}(t),\,|\xi_{\alpha,i}(0)|\}$; thresholds $\dref$ and $\dcrs$
\Ensure Adaptively redistributed interface and auxiliary tangent state

\State Compute segment lengths 
$d_{\alpha,i}=|\gamma_{\alpha,i+1}(t)-\gamma_{\alpha,i}(t)|$, $i=1,\ldots,N_\alpha$.

\If{$\max_i d_{\alpha,i}>\dref$}

    \For{$i=1,\ldots,N_\alpha$}
        \If{$d_{\alpha,i}>\dref$}
            \State Set
            $
            s_{\alpha,i+1/2}
            =
            (s_{\alpha,i}+s_{\alpha,i+1})/2
            $
            \State Define
            $
            j_1=i-1,\ 
            j_2=i,\ 
            j_3=i+1,\ 
            j_4=i+2
            $
            with periodic indexing
            \State Insert the cubic interpolated node
            \[
            \gamma_{\alpha,i+1/2} (t)
            =
            \sum_{q=1}^{4}
            \gamma_{\alpha,j_q} (t)
            \prod_{\substack{r=1\\ r\neq q}}^{4}
            \tfrac{
            s_{\alpha,i+1/2}-s_{\alpha,j_r}
            }{
            s_{\alpha,j_q}-s_{\alpha,j_r}
            } .
            \]
            \State Interpolate the auxiliary tangent state to the new node
            \[
            \xi_{\alpha,i+1/2}(t)
            =
            \tfrac12
            \left(
            \xi_{\alpha,i}(t)
            +
            \xi_{\alpha,i+1}(t)
            \right)
            \quad
            \text{and}
            \quad 
            |\xi_{\alpha,i+1/2}(0)|
            =
            \tfrac12
            \left(
            |\xi_{\alpha,i}(0)|
            +
            |\xi_{\alpha,i+1}(0)|
            \right).
            \]
        \EndIf
    \EndFor

\Else

    \For{$i=1,\ldots,N_\alpha$}
        \If{$d_{\alpha,i}<\dcrs$}
            \State Delete $\gamma_{\alpha,i+1} (t)$ and its auxiliary tangent state
            $\xi_{\alpha,i+1}(t)$ and $|\xi_{\alpha,i+1}(0)|$
            \State Connect $\gamma_{\alpha,i} (t)$ directly to $\gamma_{\alpha,i+2} (t)$
        \EndIf
    \EndFor

\EndIf

\end{algorithmic}
\end{breakablealgorithm}


\begin{breakablealgorithm}
\caption{Narrow-band SDF construction}
\label{alg:sdf}

\begin{algorithmic}[1]
\Require Polygonal interface family
$\{\Gammalgr_\alpha\}_{\alpha=1}^{\Nlgr}$; refined mesh spacing $h$
\Ensure Signed-distance values $\phi(x_i)$ on a sparse local refined grid

\State Initialize storage masks by $m(x_i)=0$
\State Initialize temporary squared distances by $D(x_i)=+\infty$

\For{$\alpha=1,\ldots,\Nlgr$}
    \For{$k=1,\ldots,N_\alpha$}
        \State Define the segment
        $
        S_{\alpha,k}=[\gamma_{\alpha,k},\gamma_{\alpha,k+1}]
        $
        \State Compute the coarse-cell index box intersecting the
        $3h$-neighborhood of $S_{\alpha,k}$

        \For{each refined cell in this index box}
            \If{$\mathrm{dist}(\mathrm{cell},S_{\alpha,k})>3h$}
                \State \textbf{continue}
            \EndIf

            \For{each corner node $x_i$ of the refined cell}
                \State Mark $x_i$ as stored: $m(x_i)=1$
                \State Update the local squared distance:
                $
                D(x_i)
                \gets
                \min
                \left\{
                D(x_i),
                \mathrm{dist}(x_i,S_{\alpha,k})^2
                \right\}
                $
            \EndFor
        \EndFor
    \EndFor
\EndFor

\State Compress the marked coarse cells into a sparse local-block representation
\State Set $d(x_i)=\sqrt{D(x_i)}$ on all stored nodes
\State Bin non-horizontal interface segments by the refined horizontal grid rows they intersect

\For{each stored node $x_i$}
    \State Count the number $n_i$ of intersections between the horizontal ray
    $
    R_i=\{x_i+s(1,0):s>0\}
    $
    and the row-binned candidate segments
    \State Assign the signed-distance value:
    $
    \phi(x_i)=(-1)^{n_i}d(x_i)
    $
\EndFor

\end{algorithmic}
\end{breakablealgorithm}


\begin{breakablealgorithm}
\caption{Identification of zero-level set components}
\label{alg:identify}

\begin{algorithmic}[1]
\Require Signed-distance values $\phi(x_i)$ on the refined Cartesian grid;
active-node mask
\Ensure Embedded marching-squares graph
$G_\phi=(V_\phi,E_\phi)$;
vertex-coordinate map
$X:V_\phi\to\mathbb{R}^2$;
component labels $c(v)$ for $v\in V_\phi$

\State Initialize an empty graph $G_\phi=(V_\phi,E_\phi)$
\State Initialize horizontal and vertical edge-vertex arrays for refined-grid edges

\For{each refined grid cell whose four corner nodes are active}

    \State Read the nodal values
    $
    \phi_{00},\phi_{10},\phi_{01},\phi_{11}
    $
    at the four cell corners

    \State Determine which of the bottom, right, top, and left cell edges
    contain sign changes

    \State Compute the zero-crossing points on all sign-changing edges by
    linear interpolation

    \State Let $m$ denote the number of sign-changing edges

    \If{$m=0$}
        \State \textbf{continue}

    \ElsIf{$m=2$}
        \State Obtain the two corresponding graph vertices from the
        edge-vertex arrays, creating them if necessary and storing their
        coordinates in $X$

        \State Connect the two vertices by an edge in $E_\phi$

    \ElsIf{$m=4$}
        \State Obtain the four corresponding graph vertices from the
        edge-vertex arrays, creating them if necessary and storing their
        coordinates in $X$

        \State Evaluate the asymptotic decider: 
        $
        a
        =
        \phi_{00}\phi_{11}
        -
        \phi_{10}\phi_{01}
        $

        \State Connect the four vertices using the pairing determined by
        the sign of $a$
    \EndIf

\EndFor

\State Compute the connected components of $G_\phi$ using breadth-first search

\State Define
$
\Neul
=
\#\,\mathrm{CC}(G_\phi),
$
where $\mathrm{CC}(G_\phi)$ denotes the collection of connected
components of $G_\phi$

\State Assign a component label
$
c(v)\in\{1,\ldots,\Neul\}
$
to each vertex $v\in V_\phi$

\end{algorithmic}
\end{breakablealgorithm}


\begin{breakablealgorithm}
\caption{Tracing a graph component into an ordered curve}
\label{alg:trace}

\begin{algorithmic}[1]
\Require Embedded graph $G_\phi=(V_\phi,E_\phi)$ with vertex-coordinate map
$X:V_\phi\to\mathbb{R}^2$; component labels $c(v)$; component index $\beta$
\Ensure Ordered polygonal curve
$
\Gammaeul_\beta
=
\{\zeul_{\beta,1},\ldots,\zeul_{\beta,N_\beta}\}
$

\State Choose a starting vertex $v_0\in V_\phi$ with $c(v_0)=\beta$
\State Set $v_{\mathrm{prev}}=0$ and $v_{\mathrm{curr}}=v_0$
\State Initialize $\Gammaeul_\beta$ as an empty ordered list

\Repeat
    \State Append the physical coordinate of $v_{\mathrm{curr}}$ to
    $\Gammaeul_\beta$:
    \[
    \zeul_{\beta,i}
    \coloneqq
    X(v_{\mathrm{curr}})
    \]
    \State Choose $v_{\mathrm{next}}$ as the neighbor of $v_{\mathrm{curr}}$
    that is different from $v_{\mathrm{prev}}$
    \State Set $v_{\mathrm{prev}}=v_{\mathrm{curr}}$ and
    $v_{\mathrm{curr}}=v_{\mathrm{next}}$
\Until{$v_{\mathrm{curr}}=v_0$}

\State Append the first point again to enforce periodic closure

\end{algorithmic}
\end{breakablealgorithm}


\begin{breakablealgorithm}
\caption{Refinement and filtering of extracted curves}
\label{alg:ref-and-filter}

\begin{algorithmic}[1]
\Require Extracted polygonal curves
$\{\Gammaeul_\beta\}_{\beta=1}^{\Neul}$;
filtering thresholds $\Nmin$, $\Lmin$, and $\Amin$
\Ensure Refined and filtered extracted interface family

\For{$\beta=1,\ldots,\Neul$}
    \State Apply the adaptive refinement procedure in \Cref{alg:air}
    to $\Gammaeul_\beta$
\EndFor

\For{$\beta=1,\ldots,\Neul$}

    \State Define
    $
    N_\beta
    \coloneqq
    \#\Gammaeul_\beta
    $

    \State Compute the polygonal arclength
    $
    L_\beta
    \coloneqq
    \sum_{i=1}^{N_\beta}
    |\zeul_{\beta,i+1}-\zeul_{\beta,i}|
    $

    \State Compute the enclosed polygonal area
    $
    A_\beta
    \coloneqq
    \frac12
    \left|
    \sum_{i=1}^{N_\beta}
    \det
    \bigl(
    \zeul_{\beta,i},
    \zeul_{\beta,i+1}
    \bigr)
    \right|
    $

\EndFor

\For{$\beta=1,\ldots,\Neul$}
    \If{$N_\beta<\Nmin$
        \textbf{or}
        $L_\beta<\Lmin$
        \textbf{or}
        $A_\beta<\Amin$}
        \State Discard $\Gammaeul_\beta$
    \EndIf
\EndFor

\State Sort the remaining curves in decreasing order of $N_\beta$

\State Relabel the remaining curves as
$\{\Gammaeul_\beta\}_{\beta=1}^{\Neul}$

\end{algorithmic}
\end{breakablealgorithm}



\begin{breakablealgorithm}
\caption{Greedy event-local surgical reconstruction}
\label{alg:greedy-event-surgery}

\begin{algorithmic}[1]
\Require Event-local surgical pieces
$
\Plgr
=
\{\plgr_i\}_{i=1}^{\Npiece}
$
and
$
\Peul
=
\{\peul_j\}_{j=1}^{\Mpiece}
$;
greedy continuation rule \eqref{eq:greedy-continuation}

\Ensure Closed reconstructed curves
$
\{\Gamma^{\mathrm{reb}}_1,\ldots,\Gamma^{\mathrm{reb}}_{N_{\mathrm{reb}}}\}
$

\State Mark all surgical pieces as unused
\State Initialize the reconstructed curve collection as empty
\State Set $\texttt{require\_single\_cycle}$ to be true for many-to-one merge events

\ForAll{unused Eulerian pieces $\peul_{j_0}$}

    \State Set the starting piece to $\peul_{j_0}$ and orient it forward
    \State Initialize an empty temporary curve $\Gamma^{\mathrm{tmp}}$
    \State Initialize an empty ordered piece sequence $\mathcal S$
    \State Set the current piece to the oriented starting piece

    \Repeat

        \State Append the current oriented piece to $\Gamma^{\mathrm{tmp}}$
        \State Append the current oriented piece to $\mathcal S$
        \State Mark the current piece as used
        \State Select the next oriented piece using \eqref{eq:greedy-continuation}

        \If{the selected piece is the starting piece}

            \If{$\texttt{require\_single\_cycle}$ is true and unused pieces remain}
                \State Reject this premature closure
                \State Recompute the continuation while forbidding closure to the starting piece
            \Else
                \State Declare $\Gamma^{\mathrm{tmp}}$ to be a completed cycle
            \EndIf

        \Else
            \State Set the selected oriented piece to be the current piece
        \EndIf

    \Until{$\Gamma^{\mathrm{tmp}}$ is a completed cycle or no admissible continuation exists}

    \If{$\Gamma^{\mathrm{tmp}}$ is a completed cycle and satisfies the geometric closure criterion}
        \State Store $\Gamma^{\mathrm{tmp}}$ as a reconstructed curve
    \EndIf

\EndFor

\State Apply the smoothing and untangling procedure in
\Cref{alg:smooth-interface} to the reconstructed curves

\end{algorithmic}
\end{breakablealgorithm}


\begin{breakablealgorithm}
\caption{Smoothing of reconstructed interfaces}
\label{alg:smooth-interface}

\begin{algorithmic}[1]
\Require Event-local reconstructed curves
$
\{\Gamma^{\mathrm{reb}}_c\}_{c=1}^{N_{\mathrm{reb}}}
$
with piece-source labels identifying extracted and retained segments
\Ensure Smoothed reconstructed curves
$
\{\Gamma^{\mathrm{reb}}_c\}_{c=1}^{N_{\mathrm{reb}}}
$

\ForAll{reconstructed curves $\Gamma^{\mathrm{reb}}_c$}

    \If{$\Gamma^{\mathrm{reb}}_c$ is a closed extracted loop}

        \State Store the initial enclosed area
        $
        A_0=A(\Gamma^{\mathrm{reb}}_c)
        $
        \State Replace $\Gamma^{\mathrm{reb}}_c$ by the minimizer of
        \eqref{eq:closed-loop-smoothing}
        \State Apply area correction to approximately recover $A_0$

    \Else

        \ForAll{maximal extracted subchains of $\Gamma^{\mathrm{reb}}_c$}
            \State Replace the extracted subchain by the minimizer of
            \eqref{smooth-inter-piece}
        \EndFor

    \EndIf

\EndFor

\end{algorithmic}
\end{breakablealgorithm}


\bibliographystyle{abbrvnatmod}
\bibliography{references}

\end{document}